\documentclass[a4paper,12pt,twoside,openright]{book}

\usepackage[all]{xy}
\usepackage[english, italian]{babel}
\usepackage[utf8]{inputenc}
\usepackage[T1]{fontenc}
\usepackage{hyperref}
\usepackage{cite}
\usepackage{ae,aecompl,aeguill} 
\usepackage{graphicx,color,setspace}
\usepackage{amsmath}
\usepackage{amssymb}
\usepackage{amsthm}
\usepackage{mathrsfs}
\usepackage{bbm}
\usepackage{amsfonts}  
\usepackage{makeidx}   
\usepackage{enumitem}
\usepackage{tikz}
\usepackage[ISBN=978-88-95767-78-9]{ean13isbn}
\usepackage{calligra}

\usepackage{sectsty} 
\chapterfont{\thispagestyle{empty}}
\paragraphfont{\sl\bf}

\newcommand{\clearemptydoublepage}{\newpage{\pagestyle{empty}\cleardoublepage}}
\newcommand{\paginavuota}{\newpage\thispagestyle{empty}{\ }} 

\usepackage[Lenny]{fncychap} 
\usepackage{fancyhdr}
\ChNameVar{\fontsize{16}{14}\usefont{OT1}{ptm}{m}{n}\selectfont\bf}
\ChTitleVar{\Huge\bfseries\rm\bf}
\definecolor{gray50}{gray}{.5} 
\ChRuleWidth{1pt\color{black}}
\ChNumVar{\color{gray50}\fontsize{70}{60}\usefont{OT1}{ptm}{m}{n}\selectfont\bf}

\hypersetup{colorlinks=true, linkcolor=blue, citecolor=blue, urlcolor=blue}

\newtheorem{theorem}{Theorem}[section]
\newtheorem{lemma}[theorem]{Lemma}
\newtheorem{proposition}[theorem]{Proposition}
\newtheorem{corollary}[theorem]{Corollary}

\theoremstyle{definition}
\newtheorem{definition}[theorem]{Definition}
\newtheorem{example}[theorem]{Example}
\newtheorem{remark}[theorem]{Remark}

\numberwithin{equation}{section}

\def\bbZ{\mathbb{Z}}
\def\bbR{\mathbb{R}}
\def\bbC{\mathbb{C}}

\def\bbT{\mathbb{T}}
\def\lieg{\mathfrak{g}}
\def\bw{\begingroup\textstyle\bigwedge\endgroup}

\def\lra{\longrightarrow}
\def\la{\langle}
\def\ra{\rangle}
\def\ul{\underline}
\def\tint{\begingroup\textstyle\int\endgroup}

\def\mfo{\mathfrak{o}}
\def\mft{\mathfrak{t}}
\def\vol{\mathrm{vol}}
\def\id{\mathrm{id}}
\def\rank{\mathrm{rank}}
\def\Ad{\mathrm{Ad}}
\def\ad{\mathrm{ad}}
\def\op{\mathrm{op}}
\def\myexp{\mathrm{e}}
\def\ttone{\mathtt{1}}
\def\bfone{\mathbf{1}}
\def\bbone{\mathbbm{1}}
\def\triv{\mathsf{Triv}}
\def\exptriv{\triv^\bbZ}
\def\weyl{\mathsf{W}}
\def\supp{\mathrm{supp}}
\def\ker{\mathrm{ker}}
\def\im{\mathrm{im}}
\def\dd{\mathrm{d}}
\def\de{\delta}
\def\hom{\mathrm{Hom}}
\def\endo{\mathrm{End}}
\def\pr{\mathrm{pr}}
\def\gauge{\mathrm{Gau}}
\def\eqv{\mathrm{eqv}}
\def\hor{\mathrm{h}}
\def\base{\mathrm{base}}

\def\balg{\Delta^\mathrm{B}}
\def\states{\mathscr{S}}

\def\conn{\mathcal{C}}
\def\ev{\mathcal{O}}
\def\expev{\mathcal{W}}

\def\curv{\mathcal{F}}
\def\MW{\mathcal{M}}
\def\curvstar{\mathscr{F}^\ast}
\def\dense{\mathscr{D}}
\def\hilb{\mathscr{H}}
\def\bounded{\mathscr{L}}

\def\obs{\mathcal{E}}
\def\kin{\obs^{\mathrm{kin}}}
\def\inv{\obs^{\mathrm{inv}}}
\def\van{\obs^{\mathrm{van}}}
\def\expobs{\mathcal{D}}
\def\expkin{\expobs^{\mathrm{kin}}}
\def\expinv{\expobs^{\mathrm{inv}}}
\def\expvan{\expobs^{\mathrm{van}}}
\def\rad{\mathcal{N}}

\def\exprad{\mathcal{R}}
\def\expcnt{\mathcal{Z}}

\def\GHyp{\mathsf{GHyp}}
\def\GHypF{\mathsf{GHypF}}
\def\Vec{\mathsf{Vec}}
\def\Aff{\mathsf{Aff}}
\def\VBun{\mathsf{VBun}}
\def\PrBun{\mathsf{PrBun}}
\def\PSymV{\mathsf{PSymV}}
\def\PSymA{\mathsf{PSymA}}
\def\Alg{{}^\ast\mathsf{Alg}}
\def\BAlg{\mathsf{B}^\ast\mathsf{Alg}}
\def\CAlg{\mathsf{C}^\ast\mathsf{Alg}}

\def\PSV{\mathfrak{PSV}}
\def\PSA{\mathfrak{PSA}}
\def\mfker{\mathfrak{Ker}}
\def\CCR{\mathfrak{CCR}}
\def\QFT{\mathfrak{A}}

\def\sol{\mathcal{S}}
\newcommand{\solsc}[1]{\sol_{\mathrm{sc}\,#1}}
\def\gau{\mathcal{G}}
\newcommand{\gausc}[1]{\gau_{\mathrm{sc}\,#1}}

\def\herm{\mathcal{H}}
\def\hermt{\tilde{\herm}}
\def\cone{\mathcal{K}}

\def\sect{\Gamma}
\def\sc{\sect_\mathrm{c}}

\def\stc{\sect_\mathrm{tc}}
\def\c{\mathrm{C}^\infty}
\def\cc{\c_\mathrm{c}}
\def\csc{\c_\mathrm{sc}}
\def\ctc{\c_\mathrm{tc}}
\def\f{\Omega}
\def\fdd{\f_\dd}
\def\fde{\f_\de}
\def\fc{\f_\mathrm{c}}
\def\fcdd{\f_{\mathrm{c}\,\dd}}
\def\fcde{\f_{\mathrm{c}\,\de}}
\def\ffsc{\f_\mathrm{fsc}}
\def\fpsc{\f_\mathrm{psc}}
\def\fsc{\f_\mathrm{sc}}
\def\fscdd{\f_{\mathrm{sc}\,\dd}}
\def\fscde{\f_{\mathrm{sc}\,\de}}
\def\fpc{\f_\mathrm{pc}}
\def\ffc{\f_\mathrm{fc}}
\def\ftc{\f_\mathrm{tc}}
\def\ftcdd{\f_{\mathrm{tc}\,\dd}}
\def\ftcde{\f_{\mathrm{tc}\,\de}}
\def\coho{\mathrm{H}}
\def\cech{\check{\mathrm{H}}}
\def\hdd{\mathrm{H}_\dd}
\def\hde{\mathrm{H}_\de}
\def\hcdd{\mathrm{H}_{\mathrm{c}\,\dd}}
\def\hcde{\mathrm{H}_{\mathrm{c}\,\de}}
\def\hscdd{\mathrm{H}_{\mathrm{sc}\,\dd}}
\def\hscde{\mathrm{H}_{\mathrm{sc}\,\de}}
\def\htcdd{\mathrm{H}_{\mathrm{tc}\,\dd}}
\def\htcde{\mathrm{H}_{\mathrm{tc}\,\de}}
\def\connf{\f^1_\conn}

\newcommand{\quotes}[1]{``#1''}
\newcommand{\backtick}[1]{\`#1}
\newcommand{\norm}[1]{\Vert #1\Vert}
\newcommand{\Norm}[1]{\left\Vert #1\right\Vert}
\newcommand{\normb}[1]{\norm{#1}^\bullet}
\newcommand{\Normb}[1]{\Norm{#1}^\bullet}
\newcommand{\GNS}[1]{(\pi_{#1},\mathscr{H}_{#1},\Omega_{#1})}

\def\sk{\vspace{1em}}

\makeindex

\begin{document}

\frontmatter

\selectlanguage{english}


\thispagestyle{empty}

\begin{center}

%

{\large\textsc{Universit\backtick{a} degli Studi di Pavia}}

\vspace{0.1cm}

{\large\textsc{Dottorato di Ricerca in Fisica -- XXVII Ciclo}}

\rule[0.5ex]{\columnwidth}{1pt}

\vspace{2.5cm}

{\LARGE\textbf{Locality in Abelian gauge theories}}

\vspace{0.3cm}

{\LARGE\textbf{over globally hyperbolic spacetimes}}

\vspace{2.2cm}

{\Large{Marco Benini}}

\vspace{1.5cm}

\begin{tikzpicture}

\shade[left color=teal, right color=teal!40] 
	(0,0) to[out=10, in=130] (4,-2) -- (8,1) to[out=120, in=40] (3,3) -- cycle;

\fill[white] (4,1.3) circle [x radius=1, y radius=0.5, rotate=30];

\draw[very thick] (4,1.3) circle [x radius=1.5, y radius=1, rotate=30];

\shade[left color=teal, right color=teal!40] 
	(6,-3) to[out=10, in=130] (10,-5) -- (14,-2) to[out=120, in=40] (9,0) -- cycle;

\draw[dotted] (10,-1.7) circle [x radius=1, y radius=0.5, rotate=30];

\draw[dashed] (10,-1.7) circle [x radius=1.5, y radius=1, rotate=30];

\draw[-latex, very thick] (6,2.5) to[out=20, in=100] (10,0);

\path	(12.4,3.55)	node(x0)	{$0$}
		(12.4,2.4)	node(x1)	{$\mathrm{H}^1(M,\mathrm{C}_M^\infty(\cdot,U(1)))$}
		(12.4,1.3)	node(x2)	{$\mathrm{H}^2(M,\mathbb{Z})$}
		(12.4,0.25)	node(x3)	{$0$};
\draw[->]	(x0) -- (x1);
\draw[->]	(x1) -- (x2);
\draw[->]	(x2) -- (x3);

\path	(1.4,-1.35)	node(y0)	{$0$}
		(1.4,-2.5)	node(y1)	{$\mathrm{C}^\infty(M,\mathbb{Z})$}
		(1.4,-3.6)	node(y2)	{$\mathrm{C}^\infty(M,\mathbb{R})$}
		(1.4,-4.7)	node(y3)	{$\mathrm{C}^\infty(M,U(1))$}
		(4.3,-4.7)	node(y4)	{$\mathrm{H}^1(M,\mathbb{Z})$}
		(6.1,-4.65)	node(y5)	{$0$}
		(6.1,-4.7)	node(y6)	{$\phantom{0}$};
\draw[->]	(y0) -- (y1);
\draw[->]	(y1) -- (y2);
\draw[->]	(y2) -- (y3);
\draw[->]	(y3) -- (y4);
\draw[->]	(y4) -- (y6);

\end{tikzpicture}

\vfill

{\Large{Tesi per il conseguimento del titolo}}

\end{center}

\clearemptydoublepage


\thispagestyle{empty}

\noindent
\begin{minipage}{2cm}
\includegraphics[height=3.5cm]{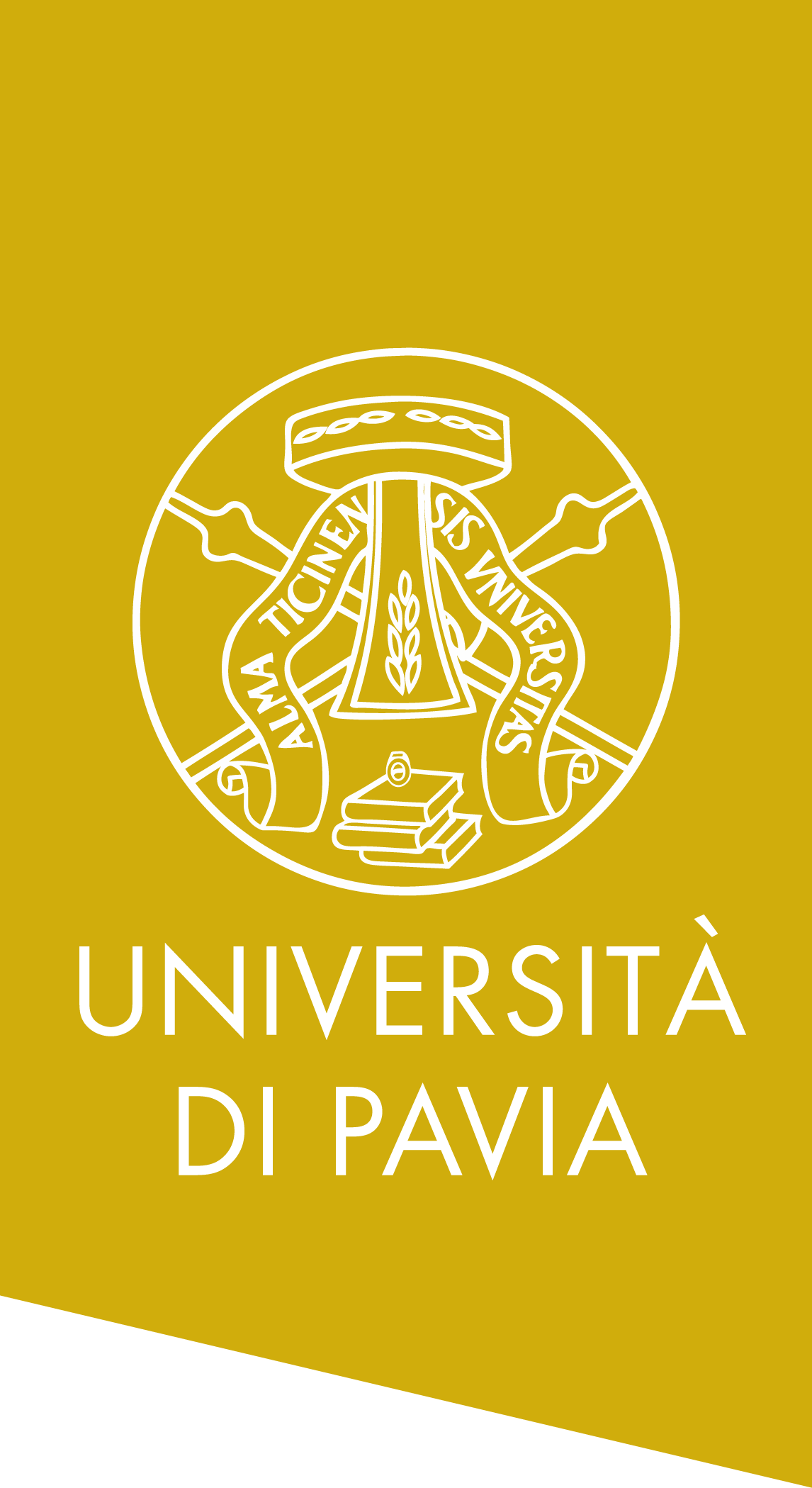}
\end{minipage}
\begin{minipage}{2.35cm}
{\large Universit\backtick{a} degli Studi di Pavia}
\end{minipage}
\hfill
\begin{minipage}{2.9cm}
{\flushright\large Dipartimento di Fisica}
\end{minipage}
\begin{minipage}{3.5cm}
\includegraphics[height=3.5cm]{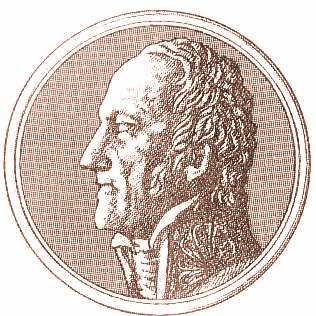}
\end{minipage}

\begin{center}

%
%
%

{\large\textsc{Dottorato di Ricerca in Fisica -- XXVII Ciclo}}

\vspace{2.7cm}

{\LARGE\textbf{Locality in Abelian gauge theories}}

\vspace{0.3cm}

{\LARGE\textbf{over globally hyperbolic spacetimes}}

\vspace{2.7cm}

{\Large{Marco Benini}}

\vspace{2.7cm}

{\large{Submitted to the Graduate School in Physics}}

\vspace{0.1cm}

{\large{in partial fulfillment of the requirements for the degree of}}

%

\vspace{0.3cm}

{\large\textsc{Dottore di Ricerca in Fisica}}

\vspace{0.1cm}

{\large\textsc{Doctor of Philosophy in Physics}}

\vspace{0.3cm}

{\large{at the University of Pavia}}

%

\end{center}

\vfill

\noindent
\begin{minipage}{0.25\textwidth}
Adviser:

Claudio Dappiaggi
\end{minipage}
\hfill
\begin{minipage}{0.25\textwidth}
Co-adviser:

Alexander Schenkel
\end{minipage}

%

\paginavuota

\vfill

\noindent\textbf{Cover:} The embedding of a surface with a hole into a contractible surface is depicted 
together with the exact sequences which are relevant for the description of the gauge theory of electromagnetism. 

\vspace{1cm}

\noindent Locality in Abelian gauge theories over globally hyperbolic spacetimes

\noindent\textit{Marco Benini}

\noindent PhD thesis -- University of Pavia

\noindent Pavia, Italy, November 2014.

\noindent\ISBN

\paginavuota


\vfill
\hfill{\Large\calligra{To Giulia}}

\paginavuota


\title{\LARGE\textbf{Locality in Abelian gauge theories\\over globally hyperbolic spacetimes}}
\author{\Large{\textit{Marco Benini}}}
\date{\Large{2014}}

\maketitle

\setcounter{page}{8}

\paginavuota


\phantomsection
\addcontentsline{toc}{chapter}{Abstract}
\chapter*{Abstract}
Being responsible for interactions between matter fields, 
gauge field theories play a major role in our understanding of nature. 
The most prominent demonstration of this fact is provided by the gauge bosons of the Standard Model of Particle Physics, 
among which the photon is the carrier of the electromagnetic interaction. 
Notably, electromagnetism is described by an Abelian gauge theory. 
Roughly, this means that photons do not exhibit self-interactions, 
the reason why electromagnetism is very well understood also mathematically. 

The aim of this thesis is to investigate the locality properties of Abelian gauge field theories 
in the mathematically rigorous framework provided by the principle of {\em general local covariance}, 
consisting in a set of axioms to describe field theories on curved spacetimes (both classically and at the quantum level). 
Roughly speaking, to each spacetime, one is supposed to assign (in a mathematically coherent way) 
a family of observables for the field theory of interest. 
We will analyze two examples, {\em Maxwell $k$-forms}, a generalization of the vector potential of electromagnetism, 
and the {\em $U(1)$ Yang-Mills model}, which describes photons in the Standard Model. 
As backgrounds, we will consider globally hyperbolic spacetimes, 
the most general class of spacetimes where the dynamics of wave-type field equations can be understood globally. 

Our attention will be focused in particular on the {\em locality axiom} of general local covariance, 
which states that an embedding between spacetimes (compatible with the causal structures) 
should induce an inclusion at the level of observables, 
namely the observables on the embedded spacetime should form a subset of those on the target spacetime. 
Both at the classical and at the quantum level, it turns out that the models we consider violate this axiom 
depending on certain global features of the background spacetime. 
For Maxwell $k$-forms, we prove that there is no coherent way to recover the locality axiom. 
For the $U(1)$ Yang-Mills model we adopt two different approaches: 
in the first case locality can be recovered coherently, but the class of observables we introduce fails 
to detect field configurations corresponding to the Aharonov-Bohm effect; 
in the second case observables are defined mimicking Wilson loops 
so that also Aharonov-Bohm configurations are captured, yet locality cannot be recovered coherently.

\paginavuota


\tableofcontents


\mainmatter
\pagenumbering{arabic}
\setcounter{page}{1}

\chapter{Introduction}

Quantum field theory on curved spacetimes is a well-established and very promising research field 
in mathematical physics. Algebraic methods proved very successful in this context. 
The original idea can be traced back to the seminal work of Haag and Kastler\ \cite{HK64}, 
who first noted that quantum field theory could be better understood 
as a net of algebras of observables on Minkowski spacetime. 
This consists of specifying, for each spacetime region, an algebra describing observables which are localized there. 
The net naturally encodes both geometric and dynamical features of quantum field theory, 
such as locality, covariance with respect to the symmetries of Minkowski spacetime and causality. 
Notably, these properties arise already at the algebraic level, 
without the need of any choice of representation on a Hilbert space. 
This approach was later extended to curved spacetimes by Dimock\ \cite{Dim80}. 
A recent breakthrough in the context of quantum field theory on curved backgrounds 
is the axiomatic formulation proposed by Brunetti, Fredenhagen and Verch\ \cite{BFV03}, 
known as {\em generally covariant locality principle}.\index{general local covariance} 
Roughly speaking, the central feature of this approach consists in taking into account quantum field theories 
on all globally hyperbolic spacetimes at the same time. 
Specifically, to each globally hyperbolic spacetime, one assigns an algebra of observables. 
Again, covariance, locality and causality are encoded in a natural way; 
moreover, there is a further requirement, the so-called time-slice axiom, 
which is reminiscent of the global well-posedness of Cauchy problems for wave equations 
on globally hyperbolic spacetimes. Basically, this means that 
the full information about the field is already contained in the initial data. 
One might summarize the axioms of general local covariance as follows: 
\begin{description}
\item[Covariance]\index{covariance} Whenever a spacetime is embedded (preserving its causal structure) 
in a larger one, there should be a map (compatible with the algebraic structures) 
which assigns an observable on the target spacetime to each observable on the source; 
\item[Locality]\index{locality} It is required that the map described above is an injection, 
namely, on the source of the spacetime embedding, one should have a subalgebra 
of the algebra of observables associated to the target spacetime; 
\item[Causality]\index{causality} Observables localized in regions 
which cannot be joined by a causal curve should commute; 
\item[Time-slice axiom]\index{time-slice axiom} The algebra of observables localized in a suitable neighborhood 
of a Cauchy hypersurface is isomorphic to the algebra describing observables on the full spacetime. 
\end{description}
\index{locally covariant field theory}Quantum field theories 
which realize the axioms of general local covariance are called {\em locally covariant quantum field theories}. 
Similar axioms can be stated for classical field theories as well. 
Notably, general local covariance provided a successful framework to extend structural results 
about quantum field theory on Minkowski spacetime to much more general curved spacetimes, 
for example see\ \cite{Ver01} for the spin-statistics theorem on curved spacetimes; 
moreover, this proved to be a solid background to develop new ideas 
(an example is the concept of dynamical locality\ \cite{FV12}), 
eventually stimulating developments in classical field theory too, {\em cfr.}\ \cite{BFR12, FR12, Kha14b}. 
For this reason, much effort has been spent in the last decade 
to realize the axioms of general local covariance in physical models. 
First of all, free field theories were taken into account, such as the scalar field \cite{BFV03}, 
the Dirac field \cite{DHP09, San10} and the Proca field \cite{Dap11} (for a review see\ also\ \cite{BDH13}). 
Furthermore, a systematic approach for linear field equations was developed in\ \cite{BGP07, BG12a, BG12b} 
and later it was extended in\ \cite{BDS14b} to include affine field theories too. 
The successful description of free field theories paved the way for further developments in treating 
interacting field theories in the framework of perturbative algebraic quantum field theory \cite{BDF09, HW10}, 
eventually leading to a mathematically rigorous notion of local Wick polynomials \cite{BFK96, HW01}
and of the operator product expansion \cite{Hol07}, 
as well as to new insights on the renormalization of quantum field theories on curved backgrounds \cite{BF00, HW03}. 
Linear gauge field theories were considered too, see for example\ \cite{Dap11, DL12}, 
for the first presentations of the electromagnetic field in the context of general local covariance, 
and\ \cite{FH13} for linearized gravity. For a more systematic approach to linear gauge fields, refer to\ \cite{HS13}. 
Note that also interacting gauge field theories have been studied 
by means of perturbative techniques in\ \cite{Hol08, FR13, BFR13}. 

It was first observed in\ \cite{DL12} in the attempt to quantize Maxwell equations 
in the framework of general local covariance that the Faraday tensor does not satisfy the locality axiom, 
such behavior being related to the second de Rham cohomology group of the background spacetimes. 
This observation motivated further investigations on the interplay between general local covariance 
and electromagnetism under different perspectives. 
The quantization of the vector potential of electromagnetism was shown to satisfy locality in\ \cite{Dap11, DS13}, 
but restricted to those globally hyperbolic spacetimes for which certain de Rham cohomology groups are trivial. 
Explicit counterexamples to locality of the vector potential (as well as of its higher analogues) 
were exhibited in\ \cite{SDH14} and physically interpreted as being caused by Gauss' law. 
Therefore de Rham cohomology classes representing the fluxes of an electric field 
are recognized as causing the failure of the locality axiom in Maxwell field theory. 
Similar issues appear when describing electromagnetism as a Yang-Mills field theory as in\ \cite{BDS14a, BDHS14}. 

So far, we did not consider states for the algebras of observables discussed above. 
Even if we will not deal with this topic, states play a major role when one wants 
to extract physical information from a quantum field theory since they provide expectation values. 
Let us mention that different approaches are available to construct states for the algebra of observables 
of electromagnetism over a curved spacetime\ \cite{FP03, DS13, FS13}. 
More recently, techniques have been developed 
also to construct states for more general linear gauge field theories including the Maxwell case\ \cite{GW14, WZ14}. 

\subsubsection*{Interplay between gauge symmetry, locality and spacetime topology}

The main goal of this thesis is twofold. On the one hand, we want to better understand 
to what extent Abelian gauge field theories violate the locality axiom of general local covariance. 
On the other hand, we try to highlight certain global topological aspects 
which are accessible also in terms of local observables. These two aspects turn out to be tightly related, 
the link between the two being provided by the interplay between dynamics and gauge symmetry: 
Since field configurations only matter up to gauge transformations, 
one is induced to restrict the class of functionals to be used to test field configurations to gauge invariant ones. 
Doing so, one can test gauge equivalence classes of field configurations, rather then their specific representatives. 
Furthermore, dynamics is such that on-shell field configurations have certain associated quantities 
which do not depend on the choice of representative in a gauge equivalence class 
and which are represented by suitable cohomology classes. 
Despite the restrictions imposed on functionals by gauge invariance, 
this cohomological information is accessible by testing observables on field configurations. 
This is reminiscent of the fact that the quantities encoding the topological information 
only depend on the gauge equivalence class of field configurations rather than on a specific representative. 
Therefore, gauge symmetry and dynamics together introduce global information in field theory 
and this feature is captured by means of gauge invariant functionals, which are used to define observables. 
This kind of topological information does not necessarily comply with the locality axiom of general local covariance. 
As a simple (non-dynamical) example, 
consider the inclusion of the pointed plane $\bbR^2\setminus\{0\}$ into the plane $\bbR^2$. 
The map induced at the level of cohomology groups by the pull-back under the inclusion cannot be surjective 
since the first de Rham cohomology group of the plane is trivial, 
while this is not the case for the pointed plane (which is diffeomorphic to a cylinder). 
Roughly, this means that there are field configurations on the pointed plane 
carrying non-trivial topological information which have no counterpart on the full plane. 
Dually, on the smaller spacetime one has observables sensitive to some field configurations reminiscent of 
specific topological aspects. These observables, after being pushed forward to the larger spacetime, 
cannot measure anything. This is simply because there cannot be any field configuration 
carrying the kind of topological information to which the considered observables are sensitive. 
Therefore observables of this type become trivial under suitable spacetime embeddings, 
thus leading to violations of locality. 

One of the main guiding principles of the investigations presented here stems from the idea that 
observables should be able to extract any possible information 
about dynamically allowed field configurations (up to gauge equivalence). 
To achieve this result, already at the classical level, 
we want to specify observables via functionals on field configurations in such a way that 
these are sufficiently many to distinguish all different (gauge classes of) field configurations. 
If this result is achieved, observables can really capture the full information 
encoded in the dynamics, including topological aspects, if any. 

So far, only the pairing between field configurations and observables has been taken into account. 
However, from the point of view of classical field theory, 
a suitable presymplectic structure is supposed to be defined on the space of observables. 
As usual, the correct presymplectic structure can be read off from the Lagrangian. 
It turns out that, depending on the spacetime topology, the presymplectic form has non-trivial degeneracies, 
which provide potential sources for the failure of locality. 
Consider once again the situation of a spacetime embedded into a larger one. 
Assuming that an observable on the embedded spacetime lies in the null space of the presymplectic structure 
for the target spacetime implies that this observable is degenerate for the presymplectic structure of the source too. 
However, the converse might fail since on the larger spacetime there is in general room for more observables, 
which might reduce the degeneracy of the presymplectic structure. 
Therefore, it might happen that observables, which are degenerate 
with respect to the presymplectic form on a spacetime, are mapped 
along a push-forward outside the null space of the presymplectic structure of the target spacetime. 
This way the dynamics eventually prevents us from recovering locality by quotients (in a suitable sense). 
In fact, one might consider spacetimes which are embedded in two different targets. 
On the source one has a degenerate observable which becomes trivial when mapped to the first target spacetime. 
However, the same observable, when mapped to the other target spacetime, is no longer degenerate 
with respect to the relevant presymplectic structure.
Therefore, to recover injectivity along the first embedding, one should take a quotient on the source 
to \quotes{remove} the observable which becomes trivial. Yet, the second embedding is such that 
this quotient is not compatible with the presymplectic structure on its target, 
hence such a quotient cannot be performed coherently on all spacetimes. 
This is the extent to which gauge theories violate the locality axiom of general local covariance. 
Dynamics and gauge symmetry together not only produce the lack of locality, 
but they even prevent its recovery by means of quotients (in a coherent sense). 

\subsubsection*{Results achieved}

In this thesis the features of Abelian gauge theories presented above are analyzed in detail for two different models. 
First, Maxwell $k$-forms are considered, which describe analogues of the vector potential 
(corresponding to degree $k=1$) in higher degree $k$. 
In the first place, we obtain a suitable space of observables for this model
and a non-degenerate pairing with gauge equivalence classes of dynamically allowed field configurations. 
On the space of observables, we introduce the presymplectic structure 
induced by the Lagrangian density of the model and we study its degeneracies exhibiting examples. 
The assignment of the corresponding space of observables to each globally hyperbolic spacetime 
gives rise to a covariant classical field theory, yet locality fails. 
In fact, we exploit the degeneracies of the presymplectic form 
to determine explicit violations of the locality axiom of general local covariance 
and we prove a no-go theorem for the recovery of locality via quotients. 
We conclude our study of Maxwell $k$-forms proving that at least isotony in the sense 
of Haag and Kastler\ \cite{HK64} can be recovered by a suitable quotient after fixing a target spacetime. 
All these facts find counterparts after the quantization procedure, 
which is performed via Weyl canonical commutation relations for presymplectic vector spaces 
producing a covariant quantum field theory which violates locality. 

Secondly, we describe electromagnetism as a pure Yang-Mills field theory with $U(1)$ as its structure group. 
The degrees of freedom of this model are represented by principal bundle connections 
and its gauge symmetry arises from the geometry by means of principal bundle automorphisms 
covering the identity on the base manifold of the principal bundle. 
We take into account the affine structure of the space of connections 
as well as the topologically non-trivial gauge transformations (also known as {\em finite} gauge transformations). 
In the first place, we introduce observables by means of gauge invariant affine functionals. 
It turns out that the vector space of observables obtained along these lines fails in distinguishing 
certain gauge classes of on-shell field configurations: Because of the severe constraints imposed by gauge invariance,  
affine observables are not sensitive to flat connections. Both mathematically and physically, this seems unsatisfactory. 
On the one hand, we have solutions to the field equations which are not detectable via the functionals we consider; 
on the other hand, flat connections reproduce the Aharonov-Bohm effect, 
therefore affine observables are not able to detect a well-known physical phenomenon. 
Analyzing this model more closely, one realizes that again the presymplectic structure has non-trivial degeneracies 
depending on the spacetime topology. 
This gives us the chance to look for counterexamples to the locality axiom, which are explicitly shown. 
Therefore, via the assignment of the corresponding presymplectic space of affine observables 
to each principal $U(1)$-bundle over a globally hyperbolic spacetime, 
we obtain a covariant classical field theory which violates locality. 
However, the fact that some observables are missing allows us to recover locality by means of a suitable quotient, 
thus leading to a locally covariant classical field theory. 
The quotient which recovers locality in this context can be interpreted by saying that 
connections describing pure electromagnetism should carry no electric flux. 
In the end, quantization is performed adopting canonical commutation relations \backtick{a} la Weyl 
associated to presymplectic vector spaces. 
Doing so, the results found for the classical field theory are promoted to the quantum case. 

To cure the impossibility to detect certain gauge equivalence classes of connections, 
we analyze the $U(1)$ Yang-Mills model once again, introducing observables in a different manner. 
Motivated by the analogy with Wilson loops, 
we replace the affine functionals described above by their complex exponentials 
(we adopt the term {\em affine characters} for functionals of this type). 
This choice weakens the constraints imposed by gauge invariance. 
The result is a richer space of observables which succeeds in distinguishing on-shell connections up to gauge, 
including the flat ones, which describe the Aharonov-Bohm effect. 
The space of observables defined via affine characters is only an Abelian group, 
which can be endowed with a presymplectic structure very similar to the one in the previous case. 
As before, degeneracies are present depending on the spacetime topology, together with counterexamples to locality. 
This shows that, assigning to each principal $U(1)$-bundle over a globally hyperbolic spacetime 
the corresponding presymplectic vector space gives rise to a covariant classical field theory which violates locality. 
Contrary to the case where affine functionals are considered, 
we now have a space of observables (specified via affine characters) 
rich enough to separate gauge equivalence classes of on-shell field configurations. 
This fact reduces the degeneracies of the presymplectic structure in comparison with the previous situation, 
eventually leading to the impossibility to recover locality by quotients. 
In analogy with the case of Maxwell $k$-forms, one can still fix a target spacetime 
(with a principal $U(1)$-bundle on top) and, performing a suitable quotient, one can recover Haag-Kastler isotony. 
In conclusion, a covariant quantum field theory is obtained after quantization via Weyl relations 
for presymplectic Abelian groups. All the features of the covariant classical field theory are shown 
to have counterparts at the quantum level. In particular, the covariant quantum field theory 
violates the locality axiom of general local covariance. 

\subsubsection*{Outline}

In the following, we summarize the topics investigated in this thesis. 
In\ Chapter\ \ref{chPreliminaries} we set the bases for the subsequent developments. 
In particular, Section\ \ref{secGlobHyp} introduces globally hyperbolic spacetimes, 
which provide the background where the field dynamics takes place. 
We proceed with\ Section\ \ref{secForms}, where our notation for differential forms is established, 
de Rham cohomology with unrestricted and compact support is briefly recalled, together with its Poincar\'e duality, 
and analogues of de Rham cohomology with causally restricted support systems are developed, 
together with the corresponding analogue of Poincar\'e duality. 
To conclude this section, we briefly recall the dynamical properties of the Hodge-d'Alembert operator 
for differential forms over globally hyperbolic spacetimes. 
Section\ \ref{secAffine} is devoted to prepare some material to deal with affine structures, 
which naturally appear when dealing with connections on principal bundles. 
In particular, we present a notion of affine differential operator and we study its dual. 
Principal bundles, together with the corresponding notions of connections and of gauge transformations, 
are presented in\ Section\ \ref{secPrBun}. To conclude, 
Section\ \ref{secQuantization} deals with the quantization of presymplectic Abelian groups 
adopting Weyl canonical commutation relations. 

Maxwell $k$-forms are introduced and analyzed in\ Chapter\ \ref{chMaxwell}. As a starting point, 
Section\ \ref{secGaugeDynForms} deals with their dynamics and gauge symmetry. 
In\ Section\ \ref{secClassicalFTForms}, we develop the covariant classical field theory of Maxwell $k$-forms 
on globally hyperbolic spacetimes. All axioms of general local covariance are shown to hold, except locality. 
Explicit counterexamples to the locality axioms are exhibited, 
the failure of locality is attributed to the non-injectivity of the push-forward 
on de Rham cohomology with compact support in degree $k+1$ for certain spacetime embeddings, 
a no-go theorem for the recovery of locality by quotients (in a suitable sense) is proved 
and the recovery of isotony \backtick{a} la Haag-Kastler by means of a quotient is presented. 
To conclude Chapter\ \ref{chMaxwell}, in\ Section\ \ref{secQFTForms} we perform the quantization 
of the covariant classical field theory developed in\ Section\ \ref{secClassicalFTForms}. 
We obtain a covariant quantum field theory which violates locality to the extent that 
no quotient can be performed to recover this property. 
Yet, a quantum field theory in the sense of Haag and Kastler can be obtained 
performing quantization after the recovery of isotony at the classical level. 

In\ Chapter\ \ref{chYangMills} we investigate the Yang-Mills model with structure group $U(1)$ 
over globally hyperbolic spacetimes. Section\ \ref{secGaugeDynamicsYM} specifies 
the dynamics and gauge symmetry for Yang-Mills connections on principal $U(1)$-bundles. 
In particular, we exhibit the naturality of the field equations and we prove the existence of solutions. 
Furthermore, we analyze in detail the action of gauge transformations on connections 
using cohomological techniques and paying the due attention to topologically non-trivial gauge transformations. 
We proceed with\ Section\ \ref{secAffObsYM} introducing observables by means of gauge invariant affine functionals,  
which, however, fail in detecting flat connections (corresponding to the Aharonov-Bohm effect). 
Despite this shortcoming, we endow the space of affine observables with the presymplectic structure 
induced from the Yang-Mills Lagrangian and we prove that 
the assignment of the presymplectic vector space of affine observables to each principal $U(1)$-bundle 
over a globally hyperbolic spacetime gives rise to a covariant classical field theory, 
namely all axioms of general local covariance hold, up to locality. As a matter of fact, 
violations of locality appear explicitly in relation to observables measuring the electric flux (Gauss' law). 
However, a quotient by these observables enables us to recover locality, 
thus leading to a locally covariant classical field theory. 
This quotient is interpreted as the restriction to those connections which carry no electric flux. 
The present section is completed with quantization by means of Weyl relations, 
producing a covariant quantum field theory violating locality when the full theory is considered and a locally 
covariant quantum field theory in case the electric flux observables are removed via the quotient mentioned above. 
To overcome the shortcomings of affine functionals, namely the impossibility to detect flat connections, 
in\ Section\ \ref{secYangMillsChar} we consider their exponentiated version (affine characters). 
The constraints imposed by gauge invariance are weakened by the complex exponential 
in such a way that observables defined via gauge invariant affine characters succeed 
in distinguishing connections up to gauge (including the flat ones, which reproduce the Aharonov-Bohm effect). 
In fact, we can interpret Wilson loops for $U(1)$-connections as the distributional counterparts 
of our more regular affine characters. 
Notice that the space of observables in this approach is only an Abelian group, 
which is endowed with a suitable presymplectic structure in analogy with the previous section. 
Once more, the result is a covariant classical field theory which violates locality 
because of the observables measuring the electric flux. 
However, in this approach we prove a no-go theorem for the recovery of locality by means of a suitable quotient. 
Comparing this result with the previous case in which affine functionals are considered in place of affine characters, 
one can interpret the no-go theorem as follows: The availability of observables detecting the Aharonov-Bohm effect 
prevents us from performing the quotient by electric flux observables, which are the sources of the lack of locality. 
Fixing a target principal $U(1)$-bundle over a globally hyperbolic spacetime, 
we can perform a quotient which recovers isotony in the sense of Haag and Kastler. 
To complete our analysis, we quantize the covariant classical field theory 
using Weyl relations for presymplectic Abelian groups. 
The result is a covariant quantum field theory violating the locality axiom. 
Yet, quantization of the model after the recovery of isotony 
produces a quantum field theory in the Haag-Kastler sense.

\chapter{Preliminaries}\label{chPreliminaries}
The aim of this chapter is to collect the most important mathematical tools upon which the whole thesis relies. 
Most of the material will be presented with very few details and often avoiding proofs. 
However, we will always provide the relevant references to the existing literature. 

The following definition specifies which category of manifolds we will be interested in throughout this thesis. 
For a detailed discussion on this topic we refer the reader to the literature, for example\ \cite[Section 1.1]{Jos11}.

\begin{definition}\label{defManifold}\index{manifold}\index{map}
With the term {\em manifold} we refer to an orientable, boundaryless, smooth, 
second countable, Hausdorff manifold of dimension $m\geq2$. 
Furthermore, whenever a map between manifolds is considered, it is implicitly assumed to be smooth. 
\end{definition}

\index{vector bundle}We will often consider vector bundles over a given manifold $M$, 
{\em e.g.}\ its tangent bundle $TM$. A vector bundle $(V,\pi,M)$ will be often denoted only by its total space $V$. 

\index{section}Similarly, we will usually denote the space of its sections $\sect(M,V)$ 
omitting the base manifold, namely we will write $\sect(V)$. 
For sections with compact support, we will adopt the standard convention 
which consists in adding a subscript $\mathrm{c}$, 
{\em i.e.}\ we will write $\sc(V)$ for sections with compact support of the vector bundle $V$. 
Again, the reader is invited to refer to the literature 
for the definition of vector bundles as well as of its sections, {\em e.g.}\ \cite[Section 2.1]{Jos11}.

\index{linear differential operators}\index{differential operator!linear}Let $V$ and $W$ be vector bundles 
over a manifold $M$. We will often consider {\em linear differential operators} $P:\sect(V)\to\sect(W)$ 
between the corresponding spaces of sections, defined according to\ \cite[Appendix\ A.4]{BGP07}.

\section{Globally hyperbolic spacetimes}\label{secGlobHyp}
In this section we introduce briefly some concepts of Lorentzian geometry. 
We are particularly interested in recalling the fundamental properties of globally hyperbolic spacetimes. 
Our references are\ \cite{BEE96, BGP07, ONe83, Wal12}. 

\begin{definition}\label{defLorManifold}\index{manifold!Lorentzian}
A Lorentzian manifold $(M,g)$ is a manifold $M$ endowed with a Lorentzian metric $g$, 
{\em i.e.}\ a fiberwise non-degenerate symmetric $(2,0)$-tensor field with signature of type $-+\cdots+$. 
\end{definition}

\index{timelike}\index{lightlike}\index{spacelike}Lorentzian manifolds provide the appropriate data 
to define a causal structure on a manifold. 
This is achieved specifying three classes of tangent vectors $v\in T_pM$ at each point $p\in M$: 
\begin{enumerate}
\item $v$ is {\em timelike} if $g(v,v)<0$; 
\item $v$ is {\em lightlike} if $g(v,v)=0$; 
\item $v$ is {\em spacelike} if $g(v,v)>0$. 
\end{enumerate}
\index{causal}{\em Causal} tangent vectors are either timelike or lightlike. 
This classification can be naturally extended to curves $\gamma:I\to M$ on $M$. 

\begin{definition}\index{curve}
Let $(M,g)$ be a Lorentzian manifold and consider a smooth curve $\gamma:I\to M$ on $M$, 
$I\subseteq\bbR$ being an open interval. 
$\gamma$ is said to be timelike, lightlike, causal or spacelike 
if its tangent vector field $\dot{\gamma}:I\to TM$ is respectively timelike, lightlike, causal or spacelike 
everywhere along the curve $\gamma$. 
\end{definition}

The structure presented so far can be enriched specifying a time-orientation.

\begin{definition}\label{defSpacetime}\index{time-orientation}\index{manifold!Lorentzian!time-orientable}
A Lorentzian manifold $(M,g)$ is {\em time-orientable} 
if there exists a vector field $\mft\in\sect(TM)$ which is everywhere timelike. 
A {\em time-orientation} for $(M,g)$ is specified by the choice of a vector field $\mft$ as above. 

\index{spacetime}A {\em spacetime} $(M,g,\mfo,\mft)$ consists of a time-orientable Lorentzian manifold $(M,g)$, 
the choice of an orientation $\mfo$ for $M$\footnote{Recall that only orientable manifolds are considered, 
see Definition\ \ref{defManifold}.}
and the assignment of a time-orientation $\mft$. 
\end{definition}

Once a time-orientation $\mft$ has been chosen, 
one can distinguish non-zero causal tangent vectors $0\neq v\in T_pM$ in two classes, 
depending on whether their orientation agrees with or is opposite to that of $\mft$. 
This distinction can be easily extended to causal curves with a non-vanishing tangent vector field. 

\begin{definition}\index{future-directed}\index{past-directed}
Let $(M,g,\mfo,\mft)$ be a spacetime 
and consider a non-vanishing causal tangent vector $0\neq v\in T_pM$ at the point $p\in M$. 
We say that $v$ is 
\begin{enumerate}
\item {\em future directed} if $g(\mft,v)<0$; 
\item {\em past directed} if $g(\mft,v)>0$. 
\end{enumerate}

Let $\gamma:I\to M$ be a causal curve with non-vanishing tangent vector field, 
{\em i.e.}\ $\dot{\gamma}:I\to TM$ never vanishes along the curve. 
We say that $\gamma$ is future/past directed 
if its tangent vector $\dot{\gamma}$ is such everywhere along $\gamma$.
\end{definition}

A spacetime $(M,g,\mfo,\mft)$ provides a sufficiently rich structure 
in order to introduce the {\em causal future/past} $J_M^\pm(S)$ 
as well as the {\em chronological future/past} $I_M^\pm(S)$ of a spacetime region $S\subseteq M$. 
Furthermore, looking at the causal future and past of each point, 
one can determine whether a subset $S$ of a given spacetime $(M,g,\mfo,\mft)$ is compatible 
with the causal structure of the spacetime itself. 

\begin{definition}\index{causal future}\index{causal past}\index{chronological future}\index{chronological past}
\index{causally compatible}Let $(M,g,\mfo,\mft)$ be a spacetime and consider a subset $S\subseteq M$. 
We define the {\em causal future/past} $J_M^\pm(S)$ of $S$ in $M$ as the set of those points in $M$ 
that can be reached by a future/past directed causal curve emanating from a point in $S$. 
Similarly, one defines the {\em chronological future/past} $I_M^\pm(S)$ of $S$ in $M$ as the set of points in $M$ 
that can be reached by a future/past directed timelike curve emanating from a point in $S$. 
A subset $S$ of a spacetime $(M,g,\mfo,\mft)$ is {\em causally compatible} 
if $J_S^\pm(x)=J_M^\pm(x)\cap S$ for each $x\in S$. 
\end{definition}

The causal future and the causal past of a spacetime region $S\subseteq M$ are often considered together. 
For this purpose we introduce $J_M(S)=J_M^+(S)\cup J_M^-(S)$. 

In the development of the thesis, we will often encounter wave-type equations. 
This family of equations admits a globally well-posed initial value problem on a distinguished class of spacetimes, 
which are the main topic of this section, namely globally hyperbolic spacetimes. 
In order to introduce this special class of spacetimes, 
we need to define first a special class of $(m-1)$-dimensional submanifolds, known as Cauchy hypersurfaces, 
$m$ being the dimension of $M$ as a manifold. These hypersurfaces are exactly those 
where to specify the initial data of a Cauchy problem. 

\begin{definition}\index{Cauchy hypersurface}\index{spacetime!globally hyperbolic}
A {\em Cauchy hypersurface }$\Sigma$ of a spacetime $(M,g,\mfo,\mft)$ is a subset of $M$ 
satisfying the following property: $\Sigma$ meets each inextensible timelike curve on $M$ exactly once. 
A spacetime $(M,g,\mfo,\mft)$ is {\em globally hyperbolic} if it admits a Cauchy hypersurface. 
\end{definition}

Without any further assumption, one can show that 
any Cauchy hypersurface of a given spacetime $(M,g,\mfo,\mft)$ 
is an $(m-1)$-dimensional submanifold of $M$ of class $\mathrm{C}^0$, 
see \cite[Lemma\ 29, Chapter\ 14]{ONe83}. 

One of the most prominent results on global hyperbolicity is the theorem stated below, 
which is due to Bernal and S\'anchez\ \cite{BS05,BS06}. 
This result allows for a very detailed characterization of the structure of any globally hyperbolic spacetime. 

\begin{theorem}\label{thmGlobHyp}
Let $(M,g,\mfo,\mft)$ be a spacetime. Then the statements presented below are equivalent:
\begin{enumerate}
\item $(M,g,\mfo,\mft)$ is globally hyperbolic;

\item There are no closed causal curves in $M$ 
and $J_M^+(p)\cap J_M^-(q)$ is a compact set for all $p,q\in M$;

\item $(M,g,\mfo,\mft)$ is isometric to $\bbR\times\Sigma$ 
endowed with the metric $-\beta\,\dd t\otimes\dd t+h_t$, 
$t:\bbR\times\Sigma\to\bbR$ is the projection on the first factor, $\beta\in\c(\bbR\times\Sigma)$ is strictly positive, 
$\bbR\ni t\mapsto h_t$ provides a smooth $1$-parameter family of Riemannian metrics on $\Sigma$
and $\{t\}\times\Sigma$ is a spacelike smooth Cauchy hypersurface in $(M,g,\mfo,\mft)$ for each $t\in\bbR$ 
(under the identification of $M$ with $\bbR\times\Sigma$). 
\end{enumerate}

Furthermore, any spacelike smooth Cauchy hypersurface $\Sigma$ for $(M,g,\mfo,\mft)$ 
induces a foliation for $(M,g,\mfo,\mft)$ of the type described in the third statement above. 
\end{theorem}

Notice that the second statement appearing in the theorem above is very often used 
to define globally hyperbolic spacetimes. This is for example the approach of\ \cite{BGP07, ONe83}. 
However, in these references an apparently stronger hypothesis is required to hold, 
the so-called strong causality condition. It was proven in\ \cite{BS07} that the causality condition, 
namely non-existence of closed causal curves, is actually sufficient. 

Being motivated by the last theorem, which, given a Cauchy hypersurface for a spacetime $(M,g,\mfo,\mft)$, 
ensures the existence of a whole family of Cauchy hypersurfaces which are also smooth and spacelike, 
in the following we will always implicitly assume all Cauchy hypersurfaces to be smooth and spacelike, 
unless otherwise stated. 

Since this notational abuse almost never becomes a source of misunderstanding, 
from now on it will be much more practical to refer to a globally hyperbolic spacetime $(M,g,\mfo,\mft)$ 
mentioning explicitly only the underlying manifold $M$, the rest of the data being understood. 

For the moment we specified objects in the category of globally hyperbolic spacetimes. 
The next definition provides the suitable notion of a morphism between globally hyperbolic spacetimes. 

\begin{definition}\index{causal embedding}\index{Cauchy morphism} 
Let $M$ and $N$ be $m$-dimensional globally hyperbolic spacetimes. 
A {\em causal embedding} $f$ from $M$ to $N$ is specified by an isometric embedding $f:M\to N$ 
such that both the orientation and the time-orientation are preserved by $f$ 
and the image of $f$ is a causally compatible open subset of $N$. 
A causal embedding $f$ between two $m$-dimensional globally hyperbolic spacetimes $M$ and $N$ 
is a {\em Cauchy morphism} if there exists a Cauchy hypersurface $\Sigma$ of $N$ 
which is included in the image of $M$ under $f$, that is to say $\Sigma\subseteq f(M)$. 

\index{spacetime!globally hyperbolic!category}The {\em category of globally hyperbolic spacetimes} $\GHyp$
has $m$-dimensional globally hyperbolic spacetimes as objects 
and causal embeddings between two $m$-dimensional globally hyperbolic spacetimes as morphisms. 
\end{definition}

Before concluding the present section, following\ \cite{Bar14,San13}, 
we introduce some nomenclature for subregions of a globally hyperbolic spacetime $M$. 

\begin{definition}\label{defSupportSystems}\index{past compact}\index{future compact}\index{timelike compact}
\index{past spacelike compact}\index{future spacelike compact}\index{spacelike compact}
Let $M$ be a globally hyperbolic spacetime. We say that a subset $S\subseteq M$ is:
\begin{enumerate}

\item[$\mathrm{pc}$)] {\em past compact} if $S\cap J_M^-(K)$ is compact for each $K\subseteq M$ compact; 

\item[$\mathrm{fc}$)] {\em future compact} if $S\cap J_M^+(K)$ is compact for each $K\subseteq M$ compact; 

\item[$\mathrm{tc}$)] {\em timelike compact} if it is both past and future compact; 

\item[$\mathrm{sc}$)] {\em spacelike compact} if it is closed 
and there exists $K\subseteq M$ compact such that $S\subseteq J_M(K)$; 

\item[$\mathrm{psc}$)] {\em past spacelike compact} if it is closed 
and there exists $K\subseteq M$ compact such that $S\subseteq J_M^+(K)$; 

\item[$\mathrm{fsc}$)] {\em future spacelike compact} if it is closed 
and there exists $K\subseteq M$ compact such that $S\subseteq J_M^-(K)$. 

\end{enumerate}
\end{definition}

Mimicking the standard convention which consists in denoting sections with compact support with $\sc$, 
we will use the subscripts $\mathrm{tc}$ and $\mathrm{sc}$ 
to specify the support properties of certain sections of a vector bundle over a globally hyperbolic spacetime. 
For example, given a vector bundle $V$ over a globally hyperbolic spacetime $M$, 
we will denote the space of its sections with timelike compact support with the symbol $\stc(V)$.

\section{Differential forms}\label{secForms}
In this section we introduce briefly differential forms. 
The aim is to recall the key properties of the usual differential complexes for forms, 
which eventually lead to the standard de Rham cohomologies with unrestricted and compact supports. 
This presentation is mainly devoted to the development of a slightly modified version of de Rham cohomology 
for globally hyperbolic spacetimes defined out of the differential complexes which only involve forms 
supported on spacelike compact or timelike compact regions. 
As it will be clear from Chapter\ \ref{chMaxwell}, these cohomologies are 
very much related with the dynamics of $k$-form Maxwell fields on a globally hyperbolic spacetime. 
In the last part of this section, we will recall some properties of the Hodge-d'Alembert differential operator 
over globally hyperbolic spacetimes, 
in particular the existence and uniqueness of retarded and advanced Green operators. 
Our references for differential forms and de Rham cohomology are the classic books\ \cite{dR84, BT82}, 
while we refer to\ \cite{BGP07, Wal12} for the properties of the Hodge-d'Alembert differential operator 
(more generally, of any normally hyperbolic differential operator). 

\sk

\index{differential form}\index{form}Following the standard literature, 
we introduce {\em differential forms} of degree $k$ on an $m$-dimensional manifold $M$ 
as sections of the $k$-th exterior power $\bw^k(T^\ast M)$ of the cotangent bundle $T^\ast M$. 
This means that a $k$-form is nothing but a skew-symmetric $k$-covariant tensor field on $M$. 
The space of $k$-forms is denoted by $\f^k(M)=\sect\big(\bw^k(T^\ast M)\big)$. 
\index{wedge product}Being defined out of the exterior powers of a vector bundle, 
one can naturally endow the graded vector space $\f^\ast(M)=\bigoplus_k\f^k(M)$ 
with the structure of a graded algebra by defining the so-called {\em wedge product} 
$\wedge:\f^k(M)\times\f^{k^\prime}(M)\to\f^{k+k^\prime}(M)$. 
In the same spirit, one might also consider {\em differential forms with compact support} on $M$, 
which again form a graded algebra with respect to $\wedge$, denoted by $\fc^\ast(M)$. 

\index{integration}\index{integral}As stated at the beginning of this chapter, 
{\em cfr.}\ Definition\ \ref{defManifold}, 
we always consider $m$-dimensional oriented manifolds $M$. 
This provides a natural notion of {\em integral} $\int_M:\fc^m(M)\to\bbR$, 
which is defined as a real valued linear map on compactly supported forms of top degree $k=m$. 
\index{pairing between forms}The integration map, together with the wedge product, 
provides a pairing $\la\cdot,\cdot\ra$ between $k$-forms and $(m-k)$-forms. 
This is defined according to the formula given below: 
\begin{equation}\label{eqPairing1}
\la\alpha,\beta\ra=\tint_M\alpha\wedge\beta\,,
\end{equation}
where $\alpha\in\f^k(M)$ and $\beta\in\f^{m-k}(M)$ have supports with compact overlap. 
One can show that $\la\cdot,\cdot\ra$ induces a non-degenerate bilinear pairing between $\fc^k(M)$ and $\f^{m-k}(M)$. 
Similar conclusions hold true if the restriction to compactly supported forms is imposed 
on the second argument of $\la\cdot,\cdot\ra$ in place of the first one. 
In fact $\la\alpha,\beta\ra=(-1)^{k(m-k)}\la\beta,\alpha\ra$ 
for each $\alpha\in\f^k(M)$ and $\beta\in\f^{m-k}(M)$ with compact overlapping support, 
as it follows from the definition of the wedge product. 

\index{Hodge star}Furthermore, whenever $M$ is not only oriented, but also endowed with a metric 
(for example this is the case for a Lorentzian manifold according to Definition\ \ref{defLorManifold}), 
it is possible to perform another standard construction 
to define the {\em Hodge star} operator $\ast:\bw^k(T^\ast M)\to\bw^{m-k}(T^\ast M)$ 
as a linear map from skew-symmetric $k$-covariant tensors to skew-symmetric $(m-k)$-covariant tensors. 
For an explicit definition of the Hodge star, we refer the reader to the literature, 
{\em e.g.}\ \cite[Section\ 3.3]{Jos11}. 
\index{inner product}The Hodge star operator $\ast$ induces a non-degenerate inner product 
on $\bw^k(T^\ast M)$, which consists in contracting with respect to the metric 
the corresponding indices of two skew-symmetric covariant $k$-tensors at a point $p\in M$. 
One can express this inner product as 
\begin{equation}\label{eqContraction}
(\xi,\eta)\in\bw^k(T^\ast M)\times\bw^k(T^\ast M)\mapsto\ast^{-1}(\xi\wedge\ast\eta)\in\bbR\,.
\end{equation}
\index{volume form}An oriented manifold $M$ endowed with a metric admits a canonical top degree differential form 
induced by the Hodge star operator. Specifically $\vol=\ast1\in\f^m(M)$ 
defines the standard {\em volume form} over a pseudo-Riemannian oriented manifold. 
Via integration with the volume form $\vol$, the inner product defined in eq.\ \eqref{eqContraction} specifies 
a pairing $(\cdot,\cdot)$ between $k$-forms with compact overlapping support: 
\begin{equation}\label{eqPairing2}
(\alpha,\beta)=\tint_M\ast^{-1}(\alpha\wedge\ast\beta)\,\vol=\tint_M\alpha\wedge\ast\beta
=\la\alpha,\ast\beta\ra\,,
\end{equation}
where $\alpha,\beta\in\f^k(M)$ have supports with compact overlap. 
In particular, notice that $(\cdot,\cdot)$ is a non-degenerate bilinear pairing between $\fc^k(M)$ and $\f^k(M)$ 
or, similarly, between $\f^k(M)$ and $\fc^k(M)$. This is a consequence of \eqref{eqContraction} being symmetric. 

\index{differential}\index{exterior derivative}A central role in the theory of differential forms 
over a manifold is played by the {\em exterior derivative}, also known as {\em differential}. 
For differential forms of any degree on an $m$-dimensional manifold $M$ 
this is denoted by $\dd:\f^k(M)\to\f^{k+1}(M)$. 
An explicit definition of the exterior derivative can be found for example in\ \cite[Section 1.1]{BT82}.
The most relevant facts about the differential are the following: 
\begin{enumerate}
\item $\dd$ increases the degree of a form by one, 
\item the composition of $\dd:\f^k(M)\to\f^{k+1}(M)$ 
with $\dd:\f^{k+1}(M)\to\f^{k+2}(M)$ always gives the zero map,
\item $\dd$ does not enlarge the support of any form. 
\end{enumerate}
These facts motivate the introduction of the well-known de Rham complex 
as well as its counterpart with compact support, which constitute the main topic of the next subsection. 

\subsection{de Rham cohomology}\label{subCohom}
As anticipated, the aim of this subsection is to recall some well-known facts about de Rham cohomology, 
in particular how de Rham cohomology groups with or without compact support are defined 
and how they are related to each other via Poincar\'e duality. 

To start with, we need to recall the notions of a cochain complex and of a cochain map. 
Furthermore, we need to introduce cochain homotopies between cochain maps. 
Let us mention that one can define chain complexes along the same lines, 
the only difference being the fact that the degree is decreasing from left to right in the complex, 
instead of increasing. 
Notice that we will often refer to cochain complexes simply as complexes. 

\begin{definition}\label{defComplex}\index{complex!cochain}
Let $A$ be an Abelian group. Consider a family $\{M^k:\,k\in\bbZ\}$ of modules over $A$ 
connected by $A$-module homomorphisms $\{h^k\in\hom_A(M^k,M^{k+1}):\,k\in\bbZ\}$ 
fulfilling the property $h^{k+1}h^k=0$. These data together define a {\em cochain complex}, 
which is denoted by $(M^\ast,\,h^\ast)$, or simply $M^\ast$, and which is represented pictorially as follows: 
\begin{equation*}
\cdots\overset{h^{k-2}}{\lra}M^{k-1}\overset{h^{k-1}}{\lra}M^k\overset{h^k}{\lra}M^{k+1}
\overset{h^{k+1}}{\lra}\cdots\,.
\end{equation*}

Given two complexes $(M^\ast,\,h^\ast)$ and $(N^\ast,\,l^\ast)$ of modules over $A$, 
a cochain map $f$ from $M^\ast$ to $N^\ast$ is a collection of $A$-module homomorphisms 
$\{f^k\in\hom_A(M^k,N^k):\,k\in\bbZ\}$ such that $l^kf^k=f^{k+1}h^k$ for all $k\in\bbZ$. 
Pictorially this is represented by the commutative diagram below: 
\begin{equation*}
\xymatrix@C=3.5pc@R=3.5pc{
\cdots\ar[r]^{h^{k-2}} & M^{k-1}\ar[r]^{h^{k-1}}\ar[d]_{f^{k-1}} 
& M^k\ar[r]^{h^k}\ar[d]_{f^k} & M^{k+1}\ar[r]^{h^{k+1}}\ar[d]_{f^{k+1}} & \cdots\\
\cdots\ar[r]_{l^{k-2}} & N^{k-1}\ar[r]_{l^{k-1}} & N^k\ar[r]_{l^k} & N^{k+1}\ar[r]_{l^{k+1}} & \cdots
}\,.
\end{equation*}

\index{cochain homotopy}Let $(M^\ast,\,h^\ast)$ and $(N^\ast,\,l^\ast)$ be two complexes of $A$-modules 
and consider two cochain maps $f,g:M^\ast\to N^\ast$. A {\em cochain homotopy} $\eta$ between $f$ and $g$ 
is a collection $\{\eta^k\in\hom_A(M^k,N^{k-1}):\,k\in\bbZ\}$ of $A$-module homomorphisms 
such that $f^k-g^k=l^{k-1}\eta^k+\eta^{k+1}h^k$ for all $k\in\bbZ$. Pictorially we have 
\begin{equation*}
\xymatrix@C=3.5pc@R=3.5pc{
\cdots\ar[r]^{h^{k-2}} 
& M^{k-1}\ar[r]^{h^{k-1}}\ar@<-2pt>[d]_{f^{k-1}}\ar@<2pt>[d]^{g^{k-1}}\ar[ld]|{\eta^{k-1}} 
& M^k\ar[r]^{h^k}\ar@<-2pt>[d]_{f^k}\ar@<2pt>[d]^{g^k}\ar[ld]|{\eta^k} 
& M^{k+1}\ar[r]^{h^{k+1}}\ar@<-2pt>[d]_{f^{k+1}}\ar@<2pt>[d]^{g^{k+1}}\ar[ld]|{\eta^{k+1}} 
& \cdots\ar[ld]|{\eta^{k+2}}\\
\cdots\ar[r]_{l^{k-2}} & N^{k-1}\ar[r]_{l^{k-1}} & N^k\ar[r]_{l^k} & N^{k+1}\ar[r]_{l^{k+1}} & \cdots
}\,.
\end{equation*}
In this situation one says that the cochain maps $f$ and $g$ are {\em cochain homotopic}. 
\end{definition}

According to this definition, given a complex $(M^\ast, h^\ast)$, 
it is natural to consider the associated cohomology groups. 

\begin{definition}\index{cohomology}
Let $(M^\ast, h^\ast)$ be a complex of modules over an Abelian group $A$. 
The {\em cohomology groups} $H^\ast(M)$ associated to this complex\footnote{Despite 
the standard nomenclature, cohomology groups defined this way are $A$-modules.} 
are defined by 
\begin{align*}
H^k(M)=\frac{\ker(h^k:M^k\to M^{k+1})}{\im(h^{k-1}:M^{k-1}\to M^k)},\,k\in\bbZ\,.
\end{align*}
\end{definition}

As a straightforward consequence of the last definition and of Definition\ \ref{defComplex}, 
each cochain map induces an $A$-module homomorphism between the corresponding cohomology groups. 
Furthermore, given two homotopic cochain maps, 
it turns out that the homomorphisms induced in cohomology coincide. 

One realizes immediately that, given an $m$-dimensional manifold $M$, 
the differential $\dd$ for the graded algebra of differential forms $\f^\ast(M)$ 
fits into this scheme choosing $A=\bbR$, $M^k=\f^k(M)$ for $k\in\{0,\,\dots,\,m\}$, $M^k=0$ otherwise, 
$h^k=\dd:\f^k(M)\to\f^{k+1}(M)$ for $k\in\{0,\,\dots,\,m-1\}$ and $h^k=0$ otherwise. 
Similar conclusions can be drawn when only compactly supported forms are taken into account. 
This gives rise to the so-called de Rham complex and its analog with compact support. 

\begin{definition}\label{defdRComplex}\index{complex!de Rham}\index{complex!de Rham!compact support}
Let $M$ be an $m$-dimensional manifold and consider the graded algebra of differential forms $\f^\ast(M)$ 
over $M$. Furthermore, denote the differential with $\dd:\f^k(M)\to\f^{k+1}(M)$. 
The {\em de Rham complex} is the following cochain complex: 
\begin{equation}\label{eqddComplex}
0\lra\f^0(M)\overset{\dd}{\lra}\f^1(M)\overset{\dd}{\lra}\cdots\overset{\dd}{\lra}\f^m(M)\lra0\,.
\end{equation}
The {\em de Rham complex with compact support} is exactly the same complex 
where only forms with compact support are taken into account: 
\begin{equation}\label{eqddCompactComplex}
0\lra\fc^0(M)\overset{\dd}{\lra}\fc^1(M)\overset{\dd}{\lra}\cdots\overset{\dd}{\lra}\fc^m(M)\lra0\,.
\end{equation}
\end{definition}

Out of these complexes one can define de Rham cohomology groups (possibly, with compact support)
by taking the quotient of the kernel of an arrow in the de Rham complex 
(respectively in its counterpart with compact support) by the image of the previous arrow. 

\begin{definition}\index{cohomology!de Rham}\index{cohomology!de Rham!compact support}
Let $M$ be a manifold and consider both the de Rham complex over $M$, eq.\ \eqref{eqddComplex}, 
and its counterpart with compact support, eq.\ \eqref{eqddCompactComplex}. 
Denote the kernel of $\dd:\f^k(M)\to\f^{k+1}(M)$ with $\fdd^k(M)=\ker\big(\dd:\f^k(M)\to\f^{k+1}(M)\big)$, 
and the kernel of $\dd:\fc^k(M)\to\fc^{k+1}(M)$ with $\fcdd^k(M)=\ker\big(\dd:\fc^k(M)\to\fc^{k+1}(M)\big)$.  
The $k$-th {\em de Rham cohomology group} 
and the $k$-th {\em de Rham cohomology group with compact support} over $M$\footnote{Note that 
these cohomology groups are $\bbR$-modules, hence vector spaces.} 
are defined as the quotient spaces 
\begin{align}\label{eqddCohomology}
\hdd^k(M)=\frac{\fdd^k(M)}{\dd\f^{k-1}(M)}\,, && \hcdd^k(M)=\frac{\fcdd^k(M)}{\dd\fc^{k-1}(M)},\,
&& k\in\{0,\,\dots,\,m\}\,. 
\end{align}
\end{definition}

\index{form!closed}\index{form!exact}We say that a form $\omega\in\f^k(M)$ is {\em closed} if $\dd\omega=0$, 
while we refer to a form $\theta\in\f^k(M)$ of the type $\theta=\dd\zeta$ for some $\zeta\in\f^{k-1}(M)$ as being 
{\em exact}. The same nomenclature applies to forms with any kind of support. 

De Rham cohomology and its counterpart with compact support are strictly related to each other 
due to the relation provided by {\em Stokes' theorem} 
between the differential $\dd$ and the integration map $\int_M$ on an oriented manifold $M$. 

\begin{theorem}[Stokes' theorem]\label{thmStokes}\index{Stokes' theorem}
Let $M$ be an $m$-dimensional oriented manifold (without boundary) and $\omega\in\fc^{m-1}(M)$. 
Then $\int_M\dd\omega=0$. 
\end{theorem}

Later we will exploit this theorem to introduce a pairing between de Rham cohomology groups 
and their counterparts with compact support, which is known as Poincar\'e duality. 
\index{codifferential}Before getting to this point, 
we would like to present another kind of differential operator acting on forms, 
the codifferential, which can be defined for every pseudo-Riemannian oriented manifold 
and plays a major role in describing the dynamics of Maxwell $k$-forms, see\ Chapter\ \ref{chMaxwell}. 
Assuming that $M$ is an oriented manifold endowed with a metric, we can consider the Hodge star $\ast$. 
Exploiting this tool, we can introduce the {\em codifferential}:
\begin{equation*}
\de=(-1)^k\ast^{-1}\dd\,\ast:\f^k(M)\to\f^{k-1}(M)\,.
\end{equation*}
The role played by the codifferential is explained by Stokes' theorem. 
In fact, a simple calculation exploiting Theorem\ \ref{thmStokes} shows that $\de$ is the formal adjoint of $\dd$ 
with respect to the pairing $(\cdot,\cdot)$ defined in\ \eqref{eqPairing2}. 
Explicitly, this means $(\de\alpha,\beta)=(\alpha,\dd\beta)$ 
for each $\alpha\in\f^{k+1}(M)$ and $\beta\in\f^k(M)$ such that $\supp(\alpha)\cap\supp(\beta)$ is compact: 
\begin{align}
(\de\alpha,\beta) & =(-1)^{k+1}\int_M(\ast^{-1}\dd\ast\alpha)\wedge\ast\beta
=(-1)^{k+1}\int_M\beta\wedge\dd\ast\alpha\nonumber\\
& =-\int_M\big(\dd(\beta\wedge\ast\alpha)-\dd\beta\wedge\ast\alpha\big)
=(\alpha,\dd\beta)\,.\label{eqAdjoint}
\end{align}
Contrary to the differential $\dd$, the codifferential $\de$ lowers the degree by one. 
All other properties which are relevant for the definition of the complexes\ \eqref{eqddComplex}\ 
and \eqref{eqddCompactComplex} hold true for $\de$ as well, 
namely the codifferential is a support preserving map and $\de\circ\de=0$. 
These facts are straightforward consequences of the properties of $\dd$ and $\ast$. 
For this reason one can associate chain complexes to the codifferential. 
The only difference with respect to Definition\ \ref{defComplex} is that now the degree is decreasing. 
It turns out that the chain complexes for $\de$ presented below are isomorphic 
to those of Definition\ \ref{defdRComplex}, the isomorphism being provided by $\ast$: 
\begin{align}
0\lra\f^m(M)\overset{\de}{\lra}\f^{m-1}(M)\overset{\de}{\lra}\cdots
\overset{\de}{\lra}\f^0(M)\lra0\,,\label{eqdeComplex}\\
0\lra\fc^m(M)\overset{\de}{\lra}\fc^{m-1}(M)\overset{\de}{\lra}\cdots
\overset{\de}{\lra}\fc^0(M)\lra0\,.\label{eqdeCompactComplex}
\end{align}
Mimicking the notation introduced for $\dd$, we define $\fde^k(M)$ as the kernel of $\de:\f^k(M)\to\f^{k-1}(M)$ 
and $\fcde^k(M)$ as the kernel of $\de:\fc^k(M)\to\fc^{k-1}(M)$. 
\index{form!coclosed}\index{form!coexact}We say that a form $\omega\in\f^k(M)$ is {\em coclosed} 
if $\de\omega=0$, while a form $\theta\in\f^k(M)$ of the type $\theta=\de\zeta$ 
for some $\zeta\in\f^{k+1}(M)$ is {\em coexact}. 
The same nomenclature applies to any kind of support system. 

One gets cohomology groups and cohomology groups with compact support for $\de$ 
in the same spirit as the usual ones, {\em cfr.}\ \eqref{eqddCohomology}:
\begin{align}\label{eqdeCohomology}
\hde^k(M)=\frac{\fde^k(M)}{\de\f^{k+1}(M)}\,, && \hcde^k(M)=\frac{\fcde^k(M)}{\de\fc^{k+1}(M)}, 
&& k\in\{0,\,\dots,\,m\}\,.
\end{align}
It is straightforward to check that the Hodge star operator induces an isomorphism
between the cohomology groups for $\de$ and those defined for $\dd$, 
namely $\hdd^k(M)\simeq\hde^{m-k}(M)$ and $\hcdd^k(M)\simeq\hcde^{m-k}(M)$ via $\ast$. 

As anticipated, Theorem\ \ref{thmStokes} entails that, 
in the case of an oriented manifold, the pairing $\la\cdot,\cdot\ra$ descends to $\dd$-cohomology groups. 
Similarly, on any oriented pseudo-Riemannian manifold, $(\cdot,\cdot)$ induces 
a pairing between the relevant cohomology groups. We summarize these facts below: 
\begin{subequations}\label{eqPoincarePairing}
\begin{align}
\la\cdot,\cdot\ra & :\hcdd^k(M)\times\hdd^{m-k}(M)\to\bbR\,,\\
{}_\de(\cdot,\cdot) & :\hcde^k(M)\times\hdd^k(M)\to\bbR\,,\\
(\cdot,\cdot)_\de & :\hcdd^k(M)\times\hde^k(M)\to\bbR\,.
\end{align}
\end{subequations}
For the first pairing we used the notation for the original pairing between forms, 
while for the other ones we introduced a subscript $\de$ to specify 
which of the arguments is defined on $\de$-cohomology groups. 
Furthermore, notice that compact supports always appear in the first argument. 

We conclude the present section stating a version of Poincar\'e duality 
for de Rham cohomologies, see\ \cite[Section\ 1.5]{BT82}. 
This requires the notion of a good cover of a manifold.

\begin{definition}\index{good cover}\index{manifold!finite type}\index{finite type}
Let $M$ be an $m$-dimensional manifold. A {\em good cover} of $M$ is a cover of $M$ by open sets 
such that the intersection of the elements of any finite subset of the cover is either empty 
or diffeomorphic to $\bbR^m$. A manifold $M$ is {\em of finite type} if it admits a finite good cover.
\end{definition}

Taking into account manifolds of finite type, following \cite{BT82} 
it is possible to show that the pairings\ \eqref{eqPoincarePairing} are actually non-degenerate. 

\begin{theorem}\label{thmPoincareDuality}\index{Poincar\'e duality}
Let $M$ be an oriented manifold of finite type. 
Then the pairing $\la\cdot,\cdot\ra$ between $\hcdd^k(M)$ and $\hdd^{m-k}(M)$ is non-degenerate. 
If $M$ is also endowed with a metric, {\em i.e.}\ $M$ is a pseudo-Riemannian oriented manifold of finite type, 
both the pairings ${}_\de(\cdot,\cdot):\hcde^k(M)\times\hdd^k(M)\to\bbR$ 
and $(\cdot,\cdot)_\de:\hcdd^k(M)\times\hde^k(M)\to\bbR$ are non-degenerate. 
\end{theorem}

The second part of the last theorem is a direct consequence of the first part 
simply taking into account that the Hodge star $\ast$ induces isomorphisms 
$\hcde^k(M)\simeq\hcdd^{m-k}(M)$ and $\hde^k(M)\simeq\hdd^{m-k}(M)$. 

\begin{remark}\label{remPoincareDualityImproved}
Actually, a partial positive result, related to Theorem\ \ref{thmPoincareDuality}, 
holds true under slightly weaker hypotheses, namely without the requirement of existence of a finite good cover. 
Specifically, for any oriented manifold $M$ (possibly not of finite type), 
the map defined below turns out to be a vector space isomorphism, see\ \cite[Section\ V.4]{GHV72}: 
\begin{align*}
\hdd^{m-k}(M)\to\big(\hcdd^k(M)\big)^\ast\,, && [\beta]\mapsto\la\cdot,[\beta]\ra\,.
\end{align*}
However, when $M$ is not of finite type, it might be the case that 
the linear map defined below fails to be an isomorphism, see\ \cite[Remark 5.7]{BT82}: 
\begin{align*}
\hcdd^k(M)\to\big(\hdd^{m-k}(M)\big)^\ast\,, && [\alpha]\mapsto\la[\alpha],\cdot\ra\,.
\end{align*}
In summary, for any oriented manifold $M$ the first pairing\ in\ \eqref{eqPoincarePairing} turns out 
to be always non-degenerate in the second argument, 
moreover this is the case for the first argument too if $M$ is of finite type. 
Similar conclusions hold true for the pairings $_\de(\cdot,\cdot)$ and $(\cdot,\cdot)_\de$. 
\end{remark}

\subsection{Causally restricted de Rham cohomology}\label{subRestrictedCohom}
One might observe that a complex similar to the de Rham one\ \eqref{eqddComplex} 
can be defined for any support system. 
This fact simply relies on the support preserving property of $\dd$, which was already exploited 
in order to define the de Rham complex with compact support in\ \eqref{eqddCompactComplex}. 
When one is dealing with a globally hyperbolic spacetime $M$, other choices of support systems may be considered 
besides the compact case, for example those introduced in Definition\ \ref{defSupportSystems}. 
For the moment we will focus our attention only on the timelike compact and spacelike compact cases, 
which, besides their major role in discussing the dynamics of field equations, 
turn out to be relevant from the perspective of cohomology theory too. 
As we will see later, {\em cfr.}\ Remark\ \ref{remTrivCohom}, all other support systems 
presented in Definition\ \ref{defSupportSystems} are trivial from the point of view of de Rham cohomology. 
Let us mention that the content of the present subsection summarizes the results of\ \cite{Ben14}. 
An approach to this topic more in the spirit of homological algebra is also available in\ \cite{Kha14a}. 

\begin{definition}\index{complex!de Rham!timelike compact}\index{complex!de Rham!spacelike compact}
\index{cohomology!de Rham!timelike compact}\index{cohomology!de Rham!spacelike compact}
Let $M$ be an $m$-dimensional globally hyperbolic spacetime and consider the differential $\dd$ on $M$. 
The diagrams presented below define the {\em timelike compact and spacelike compact de Rham complexes}: 
\begin{align*}
0\lra\ftc^0(M)\overset{\dd}{\lra}\ftc^1(M)\overset{\dd}{\lra}\cdots\overset{\dd}{\lra}\ftc^m(M)\lra0\,,\\
0\lra\fsc^0(M)\overset{\dd}{\lra}\fsc^1(M)\overset{\dd}{\lra}\cdots\overset{\dd}{\lra}\fsc^m(M)\lra0\,.
\end{align*}
Furthermore, {\em timelike compact and spacelike compact de Rham cohomology groups} are defined according to 
\begin{align*}
\htcdd^k(M)=\frac{\ftcdd^k(M)}{\dd\ftc^{k-1}(M)}\,, && \hscdd^k(M)=\frac{\fscdd^k(M)}{\dd\fsc^{k-1}(M)},\,
&& k\in\{0,\,\dots,\,m\}\,, 
\end{align*}
where $\ftcdd^k(M)$ and $\fscdd^k(M)$ denote the kernels of $\dd:\ftc^k(M)\to\ftc^{k+1}(M)$ 
and respectively of $\dd:\fsc^k(M)\to\fsc^{k+1}(M)$. 
\end{definition}

Being a globally hyperbolic spacetime, $M$ comes endowed with a metric and an orientation. 
Therefore we can consider the codifferential $\de$ and, in full analogy with the last definition, 
we can introduce similar complexes involving $\de$ in place of $\dd$ 
as well as cohomology groups $\htcde^k(M)=\ftcde^k(M)/\de\ftc^{k+1}(M)$ 
and $\hscde^k(M)=\fscde^k(M)/\de\fsc^{k+1}(M)$ for $k\in\{0,\,\dots,\,m\}$. 
Notice that once again the Hodge star $\ast$ induces isomorphisms 
between the $\de$-complexes and the $\dd$-complexes as well as between the associated cohomology groups. 

Various techniques are available to compute the classical de Rham cohomology groups of a manifold. 
As for the timelike compact and spacelike compact cohomologies of a globally hyperbolic spacetime $M$, 
it turns out that they can be determined very conveniently from the classical ones. 
This fact relies on two isomorphisms which we are going to establish. 
Denoting with $\Sigma$ one of the Cauchy hypersurfaces of $M$, one has the following: 
\begin{enumerate}
\item The first isomorphism relates the $k$-th de Rham cohomology group with timelike compact support 
$\htcdd^k(M)$ on $M$ to the $(k-1)$-th de Rham cohomology group $\hdd^{k-1}(\Sigma)$ on $\Sigma$; 
\item The second one relates the $k$-th de Rham cohomology group with spacelike compact support $\hscdd^k(M)$ 
on $M$ to the $k$-th de Rham cohomology group with compact support $\hcdd^k(\Sigma)$ on $\Sigma$. 
\end{enumerate}
The remainder of this subsection will be devoted to the explicit construction of such isomorphisms. 
In the end we will see how these isomorphisms provide an analogue of Poincar\'e duality 
between spacelike compact and timelike compact cohomology groups. 

\begin{remark}
We are taking care of causally restricted de Rham cohomologies only for the differential $\dd$. 
As a matter of fact, it is possible to repeat all the arguments which follow for the codifferential $\delta$ as well, 
the Hodge star $\ast$ inducing natural isomorphisms between $\htcdd^\ast(M)$ and $\htcde^{m-\ast}(M)$ in the 
timelike compact case and between $\hscdd^\ast(M)$ and $\hscde^{m-\ast}(M)$ in the spacelike compact case. 
\end{remark}

\subsubsection{Timelike compact vs. unrestricted cohomology}
Let us start from the timelike compact de Rham cohomology group $\htcdd^k(M)$ 
of a globally hyperbolic spacetime $M$. 
We present here the strategy which we are going to adopt 
in order to exhibit an isomorphism with the de Rham cohomology group $\hdd^{k-1}(\Sigma)$ 
of a Cauchy hypersurface $\Sigma$ of $M$ in degree lowered by $1$: 
\begin{enumerate}
\item $M$ being globally hyperbolic, one can exploit Theorem\ \ref{thmGlobHyp} 
to exhibit an isometry from $M$ to a globally hyperbolic spacetime $M_\Sigma$ 
which is explicitly \quotes{split in time and space factors}. 
This means that the underlying manifold is $\bbR\times\Sigma$, $\Sigma$ being a Cauchy hypersurface for $M$, 
and that the metric splits in time and space components too as explained by Theorem\ \ref{thmGlobHyp}. 
Because of the isometry between $M$ and $M_\Sigma$, 
$\htcdd^k(M)$ turns out to be isomorphic to $\htcdd^k(M_\Sigma)$. 
This provides the first part of the isomorphism we are going to construct; 
\item At a second stage we want to exploit the fact that, as a manifold, $M_\Sigma$ looks like $\bbR\times\Sigma$. 
This fact will eventually enable us to construct the second part of the isomorphism we are interested in, 
namely an isomorphism $\htcdd^k(M_\Sigma)\simeq\hdd^{k-1}(\Sigma)$ 
which is provided by \quotes{integration along the time factor} of $M_\Sigma$. 
\end{enumerate}

\begin{remark}\label{remTCCohoTriv}
The isomorphism $\htcdd^k(M)\simeq\hdd^{k-1}(\Sigma)$ we would like to exhibit 
entails in particular that $\htcdd^0(M)$ is trivial. 
As a consistency check for our theorem, we prove this fact directly: 
An element in $\htcdd^0(M)$ is nothing but a constant function with timelike compact support. 
As a consequence, this function has to vanish everywhere showing that $\htcdd^0(M)=\{0\}$. 
\end{remark}

Fix a globally hyperbolic spacetime $M$ and consider a foliation $M_\Sigma$ 
provided by Theorem\ \ref{thmGlobHyp}. Since $M_\Sigma$ as a manifold looks like 
the Cartesian product $\bbR\times\Sigma$, $\Sigma$ being a Cauchy hypersurface for $M$, 
we can consider the projections on each factor: 
\begin{align}\label{eqProj}
t:M_\Sigma\to\bbR\,, && \pi:M_\Sigma\to\Sigma\,.
\end{align}
To define the maps which will give rise to the sought isomorphism in cohomology, 
we follow the approach of \cite[Section\ I.6, pp.\ 61--63]{BT82}. 
Let us stress that the context in the mentioned reference is slightly different. 
In fact, in place of forms with timelike compact support, the so-called vertical compact forms are considered there. 
As a matter of fact, each timelike compact region is indeed vertical compact in the sense of \cite{BT82}, 
but the converse is not generally true. 
For this reason, it is mandatory for us to take care that the relevant maps in the construction 
of the sought isomorphism in cohomology are well-behaved 
with respect to the stricter constraint considered here on the support of forms. 
For the sake of completeness, we reproduce below the construction of \cite{BT82}, 
adapting it to the present setting, namely to differential forms with timelike compact support. 

The starting point is the following observation: All $k$-forms with timelike compact support on $M_\Sigma$ 
are linear combinations of two types of forms:\footnote{Notice that for the remainder of this subsection 
we will often omit the wedge product in order to improve readability.} 
\begin{align}
\mbox{Type $1_\mathrm{tc}$:} && (\pi^\ast\phi)\,f\,, 
& \quad\phi\in\f^k(\Sigma),\,f\in\ctc(M_\Sigma)\,,\label{eqType1TC}\\
\mbox{Type $2_\mathrm{tc}$:} && (\pi^\ast\psi)\,h\dd t\,, & 
\quad\psi\in\f^{k-1}(\Sigma),\,h\in\ctc(M_\Sigma)\label{eqType2TC}\,.
\end{align}

\index{time-integration map}The idea is to introduce a time-integration map acting on timelike compact forms 
which integrates all terms proportional to $\dd t$, thus mapping to forms on the Cauchy hypersurface $\Sigma$. 
In order to consistently define such a map, we observe that, 
given a compact subset $K$ of $\Sigma$, $\pi^{-1}(K)$ has compact overlap 
with each timelike compact region $T$ of $M_\Sigma$. 
In fact, the overlap is indeed contained in the compact set $J_{M_\Sigma}(K)\cap T$. 
That said, we are in able to define the {\em time-integration} map as follows: 
\begin{align}\label{eqTimeInt}
i:\ftc^k(M_\Sigma) & \to\f^{k-1}(\Sigma)\,,\\
(\pi^\ast\phi)\,f & \mapsto0\,,\nonumber\\
(\pi^\ast\psi)\,h\dd t & \mapsto\psi\,\tint_\bbR h(s,\cdot)\dd s\,.\nonumber
\end{align}
Notice that we need the integration map to descend to cohomology groups. 
This happens whenever $i$ is a cochain map between the complexes $\big(\ftc^\ast(M_\Sigma),\dd\big)$ 
and $\big(\f^{\ast-1}(\Sigma),\dd\big)$, as it follows from the lemma stated below, 
{\em cfr.}\ Definition\ \ref{defComplex}. 

\begin{lemma}\label{lemTimeInt}
$\dd\,i=i\,\dd$ on $\ftc^k(M_\Sigma)$ for each $k\in\{0,\dots,m\}$. 
\end{lemma}

\begin{proof}
To start with, we choose an oriented atlas for $\Sigma$ 
and we extend it along the time factor in order to obtain an atlas for $M_\Sigma$. 
This will be used to perform explicitly the relevant computations. 
The key point of this proof relies on the possibility to interchange spatial derivatives 
and integrals along the time direction. In fact, this follows from $\pi^{-1}(K)$ having compact overlap 
with each timelike compact region of $M_\Sigma$ for any compact $K\subseteq\Sigma$. 

First we consider $k$-forms of type $1_\mathrm{tc}$, see \eqref{eqType1TC}. 
Notice that, for each fixed $x\in\Sigma$, the map $s\mapsto f(s,x)$ has compact support on $\bbR$. 
As a consequence $\int_\bbR\partial_tf(s,x)\dd s=0$ for each $x\in\Sigma$. 
This fact motivates the following chain of identities, 
eventually concluding the proof for $k$-forms of type $1_{tc}$: 
\begin{align*}
i\,\dd\big((\pi^\ast\phi)\,f\big) & =i\big((\pi^\ast\dd\phi)\,f+(-1^k)(\pi^\ast\phi)\,\dd f\big)\\
& =(-1^k)\phi\,\tint_\bbR\partial_tf(s,\cdot)\dd s=0=\dd\,i\big((\pi^\ast\phi)\,f\big)\,.
\end{align*}

To complete the proof, we consider now $k$-forms of type $2_\mathrm{tc}$, see \eqref{eqType2TC}. 
As already noted above, we are allowed to interchange the order in which integration along the time direction 
and spatial derivatives are performed. In particular, one has the identity 
$\partial_i\int_\bbR h(s,x)\dd s=\int_\bbR\partial_ih(s,x)\dd s$ for each $x\in\Sigma$. 
Therefore, the following chain of identities follows as a consequence of this fact, 
thus completing the proof of the lemma: 
\begin{align*}
\dd\,i\big((\pi^\ast\psi)\,h\dd t\big) & =\dd\big(\psi\,\tint_\bbR h(s,\cdot)\dd s\big)\\
& =\dd\psi\,\tint_\bbR h(s,\cdot)\dd s+(-1)^{k-1}\psi\,\dd x^i\,\tint_\bbR\partial_ih(s,\cdot)\dd s\\
& =i\,\big((\pi^\ast\dd\psi)\,h\dd t+(-1)^{k-1}(\pi^\ast\psi)\,\dd x^i\,\partial_ih\dd t\big)\\
& =i\,\dd\big((\pi^\ast\psi)\,h\dd t\big)\,.
\end{align*}
\end{proof}

Due to Lemma\ \ref{lemTimeInt}, the integration map $i:\ftc^\ast(M_\Sigma)\to\f^{\ast-1}(\Sigma)$ 
induces a corresponding map in cohomology. With a slight abuse of notation, 
we denote such map with $i:\htcdd^\ast(M_\Sigma)\to\hdd^{\ast-1}(\Sigma)$. 
At this stage $i$ establishes a relation between 
the timelike compact cohomology groups $\htcdd^\ast(M_\Sigma)$ of $M_\Sigma$ 
and the standard de Rham cohomology groups $\hdd^{\ast-1}(\Sigma)$ of $\Sigma$ in degree lowered by one. 
Since by construction $M_\Sigma$ is isometric to the original globally hyperbolic spacetime $M$, 
which has $\Sigma$ as one of its Cauchy hypersurfaces, 
our aim to exhibit an isomorphism between $\htcdd^\ast(M)$ and $\hdd^{\ast-1}(\Sigma)$ 
is achieved as soon as we manage to exhibit an inverse to $i$ at the level of cohomology groups. 
The standard technique to answer this question is 
to look for a cochain map from $\f^{\ast-1}(M_\Sigma)$ to $\ftc^\ast(M_\Sigma)$, 
which is an inverse of $i:\ftc^\ast(M_\Sigma)\to\f^{\ast-1}(M_\Sigma)$ up to a cochain homotopy, 
see\ Definition\ \ref{defComplex}. If this is the case, the cochain map induces an exact inverse in cohomology, 
thus showing that $i:\htcdd^\ast(M_\Sigma)\to\hdd^{\ast-1}(\Sigma)$ is an isomorphism. 

\index{time-extension map}A reasonable guess to define an inverse of 
$i:\ftc^\ast(M_\Sigma)\to\f^{\ast-1}(M_\Sigma)$ up to homotopy 
might be to extend $(k-1)$-forms on $\Sigma$ to timelike compact $k$-forms on $M_\Sigma$ 
simply by taking the wedge product with a suitable $1$-form along the time component: 
Choosing a smooth function $a\in\cc(\bbR)$ with compact support such that $\int_\bbR a(s)\dd s=1$, 
one can introduce a closed $1$-form $\omega=(t^\ast a)\dd t$ on $M_\Sigma$ with timelike compact support. 
Notice that $[a\dd s]$ generates the 1-dimensional vector space $\hcdd^1(\bbR)$. 
We can exploit $\omega$ to introduce a {\em time-extension map}: 
\begin{align}\label{eqTimeExt}
e:\f^{k-1}(\Sigma)\to\ftc^k(M_\Sigma)\,, && \phi\mapsto(\pi^\ast\phi)\,\omega\,.
\end{align}
It is straightforward to check that $e:\f^{\ast-1}(\Sigma)\to\ftc^\ast(M_\Sigma)$ is a cochain map. 

\begin{lemma}\label{lemTimeExt}
$\dd\,e=e\,\dd$ on $	\f^{k-1}(\Sigma)$ for each $k\in\{1,\,\dots,\,m\}$. 
\end{lemma}

\begin{proof}
Take $\phi\in\f^{k-1}(\Sigma)$ for some $k\in\{1,\dots,m\}$. 
Then, $\omega$ being closed, the following chain of identities holds true: 
\begin{equation*}
\dd\,e\,\phi=\dd\big((\pi^\ast\phi)\,\omega\big)=(\pi^\ast\dd\phi)\,\omega=e\,\dd\,\phi\,.
\end{equation*}
This concludes the proof. 
\end{proof}

Due to this lemma, the time-extension map $e:\f^{\ast-1}(\Sigma)\to\ftc^\ast(M_\Sigma)$ 
descends to cohomology groups. With a slight abuse of notation, 
we denote the induced map with $e:\hdd^{\ast-1}(\Sigma)\to\htcdd^\ast(M_\Sigma)$. 

As anticipated, our aim is to prove that, up to a chain homotopy, 
the time-extension map $e:\f^{\ast-1}(\Sigma)\to\ftc^\ast(M_\Sigma)$ 
is an inverse of the time-integration map $i:\ftc^\ast(M_\Sigma)\to\f^{\ast-1}(\Sigma)$. 
In fact, from \eqref{eqTimeInt} and \eqref{eqTimeExt} one can directly check that $e$ is a right-inverse of $i$, 
namely that $i\,e=\id_{\f^{\ast-1}(\Sigma)}$, which of course entails a similar identity in cohomology. 
As a matter of fact, for arbitrary $\phi\in\f^{k-1}(\Sigma)$, $k\in\{1,\dots,m\}$, 
on account of the normalization of $a\in\cc(\bbR)$, one has 
\begin{equation}\label{eqIntExt}
i\,e\,\phi=i\big((\pi^\ast\phi)\,\omega\big)=\phi\int_\bbR a(s)\dd s=\phi\,.
\end{equation}
Therefore, to conclude that $e$ and $i$ induce the sought isomorphisms in cohomology, 
one has to find a cochain homotopy $Q:\ftc^k(M_\Sigma)\to\ftc^{k-1}(M_\Sigma)$ 
between $e\,i$ and $\id_{\ftc^\ast(M_\Sigma)}$. 

As it was already mentioned, in order to find the relevant cochain homotopy, we mimic 
the approach of\ \cite[Section\ I.4, p.\ 38]{BT82}. However, we must pay attention 
to the fact that timelike compact $k$-forms must be mapped to timelike compact $(k-1)$-forms 
by the cochain homotopy we are going to introduce. This fact relies on the following observation: 
For a given smooth function $f\in\ctc(M_\Sigma)$ with timelike compact support on $M_\Sigma$, 
we can define a new function $\widehat{f}:M_\Sigma\to\bbR$, according to the formula 
\begin{equation}\label{eqHat}
\widehat{f}(t,x)=\int_{-\infty}^tf(s,x)\dd s-\int_\bbR f(r,x)\dd r\,\int_{-\infty}^ta(s)\dd s\,.
\end{equation}
This is well-defined since, for any arbitrary but fixed $x\in\Sigma$, the function $s\in\bbR\mapsto f(s,x)$ 
has compact support due to the support of $f$ being timelike compact. In fact, 
the support of $f(\cdot,x)$ is both included in $\supp(f)$ and in $\bbR\times\{x\}$, 
which in turn is included in $J_{M_\Sigma}(\bar{s},x)$ for any choice of $\bar{s}\in\bbR$; 
therefore $\supp(f(\cdot,x))\subseteq\supp(f)\cap J_{M_\Sigma}(\bar{s},x)$ is compact. 

However, this is not enough for our scope. We also need $\widehat{f}$ to have timelike compact support. 
Recalling that $a$ is a smooth function with compact support on $\bbR$, 
one can find an interval $[c,d]\subseteq\bbR$ which includes the support of $a$; 
therefore, by construction $t^\ast a$ is supported in a time slab 
$[c,d]\times\Sigma=J_{M_\Sigma}^+(\Sigma_c)\cap J_{M_\Sigma}^-(\Sigma_d)$, 
$\Sigma_c$ and $\Sigma_d$ being respectively the constant time Cauchy hypersurfaces of $M_\Sigma$ 
corresponding to $c$ and $d$, {\em i.e.}\ $\Sigma_c=\{c\}\times\Sigma$ 
and $\Sigma_d=\{d\}\times\Sigma$ both lying in $M_\Sigma$. 
Exploiting \cite[Theorem 3.1]{San13}, it is possible to find two Cauchy hypersurfaces $\tilde{\Sigma}_\pm$ 
such that $\supp(f)\subseteq J_{M_\Sigma}^+(\tilde{\Sigma}_-)\cap J_{M_\Sigma}^-(\tilde{\Sigma}_+)$. 
From \eqref{eqHat} one deduces that $\widehat{f}$ vanishes 
in the intersection of the chronological futures of $\Sigma_d$ and $\tilde{\Sigma}_+$ 
as well as in the intersection of the chronological pasts of $\Sigma_c$ and $\tilde{\Sigma}_-$. 
This is a direct consequence of the following properties of the integrals along time of both $a$ and $f$: 
\begin{align*}
\tint_{-\infty}^ta(s)\dd s & =
\begin{cases}
1\,, & \hspace{4.5pc}(t,x)\in I_{M_\Sigma}^+(\Sigma_d)\,,\\
0\,, & \hspace{4.5pc}(t,x)\in I_{M_\Sigma}^-(\Sigma_c)\,;
\end{cases}\\
\tint_{-\infty}^tf(s,x)\dd s & =
\begin{cases}
\tint_\bbR f(s,x)\dd s\,, & (t,x)\in I_{M_\Sigma}^+(\tilde{\Sigma}_+)\,,\\
0\,, & (t,x)\in I_{M_\Sigma}^-(\tilde{\Sigma}_-)\,.
\end{cases}
\end{align*}
We deduce that the support of $\widehat{f}$ lies in the intersection between 
the union of the causal futures of $\Sigma_c$ and $\tilde{\Sigma}_-$ 
and the union of the causal pasts of $\Sigma_d$ and $\tilde{\Sigma}_+$,\footnote{This 
is the complement in $M_\Sigma$ of the union between 
$I_{M_\Sigma}^+(\Sigma_d)\cap I_{M_\Sigma}^+(\tilde{\Sigma}_+)$ 
and $I_{M_\Sigma}^-(\Sigma_c)\cap I_{M_\Sigma}^-(\tilde{\Sigma}_-)$. 
This is due to the fact that, by its definition, 
each Cauchy hypersurface $\Sigma$ for a globally hyperbolic spacetime $M$ 
splits it into the two disjoint parts $I_M^\pm(\Sigma)$ and $J_M^\mp(\Sigma)=M$.} 
that is to say 
\begin{equation*}
\supp(\widehat{f}\,)\subseteq\big(J_{M_\Sigma}^+(\Sigma_c)\cup J_{M_\Sigma}^+(\tilde{\Sigma}_-)\big)
\cap\big(J_{M_\Sigma}^-(\Sigma_d)\cup J_{M_\Sigma}^-(\tilde{\Sigma}_+)\big)\,. 
\end{equation*}
From the inclusion above we deduce that $\supp(\widehat{f}\,)$ is timelike compact. 
This follows from Definition\ \ref{defSupportSystems} 
and the fact that, on a globally hyperbolic spacetime $M$, $J_M^\pm(K)\cap J_M^\mp(\Sigma)$ is compact 
for each compact $K\subseteq M$ and for each Cauchy hypersurface $\Sigma$ for $M$. 

The argument presented above shows that the following map is well defined: 
\begin{align}\label{eqTCHomotopy}
Q:\ftc^k(M_\Sigma) & \to\ftc^{k-1}(M_\Sigma)\,,\\
(\pi^\ast\phi)\,f & \mapsto0\,,\nonumber\\
(\pi^\ast\psi)\,h\dd t & \mapsto(-1)^k(\pi^\ast\psi)\,\widehat{h}\,,\nonumber
\end{align}
where $\widehat{h}$ is defined according to eq.\ \eqref{eqHat}. 
The next lemma shows that $Q$ is the sought cochain homotopy between $e\,i$ and $\id_{\ftc^\ast(M_\Sigma)}$. 

\begin{lemma}\label{lemTCHomotopy}
Consider the map $Q:\ftc^k(M_\Sigma)\to\ftc^{k-1}(M_\Sigma)$ 
defined in eq.\ \eqref{eqTCHomotopy} for $k\in\{0,\,\dots,\,m\}$. 
Then $Q$ provides a cochain homotopy between $e\,i$ and $\id_{\ftc^\ast(M_\Sigma)}$, that is to say 
$e\,i-\id_{\ftc^\ast(M_\Sigma)}=\dd\,Q+Q\,\dd$ on $\ftc^k(M_\Sigma)$ for each $k\in\{0,\dots,m\}$. 
\end{lemma}

\begin{proof}
Similarly to the proof of Lemma\ \ref{lemTimeInt}, we perform all computations using a coordinate system. 
To do so, we choose an oriented atlas for $\Sigma$ and we extend it in time to an atlas for $M_\Sigma$. 

First, we check the thesis on $k$-forms of type $1_\mathrm{tc}$, \eqref{eqType1TC}. 
Recalling that both $Q$ and $i$ vanish on forms of this type, 
{\em cfr.}\ \eqref{eqTCHomotopy} and \eqref{eqTimeInt}, 
and noting that $\widehat{\partial_tf}=f$ since $f(\cdot,x)$ has compact support for each $x\in\Sigma$, 
{\em cfr.}\ \eqref{eqHat}, we get the following chain of identities: 
\begin{align*}
(\dd\,Q+Q\,\dd)\big((\pi^\ast\phi)\,f\big) & =Q\big((\pi^\ast\dd\phi)\,f+(-1)^k(\pi^\ast\phi)\,\dd f\big)\\
& =-(\pi^\ast\phi)\,\widehat{\partial_tf}
=(e\,i-\id_{\ftc^\ast(M_\Sigma)})\big((\pi^\ast\phi)\,f\big)\,.
\end{align*}

We now consider the only possibility left, namely to have a $k$-form of type $2_\mathrm{tc}$, \eqref{eqType2TC}. 
Simply recalling the definitions of $Q$, $e$, and $i$, respectively\ \eqref{eqTCHomotopy}, \eqref{eqTimeExt} 
and \eqref{eqTimeInt}, one deduces the three identities presented below: 
\begin{align*}
\dd\,Q\big((\pi^\ast\psi)\,h\dd t\big) & =(-1)^k\dd\big((\pi^\ast\psi)\,\widehat{h}\big)
=(-1)^k(\pi^\ast\dd\psi)\,\widehat{h}-(\pi^\ast\psi)\,\dd\widehat{h}\,,\\
Q\,\dd\big((\pi^\ast\psi)\,h\dd t\big) & 
=Q\big((\pi^\ast\dd\psi)\,h\dd t+(-1)^{k-1}(\pi^\ast\psi)\,\dd x^i\,\partial_ih\dd t\big)\\
& =(-1)^{k+1}(\pi^\ast\dd\psi)\,\widehat{h}+(\pi^\ast\psi)\,\dd x^i\,\widehat{\partial_ih}\,,\\
e\,i\big((\pi^\ast\psi)\,h\dd t\big) & =e\left(\psi\,\int_\bbR h(s,\cdot)\dd s\right)
=\left(\pi^\ast\left(\psi\,\int_\bbR h(s,\cdot)\dd s\right)\right)\,\omega\,.
\end{align*}
Furthermore, eq.\ \eqref{eqHat} allows one to compute $\dd\widehat{h}$: 
\begin{equation*}
\dd\widehat{h}=h\dd t-\pi^\ast\left(\int_\bbR h(r,\cdot)\dd r\right)\,\omega+\dd x^i\,\partial_i\widehat{h}\,.
\end{equation*}
Putting all these data together, one reads 
\begin{align*}
(\dd\,Q+Q\,\dd)\big((\pi^\ast\psi)\,h\dd t\big) 
& =\left(\pi^\ast\left(\psi\,\int_\bbR h(r,\cdot)\dd r\right)\right)\,\omega-(\pi^\ast\psi)\,h\dd t\\
& =\big(e\,i-\id_{\ftc^k(M_\Sigma)}\big)\big((\pi^\ast\psi)\,h\dd t\big)\,,
\end{align*}
thus concluding the proof of the lemma. 
\end{proof}

Applying the last lemma, we get the sought isomorphism between cohomology groups with timelike compact support 
$\htcdd^\ast(M)$ of a globally hyperbolic spacetime $M$ and de Rham cohomology groups $\hdd^{\ast-1}(\Sigma)$ 
of a Cauchy hypersurface $\Sigma$ for $M$ in degree lowered by one. 

\begin{theorem}\label{thmTCCohomology}
Let $M$ be a globally hyperbolic spacetime and consider a Cauchy hypersurface $\Sigma$ for $M$. 
Then $i$ and $e$, defined respectively in \eqref{eqTimeInt} and \eqref{eqTimeExt}, 
induce isomorphisms at the level of cohomology groups: 
\begin{equation*}
\xymatrix{\htcdd^\ast(M)\ar@/^1.5pc/[r]^{i} & \hdd^{\ast-1}(\Sigma)\ar@/^1.5pc/[l]^{e}}\,.
\end{equation*}
\end{theorem}

\begin{proof}
Theorem\ \ref{thmGlobHyp} entails that $M$ is isometric to the globally hyperbolic spacetime $M_\Sigma$, 
whose underlying manifold is the Cartesian product $\bbR\times\Sigma$, 
$\Sigma$ being the given Cauchy hypersurface. This isometry induces an isomorphism 
at the level of cohomology groups, $\htcdd^\ast(M)\simeq\htcdd^\ast(M_\Sigma)$. 
$M_\Sigma$ is exactly a globally hyperbolic spacetime of the type considered in the construction 
of the time-integration map $i:\ftc^\ast(M_\Sigma)\to\f^{\ast-1}(\Sigma)$, 
see\ \eqref{eqTimeInt} and Lemma\ \ref{lemTimeInt}, 
and a time-extension map $e:\f^{\ast-1}(\Sigma)\to\ftc^\ast(M_\Sigma)$ can be introduced 
choosing a function $a\in\cc(\bbR)$ such that $\int_\bbR a(s)\dd s=1$, 
see\ \eqref{eqTimeExt} and Lemma\ \ref{lemTimeExt}. 
Furthermore, one can define a cochain homotopy $Q$ between the cochain maps $e\,i$ 
and $\id_{\ftc^\ast(M_\Sigma)}$, see\ \eqref{eqTCHomotopy} and Lemma\ \ref{lemTCHomotopy}. 
Recalling eq.\ \eqref{eqIntExt} and Lemma\ \ref{lemTCHomotopy}, we deduce 
that $i$ and $e$ give rise to inverse maps of each other at the level of cohomology, thus completing the proof. 
This is due to the fact that the term involving the cochain homotopy $Q$ in Lemma\ \ref{lemTCHomotopy} 
maps closed forms with timelike compact support to exact ones 
and therefore this contribution vanishes passing to timelike compact cohomology groups. 
\end{proof}

\subsubsection{Spacelike compact vs. compactly supported cohomology}
We now focus our attention on the spacelike compact cohomology $\hscdd^\ast(M)$ 
associated to a globally hyperbolic spacetime $M$. In analogy with the analysis in the timelike compact case, 
we would like to establish an isomorphism between the spacelike compact cohomology groups of $M$ 
and the de Rham cohomology groups with compact support $\hcdd^\ast(\Sigma)$ 
of a Cauchy hypersurface $\Sigma$ of $M$. 

Let us consider an $m$-dimensional globally hyperbolic spacetime $M$. 
Consider moreover a Cauchy hypersurface $\Sigma$ for $M$. 
To exhibit an isomorphism $\hscdd^\ast(M)\simeq\hcdd^\ast(\Sigma)$ we follow a strategy 
which is similar to the one for the timelike compact case: 
\begin{enumerate}
\item[1.] First, Theorem\ \ref{thmGlobHyp} is exploited to foliate the globally hyperbolic spacetime $M$ 
in such a way that $\Sigma$ is the typical folium. In fact, Theorem\ \ref{thmGlobHyp} provides an isometry 
between $M$ and an auxiliary globally hyperbolic spacetime $M_\Sigma$, 
whose underlying manifold is given by the Cartesian product $\bbR\times\Sigma$, 
$\Sigma$ being the given Cauchy hypersurface for $M$. 
Indeed this isometry induces an isomorphism $\hscdd^k(M)\simeq\hscdd^k(M_\Sigma)$ 
between cohomology groups. This is the first part of the sought isomorphism; 
\item[2.] The second part is constructed exploiting the structure of the auxiliary globally hyperbolic spacetime 
$M_\Sigma$, which is explicitly decomposed in time and space factors. 
This decomposition eventually leads to an isomorphism $\hscdd^k(M_\Sigma)\simeq\hcdd^{k}(\Sigma)$, 
thus completing our program. 
\end{enumerate}

\begin{remark}\label{remSCCohoTriv}
Writing $\hscdd^\ast(M)\simeq\hcdd^\ast(\Sigma)$, we mean that the isomorphism holds true in each degree. 
Indeed $\hcdd^m(\Sigma)$ vanishes, $\Sigma$ being an $(m-1)$-dimensional manifold. 
Therefore the sought isomorphism in degree $k=m$ tells us that the top cohomology group 
with spacelike compact support $\hscdd^m(M)$ on a globally hyperbolic spacetime $M$ is always trivial. 
\end{remark}

In\ \eqref{eqProj} we exploited the Cartesian product structure $\bbR\times\Sigma$ underlying $M_\Sigma$ 
to define projection maps on each factor. It is now convenient to introduce also the time-zero section $s$ of $\pi$: 
\begin{align}\label{eqSect}
s:\Sigma\to M_\Sigma\,, && x\mapsto(0,x)\,,
\end{align}

We note that, for each compact subset $K$ of the Cauchy hypersurface $\Sigma$, 
the preimage of $K$ under the projection $\pi:M_\Sigma\to\Sigma$ onto the space factor is spacelike compact. 
In fact, $\pi^{-1}(K)$ is closed, $\pi$ being continuous, and it is contained in $J_{M_\Sigma}(\{0\}\times K)$. 
This entails that the pull-back via $\pi$ of a compactly supported differential form on $\Sigma$ 
is a spacelike compact form on $M_\Sigma$. 
Furthermore, we observe that the image of $s$ is the Cauchy hypersurface $\{0\}\times\Sigma$ of $M_\Sigma$. 
Due to \cite[Corollary\ A.5.4]{BGP07}), each spacelike compact region has compact intersection 
with any Cauchy hypersurface. Therefore, pulling a spacelike compact form on $M_\Sigma$ 
back to $\Sigma$ via $s$ gives a compactly supported form on $\Sigma$. We summarize these facts below: 
\begin{align}\label{eqSCCochainMaps}
\pi^\ast:\fc^k(\Sigma)\to\fsc^k(M_\Sigma)\,, && s^\ast:\fsc^k(M_\Sigma)\to\fc^k(\Sigma)\,.
\end{align}
It is straightforward to check that $\pi\circ s=\id_\Sigma$. It follows immediately that 
\begin{align}\label{eqSectProj}
s^\ast\pi^\ast=\id_{\fc^k(\Sigma)}\,, && k\in\{0,\,\dots,\,m-1\}\,.
\end{align}
Being pull-backs along a smooth map, both $\pi^\ast$ and $s^\ast$ intertwine 
the differentials on $M_\Sigma$ and $\Sigma$, that is to say 
\begin{align*}
\pi^\ast\dd=\dd\pi^\ast\,, && s^\ast\dd=\dd s^\ast\,.
\end{align*}
This means that we have cochain maps $\pi^\ast:\fc^\ast(\Sigma)\to\fsc^\ast(M_\Sigma)$ 
and $s^\ast:\fsc^\ast(\Sigma)\to\fc^\ast(\Sigma)$ and moreover $s^\ast$ is a left inverse of $\pi^\ast$. 
Automatically, the same holds true at the level of cohomology groups as well. 
Therefore, in order to exhibit an isomorphism between $\hscdd^\ast(M_\Sigma)$ and $\hcdd^\ast(\Sigma)$, 
we are left with the proof of the fact that $\pi^\ast s^\ast$ is cochain homotopic to $\id_{\fsc^\ast(M_\Sigma)}$. 
If this were the case, then $\pi^\ast s^\ast$ would be the identity map 
at the level of spacelike compact cohomology groups, 
thus showing that $s^\ast$ is also a right inverse of $\pi^\ast$ at this cohomological level. 
In particular, this would prove that $\pi^\ast$ actually gives the sought isomorphism 
$\hcdd^\ast(\Sigma)\simeq\hscdd^\ast(M_\Sigma)$. 

The cochain homotopy between $\pi^\ast s^\ast$ and $\id_{\fsc^\ast(M_\Sigma)}$ 
is defined along the lines of \cite[Section\ I.4, pp.\ 33--35]{BT82}. 
However, the setting considered here is slightly different. 
In particular, one has to be careful with the support properties of the cochain homotopy we are going to consider. 
In fact, we want our candidate cochain homotopy to map spacelike compact forms in degree $k$ 
to spacelike compact forms in degree $k-1$. 
If this is the case, then the formula presented in \cite{BT82} provides 
an appropriate cochain homotopy for spacelike compact cohomologies too. 

Similarly to the timelike compact case, $k$-forms with spacelike compact support on $M_\Sigma$ are always  
a linear combination of $k$-forms of two types: 
\begin{align}
\mbox{$1_\mathrm{sc}$:} && (\pi^\ast\phi)\,f\,, & \quad\phi\in\f^k(\Sigma),\,f\in\csc(M_\Sigma)\,,
\label{eqType1SC}\\
\mbox{$2_\mathrm{sc}$:} && (\pi^\ast\psi)\,h\dd t\,, & \quad\psi\in\f^{k-1}(\Sigma),\,
h\in\csc(M_\Sigma)\label{eqType2SC}\,.
\end{align}

Given a smooth function $f\in\csc(M_\Sigma)$ with spacelike compact support, 
we can consider a new smooth function $\tilde{f}$ on $M_\Sigma$ defined by $\tilde{f}(t,x)=\int_0^tf(s,x)\dd s$. 
For our purposes, the most relevant feature of $\tilde{f}$ is that its support is spacelike compact. 
In fact, on account of the spacelike compact support of $f$, one can find a compact subset $K$ of $M_\Sigma$ 
such that the support of $f$ lies inside $J_{M_\Sigma}(K)$. 
Since $\{0\}\times\Sigma$ is a Cauchy hypersurface of $M_\Sigma$, 
the intersection $\tilde{K}$ between $\{0\}\times\Sigma$ and $J_{M_\Sigma}(K)$ is compact 
and therefore $J_{M_\Sigma}(\tilde{K})$ is spacelike compact. 
If we can prove that $\tilde{f}$ vanishes outside $J_{M_\Sigma}(\tilde{K})$, 
then $\tilde{f}$ has spacelike compact support as argued. 
Given a point $(t,x)\in M_\Sigma\setminus J_{M_\Sigma}(\tilde{K})$, 
we can consider the timelike curve $\gamma:s\in[0,t]\mapsto(s,x)\in M_\Sigma$ 
and deduce that the curve $\gamma$ does not meet $\tilde{K}$. 
By construction $J_{M_\Sigma}(K)$ includes the support of $f$ and is contained in $J_{M_\Sigma}(\tilde{K})$. 
It follows that $f$ vanishes along the curve $\gamma$, that is to say $f(s,x)=0$ for each $s\in[0,t]$. 
Recalling the formula which gives the value of $\tilde{f}$ at the point $(t,x)$, we get $\tilde{f}(t,x)=0$. 
Therefore we conclude that $\tilde{f}$ vanishes outside $J_{M_\Sigma}(\tilde{K})$, 
hence the support of $\tilde{f}$ is actually spacelike compact. 
 
Bearing this fact in mind, one defines the candidate for the cochain homotopy 
between $\pi^\ast s^\ast$ and $\id_{\fsc^\ast(M_\Sigma)}$ as follows: 
\begin{align}\label{eqSCHomotopy}
P:\fsc^k(M_\Sigma) & \to\fsc^{k-1}(M_\Sigma)\,,\\
(\pi^\ast\phi)\,f & \mapsto0\,,\\
(\pi^\ast\psi)\,h\dd t & \mapsto(-1)^k(\pi^\ast\psi)\,\tint_0^\cdot h(s,\cdot)\dd s\,.
\end{align}
The next lemma shows that $P$ is indeed the sought cochain homotopy. 

\begin{lemma}\label{lemSCHomotopy}
The map $P:\fsc^k(M_\Sigma)\to\fsc^{k-1}(M_\Sigma)$ defined according to\ \eqref{eqSCHomotopy}
provides a cochain homotopy between $\pi^\ast s^\ast$ and $\id_{\fsc^\ast(M_\Sigma)}$, 
namely $\pi^\ast s^\ast-\id_{\fsc^k(M_\Sigma)}=\dd\,P+P\,\dd$ for each $k\in\{0,\dots,m\}$. 
\end{lemma}

\begin{proof}
The proof is an explicit computation performed choosing an oriented atlas for $\Sigma$ 
and extending it to an atlas for $M_\Sigma$. We will first consider forms of type $1_\mathrm{sc}$ 
and then forms of type $2_\mathrm{sc}$. 

Let us consider a $k$-form $(\pi^\ast\phi)f$ of type $1_\mathrm{sc}$, \eqref{eqType1SC}. 
Bearing in mind the identity $\int_0^t\partial_tf(s,x)\dd s=f(t,x)-f(0,x)=(f-\pi^\ast s^\ast f)(t,x)$ 
for all $(t,x)\in M_\Sigma$, 
one gets the following chain of identities: 
\begin{align*}
(\dd\,P+P\,\dd)\big((\pi^\ast\phi)f\big) & =P\big((\pi^\ast\dd\phi)\,f+(-1)^k(\pi^\ast\phi)\,\dd f\big)\\
& =-(\pi^\ast\phi)\,\int_0^\cdot\partial_tf(s,\cdot)\dd s\\
& =(\pi^\ast\phi)\,(\pi^\ast s^\ast f-f)\\
& =\big(\pi^\ast s^\ast-\id_{\fsc^k(M_\Sigma)}\big)\big((\pi^\ast\phi)\,f\big)\,.
\end{align*}
Note that for the last equality we exploited eq.\ \eqref{eqSectProj}. 

For $k$-forms of the type $2_\mathrm{sc}$, \eqref{eqType2SC}, the computation is a bit more involved, 
therefore we prefer to compute each contribution separately. 
Only in the end we will put all terms together to get the final result. 
For the $\dd\,P$-term we have the following expression: 
\begin{align*}
\dd\,P & \big((\pi^\ast\psi)\,h\dd t\big)=(-1)^k\dd\left((\pi^\ast\psi)\,\int_0^\cdot h(s,\cdot)\dd s\right)\\
& =(-1)^k(\pi^\ast\dd\psi)\,\int_0^\cdot h(s,\cdot)\dd s
-(\pi^\ast\psi)\,\left(\dd x^i\partial_i\int_0^\cdot h(s,\cdot)\dd s+h\dd t\right)\,.
\end{align*}
As a second step we compute the $P\,\dd$-term explicitly: 
\begin{align*}
P\,\dd\big((\pi^\ast\psi)\,h\dd t\big)
& =P\,\big((\pi^\ast\dd\psi)\,h\dd t+(-1)^{k-1}(\pi^\ast\psi)\,\dd x^i\,\partial_ih\dd t\big)\\
& =(-1)^{k+1}(\pi^\ast\dd\psi)\,\int_0^\cdot h(s,\cdot)\dd s
+(\pi^\ast\psi)\,\dd x^i\,\int_0^\cdot\partial_ih(s,\cdot)\dd s\,.
\end{align*}
Exploiting the possibility to interchange the order 
in which the integral along the time direction and the spatial derivatives are performed 
and bearing in mind that $s^\ast\dd t=0$, we come to the conclusion of the proof: 
\begin{align*}
(\dd\,P+P\,\dd)\big((\pi^\ast\psi)\,h\dd t\big) & =-(\pi^\ast\psi)\,h\dd t
=(\pi^\ast s^\ast-\id_{\fsc^k(M_\Sigma)})\big((\pi^\ast\psi)\,h\dd t\big)\,.
\end{align*}
\end{proof}

Exploiting the last lemma, we can exhibit an isomorphism between the spacelike compact cohomology groups 
$\hscdd^\ast(M)$ of a globally hyperbolic spacetime $M$ and the cohomology groups $\hcdd^\ast(\Sigma)$ 
with compact support of a Cauchy hypersurface $\Sigma$ for $M$. 

\begin{theorem}\label{thmSCCohomology}
Let $M$ be a globally hyperbolic spacetime and consider a Cauchy hypersurface $\Sigma$ for $M$. 
Then $\pi^\ast:\fc^\ast(\Sigma)\to\fsc^\ast(M_\Sigma)$ and $s^\ast:\fsc^\ast(M_\Sigma)\to\fc^\ast(\Sigma)$, 
introduced in\ \eqref{eqSCCochainMaps}, induce isomorphisms in cohomology: 
\begin{equation*}
\xymatrix{\hscdd^\ast(M)\ar@/^1.5pc/[r]^{s^\ast} & \hcdd^\ast(\Sigma)\ar@/^1.5pc/[l]^{\pi^\ast}}\,.
\end{equation*}
\end{theorem}

\begin{proof}
From Theorem\ \ref{thmGlobHyp} we get an isometry between $M$ 
and the explicitly foliated globally hyperbolic spacetime $M_\Sigma$, 
whose underlying manifold is the Cartesian product $\bbR\times\Sigma$. 
On account of this isometry, there is an isomorphism $\hscdd^\ast(M)\simeq\hscdd^\ast(M_\Sigma)$. 
From\ \eqref{eqSectProj} we read that the cochain map $s^\ast:\fsc^\ast(M_\Sigma)\to\fc^\ast(\Sigma)$ 
is a left inverse of the cochain map $\pi^\ast:\fc^\ast(\Sigma)\to\fsc^\ast(M_\Sigma)$. 
In particular $s^\ast\pi^\ast$ induces the identity map in $\hcdd^\ast(\Sigma)$. 
Furthermore, Lemma\ \ref{lemSCHomotopy} shows 
that the cochain maps $\pi^\ast s^\ast$ and $\id_{\fsc^k(M_\Sigma)}$ are the same up to a cochain homotopy. 
Thus $\pi^\ast s^\ast$ induces the identity map in $\hscdd^\ast(\Sigma)$. This concludes the proof. 
\end{proof}

\subsubsection{Poincar\'e duality for causally restricted cohomology}
To complete our analysis on causally restricted de Rham cohomology groups, 
we present an analogue of the usual Poincar\'e duality between de Rham cohomology groups $\hdd^\ast(M)$ 
and de Rham cohomology groups $\hcdd^\ast(M)$ with compact support, see\ Theorem\ \ref{thmPoincareDuality}. 
This modified version relates cohomology groups with timelike compact support $\htcdd^\ast(M)$ 
to cohomology groups with spacelike compact support $\hscdd^{m-\ast}(M)$ 
on an $m$-dimensional globally hyperbolic spacetime $M$. 

As a starting point we observe the following fact: On an $m$-dimensional globally hyperbolic spacetime $M$, 
both eq.\ \eqref{eqPairing1} and eq.\ \eqref{eqPairing2} are well-defined 
for differential forms with spacelike compact support in one argument 
and with timelike compact support in the other argument. 
In fact, by Definition\ \ref{defSupportSystems} the intersection 
between a spacelike compact region and a timelike compact one is always a compact set. 
Bearing in mind that compact sets are both spacelike compact and timelike compact 
and that it is enough to consider forms with compact support to prove non-degeneracy 
even in the general case where no restriction is imposed on the support of the other argument of the pairing, 
it turns out that both pairings\ \eqref{eqPairing1} and \eqref{eqPairing2} are non-degenerate 
when only spacelike compact and timelike compact forms are taken into account. 

In \eqref{eqPoincarePairing} we were able to pair cohomology with compact support 
to cohomology with arbitrary support using Stokes' theorem. 
Adopting exactly the same approach, one can exploit Theorem\ \ref{thmStokes} 
to prove that $\la\cdot,\cdot\ra:\fsc^k(M)\times\ftc^{m-k}(M)\to\bbR$, 
and $(\cdot,\cdot):\fsc^k(M)\times\ftc^k(M)\to\bbR$ descend to cohomologies 
with spacelike and timelike compact support 
($\dd$-cohomology groups in both arguments for the first pairing, while $\de$-cohomology groups in one argument 
and $\dd$-cohomology groups in the other argument for the second pairing). 
As an example, take $\alpha\in\fsc^{k-1}(M)$ and $\beta\in\ftcdd^{m-k}(M)$, 
that is to say $\dd\beta=0$. Bearing in mind that $\alpha\wedge\beta$ has compact support, 
which in turn allows us to apply Stokes' theorem, we get 
\begin{equation*}
\la\dd\alpha,\beta\ra=\int_M\dd\alpha\wedge\beta=\int_M\dd(\alpha\wedge\beta)=0\,.
\end{equation*}
Summing up, we are interested in the following pairings 
between spacelike compact and timelike compact cohomology groups: 
\begin{subequations}\label{eqSCTCPairing}
\begin{align}
\la\cdot,\cdot\ra & :\hscdd^k(M)\times\htcdd^{m-k}(M)\to\bbR\,,\\
{}_\de(\cdot,\cdot) & :\hscde^k(M)\times\htcdd^k(M)\to\bbR\,,\\
(\cdot,\cdot)_\de & :\hscdd^k(M)\times\htcde^k(M)\to\bbR\,.
\end{align}
\end{subequations}

Our aim is to exhibit an analogue of Poincar\'e duality in the case of causally restricted cohomologies. 
Note that it is sufficient to prove non-degeneracy for one of the pairings listed above. 
In fact, the others are related to the chosen one via Hodge duality $\ast$. 
Our choice is to prove non-degeneracy for the first pairing in \eqref{eqSCTCPairing}. 
To achieve this result we prove the following preliminary lemma. 

\begin{lemma}\label{lemEquivalence}
Let $M$ be an $m$-dimensional globally hyperbolic spacetime and consider a Cauchy hypersurface $\Sigma$ for $M$. 
Recalling Theorem\ \ref{thmSCCohomology} and Theorem\ \ref{thmTCCohomology}, 
consider the isomorphisms $\pi^\ast:\hcdd^\ast(\Sigma)\to\hscdd^\ast(M)$ 
(pull-back along the projection to the Cauchy hypersurface)
and $e:\hdd^{m-1-\ast}(\Sigma)\to\htcdd^{m-\ast}(M)$ (time-extension). 
Then $\big\la\pi^\ast[\phi],e\,[\psi]\big\ra=\big\la[\phi],[\psi]\big\ra$ 
for each $[\phi]\in\hcdd^k(\Sigma)$ and $[\psi]\in\hdd^{m-1-k}(\Sigma)$, $k\in\{0,\dots,m-1\}$. 
Note that the pairing on the left-hand-side is given by\ \eqref{eqSCTCPairing}, 
while on the right-hand-side we have the pairing 
defined in\ \eqref{eqPoincarePairing} for the oriented manifold $\Sigma$. 
\end{lemma}

\begin{proof}
Denote with $M_\Sigma$ the foliation of the globally hyperbolic spacetime $M$ 
induced by the given Cauchy hypersurface $\Sigma$ for $M$ according to Theorem\ \ref{thmGlobHyp}. 
Consider any $[\phi]\in\hcdd^k(\Sigma)$ and $[\psi]\in\hdd^{m-k-1}(M)$ for arbitrary $k\in\{0,\dots,m-1\}$. 
Recalling the definitions of the maps $\pi^\ast:\fc^k(\Sigma)\to\fsc^k(M_\Sigma)$, \eqref{eqProj}, 
and $e:\f^{m-k-1}(\Sigma)\to\ftc^{m-k}(M_\Sigma)$, \eqref{eqTimeExt}, 
it is possible to perform the following calculation taking advantage of the \quotes{time-space} factorization: 
\begin{equation*}
\big\la\pi^\ast[\phi],e[\psi]\big\ra=\int_{M_\Sigma}(\pi^\ast\phi)\wedge\big((\pi^\ast\psi)\wedge\omega\big)
=\int_\Sigma\phi\wedge\psi\;\int_\bbR a(s)\dd s=\big\la[\phi],[\psi]\big\ra\,,
\end{equation*}
where we exploited also the normalization of $a$, $\tint_\bbR a(s)\dd s=1$. 
\end{proof}

The last lemma, together with the standard version of Poincar\'e duality, Theorem\ \ref{thmPoincareDuality}, 
provides the counterpart of Poincar\'e duality for causally restricted cohomology. 
This is the content of the next theorem. 

\begin{theorem}\label{thmSCTCPoincareDuality}
Let $M$ be an $m$-dimensional globally hyperbolic spacetime of finite type. 
Then the pairing $\la\cdot,\cdot\ra$ between $\hscdd^k(M)$ and $\htcdd^{m-k}(M)$ is non-degenerate. 
Therefore the same holds true for ${}_\de(\cdot,\cdot)$ between $\hscde^k(M)$ and $\htcdd^k(M)$ 
and $(\cdot,\cdot)_\de$ between $\hscdd^k(M)$ and $\htcde^k(M)$ as well. 
\end{theorem}

\begin{proof}
Choose a Cauchy hypersurface $\Sigma$ of $M$. According to the hypotheses, 
$M$ admits a finite good cover, therefore so does $\Sigma$. 
As a consequence of this fact, the pairing between $\hcdd^k(\Sigma)$ and $\hdd^{m-k-1}(\Sigma)$ 
is non-degenerate for each $k\in\{0,\dots,m-1\}$, 
see Theorem\ \ref{thmPoincareDuality} applied to the oriented manifold $\Sigma$. 
From Lemma\ \ref{lemEquivalence}, we deduce that the pairing between the causally restricted 
cohomology groups $\hscdd^k(M)$ and $\htcdd^{m-k}(M)$ is equivalent to the one between 
the standard cohomology groups $\hcdd^k(\Sigma)$ and $\hdd^{m-k-1}(\Sigma)$ for each $k\in\{0,\dots,m-1\}$. 
For definiteness, we remind the reader that this result is achieved via time-extension on the first argument 
and pull-back along the projection onto the Cauchy hypersurface on the second argument. 
Therefore, non-degeneracy of the pairing between $\hcdd^k(\Sigma)$ and $\hdd^{m-k-1}(\Sigma)$ 
carries over to the pairing between $\hscdd^k(M)$ and $\htcdd^{m-k}(M)$ for each $k\in\{0,\dots,m-1\}$. 
For $k=m$, recalling Remark\ \ref{remSCCohoTriv} and\ Remark\ \ref{remTCCohoTriv}, 
we have $\hscdd^m(M)=\{0\}$ and $\htcdd^0(M)=\{0\}$, 
whence the pairing is obviously non degenerate in this case as well. 
\end{proof}

\begin{remark}\label{remSCTCPoincareDualityImproved}
Recalling Remark \ref{remPoincareDualityImproved}, one can obtain a partial positive result 
under slightly weaker hypotheses. In fact, even when $M$ fails to be of finite type, one can exhibit an isomorphism 
$\htcdd^{m-\ast}(M)\to\big(\hscdd^\ast(M)\big)^\ast$ defined by $[\beta]\mapsto\la\cdot,[\beta]\ra$ 
exploiting the map $\hdd^{m-1-\ast}(\Sigma)\to\big(\hcdd^\ast(\Sigma)\big)^\ast$ 
defined by $[\psi]\mapsto\la\cdot,[\psi]\ra$, which is always an isomorphism 
on account of Remark\ \ref{remPoincareDualityImproved}. 
We get similar results for the other pairings listed in \eqref{eqSCTCPairing}. 
\end{remark}

\subsection{Dynamics of the Hodge-d'Alembert operator}\label{subDynamics}
In this subsection certain features of the Hodge-d'Alembert partial differential operator are presented. 
In particular, we collect here some well-known facts about Green-hyperbolic differential operators 
on globally hyperbolic spacetimes specializing them to the case of interest. 
As general references we adopt\ \cite{BGP07, Fri75}. 
Note that some of the forthcoming statements rely on the extensions presented in\ \cite{Bar14, San13, Kha14b}. 

\begin{definition}\index{Hodge-d'Alembert operator}
Let $(M,g,\mfo)$ be an $m$-dimensional oriented Lorentzian manifold. 
The {\em Hodge-d'Alembert operator} $\Box:\f^k(M)\to\f^k(M)$ 
acting on differential forms of degree $k$ over $M$ is defined by $\Box=\de\dd+\dd\de$, 
$\dd$ and $\de$ being respectively the differential and the codifferential on $M$. 
\end{definition}

Notice that we are going to denote the Hodge-d'Alembert differential operator with the same symbol 
regardless of the degree of the form upon which it acts. 
Nevertheless, if relevant, it will be clear from the context which degree is considered. 

\index{differential operator!formally self-adjoint}It is fairly easy to check that $\Box$ is a {\em formally self-adjoint} 
differential operator with respect to the pairing $(\cdot,\cdot)$ defined in\ \eqref{eqPairing2} 
for $k$-forms with compact overlapping support. In fact, this directly follows from Stokes' theorem: 
For each $\alpha,\beta\in\f^k(M)$ with compact overlapping support, exploiting\ \eqref{eqAdjoint}, one reads: 
\begin{equation*}
(\Box\alpha,\beta)=(\de\dd\alpha,\beta)+(\dd\de\alpha,\beta)=(\alpha,\de\dd\beta)+(\alpha,\dd\de\beta)
=(\alpha,\Box\beta)\,.
\end{equation*}
\index{differential operator!normally hyperbolic}Furthermore, $\Box$ turns out to be {\em normally hyperbolic}. 
This amounts to say that $\Box$ is a second order linear differential operator 
whose principal symbol is of metric type. For further details, see\ \cite[Section\ 1.5]{BGP07}. 

We are interested in the dynamics which is described via hyperbolic partial differential equations such as 
the one expressed in terms of the Hodge-d'Alembert operator $\Box$ on a globally hyperbolic spacetime $M$. 
In fact, in this situation initial value problems ruled by a normally hyperbolic differential operator 
(hence $\Box$ in particular) are globally well-posed. This is the feature which motivates the physical interest 
towards both globally hyperbolic spacetimes and normally hyperbolic partial differential equations. 
More details about this topic can be found in\ \cite[Chapter 3]{BGP07}. 
\index{Green operators}For our aims, it is enough to mention that, being a normally hyperbolic differential operator 
on a globally hyperbolic spacetime $M$, $\Box$ admits unique {\em retarded and advanced Green operators} 
$G^+,G^-:\fc^k(M)\to\f^k(M)$ (as for $\Box$, we use the same symbol regardless of the degree $k$). 
These linear maps are uniquely specified by their properties, which are listed below: 
\begin{enumerate}
\item[1.] $\Box G^\pm\alpha=\alpha$ for each $\alpha\in\fc^k(M)$;
\item[2.] $G^\pm\Box\alpha=\alpha$ for each $\alpha\in\fc^k(M)$;
\item[3.] $\supp(G^\pm\alpha)\subseteq J_M^\pm\big(\supp(\alpha)\big)$ for each $\alpha\in\fc^k(M)$.
\end{enumerate}
The third property allows one to consider retarded and advanced Green operators 
as mapping to subspaces of forms with past spacelike compact and respectively future spacelike compact support: 
\begin{align*}
G^+:\fc^k(M)\to\fpsc^k(M)\,, && G^-:\fc^k(M)\to\ffsc^k(M)\,.
\end{align*}

Let us also mention that formal self-adjointness of $\Box$, together with the properties listed above, entails that 
the retarded and advanced Green operators $G^+$ and $G^-$ are formal adjoints of each other 
with respect to\ \eqref{eqPairing2}. Specifically, given $\alpha,\beta\in\fc^k(M)$, one has 
\begin{equation}\label{eqGreenAdjoint}
(G^\pm\alpha,\beta)=(G^\pm\alpha,\Box G^\mp\beta)=(\Box G^\pm\alpha,G^\mp\beta)=(\alpha,G^\mp\beta)\,.
\end{equation}

\index{causal propagator}Using the retarded and advanced Green operators $G^\pm$ for $\Box$ 
on a globally hyperbolic spacetime $M$, one can introduce the {\em causal propagator}:
\begin{equation}\label{eqCausalProp}
G=G^+-G^-:\fc^k(M)\to\fsc^k(M)\,.
\end{equation}
Notice that we used the observation that the retarded and advanced Green operators map compactly 
supported forms to forms with past and respectively future spacelike compact support to conclude that 
the causal propagator maps compactly supported forms to spacelike compact forms. 
The causal propagator $G$ for $\Box$ is formally antiself-adjoint, 
$G^+$ and $G^-$ being formal adjoints of each other. 
The role played by the causal propagator in describing the dynamics of $\Box$ is explained by the next theorem. 
Indeed, this is a general feature of any Green-hyperbolic differential operator 
on a globally hyperbolic spacetime.\footnote{In fact, Theorem\ \ref{thmCausalProp} relies 
only on the properties of retarded and advanced Green operators, whose existence is ensured 
by definition for Green hyperbolic differential operators, see\ \cite[Definition\ 3.2]{Bar14}.} 

\begin{theorem}\label{thmCausalProp}
Let $M$ be an $m$-dimensional globally hyperbolic spacetime. Consider the causal propagator 
$G:\fc^k(M)\to\fsc^k(M)$ for the Hodge-d'Alembert operator $\Box$ acting on $k$-forms. 
Then the complex below is actually an exact sequence: 
\begin{equation*}
0\lra\fc^k(M)\overset{\Box}{\lra}\fc^k(M)\overset{G}{\lra}\fsc^k(M)\overset{\Box}{\lra}\fsc^k(M)\lra0\,.
\end{equation*}
\end{theorem}

\begin{proof}
The proof can be found in\ \cite[Theorem\ 3.4.7]{BGP07}, except for surjectivity of $\Box:\fsc^k(M)\to\fsc^k(M)$. 
This last property is a byproduct of\ \cite[Theorem\ 3.2.11]{BGP07}. 
For a different proof, which holds true for any Green-hyperbolic differential operator, 
refer to\ \cite[Section\ 2.3]{Kha14b}. 
\end{proof}

As explained in\ \cite{Bar14, San13}, there exist unique linear extensions 
for the retarded and advanced Green operators: 
\begin{align*}
G^+:\fpc^k(M)\to\f^k(M)\,, && G^-:\ffc^k(M)\to\f^k(M)\,,
\end{align*}
where the subscripts $\mathrm{pc}$ and $\mathrm{fc}$ refer to 
past compact and respectively future compact supports. 
Such extensions are obtained mainly exploiting the support properties of $G^\pm$, 
see {\em e.g.}\ \cite[Theorem\ 3.8]{Bar14}. 
Let us stress that the unique extensions of the retarded and advanced Green operators exhibit properties 
analogous to the original retarded and advanced Green operators. 
We list only the properties for the extension of the retarded Green operator $G^+$. 
The case of $G^-$ can be obtained interchanging future and past: 
\begin{enumerate}
\item[1.] $\Box G^+\alpha=\alpha$ for each $\alpha\in\fpc^k(M)$;
\item[2.] $G^+\Box\alpha=\alpha$ for each $\alpha\in\fpc^k(M)$;
\item[3.] $\supp(G^+\alpha)\subseteq J_M^+\big(\supp(\alpha)\big)$ for each $\alpha\in\fpc^k(M)$.
\end{enumerate}
In particular, $G^+$ maps $\fpc^k(M)$ to itself, while $G^-$ maps $\ffc^k(M)$ to itself. 
Furthermore, \eqref{eqGreenAdjoint} is now extended to those $\alpha$ having past or future compact support 
and those $\beta$ having respectively past or future spacelike compact support. 
As a consequence, we also get a unique extension of the causal propagator $G$, 
which can be now defined for $k$-forms with timelike compact support 
as the difference between the extensions of retarded and advanced Green operators: 
\begin{equation*}
G=G^+-G^-:\ftc^k(M)\to\f^k(M)\,.
\end{equation*}
An extended version of Theorem\ \ref{thmCausalProp} holds true for the extended causal propagator as well. 

\begin{theorem}\label{thmExtCausalProp}
Let $M$ be an $m$-dimensional globally hyperbolic spacetime. Consider the extended causal propagator 
$G:\ftc^k(M)\to\f^k(M)$ for the Hodge-d'Alembert operator $\Box$ acting on $k$-forms. 
Then the complex below is actually an exact sequence: 
\begin{equation*}
0\lra\ftc^k(M)\overset{\Box}{\lra}\ftc^k(M)\overset{G}{\lra}\f^k(M)\overset{\Box}{\lra}\f^k(M)\lra0\,.
\end{equation*}
\end{theorem}

\begin{proof}
The proof, up to surjectivity of the last non-trivial arrow, is obtained from\ \cite[Theorem\ 4.3]{Bar14} 
taking only smooth sections into account, see\ also\ \cite[Theorem\ 3,8]{Bar14}. 
Surjectivity of $\Box:\f^k(M)\to\f^k(M)$ turns out to be a byproduct of\ \cite[Section\ 3.5.3, Corollary\ 5]{BF09}. 
A slightly more general proof, which holds true for any Green hyperbolic differential operator, 
can be obtained extending the argument in\ \cite[Section\ 2.3]{Kha14b}: 
Let $\omega\in\f^k(M)$ and consider a partition of unity $\{\chi_+,\chi_-\}$ on $M$ 
such that $\chi_\pm$ is past/future compact. Therefore $\zeta=G^+(\chi_+\omega)+G^-(\chi_-\omega)$, 
defined using the extended retarded and advanced Green operators, is such that $\Box\zeta=\omega$. 
\end{proof}

Intertwining properties of $\dd$ and $\de$ with $\Box$, as well as their consequences 
on the corresponding retarded and advanced Green operators, can be found in\ \cite[Proposition\ 2.1]{Pfe09}. 
For the sake of completeness, we provide a proof of these results for the extended Green operators. 
Note that we will adopt an approach slightly different from\ \cite{Pfe09}. 

\begin{proposition}\label{prpIntertwiners}
Let $M$ be an $m$-dimensional globally hyperbolic spacetime 
and denote with $G^+:\fpc^k(M)\to\f^k(M)$ and with $G^-:\ffc^k(M)\to\f^k(M)$
the extended retarded and advanced Green operators for the Hodge-d'Alembert operator $\Box$. 
The following identities hold true for all $\alpha\in\f^k(M)$, $\beta\in\fpc^k(M)$ and $\gamma\in\ffc^k(M)$: 
\begin{align*}
\dd\Box\alpha & =\Box\dd\alpha\,, & \de\Box\alpha & =\Box\de\alpha\,,\\
\dd G^+\beta & =G^+\dd\beta\,, & \de G^+\beta & =G^+\de\beta\,,\\
\dd G^-\gamma & =G^-\dd\gamma\,, & \de G^-\gamma & =G^-\de\gamma\,.
\end{align*}
\end{proposition}

\begin{proof}
The formulas involving $\Box=\de\dd+\dd\de$ are immediate consequences of its definition 
and of the identities $\dd\dd=0$ and $\de\de=0$. 

We prove only the formula involving $\dd$ and $G^+$ since all others can be obtained along the same lines. 
Consider $\beta\in\fpc^k(M)$ and $\omega\in\fc^{k+1}(M)$. 
Exploiting the properties of the retarded and advanced Green operators, formal self-adjointness of $\Box$ 
and the fact that $\dd$ intertwines $\Box:\f^k(M)\to\f^k(M)$ and $\Box:\f^{k+1}(M)\to\f^{k+1}(M)$, 
one obtains the following chain of identities: 
\begin{align*}
(\dd G^+\beta,\omega) & =(\dd G^+\beta,\Box G^-\omega)=(\Box\dd G^+\beta,G^-\omega)\\
& =(\dd\Box G^+\beta,G^-\omega)=(\dd\beta,G^-\omega)=(G^+\dd\beta,\omega)\,.
\end{align*}
Since $\omega\in\fc^{k+1}(M)$ can be chosen arbitrarily and the pairing 
$(\cdot,\cdot):\f^{k+1}(M)\times\fc^{k+1}(M)\to\bbR$ defined in\ \eqref{eqPairing2} is non-degenerate, 
we conclude that $\dd G^+\beta=G^+\dd\beta$. 
\end{proof}

\begin{remark}\label{remTrivCohom}
At the beginning of Subsection\ \ref{subRestrictedCohom} we anticipated that 
the de Rham complex on a globally hyperbolic spacetime $M$ provides trivial cohomology groups 
when restricted to forms with support which is either past/future compact or past/future spacelike compact. 
In fact, this follows from Proposition\ \ref{prpIntertwiners}, together with the properties of 
retarded and advanced Green operators for $\Box$. Consider for example $\omega\in\fpc^k(M)$ 
such that $\dd\omega=0$, namely a closed form with past compact support. 
Introducing the past compact $(k-1)$-form $\eta=G^+\de\omega$, 
we conclude that 
\begin{equation*}
\dd\eta=\dd G^+\de\omega=G^+\dd\de\omega=G^+\Box\omega=\omega\,.
\end{equation*}
This shows that $\omega$ is an exact form in the past compact sense too, 
therefore $\mathrm{H}_{\mathrm{pc}\,\dd}^k(M)=\f_{\mathrm{pc}\,\dd}^k(M)/\dd\fpc^{k-1}(M)=\{0\}$. 
For the support system $\mathrm{psc}$ (see Definition\ \ref{defSupportSystems}) one can argue similarly, 
while for the support systems $\mathrm{fc}$ and $\mathrm{fsc}$ one comes to the same conclusions 
by reversing the time-orientation, see\ Chapter\ \ref{chMaxwell} for example. 
\end{remark}

We conclude the present subsection discussing certain naturality properties of the Hodge-d'Alembert operator 
and the corresponding Green functions. This will turn out useful when discussing functors 
describing classical or quantum field theories. 

\begin{remark}[Naturality of $\dd$, $\de$ and $\Box$]\label{remdddeBoxNaturality}
\index{vector space!category}Denote with $\Vec$ the {\em category of $\bbR$-vector spaces}, 
whose objects are vector spaces over the field of real numbers 
and whose morphisms are $\bbR$-linear maps between $\bbR$-vector spaces. 
Given a causal embedding $f:M\to N$ between $m$-dimensional globally hyperbolic spacetimes, 
for each degree $k\in\{0,\dots,m\}$ we can consider both the pull-back $f^\ast:\f^k(N)\to\f^k(M)$ 
and the push-forward $f_\ast:\fc^k(M)\to\fc^k(N)$. Therefore we can think of $\f^\ast(\cdot):\GHyp\to\Vec$ 
and $\fc^\ast(\cdot):\GHyp\to\Vec$ respectively as a contravariant and a covariant functor. 
It is well-known that the pull-back and the push-forward along $f$ intertwine the differentials on $M$ and on $N$. 
Furthermore, since $f$ is an isometry compatible with orientations, both $f^\ast:\f^k(N)\to\f^k(M)$ 
and $f_\ast:\fc^k(M)\to\fc^k(N)$ intertwine the Hodge star operators on $M$ and on $N$ as well. 
For any value of $k\in\{0,\dots,m\}$, we collect these statements in the following equations: 
\begin{align*}
f^\ast\dd_N & =\dd_Mf^\ast & \mbox{on}\;\f^k(N)\,, &&&& f_\ast\dd_M & =\dd_Nf_\ast 
& \mbox{on}\;\fc^k(M)\,,\\
f^\ast\ast_N & =\ast_Mf^\ast & \mbox{on}\;\f^k(N)\,, &&&& f_\ast\ast_M & =\ast_Nf_\ast 
& \mbox{on}\;\fc^k(M)\,,
\end{align*}
where the subscripts are introduced to specify to which manifold the differentials and the Hodge stars are referred. 
These two properties together entail that the pull-back and the push-forward along $f$ intertwine 
the codifferentials for $M$ and $N$ too, and hence the Hodge-d'Alembert operators as well, namely 
\begin{align*}
f^\ast\de_N & =\de_Mf^\ast & \mbox{on}\;\f^k(N)\,, &&&& f_\ast\de_M&=\de_Nf_\ast 
& \mbox{on}\;\fc^k(M)\,,\\
f^\ast\Box_N & =\Box_Mf^\ast & \mbox{on}\;\f^k(N)\,, &&&& f_\ast\Box_M&=\Box_Nf_\ast 
& \mbox{on}\;\fc^k(M)\,.
\end{align*}
In the language of category theory, we have four natural transformations between contravariant functors: 
\begin{align*}
\dd:\f^\ast(\cdot)\to\f^{\ast+1}(\cdot)\,, && \ast:\f^\ast(\cdot)\to\f^\ast(\cdot)\,,\\
\de:\f^\ast(\cdot)\to\f^{\ast-1}(\cdot)\,, && \Box:\f^\ast(\cdot)\to\f^\ast(\cdot)\,,
\end{align*}
as well as four natural transformations between covariant functors: 
\begin{align*}
\dd:\fc^\ast(\cdot)\to\fc^{\ast+1}(\cdot)\,, && \ast:\fc^\ast(\cdot)\to\fc^\ast(\cdot)\,,\\
\de:\fc^\ast(\cdot)\to\fc^{\ast-1}(\cdot)\,, && \Box:\fc^\ast(\cdot)\to\fc^\ast(\cdot)\,.
\end{align*}
\end{remark}

\begin{proposition}\label{prpGreenNat}
Let $M$ and $N$ be $m$-dimensional globally hyperbolic spa\-ce\-ti\-mes 
and consider a causal embedding $f:M\to N$. 
Denote with $G^+_M,G^-_M:\fc^\ast(M)\to\f^\ast(M)$ the retarded and advanced Green operators 
for the Hodge-d'Alembert operator $\Box_M$ on $M$ and with $G^+_N,G^-_N:\fc^\ast(M)\to\f^\ast(M)$ 
the retarded and advanced Green operators for the Hodge-d'Alembert operator $\Box_N$ on $N$. Then 
\begin{align*}
f^\ast G^\pm_Nf_\ast=G^\pm_M\;\mbox{on}\;\fc^k(M)\,, && k\in\{0,\dots,m\}\,.
\end{align*}
In particular, the same property holds true for the causal propagators $G_M$ and $G_N$. 
\end{proposition}

\begin{proof}
We can check directly that $f^\ast G^\pm_Nf_\ast$ satisfies all the defining properties 
for being a retarded/advanced Green operator for $\Box_M$. This fact follows from $G^\pm_N$ 
being the retarded/advanced Green operator for $\Box_N$ and $f$ being a causal embedding 
(in particular, we exploit both causal compatibility and the fact that the time-orientation is preserved). 
Since the retarded/advanced Green operator $G^\pm_M$ for $\Box_M$ is unique, the statement is proved. 
\end{proof}

\section{Affine spaces and affine bundles}\label{secAffine}
In the present section we briefly recall the material developed in\ \cite{BDS14a}, 
which will be employed in\ Section\ \ref{secPrBun} and especially in\ Chapter\ \ref{chYangMills}. 
Starting from the definition of an affine space and of its vector dual, 
we will introduce affine bundles and the corresponding vector dual bundles. As an example, 
we will present an affine bundle describing fiberwise splittings of an exact sequence of vector bundles. 
A special case of this construction plays a central role in\ Subsection\ \ref{subConnBun}. 
We will conclude introducing affine differential operators 
and we will address the problem of specifying their formal duals. 

\begin{definition}\label{defAffSp}\index{affine!space}\index{affine!map}\index{affine!map!linear part}
\index{linear part!affine map}\index{affine!space!category}
Let $V$ be a vector space. An {\em affine space} $A$ modeled on the vector space $V$ 
is a set endowed with a free and transitive right action $+:A\times V\to A$ of the group $V$ on $A$, 
where $V$ is regarded as an Abelian group with respect to addition of vectors. 

Let $A$ and $B$ be affine spaces modeled respectively on the vector spaces $V$ and $W$. 
An {\em affine map} $f:A\to B$ is a map between the underlying sets such that 
there exists a linear map $f_V:V\to W$, called {\em linear part} of $f$, 
which fulfils the following requirement: 
\begin{align*}
f(a+v)=f(a)+f_V(v)\,, && \forall\,a\in A\,,\;v\in V\,,
\end{align*}
where $+$ denotes the right action of $V$ on $A$ in the left-hand-side, 
while on the right-hand-side the same symbol denotes the right action of $W$ on $B$. 

The {\em category of affine spaces} $\Aff$ has affine spaces as objects and affine maps as morphisms. 
\end{definition}

Note that the linear part of an affine map is unique due to the fact that 
the right action on an affine space of the underlying vector space is free. 

\begin{remark}\label{remAffDiff}
Let $A$ be an affine space modeled on the vector space $V$. 
Since the right action of $V$ on $A$ is both free and transitive, 
for each $a,b\in A$ there exists a unique $v\in V$ such that $a+v=b$. 
This property entails that the following map is well defined: 
\begin{align*}
-:A\times A\to V\,, && (a,b)\mapsto v\,,
\end{align*}
where $v\in V$ is such that $a+v=b$. 
\end{remark}

In the next example, and in the following as well, we will use the symbol $\hom(V,W)$ 
to denote the vector space of linear maps between the vector spaces $V$ and $W$. 
Furthermore, given two linear maps $f:U\to V$ and $h:W\to Z$, 
define $\hom(f,h):\hom(V,W)\to\hom(U,Z)$ to be the linear map $g\mapsto h\circ g\circ f$. 
Then $\hom(\cdot,\cdot):\Vec^\op\times\Vec\to\Vec$ becomes an internal $\hom$ functor. 

\begin{example}\label{exaAffSp}
A trivial example of an affine space $A$ is provided by any vector space $V$. 
As sets $A=V$ and we endow $A$ with the right group action induced by addition of vectors. 
Basically, $A$ is the same as the vector space $V$ where one forgets about the origin. 
In this case one says that $A$ is the vector space $V$ regarded as an affine space modeled on itself. 

A less trivial example comes from short exact sequences of vector spaces. 
Consider the following short exact sequence, where $V$, $W$ and $Z$ are vector spaces: 
\begin{equation*}
0\lra V\overset{\iota}{\lra}W\overset{\pi}{\lra}Z\lra0\,.
\end{equation*}
It is well-known that short exact sequences of vector spaces are always split exact, 
namely there always exists a linear map $\rho:Z\to W$, called right splitting, such that $\pi\rho=\id_Z$ 
or equivalently a linear map $\lambda:W\to V$, called left splitting, such that $\lambda\iota=\id_V$. 
For example, let us consider the set $R$ of right splittings of the exact sequence displayed above. 
Clearly, $R$ is a subset of $\hom(Z,W)$ and, given $\rho\in R$ and $\phi\in\hom(Z,V)$, 
$\rho+\iota\phi\in\hom(W,Z)$ still lies in $R$. Furthermore, two right splittings always differ 
by a linear map $Z\to V$. Consider $\rho,\rho^\prime\in R$ and note that $\pi(\rho^\prime-\rho)=0$. 
Therefore $\rho^\prime-\rho\in\hom(Z,W)$ factors through $\iota$, the sequence above being exact. 
This means that the right action of the Abelian group $\hom(Z,V)$ 
on the set $R\subseteq\hom(Z,W)$ defined below is transitive: 
\begin{align*}
+:R\times\hom(W,V)\to R\,, && (\rho,\phi)\mapsto\rho+\iota\phi\,,
\end{align*}
Since $\iota$ is injective, this right action is also free, hence we conclude that 
$R$ is an affine space modeled on the vector space $\hom(Z,V)$. 
\end{example}

The next definition introduces the notion of the vector dual of an affine space. 

\begin{definition}\label{defAffSpDual}\index{affine!space!vector dual}\index{vector dual of an affine space}
Let $A$ be an affine space modeled on the vector space $V$. 
The {\em vector dual} $A^\dagger$ of the affine space $A$ is the vector space of real valued affine maps on $A$. 
\end{definition}

\begin{remark}\label{remAffSpDualDim}
$A^\dagger$ inherits the structure of a vector space from $\bbR$. 
Indeed, if $n$ is the dimension of the underlying vector space $V$, then $\dim(A^\dagger)=n+1$. 
In fact, consider a basis $\{v^\ast_1,\dots,v^\ast_n\}$ of the dual $V^\ast$ of $V$. 
Choosing an arbitrary point $\tilde{a}\in A$, we can introduce $\varphi_i:A\to\bbR$ 
by setting $\varphi_i(\tilde{a}+v)=v_i^\ast(v)$ for each $v\in V$. 
Clearly, $\{\varphi_1,\dots,\varphi_n\}$ is a family of $n$ linearly independent elements of $A^\dagger$. 
Furthermore, consider the function $\ttone:A\to\bbR$ which takes the constant value $1$. 
Indeed $\ttone\in A^\dagger$ and we can express each $\varphi\in A^\dagger$ 
as a linear combination of $\{\ttone,\varphi_1,\dots,\varphi_n\}$ 
noting that the linear part $\varphi_V$ lies in $V^\ast$ 
and that $\varphi(\tilde{a}+v)=\varphi(\tilde{a})+\varphi_V(v)$ for each $v\in V$. 
It remains only to check that $\{\ttone,\varphi_1,\dots,\varphi_n\}$ are linearly independent in $A^\dagger$. 
This follows from the fact that $\ttone$ is constant and $\{\varphi_1,\dots,\varphi_n\}$ are linearly independent. 
\end{remark}

We turn our attention to bundles whose fibers are affine spaces in such a way 
that the affine structures of neighboring fibers agree in a suitable sense. 
Our definition mimics the one of\ \cite[Section\ 6.22]{KMS93}. 

\begin{definition}\label{defAffBun}\index{affine!bundle}\index{affine!bundle map}
Let $\pi_V:V\to M$ be a vector bundle with typical fiber $\bbR^n$ 
and consider an affine space $S$ modeled on the vector space $\bbR^n$. 
An {\em affine bundle} $\pi_A:A\to M$ with typical fiber $S$ modeled on $V$ 
is a fiber bundle $A$ satisfying the following requirements: 
\begin{enumerate}
\item The fiber over any point $x\in M$ is an affine space modeled on the corresponding fiber of $V$; 
\item For each $x\in M$, there exists an affine bundle trivialization in a neighborhood of $x$. 
This consists of an open neighborhood $U\subseteq M$ of $x$, 
a vector bundle trivialization $\psi_V:\pi_V^{-1}(U)\to U\times\bbR^n$
and a fiber-preserving diffeomorphism $\psi:\pi_A^{-1}(U)\to U\times S$, 
namely $\pr_1\circ\psi=\pi$ on $\pi_A^{-1}(U)$, where $\pr_1:U\times S\to U$ denotes 
the projection on the first factor. The triple $(U,\psi_V,\psi)$ is such that, 
for each $y\in U$, the restriction $\psi\vert_y:A_y\to \{y\}\times S$ is an isomorphism of affine spaces 
whose linear part coincides with $\psi_V\vert_y:V_y\to\{y\}\times\bbR^n$. 
\end{enumerate}

Let $\pi_A:A\to M$ and $\pi_B:B\to N$ be affine bundles with typical fibers $S$ and respectively $T$ 
modeled on the vector bundles $\pi_V:V\to M$ and respectively $\pi_W:W\to M$. 
An {\em affine bundle map} $f:A\to B$ is a bundle map covering $\ul{f}:M\to N$ such that, 
for each $x\in M$, the restriction of $f$ to the fiber $A_x$ is an affine map $f\vert_x:A_x\to B_{\ul{f}(x)}$. 
\end{definition}

\begin{remark}[Linear part of an affine bundle map]\label{remAffBunMapLinPart}
\index{affine!bundle map!linear part}\index{linear part!affine bundle map}
The existence and uniqueness of the linear part of any affine map, {\em cfr.}\ Definition\ \ref{defAffSp}, 
entails the existence and uniqueness of a vector bundle map 
which can be regarded as the linear part of an affine bundle map. 
Let $f:A\to B$ be an affine bundle map between the affine bundles $A$ and $B$ 
modeled on the vector bundles $V$ and respectively $W$. Since on each fiber $f$ provides an affine map, 
we can uniquely specify a vector bundle map $f_V:V\to W$ by saying that 
its restriction to each fiber of $V$ must be the linear part of the restriction of $f$ to the corresponding fiber of $A$. 
$f_V:V\to W$ is also called the {\em linear part} of the affine bundle map $f:A\to B$. 
\end{remark}

We present two examples of affine bundles. We will adopt the following notation: 
Given two vector bundles $V,W$ over $M$, we denote with $\hom(V,W)$ the corresponding $\hom$-bundle, 
namely the vector bundle over $M$ whose fiber over $x\in M$ is given by 
the vector space $\hom(V_x,W_x)$ of linear maps between the fibers over $x$ of $V$ and $W$. 
Note that the vector bundle structure of $\hom(V,W)$ is inherited from $V$ and $W$. 
Furthermore, given two bundle maps $f:U\to V$ and $h:W\to Z$ covering the identity, 
define $\hom(f,h):\hom(V,W)\to\hom(U,Z)$ to be the vector bundle map $\gamma\mapsto h\circ\gamma\circ f$. 
Then $\hom(\cdot,\cdot):\VBun_M^\op\times\VBun_M\to\VBun_M$ 
becomes a bi-functor on the category of vector bundles $\VBun_M$ over a fixed manifold $M$, 
which is contravariant in the first entry and covariant in the second. 

\begin{example}\label{exaAffBun}
In full analogy with\ Example\ \ref{exaAffSp}, each vector bundle can be regarded as an affine bundle 
modeled on itself. In fact, given a vector bundle $V$ over $M$, we can set $A=V$ only as fiber bundles. 
Obviously, each fiber of $A$ can be regarded as an affine space modeled on the corresponding fiber of $V$ 
and indeed the typical fiber of $A$ is the typical fiber of $V$ regarded as an affine space modeled on itself. 
It is straightforward to check that vector bundle trivializations of $V$ provide affine bundle trivializations of $A$. 
Basically, to define $A$, we are forgetting of the origin of each fiber of $V$. 

Let $V$, $W$ and $Z$ be vector bundles over $M$ of rank $a$, $b$ and respectively $c=b-a$ 
forming the following short exact sequence in $\VBun_M$: 
\begin{equation*}
0\lra V\overset{\iota}{\lra}W\overset{\pi}{\lra}Z\lra0\,.
\end{equation*}
Local triviality of vector bundles, together with the existence of a partition of unity for the base space $M$, 
entails that the property of short exact sequences of vector spaces to be always split exact 
carries over to short exact sequences of vector bundles over a fixed manifold $M$. 
Hence, there always exists a vector bundle map $\rho:Z\to W$ covering $\id_M$ such that $\pi\circ\rho=\id_Z$. 
Note that we can equivalently regard $\rho$ as a section of the vector bundle $\hom(Z,W)$. 
At the moment we would like to better understand the structure of 
the bundle $R$ over $M$ of right splittings of the short exact sequence displayed above, 
whose typical fiber at $x\in M$ is given by the affine space $R_x$ of right splittings 
of the sequence restricted to the base point $x$, {\em cfr.}\ Example\ \ref{exaAffSp}. 
Indeed, we note that the typical fibers of the vector bundles fit 
into the short exact sequence of vector spaces $0\to\bbR^a\to\bbR^b\to\bbR^c\to0$. 
We denote by $S$ the unique (up to isomorphisms) affine space of right splittings of this sequence 
and we observe that $S$ is modeled on the vector space $\hom(\bbR^c,\bbR^a)$, 
{\em cfr.}\ Example\ \ref{exaAffSp}. 
Taking into account local trivializations $\alpha:\pi_V^{-1}(U)\to U\times\bbR^a$, 
$\beta:\pi_W^{-1}(U)\to U\times\bbR^b$ and $\gamma:\pi_Z^{-1}(U)\to U\times\bbR^c$
of the vector bundles $V$, $W$ and respectively $Z$ on a neighborhood $U$ of $x\in M$, 
we can construct an affine bundle trivialization over $U$ of the bundle $R$ of right splittings 
of the short exact sequence of vector bundles displayed above. 
First of all, note that the local trivialization $\hom(\gamma^{-1},\beta)$ over $U$ of the vector bundle $\hom(Z,W)$ 
maps $\pi_R^{-1}(U)$ to $U\times S$, thus providing a candidate $\phi:\pi_R^{-1}(U)\to U\times S$ 
for an affine bundle trivialization of $R$ over $U$. 
Introducing also the local trivialization $\hom(\gamma^{-1},\alpha)$ over $U$ of the vector bundle $\hom(Z,V)$, 
we realize that its restriction to each fiber over $x\in U$ fulfils the condition 
to be the linear part of the restriction to the fiber over $x$ of $\phi$, 
whence $\phi$ is an affine isomorphism on each fiber. 
This shows that $A$ is an affine bundle over $M$ modeled on the vector bundle $\hom(Z,V)$ over $M$ 
whose typical fiber $S$ is an affine space modeled on the vector space $\hom(\bbR^c,\bbR^a)$. 
\end{example}

In analogy with\ Definition\ \ref{defAffSpDual}, we introduce now the vector dual of an affine vector bundle. 

\begin{definition}\index{affine!bundle!vector dual}\index{vector dual of an affine bundle}
Let $A$ be an affine bundle over $M$ modeled on the vector bundle $V$ over $M$ 
with typical fiber $S$ modeled on the typical fiber $\bbR^n$ of $V$. 
The {\em vector dual} $A^\dagger$ of the affine bundle $A$ is 
the vector bundle over $M$ with typical fiber $S^\dagger$ 
whose fiber over each point $x\in M$ is given by the vector dual of the affine space $A_x$. 
\end{definition}

Local trivializations of $A^\dagger$ can be obtained from local trivializations of $A$. Let $x\in M$ 
and consider a local trivialization $\psi:\pi_A^{-1}(U)\to U\times S$ of $A$ over a neighborhood $U$ of $x$. 
It is easy to check that $\hom(\psi^{-1},\id_{U\times\bbR}):\pi_{A^\dagger}^{-1}(U)\to U\times S^\dagger$ 
defines a local trivialization of $A^\dagger$ over $U$. 
Note that $\rank(A^\dagger)=\rank(V)+1$ as it follows from\ Remark\ \ref{remAffSpDualDim}. 

\subsection{Sections of an affine bundle}
The space $\sect(A)$ of sections of an affine bundle $A$ over $M$ modeled on the vector bundle $V$ 
can be endowed with the affine structure naturally induced by the affine structure on $A$. 
We have a right action $+:\sect(A)\times\sect(V)\to\sect(A)$ of the Abelian group $\sect(V)$ 
on the set $\sect(A)$ obtained by pointwise translation via the affine structure of $A$. 
This action is free and transitive since the fiberwise right action of $V$ on $A$ is such. 
We will always endow $\sect(A)$ with this affine structure and thus regard it as an affine space. 

At each point $x\in M$, a sections $\phi$ of the vector dual $A^\dagger$ of an affine bundle $A$ over $M$ 
gives an affine map $\phi(x):A_x\to\bbR$. Taking the linear part of $\phi(x)$ at each point $x\in M$, 
we obtain a section $\phi_V\in\sect(V^\ast)$ of the dual of the vector bundle $V$ 
underlying the affine structure of $A$. $\phi_V$ is the {\em linear part} of the section $\phi$. 
Notice that this construction amounts to the application of\ Remark\ \ref{remAffBunMapLinPart} 
to the affine bundle map $A\to M\times\bbR$ corresponding to the section $\phi\in\sect(A^\dagger)$, 
which is defined by $a\in A_x\mapsto(x,\phi(x)(a))$ for each $x\in M$. 

There is a distinguished section $\bfone\in\sect(A^\dagger)$ of the vector dual bundle of $A$ 
which is defined imposing that $\bfone$ is the affine map $a\in A_x\mapsto 1\in\bbR$ for each point $x\in M$. 

\begin{remark}\label{remAffDualSectDecomp}
A choice of a section $\tilde{\alpha}\in\sect(A)$ allows us to decompose each section $\phi\in\sect(A^\dagger)$ 
in terms of a linear combination with $\c(M)$-coefficients of its linear part $\phi_V\in\sect(V^\ast)$ 
and $\bfone\in\sect(A^\dagger)$. In fact, consider the real valued smooth function $\phi(\tilde{\alpha})$ on $M$ 
whose value at $x\in M$ is given by the evaluation of $\phi(x)\in A^\dagger_x$ on $\tilde{\alpha}(x)\in A_x$. 
Then the following identity follows from the fact that $\phi_V$ is the linear part of $\phi$ 
and provides the sought decomposition: 
\begin{equation*}
\phi=\phi(\tilde{\alpha})\,\bfone+\phi_V(\cdot-\tilde{\alpha})\,,
\end{equation*} 
where $-:\sect(A)\times\sect(A)\to\sect(V)$ is defined along the lines of\ Remark\ \ref{remAffDiff}. 
\end{remark}

The section $\bfone\in\sect(A^\dagger)$ is responsible for the degeneracy in the left argument 
of the integral pairing between $\sc(A^\dagger)$ and $\sect(A)$: 
\begin{align}\label{eqAffPairing}
\sc(A^\dagger)\times\sect(A)\to\bbR\,, && (\phi,\alpha)\mapsto\int_M\phi(\alpha)\,\vol\,,
\end{align}
where $M$ is now assumed to be an $m$-dimensional oriented manifold, $\vol\in\f^m(M)$ is a volume form on $M$ 
and $\phi(\alpha)\in\c(M)$ has compact support since $\supp(\phi)$ is compact. 

\begin{proposition}\label{prpAffSeparability}
Let $M$ be an oriented manifold and consider a volume form $\vol\in\f^m(M)$. 
Furthermore, let $A$ be an affine bundle over $M$ modeled on the vector bundle $V$ 
and consider the integral pairing between $\sc(A^\dagger)$ and $\sect(A)$ introduced in\ \eqref{eqAffPairing}. 
Then $\sc(A^\dagger)$ separates points in $\sect(A)$ via this pairing, 
while $\sect(A)$ does not separate points in $\sect(A^\dagger)$. 
More specifically, $\phi\in\sect(A^\dagger)$ is such that $\int_M\phi(\alpha)\,\vol=0$ 
for each $\alpha\in\sect(A)$ if and only if there exists $f\in\cc(M)$ 
such that $\int_Mf\,\vol=0$ and $f\,\bfone=\phi$. 
\end{proposition}

\begin{proof}
For the first part of the statement, take $\alpha,\beta\in\sect(A)$ such that 
\begin{align*}
\int_M\phi(\alpha)\,\vol=\int_M\phi(\beta)\,\vol\,, && \forall\,\phi\in\sc(A^\dagger)\,.
\end{align*} 
For each compactly supported section $\eta\in\sc(V^\ast)$ of the dual of the vector bundle $V$, 
consider $\phi\in\sc(A^\dagger)$ defined by $\phi(\alpha+\mu)=\eta(\mu)$ for each $\mu\in\sect(A)$. 
From the hypothesis on $\alpha$ and $\beta$, it follows that $\int_M\eta(\nu)\,\vol=0$ for all $\eta\in\sc(V^\ast)$, 
where $\nu\in\sect(V)$ is the unique section fulfilling $\alpha+\nu=\beta$. 
It follows that $\nu$ vanishes and therefore $\alpha=\beta$. 

For the second part, suppose that $\phi\in\sc(A^\ast)$ is such that 
$\int_M\phi(\alpha)\,\vol$ vanishes for all $\alpha\in\sect(A)$. 
Let us fix $\tilde{\alpha}\in\sect(A)$. Exploiting Remark \ref{remAffDualSectDecomp}, 
we deduce that $\phi=\phi(\tilde{\alpha})\,\bfone+\phi_V(\cdot-\tilde{\alpha})$. 
Therefore our assumption on $\phi$ implies that $\int_M\phi(\tilde{\alpha})\,\vol$ vanishes 
and moreover $\int_M\phi_V(\mu)\,\vol=0$ for all $\mu\in\sect(V)$. 
It follows that $\phi_V=0$ and therefore we found $f=\phi(\tilde{\alpha})\in\cc(M)$ 
such that $\int_Mf\,\vol=0$ and $f\,\bfone=\phi$. Clearly, the converse is true as well: 
For each $\phi=f\,\bfone\in\sc(A^\dagger)$ with $f\in\cc(M)$ such that $\int_Mf\,\vol=0$, 
the pairing in\ \eqref{eqAffPairing} vanishes upon evaluation on $\phi$ and any section of $A$. 
This concludes the proof. 
\end{proof}

\begin{remark}[Trivial sections of the vector dual of an affine bundle]\label{remAffTriv}
We can recover separability in both arguments of the pairing\ \eqref{eqAffPairing} 
taking a suitable quotient on $\sc(A^\dagger)$. 
More precisely, consider the vector subspace of $\sc(A^\dagger)$ defined below: 
\begin{equation}\label{eqTriv}
\triv=\left\{f\,\bfone\in\sc(A^\dagger)\,:\;f\in\cc(M)\,,\;\int_Mf\,\vol=0\right\}\,.
\end{equation}
We regard $\triv$ as the space of \quotes{trivial} section of $A^\dagger$ since these sections 
do not enable us to extract any information about sections of $A$ via the pairing\ \eqref{eqAffPairing}. 
On account of\ Proposition\ \ref{prpAffSeparability}, the pairing of\ \eqref{eqAffPairing} automatically 
descends to a non-degenerate pairing between $\sc(A^\dagger)/\triv$ and $\sect(A)$. 
Let us stress that the quotient by $\triv$ does not affect the linear part of a section in $\sc(A^\dagger)$. 
More precisely, all representatives of an element in the quotient $\sc(A^\dagger)/\triv$ have the same linear part. 
\end{remark}

\subsection{Affine differential operators}\label{subAffDiffOp}
In view of\ Chapter\ \ref{chYangMills}, we must face the problem of extending 
the notion of a linear differential operator to the case where the source is the space of sections of an affine bundle. 
While the extension to affine bundles is straightforward, 
it turns out that the formal dual of an affine differential operator is not uniquely defined. 
Yet, Proposition\ \ref{prpAffSeparability} will enable us to cure this ambiguity. 
Let us mention that further details about the results presented here can be found in\ \cite{BDS14b}. 

\begin{definition}\label{defAffDiffOp}\index{affine!differential operator}\index{differential operator!affine}
Let $V$ and $W$ be vector bundles over $M$ and consider an affine bundle $A$ over $M$ modeled on $V$. 
An {\em affine differential operator} $P:\sect(A)\to\sect(W)$ is an affine map 
whose linear part $P_V:\sect(V)\to\sect(W)$ is a linear differential operator in the usual sense. 
\end{definition}

\label{dual of an affine differential operator}\index{affine!differential operator!dual}Suppose 
$M$ is an $m$-dimensional oriented manifold and $\vol\in\f^m(M)$ is a volume form on $M$. 
Let $V$ and $W$ be vector bundles over $M$ and consider an affine bundle $A$ over $M$ modeled on $V$. 
To define the formal dual of an affine differential operator $P:\sect(A)\to\sect(W)$, 
we look for a linear differential operator $P^\ast:\sect(W^\ast)\to\sect(A^\dagger)$ 
fulfilling the condition stated below: 
\begin{align}\label{eqAffDiffOpDual}
\int_M\left(P^\ast(\nu)\right)(\alpha)\,\vol=\int_M\nu\left(P(\alpha)\right)\,\vol\,, && 
\forall\,\nu\in\sc(W^\ast)\,,\;\forall\,\alpha\in\sect(A)\,.
\end{align}

\begin{theorem}\label{thmAffDiffOpDual}
Suppose $M$ is an $m$-dimensional oriented manifold and $\vol\in\f^m(M)$ is a volume form on $M$. 
Let $V$ and $W$ be vector bundles over $M$ and consider an affine bundle $A$ over $M$ modeled on $V$. 
Furthermore, consider an affine differential operator $P:\sect(A)\to\sect(W)$. 
Then a linear differential operator $P^\ast:\sect(W^\ast)\to\sect(A^\dagger)$ 
fulfilling the condition stated in\ \eqref{eqAffDiffOpDual} for the affine differential operator $P$ exists. 

Let $P^\ast,\tilde{P}^\ast:\sect(W^\ast)\to\sect(A^\dagger)$ be linear differential operators 
fulfilling the condition mentioned above. Then there exists a linear differential operator $Q:\sect(W^\ast)\to\c(M)$ 
such that $Q(\cdot)\,\bfone=\tilde{P}^\ast-P^\ast$ and $\int_MQ(\phi)\,\vol=0$ for all $\phi\in\sc(A^\dagger)$. 
\end{theorem}

\begin{proof}
Let us fix a section $\tilde{\alpha}\in\sect(A)$. Using $\tilde{\alpha}$ we can decompose $P$ 
following a strategy similar to the one presented in\ Remark\ \ref{remAffDualSectDecomp}: 
\begin{equation*}
P=P(\tilde{\alpha})\,\bfone+P_V(\cdot-\tilde{\alpha})\,,
\end{equation*}
where $P_V$ is the linear part of $P$. Since $P_V:\sect(V)\to\sect(W)$ is a linear differential operator, 
there exists a unique formal dual $P_V^\ast:\sect(W^\ast)\to\sect(V^\ast)$. 
Using $P_V^\ast$, we introduce the following linear operator, which is the candidate to define a formal dual of $P$ 
since per construction it fulfils the condition stated in\ \eqref{eqAffDiffOpDual}:
\begin{align*}
P^\ast:\sect(W^\ast)\to\sect(A^\dagger)\,, 
&& \nu\mapsto\nu(P(\tilde{\alpha}))\,\bfone+\left(P_V^\ast(\nu)\right)(\cdot-\tilde{\alpha})\,.
\end{align*}
To conclude the proof, we still have to check that $P^\ast$ is a differential operator. 
To this aim, let us choose an open set $U$ of $M$ where both $V^\ast$, $W^\ast$ are trivial. 
Therefore we can find $\c(U)$-module bases for sections over $U$ of both these vector bundles. 
Let us denote with $\{\mu_i\}$ the basis of the $\c(U)$-module $\sect(V^\ast_U$) 
and with $\{\nu_j\}$ the basis of the $\c(U)$-module $\sect(W^\ast_U$). 
In analogy with\ Remark\ \ref{remAffSpDualDim}, we obtain a $\c(U)$-module basis $\{\bfone\vert_U,\phi_i\}$ 
of $\sect(A^\dagger_U)$ setting $\phi_i=\mu_i(\cdot-\tilde{\alpha})$. 
Given a section $\nu\in\sect(W^\ast)$ there is a unique family of functions $\{f^j\}$ in $\c(U)$ 
such that $f^j\,\nu_j=\nu\vert_U$, the sum over repeated indices being understood. 
Furthermore, since $P_V^\ast$ is a linear differential operator, 
there is a unique family $\{P^i_j:\c(U)\to\c(U)\}$ of linear differential operators such that 
\begin{align*}
P_V^\ast(\nu)\vert_U=\left(P^i_j(f^j)\right)\,\mu_i\,, && \forall\,\nu\in\sect(W^\ast)\,, 
\end{align*}
where $\{f^j\}$ are the coefficients of the linear combination of $\{\nu_j\}$ which reproduces $\nu\vert_U$. 
We deduce that $P^\ast$ is a linear differential operator. In fact, we have 
an expansion over $U$ of $P^\ast$ in terms of linear differential operators on $\c(U)$: 
\begin{align*}
P^\ast(\nu)\vert_U=
\left(\nu_j\left(P(\tilde{\alpha})\vert_U\right)\,f^j\right)\,\bfone\vert_U+\left(P^i_j(f^j)\right)\,\phi_i\,, 
&& \forall\,\nu\in\sect(W^\ast)\,,
\end{align*}
where $\{f^j\}$ are the coefficients of the linear combination of $\{\nu_j\}$ which reproduces $\nu\vert_U$. 
\end{proof}

We exhibit now an example which explicitly violates uniqueness. 
This example is similar to the situation we will encounter in\ Chapter\ \ref{chYangMills} 
in the attempt to define the dual of the curvature map, see\ Remark\ \ref{remCurvAb}. 

\begin{example}[Non-unique formal dual for affine differential operators]\label{exaAffDiffOpNonUniqueDual}
Consider an $m$-dimensional pseudo-Riemannian oriented manifold $M$. 
Let $k\in\{0,\dots,m-1\}$, take $\theta\in\fc^{k+1}(M)$ and define the map 
\begin{align*}
\dd_\theta:\f^k(M)\to\f^{k+1}(M)\,, && \xi\mapsto\theta+\dd\xi\,.
\end{align*}
Clearly, this is an affine differential operator with linear part $\dd_{\theta\,V}=\dd:\f^k(M)\to\f^{k+1}(M)$. 
Let us stress that the affine bundle considered here is the bundle $\bw^kT^\ast M$ of $k$-forms on $M$, 
which is indeed a vector bundle, but it is now regarded as an affine bundle modeled on itself, 
{\em cfr.}\ Example\ \ref{exaAffBun}. Accordingly, we regard $\f^k(M)$ as an affine space modeled on itself. 
Taking into account the pairing $(\cdot,\cdot)$ between $k$-forms introduced in\ \eqref{eqPairing2}, 
the condition \eqref{eqAffDiffOpDual} to define a formal dual 
$\dd_\theta^\ast:\f^{k+1}(M)\to\sect((\bw^kT^\ast M)^\dagger)$ of $\dd_\theta$ becomes 
\begin{align*}
\int_M\left(\dd_\theta^\ast(\eta)\right)(\xi)\,\vol=(\eta,\dd_\theta(\xi))\, 
&& \forall\,\eta\in\fc^{k+1}(M)\,,\;\forall\,\xi\in\f^k(M)\,,
\end{align*}
where $\vol=\ast1$ is the volume form induced by the metric and the choice of the orientation on $M$. 
Notice that we are identifying the dual $\bw^lTM$ of the vector bundle $\bw^lT^\ast M$ with $\bw^lT^\ast M$ 
via the fiberwise inner product defined in\ \eqref{eqContraction}. Under this identification, 
the linear part of each $\phi\in\sect((\bw^kT^\ast M)^\dagger)$ lies in $\f^k(M)$.\footnote{On account of 
this identification, maybe it would be more appropriate to replace the term \quotes{formal dual} 
with the term \quotes{formal adjoint} in the present case.} 
Using Stokes' theorem, one can check that, for each $\omega\in\fdd^k(M)$, 
the linear differential operator defined below provides a formal dual of $\dd_\theta$: 
\begin{align*}
\f^{k+1}(M)\to\sect\big(\big(\bw^kT^\ast M\big)^\dagger\big)\,, 
&& \eta\mapsto\ast^{-1}(\eta\wedge\ast\theta)\,\bfone
+\ast^{-1}\left(\de\eta\wedge\ast(\cdot-\omega)\right)\,.
\end{align*}
For two different choices $\omega,\omega^\prime\in\fdd^k(M)$, 
the corresponding linear differential operators differ by $Q\,\bfone$, 
where $Q$ denotes the following linear differential operator: 
\begin{align*}
Q:\f^{k+1}(M)\mapsto\c(M)\,, 
&& \eta\mapsto\ast^{-1}\left(\de\eta\wedge\ast(\omega-\omega^\prime)\right)\,.
\end{align*}
Again, it is a matter of using Stokes' theorem to deduce that $\int_MQ(\eta)\,\vol=0$ for all $\eta\in\fc^{k+1}(M)$. 
\end{example}

Below we present a strategy to \quotes{cure} the non-uniqueness of the formal dual 
of an affine differential operator. This strategy will be adopted in\ Chapter\ \ref{chYangMills} 
to define the dual of the equation of motion in a suitable sense. 

\begin{remark}\label{remAffDiffOpUniqueDual}
Let $V$ and $W$ be vector bundles over an oriented manifold $M$ with volume form $\vol$ 
and consider an affine bundle $A$ over $M$ modeled on $V$. 
Theorem\ \ref{thmAffDiffOpDual} ensures the existence of a dual $P^\ast:\sect(W^\ast)\to\sect(A^\dagger)$ 
for each affine differential operator $P:\sect(A)\to\sect(W)$ 
and captures the failure of\ \eqref{eqAffDiffOpDual} to uniquely define this dual at least on $\sc(W^\ast)$. 
In particular, we note that, if on the target space we identify sections of the form $f\,\bfone$, 
$f\in\cc(M)$ such that $\int_Mf\,\vol=0$, $P^\ast$ becomes uniquely defined on $\sc(W^\ast)$. 
More precisely, consider the vector subspace $\triv$ of $\sc(A^\dagger)$ defined in\ \eqref{eqTriv}. 
Then there is no ambiguity in the linear map $P^\ast:\sc(W^\ast)\to\sc(A^\dagger)/\triv$ 
obtained restricting any dual of $P$ to $\sc(W^\ast)$ and then identifying $\triv$ with zero on the target space. 
\end{remark}

\section{Principal bundles}\label{secPrBun}
In this section we recall the basic definitions and some important features of principal bundles and their connections. 
For details, the reader is referred to the literature, {\em e.g.}\ \cite{Hus94, KN96, Ish99, Bau14}. 
See also\ \cite{BDS14a} for a similar presentation. 
This material will be widely used in\ Chapter\ \ref{chYangMills} for the specific case of $U(1)$ as structure group. 

Let us start from the definition of a principal $G$-bundle. 

\begin{definition}\index{principal bundle}
Let $M$ be an $m$-dimensional manifold and $G$ an $n$-di\-men\-sio\-nal Lie group. 
A {\em principal $G$-bundle} $P$ over $M$ is an $(m+n)$-dimensional manifold 
endowed with a smooth right action $r:P\times G\to P$ of the Lie group $G$ on $P$. 
Furthermore the following properties hold true: 
\begin{enumerate}
\item The right $G$-action $r$ is free; 
\item $M$ coincides with the orbit space $P/G$ of the right action $r$ 
and the canonical projection $\pi:P\to M$ is smooth; 
\item $\pi:P\to M$ admits $G$-equivariant local trivializations. 
This means that, for each $x\in M$, there exists an open neighborhood $U\subseteq M$ of $x$ 
and a $G$-equivariant, fiber-preserving diffeomorphism $\psi:\pi^{-1}(U)\to U\times G$, 
namely $\psi(r(p,g))=r(\psi(p),g)$ for each $p\in\pi^{-1}(U)$ and $g\in G$ 
and $\pr_1\circ\psi=\pi$ as maps from $\pi^{-1}(U)$ to $U$, 
where $\pr_1:U\times G\to U$ denotes the projection on the first factor. 
\end{enumerate}
$P$, $M$ and $G$ are called respectively {\em total space}, {\em base space} and {\em structure group} 
of the principal $G$-bundle $P$ over $M$. 
\end{definition}

We will usually denote a principal bundle simply denoting its total space. 
Furthermore, in most cases the right $G$-action will be simply denoted by juxtaposition, 
namely we will write $pg$ in place of $r(p,g)$ for each $p\in P$ and $g\in G$. 

We now come to the notion of a principal bundle map. 

\begin{definition}\index{principal bundle!map}
Let $M$ and $N$ be $m$-dimensional manifolds and take a Lie group $G$. 
Consider principal $G$-bundles $P$ over $M$ and $Q$ over $N$. 
A principal $G$-bundle map $f:P\to Q$ is a $G$-equivariant smooth map, 
namely such that $f(pg)=f(p)g$ for each $p\in P$ and $g\in G$. 
\end{definition}

\begin{remark}\label{remPrBunMapBase}
A principal $G$-bundle map $f:P\to Q$ induces a unique smooth map $\ul{f}:M\to N$ 
between the corresponding base spaces such that the following diagram commutes: 
\begin{equation*}
\xymatrix{
P\ar[d]\ar[r]^f & Q\ar[d]\\
M\ar[r]_{\ul{f}} & N
}
\end{equation*}
where the vertical arrows are the bundle projections. 
More explicitly, $\ul{f}:M\to N$ can be defined setting, for $x\in M$, $\ul{f}(x)=y$, 
$y\in N$ being the base point of $f(p)\in Q$ for an arbitrary point $p\in P$ in the fiber over $x$. 
\end{remark}

Later we will be interested in connections on principal bundles. 
To discuss this topic, it will be very useful to recall the associated bundle construction and its functorial behavior. 
\index{associated bundle}Let $P$ be a principal $G$-bundle over $M$ 
and $\rho:G\times F\to F$, $(g,f)\mapsto gf$ be a smooth left action of the Lie group $G$ on a manifold $F$. 
We introduce the right action $(P\times F)\times G\to P\times F$, $((p,f),g)\mapsto(pg,g^{-1}f)$ 
of $G$ on $P\times F$. Denoting with $P_F$ the orbit space of $P\times F$ under this action 
and with $\pi_F:P_F\to M$ the map induced by the projection on the first factor $\pr_1:P\times F\to P$, 
we obtain the {\em associated bundle} $P_F$ over $M$ with typical fiber $F$. 
If $F$ is a vector space and $\rho$ a linear representation of $G$ on the vector space $F$, 
then the associated bundle turns out to be a vector bundle. 
If a principal bundle map $f:P\to Q$ is given, 
the associated bundle construction provides a canonical bundle map $f_F:P_F\to Q_F$ 
induced by $f\times\id_F:P\times F\to Q\times F$. 
For a given left action $\rho:G\times F\to F$, this construction gives rise to a covariant functor 
from the category of principal bundles to the category of fiber bundles, which restricts to 
a covariant functor to the category of vector bundles whenever $\rho$ is a representation on a vector space $F$. 

In particular, consider to the associated bundle construction for two particular choices of $\rho$. 
In the first case the left adjoint action $\Ad:G\times G\to G$, $(g,h)\mapsto ghg^{-1}$ of $G$ on itself 
is taken into account, while in the second case we exploit the adjoint representation $\ad:G\times\lieg\to\lieg$, 
$(g,\xi)\mapsto\Ad_{g\,\ast}\xi$ of $G$ on its Lie algebra $\lieg$. 

\begin{definition}\index{adjoint bundle}
Let $G$ be a Lie group and consider a principal $G$-bundle $P$ over $M$. 
The {\em $G$-adjoint bundle} $\Ad(P)$ is the bundle with typical fiber $G$ associated to $P$ 
via the left adjoint action $\Ad:G\times G\to G$, $(g,h)\mapsto ghg^{-1}$ of $G$ on itself, 
while the {\em $\lieg$-adjoint bundle} $\ad(P)$ is the bundle with typical fiber $\lieg$ 
associated to $P$ via the adjoint representation $\ad:G\times\lieg\to\lieg$, 
$(g,\xi)\mapsto\Ad_{g\,\ast}\xi$ of $G$ on its Lie algebra $\lieg$. 
\end{definition}

\begin{remark}\label{remAdPGroup}
Note that each fiber of $\Ad(P)$ is (non-canonically) isomorphic to $G$. 
This might suggest that the bundle $\Ad(P)$ inherits a fiberwise group structure from the structure group $G$. 
In fact, this is the case. Denoting with $\pi:\Ad(P)\to M$ the projection of the bundle $\Ad(P)$ 
onto its base space $M$ and introducing the fibered product $\Ad(P)\times_\pi\Ad(P)$ of the bundle $\Ad(P)$ 
with itself, we can define 
\begin{align*}
\Ad(P)\times_\pi\Ad(P)\to\Ad(P)\,, && ([p,g],[p,h])\mapsto[p,gh]\,.
\end{align*}
This map is defined everywhere since two points in the same fiber of $\Ad(P)$ 
admit representatives with a common point of $P$ by transitivity of the right $G$-action on each fiber of $P$. 
Furthermore, the definition is well-posed since the right $G$-action on $P$ is also free. 
This map clearly specifies a fiberwise group structure on $\Ad(P)$. 
In particular, for each point $x\in M$, the corresponding fiber $\Ad(P)_x$ has $[p,e]$ as identity element, 
$p\in P_x$ arbitrarily chosen, while the inverse of $[p,g]\in\Ad(P)_x$ is given by $[p,g^{-1}]$. 
\end{remark}

Note that, for a given principal $G$-bundle map $f:P\to Q$, we will denote 
by $\Ad(f):\Ad(P)\to\Ad(Q)$ the bundle map induced by the $G$-adjoint bundle construction 
and by $\ad(f):\ad(P)\to\ad(Q)$ the vector bundle map induced by the $\lieg$-adjoint bundle construction. 
In fact, both $\Ad(\cdot)$ and $\ad(\cdot)$ turn out to be covariant functors from the category of principal bundles 
to the category of fiber bundles and respectively to the category of vector bundles. 

\subsection{Gauge transformations}\label{secGaugeTr}
The geometry of principal bundles naturally encodes a notion of gauge transformations, 
which is specified by principal bundle automorphisms covering the identity on the base space. 

\begin{definition}\index{principal bundle!gauge group}
Let $G$ be a Lie group and consider a principal $G$-bundle $P$ over a manifold $M$. 
The group $\gauge(P)$ of gauge transformations on $P$ is the group of principal bundle automorphisms $f:P\to P$ 
covering the identity of the base space, namely such that $\ul{f}=\id_M$. 
\end{definition}

The group of gauge transformations on a principal bundle can also understood in terms of sections 
of the adjoint bundle $\Ad(P)$ bundle as explained by the next proposition. 

\begin{proposition}\label{prpGaugeGroup}
Let $G$ be a Lie group and consider a principal $G$-bundle $P$ over a manifold $M$. 
Denote with $\c(P,G)^\eqv$ the group (with respect to pointwise multiplication) of smooth $G$-valued functions 
$f:P\to G$ on $P$ which are equivariant with respect to the right adjoint action of the group $G$ on itself, 
namely such that $f(pg)=gf(p)g^{-1}$ for each $p\in P$ and each $g\in G$. 
Given $f\in\gauge(P)$, define $\tilde{f}\in\c(P,G)^\eqv$ such that $p\tilde{f}(p)=f(p)$ for each $p\in P$. 
Furthermore, given $\phi\in\c(P,G)^\eqv$, define $\hat{\phi}\in\sect(\Ad(P))$ 
such that, for each $x\in M$,  $\hat{\phi}(x)\in\Ad(P)_x$ is represented by $(p,\phi(p))\in P_x\times G$ 
for an arbitrary choice of $p\in P_x$. Then the maps listed below are group isomorphisms: 
\begin{align*}
\gauge(P)\to\c(P,G)^\eqv\,, && f\mapsto\tilde{f}\,,\\
\c(P,G)^\eqv\to\sect(\Ad(P))\,, && \phi\mapsto\hat{\phi}\,,
\end{align*}
where $\sect(\Ad(P))$ has the group structure induced by the fibered group structure of $\Ad(P)$, 
{\em cfr.}\ Remark\ \ref{remAdPGroup}. 
\end{proposition}

\begin{proof}
Given $f\in\gauge(P)$, for each $p\in P$ there exists a unique $g\in G$ such that $pg=f(p)$. 
This fact follows from $\ul{f}=\id_M$ and the action of $G$ on $P$ being free and transitive on each fiber. 
As a byproduct, $\tilde{f}\in\c(P,G)^\eqv$ such that $p\tilde{f}(p)=f(p)$ for each $p\in P$ exists and it is unique. 
Note that, for each $p\in P$ and each $g\in G$, the identity $f(pg)=f(p)g$ entails $pg\tilde{f}(pg)=p\tilde{f}(p)g$. 
Since the $G$-action on $P$ is free, we deduce $\tilde{f}(pg)=g^{-1}\tilde{f}(p)g$. 
This shows that the map $\gauge(P)\to\c(P,G)^\eqv$, $f\mapsto\tilde{f}$ is well-defined. 
We check now that this map is a group homomorphism. Given $f,h\in\gauge(P)$, for each $p\in P$, we have 
$h(f(p))=p\tilde{h}(p)\tilde{f}(p)$. This shows that $h\circ f$ is mapped to $\tilde{h}\,\tilde{f}$. 
To prove that $\gauge(P)\to\c(P,G)^\eqv$ is an isomorphism, we exhibit its inverse. 
For each $\phi\in\c(P,G)^\eqv$, define $\bar{\phi}:P\to P$ 
setting $\bar{\phi}(p)=p\phi(p)$ for each $p\in P$. 
It is straightforward to check that $\bar{\phi}\in\gauge(P)$, 
therefore a map $\c(P,G)^\eqv\to\gauge(P)$, $\phi\to \bar{\phi}$ is defined 
and by definition $\bar{\tilde{f}}=f$ for each $f\in\gauge(P)$, 
while $\tilde{\bar{\phi}}=\phi$ for each $\phi\in\c(P,G)^\eqv$. 

Consider now $\phi\in\c(P,G)^\eqv$. Note that for $x\in M$, by the equivariance property of $\phi$, 
$[p,\phi(p)]\in\Ad(P)_x$ does not depend on $p\in P_x$. This entails that a section $\hat{\phi}\in\sect(\Ad(P))$ 
can be defined setting $\hat{\phi}(x)=[p,\phi(p)]$ for any choice of $p\in P_x$. In particular, 
the map $\c(P,G)^\eqv\to\sect(\Ad(P))$, $\phi\mapsto\hat{\phi}$ is defined. 
Given $\phi,\psi\in\c(P,G)^\eqv$, for each $p\in P$ we have the identity $[p,(\phi\psi)(p)]=[p,\phi(p)][p,\psi(p)]$. 
Note that we are using here the group structure of $\c(P,G)^\eqv$ and the fiberwise group structure of $\Ad(P)$, 
see\ Remark\ \ref{remAdPGroup}. This shows that $\c(P,G)^\eqv\to\sect(\Ad(P))$ is a group homomorphism. 
It remains only to exhibit an inverse of this map. Let $\sigma\in\sect(\Ad(P))$ and consider $p\in P$. 
Denoting with $x\in M$ the base point of $p$, there exists a unique $g\in G$ such that $[p,g]=\sigma(x)$. 
Therefore we can define $\check{\sigma}:P\to G$ imposing $[p,\check{\sigma}(p)]=\sigma(x)$ 
for each $x\in M$ and each $p\in P_x$. Indeed $\check{\sigma}\in\c(P,G)^\eqv$ since 
\begin{equation*}
[p,g\check{\sigma}(pg)g^{-1}]=[pg,\check{\sigma}(pg)]=\sigma(x)=[p,\check{\sigma}(p)]\,,
\end{equation*}
for each $x\in M$, $p\in P_x$ and $g\in G$. This shows that the map $\sect(\Ad(P))\to\c(P,G)^\eqv$ is defined. 
Furthermore, by definition we have $\check{\hat{\phi}}=\phi$ for each $\phi\in\c(P,G)^\eqv$ 
and $\hat{\check{\sigma}}=\sigma$ for each $\sigma\in\sect(\Ad(P))$. This concludes the proof. 
\end{proof}

When the structure group is Abelian, the characterization of the group of gauge transformations 
of a given principal bundle becomes even simpler. 

\begin{proposition}\label{prpTrivAdBun}
Let $G$ be an Abelian Lie group and consider a principal $G$-bundle $P$ over a manifold $M$. 
Then there is a canonical bundle isomorphism $\Ad(P)\overset{\simeq}{\to}M\times G$ 
preserving the fibered group structures and a canonical vector bundle isomorphism 
$\ad(P)\overset{\simeq}{\to}M\times\lieg$. In particular, the first isomorphism induces 
an isomorphism $\sect(\Ad(P))\overset{\simeq}{\to}\c(M,G)$ of Abelian groups. 
\end{proposition}

\begin{proof}
If $G$ is Abelian, both $\Ad:G\times G\to G$, $(g,h)\mapsto ghg^{-1}=h$ 
and $\ad:G\times\lieg\to\lieg$, $(g,\xi)\to\Ad_{g\,\ast}\xi=\xi$ become trivial. 
Therefore, in the associated bundle construction, the fiber $F=G$ or $F=\lieg$ decouples, 
meaning that the orbit space of $P\times F$ under the joint right $G$-action reduces to the orbit space 
under the right $G$-action on the factor $P$ and the identity on the factor $F$. 
Since by definition the orbit space of $P$ under its right $G$-action is nothing but the base space $M$, 
we conclude that the associated bundles $\Ad(P)$ and $\ad(P)$ are isomorphic 
respectively to $M\times G$ and $M\times\lieg$. It is straightforward to check that 
for $F=G$ the pointwise group structure is preserved, 
while in the case $F=\lieg$ this procedure provides a vector bundle isomorphism. 
\end{proof}

\subsection{The bundle of connections}\label{subConnBun}
In\ Chapter\ \ref{chYangMills} we will discuss the dynamics of connections on principal $G$-bundles. 
In order to introduce the notion of a connection on a principal bundle, we follow the approach of\ \cite{Ati57, AB83}. 
We only sketch the construction. At the end we will also briefly mention its functorial properties. 

Let $G$ be a Lie group and consider a principal $G$-bundle $P$ over a manifold $M$. 
We denote with $\pi:P\to M$ the bundle projection. Consider the tangent bundle $TP$ to the total space $P$. 
The free right $G$-action $r:P\times G\to G$ on $P$ induces a free right $G$-action on $TP$ 
defined by $r_{g\,\ast}:TP\to TP$, $g\in G$ being any group element 
and $r_g$ denoting the map $r(\cdot,g):P\to P$. 
We can indeed consider the orbit space $TP/G$ of $TP$ under this right $G$-action. 
Note that $TP/G$ carries the structure of a vector bundle over $M$. 

To introduce the vertical subbundle $VP$ of $TP$, we consider the kernel of the push-forward $\pi_\ast:TP\to TM$ 
along the principal bundle projection. Tangent vectors in $VP$ are by definition tangent to a fiber of $P$. 
Since $r:P\times G\to G$ preserves the fibers, the vertical bundle $VP$ is invariant under the action of 
$r_{g\,\ast}:TP\to TP$ for all $g\in G$. This means that we can consider the orbit space of $VP$ as well 
and we still have an inclusion $VP/G\subseteq TP/G$ of vector bundles over $M$. 
It is possible to construct vertical tangent vectors simply specifying a base point $p\in P$ 
and an element $\xi\in\lieg$ of the Lie algebra of the structure group $G$. 
This follows from the observation that $X\in VP$ is tangent to a fiber of $P$. 
Explicitly, for each fixed $p\in P$, we can define the map $r_p:G\to P$, $g\mapsto pg$. 
The tangent map at the identity $X^{(\cdot)}_p=\dd_er_p:\lieg\to T_pP$ is an injection of the Lie algebra 
into the space of vertical tangent vectors at $p$ 
and moreover each $X\in V_pP$ has the form $X_p^\xi$ for a suitable $\xi\in\lieg$. 
Furthermore, the identity $X_{pg}^{\ad_{g^{-1}}\xi}=r_{g\,\ast}X_p^\xi$ for each $p\in P$ and $g\in G$ entails 
that $P\times\lieg\to VP$, $(p,\xi)\mapsto X_p^\xi$ descends to a vector bundle isomorphism $\ad(P)\to VP/G$. 
Therefore, we get an injection $\kappa:\ad(P)\to TP/G$ of vector bundles over $M$. 

We already observed that the right $G$-action on $P$ preserves the fibers. 
This entails that the tangent map $\pi_\ast:TP\to TM$ descends to a vector bundle map $\rho:TP/G\to TM$. 
Since $\pi$ is a projection, $\pi_\ast$ is surjective as well. Therefore $\rho$ is a surjection of vector bundles 
over $M$. Furthermore, by construction, the kernel of $\rho$ coincides with the image of $\kappa$. 

\index{Atiyah sequence}Summing up, we have an exact sequence of vector bundles over $M$, 
the so-called {\em Atiyah sequence} for a principal $G$-bundle $P$ over $M$: 
\begin{equation}\label{eqAtiyahSeq}
0\lra\ad(P)\overset{\kappa}{\lra}TP/G\overset{\rho}{\lra}TM\lra0\,,
\end{equation}
where $0$ here denotes the vector bundle over $M$ with trivial typical fiber. 

\begin{remark}[Naturality of the Atiyah sequence]
We already mentioned that $\ad(\cdot)$ is a covariant functor 
from the category of principal $G$-bundles to the category of vector bundles. 
It is well-known that $T(\cdot)$ is a covariant functor from the category of manifolds 
to the category of vector bundles. When we consider principal $G$-bundles and principal $G$-bundles maps, 
the fact that such maps are equivariant with respect to the right $G$-action entails that 
the induced tangent maps are equivariant with respect to the right $G$-action induced on the bundles 
tangent to the total spaces. This fact eventually leads to a covariant functor $T(\cdot)/G$ 
from the category of principal $G$-bundles to the category of vector bundles. 
The last covariant functor we take into account to each principal $G$-bundle $P$ over $M$ 
associates $T(P)_\base=TM$, the tangent bundle of the base space, 
and to each principal $G$-bundle map $f:P\to Q$, associates $T(f)_\base=\ul{f}_\ast:TM\to TN$, 
the push-forward along the map $\ul{f}:M\to N$ between the base spaces. 
Clearly $T(\cdot)_\base$ turns out to be a covariant functor 
from the category of principal $G$-bundles to the category of vector bundles. 

Let $\pi_P:P\to M$ and $\pi_Q:Q\to N$ be principal $G$-bundles and consider a principal bundle map $f:P\to Q$. 
Since $f$ is fiber-preserving, {\em cfr.}\ Remark\ \ref{remPrBunMapBase}, 
we have a commutative diagram involving tangent maps: 
\begin{equation*}
\xymatrix{
TP\ar[d]_{\pi_{P\,\ast}}\ar[r]^{f_\ast} & TQ\ar[d]^{\pi_{Q\,\ast}}\\
TM\ar[r]_{\ul{f}_\ast} & TN
}
\end{equation*}
In particular $f_\ast$ maps vertical vectors tangent to $P$ to vertical vectors tangent to $Q$ 
and the following diagram of vector bundles commutes: 
\begin{equation}\label{eqAtiyahSeqNat}
\xymatrix{
0\ar[r] & \ad(P)\ar[d]_{\ad(f)}\ar[r]^{\kappa_P} & TP/G\ar[d]|-{T(f)/G}\ar[r]^{\rho_P} 
& TM\ar[d]^{\ul{f}_\ast}\ar[r] & 0\\
0\ar[r] & \ad(Q)\ar[r]_{\kappa_Q} & TQ/G\ar[r]_{\rho_Q} & TN\ar[r] & 0
}
\end{equation}
Note that the first row is an exact sequence of vector bundles over $M$, 
while the second is an exact sequence of vector bundles over $N$, the base spaces being related by $\ul{f}$. 
This diagram tells us that $\kappa_{(\cdot)}$ and $\rho_{(\cdot)}$ form an exact sequence 
of natural transformations between the covariant functors $\ad(\cdot)$, $T(\cdot)/G$ and $T(\cdot)_\base$ 
from the category of principal bundles to the category of vector bundles. 
\end{remark}

We are now ready to introduce the bundle of connections associated to a principal $G$-bundle $P$ over $M$. 
As it will be shown immediately after its definition, the bundle of connections is an affine bundle 
in the sense of\ Definition\ \ref{defAffBun}. 

\begin{definition}\label{defBunConn}\index{bundle of connections}
Let $G$ be a Lie group and consider a principal $G$-bundle $P$ over $M$. 
Recalling the construction of the Atiyah sequence \eqref{eqAtiyahSeq}, we define the {\em bundle of connections} 
$\conn(P)$ as the affine bundle whose elements $\lambda\in\hom(TM,TP/G)$ in the fiber over $x\in M$ 
are right splittings of the restriction to the fiber over $x$ of the Atiyah sequence, namely 
\begin{equation*}
\conn(P)=\bigsqcup_{x\in M}\left\{\lambda\in\hom(TM,TP/G)_x\,:\;\rho\vert_x\lambda=\id_{TM}\vert_x\right\}\,. 
\end{equation*}
\end{definition}

\begin{remark}[Affine structure of the bundle of connections]\label{remBunConnAff}
Whenever $\lambda\in\conn(P)_x$, $x\in M$, is given, we can shift it by $\eta\in\hom(TM,\ad(P))_x$. 
The resulting homomorphism $\lambda+\kappa\vert_x\eta\in\hom(TM,TP/G)_x$ still lies in $\conn(P)_x$. 
Indeed, by exactness of the Atiyah sequence, $\lambda,\lambda^\prime\in\conn(P)_x$, $x\in M$, 
differ by $\eta\in\hom(TM,\ad(P))_x$, namely $\lambda^\prime-\lambda=\kappa\vert_x\eta$. 
This shows that the fibers of $\conn(P)$ are affine spaces modeled on the corresponding fibers of 
the vector bundle $\hom(TM,\ad(P))$. 
The affine bundle structure on $\conn(P)$ is induced by the vector bundle structures on $\ad(P)$, $TP/G$ and $TM$. 
In fact, this is always the case for the bundle of splittings of a short exact sequence of vector bundles, 
{\em cfr.} Example\ \ref{exaAffBun}. 
\end{remark}

The next remark focuses the attention on certain functorial properties of the bundle of connections. 
In particular, it is shown that the assignment of $\conn(P)$ to each principal $G$-bundle $P$ 
gives rise to a covariant functor, provided that the class of morphisms considered is restricted to 
principal bundle maps covering embeddings between the base manifolds. 

\begin{remark}[Functoriality of the bundle of connections]\label{remBunConnFunctor}
The assignment of the bundle of connections $\conn(P)$ to each principal $G$-bundle $P$ does not induce 
a covariant functor from the category of principal bundles to the category of affine bundles. 
This is due to $\conn(P)$ being defined as a subbundle of the bundle $\hom(TM,TP/G)$. 
In fact, $\hom(\cdot,\cdot)$ is contravariant in the first argument and covariant in the second. 

To obtain a covariant behavior the basic idea is to reverse the arrow 
which enters the first argument of $\hom(\cdot,\cdot)$.
Of course, this operation is usually not possible for arbitrary morphisms in the category of principal $G$-bundles. 
Yet, let us restrict ourselves to principal $G$-bundle maps $f:P\to Q$ covering an embedding $\ul{f}:M\to N$. 
This entails that $\ul{f}$ has an inverse $\ul{f}^{-1}:\ul{f}(M)\to M$ defined on its image 
and $f$ is an embedding too, hence we can consider $f^{-1}:f(P)\to P$. Therefore we obtain 
a vector bundle isomorphism $\hom({\ul{f}^{-1}}_\ast,T(f)/G):\hom(TM,TP/G)\to\hom(T\ul{f}(M),Tf(P)/G)$ 
which maps $\conn(P)$ to $\conn(f(P))$ by naturality of the Atiyah sequence \eqref{eqAtiyahSeq}, 
see \eqref{eqAtiyahSeqNat}. Since $\hom(T\ul{f}(M),Tf(P)/G)$ is a subbundle of $\hom(TN,TQ/G)$, 
from $\hom({\ul{f}^{-1}}_\ast,T(f)/G)$ we get the map $\conn(f):\conn(P)\to\conn(Q)$. 
Furthermore, we can consider the vector bundle isomorphism 
$\hom({\ul{f}^{-1}}_\ast,\ad(f)):\hom(TM,\ad(P))\to\hom(T\ul{f}(M),\ad(f(P)))$, 
together with the inclusion of the subbundle $\hom(T\ul{f}(M),\ad(f(P)))$ in $\hom(TN,\ad(Q))$, 
thus defining the vector bundle map $\conn(f)_V:\hom(TM,\ad(P))\to\hom(TN,\ad(Q))$. 
It is easy to check that $\conn(f)$ is an affine bundle morphism with linear part 
$\conn(f)_V:\hom(TM,\ad(P))\to\hom(TN,\ad(Q))$. 

Summing up, $\conn(\cdot)$ is a covariant functor from the subcategory of the category of principal $G$-bundles 
whose morphisms are principal $G$-bundles maps covering an embedding. 
In particular, gauge transformations (whose base map is the identity) enter this class. 
In fact, the present discussion, specialized to principal $G$-bundle automorphisms covering the identity, 
fully describes the action of the gauge group $\gauge(P)$ on the bundle of connections $\conn(P)$. 
Let us mention that, for other reasons, in\ Chapter\ \ref{chYangMills}, we will be anyway forced 
to restrict the category of principal $G$-bundles to morphisms covering embeddings on the base. 
As a matter of fact, we will only consider principal $G$-bundles over globally hyperbolic spacetimes as objects 
and principal bundle maps covering causal embeddings as morphism. 
\end{remark}

\subsection{Connections}\label{secConn}
In\ Chapter\ \ref{chYangMills} we will discuss the Yang-Mills gauge field theory 
over globally hyperbolic spacetimes with $U(1)$ as structure group. The dynamical degrees of freedom 
in field models of Yang-Mills type consist of connections on a principal $G$-bundle. 
Now we define these objects making use of the bundle of connections. 
We will also mention the relation to the more common definition of a principal bundle connection. 
In particular, this will turn out to be very useful in order to assign to a given connection its curvature $2$-form. 

\begin{definition}\label{defConn}\index{connection}\index{principal bundle!connection}
Let $G$ be a Lie group and consider a principal $G$-bundle $P$ over a manifold $M$. 
A {\em connection} $\lambda$ on $P$ is a global section of the bundle of connections $\conn(P)$. 
In particular, we denote the affine space of connections as $\sect(\conn(P))$. 
\end{definition}

Indeed, the affine structure on the space of connections is a direct consequence of $\conn(P)$ 
being an affine bundle, see\ Definition\ \ref{defBunConn} and\ Remark\ \ref{remBunConnAff}. 
More explicitly, for a principal $G$-bundle $P$ over $M$, $\sect(\conn(P))$ is an affine space 
modeled on the vector space of sections of the vector bundle $\hom(TM,\ad(P))$. 
It is convenient to identify this object with the space $\f^1(M,\ad(P))$ of $\ad(P)$-valued $1$-forms on $M$. 
Given $\lambda\in\sect(\conn(P))$ and $\sigma\in\f^1(M,\ad(P))$, 
we get a new connection $\lambda+\kappa\circ\sigma$ translating $\lambda$ by $\kappa\circ\sigma$. 
This characterizes the affine structure of $\sect(\conn(P))$. 
For convenience in the future we will suppress the natural transformation $\kappa$, 
thus writing $\lambda+\sigma$ for the affine translation of $\lambda$ by $\sigma$. 

The following remark discusses the functorial properties of the space of sections of the bundle of connections. 

\begin{remark}[$\sect(\conn(\cdot))$ as a contravariant functor]\label{remConnFunctor}
Previously we considered the space of sections $\sect(\conn(P))$ associated to each principal $G$-bundle $P$. 
We would like to obtain a contravariant functor out of this assignment. 
Let $G$ be a Lie group and consider two principal $G$-bundles $P$ and $Q$ over $M$ and respectively $N$. 
Furthermore, assume $f:P\to Q$ is a principal bundle map. Whenever a connection $\mu\in\sect(\conn(Q))$ is given, 
we want to use $f$ to define a new connection $\lambda\in\sect(\conn(P))$ 
which agrees with the original one, namely such that $T(f)/G\circ\lambda=\mu\circ\ul{f}_\ast$. 
The argument below shows that agreement with $\mu$, 
together with the requirement that $\lambda$ is a connection, uniquely specifies $\lambda$. 
Let $x\in M$ and consider $v\in T_xM$. We look for a unique element $[p,X]$ in the fiber over $x$ of $TP/G$ 
such that $\rho_P[p,X]=v$, {\em cfr.}\ \eqref{eqAtiyahSeq}, and $T(f)/G([p,X])=\mu(\ul{f}_\ast(v))$. 
A solution to this problem exists. Arbitrarily choose $p\in P_x$ and $Y\in T_pP$ such that $\pi_{P\,\ast}Y=v$
($Y$ of this type exists since $\pi_P:P\to M$ is the projection of $P$ onto its base $M$). 
It follows that $\rho_P[p,Y]=v$, hence $T(f)/G([p,Y])-\mu(\ul{f}_\ast(v))$ lies in the kernel of $\rho_Q$. 
Therefore there exists $[p,\xi]\in\ad(P)_x$ such that $\kappa_Q[f(p),\xi]=T(f)/G([p,Y])-\mu(\ul{f}_\ast(v))$, 
meaning that $[p,X]=[p,Y]-\kappa_P[p,\xi]$ is a solution to our problem. 
To prove uniqueness, consider $[p,Y]\in(TP/G)_x$ such that $\rho[p,Y]=0$ and $T(f)/G([p,Y])=0$. 
By the first equation, there exists $[p,\xi]\in\ad(P)_x$ such that $\kappa_P[p,\xi]=[p,Y]$. 
From the second equation we deduce that $\kappa_Q(\ad(f)([p,\xi]))=0$. 
Yet, both $\kappa_Q$ and $\ad(f)$ are injective on each fiber, therefore $[p,\xi]=0$ and $[p,Y]=0$ too. 
This shows that the following map is well-defined 
\begin{align*}
\sect(\conn(f)):\sect(\conn(Q))\to\sect(\conn(P))\,, && \mu\mapsto\lambda\,,
\end{align*}
where $\lambda\in\sect(\conn(P))$ is uniquely specified for each $\mu\in\sect(\conn(Q))$ 
by the conditions $\rho\circ\lambda=\id_{TP/G}$ ($\lambda$ is a connection) 
and $T(f)/G\circ\lambda=\mu\circ\ul{f}_\ast$ ($\lambda$ agrees with $\mu$). 
Furthermore, $\sect(\conn(f))$ turns out to be an affine map whose linear part is given by 
\begin{align*}
\sect(\conn(f))_V:\f^1(N,\ad(Q))\to\f^1(M,\ad(P))\,, && \tau\mapsto\sigma\,,
\end{align*}
where $\sigma\in\f^1(M,\ad(P))$ is uniquely specified by the condition $\ad(f)\circ\sigma=\tau\circ\ul{f}_\ast$, 
$\ad(f):\ad(P)\to\ad(Q)$ being fiberwise injective. 
In fact, for each $\mu\in\sect(\conn(Q))$ and each $\tau\in\f^1(N,\ad(Q))$, 
we find $\sect(\conn(f))(\mu+\tau)=\sect(\conn(f))(\mu)+\sect(\conn(f))_V(\tau)$. 
This way we obtain a contravariant $\sect(\conn(\cdot))$ 
from the category of principal $G$-bundles to the category of affine spaces. 

Whenever $f:P\to Q$ is a principal $G$-bundle map covering an embedding $\ul{f}:M\to N$, 
there is a more explicit way to define the action of $f$ on $\sect(\conn(Q))$, 
essentially based on\ Remark\ \ref{remBunConnFunctor}. 
In that case the idea was to exploit the vector bundle isomorphism 
$\hom({\ul{f}^{-1}}_\ast,T(f)/G):\hom(TM,TP/G)\to\hom(T\ul{f}(M),Tf(P)/G)$ defined out 
of a principal bundle map $f:P\to Q$ covering an embedding $\ul{f}:M\to N$. 
Now we exploit the inverse of this isomorphism, namely we consider 
$\hom(\ul{f}_\ast,T(f^{-1})/G):\hom(T\ul{f}(M),Tf(P)/G)\to\hom(TM,TP/G)$. 
Given $\mu\in\sect(\conn(Q))$, we first restrict it to $f(M)$. 
Doing so, we obtain a section of $\conn(f(P))$. 
Thinking of this section as a map from $f(M)$ into $\conn(f(P))$, 
we can compose it with $\hom(\ul{f}_\ast,T(f^{-1})/G)$ on the left and with $\ul{f}$ on the right. 
The resulting map is a connection $\lambda\in\sect(\conn(P))$. 
This procedure provides a more explicit expression for $\sect(\conn(f))$ whenever $\ul{f}$ is an embedding: 
\begin{align*}
\sect(\conn(f)):\conn(Q)\to\conn(P)\,, && \mu\mapsto\hom(\ul{f}_\ast,T(f^{-1})/G)\circ\mu\circ\ul{f}\,.
\end{align*}
Similarly, the linear part becomes 
\begin{align*}
\sect(\conn(f))_V:\f^1(N,\ad(Q))\to\f^1(M,\ad(P))\,, 
&& \tau\mapsto\hom(\ul{f}_\ast,\ad(f^{-1}))\circ\tau\circ\ul{f}\,.
\end{align*}
\end{remark}

We specialize the last remark to gauge transformations in order to obtain a quite explicit formula 
for the action of the gauge group on the space of sections of the bundle of connections. 
A similar approach, with further details on regularity properties of the action of gauge transformations 
on connections, can be found in\ \cite{ACMM86}. 
Later we will specialize further to the case of Abelian structure groups. 

Let $G$ be a Lie group and consider a principal $G$-bundle $\pi:P\to M$. 
Furthermore, consider $f\in\gauge(P)$. Note that, according to\ Proposition\ \ref{prpGaugeGroup}, 
we can equivalently consider $\hat{f}\in\sect(\Ad(P))$. 
Our aim is to rewrite the map $\sect(\conn(f))$ in terms of the action of $\hat{f}$ on $\sect(\conn(P))$. 
Regarding the last arrow in the Atiyah sequence as defining the affine bundle $\rho:TP/G\to TM$ 
modeled over the pull-back of the vector bundle $\ad(P)$ under $\pi_{TM}:TM\to M$, 
we can reinterpret each connection $\lambda\in\sect(\conn(P))$ 
as a section $\lambda:TM\to TP/G$ of the bundle $\rho:TP/G\to TM$. 
From the fiber bundle $\pi_\Ad:\Ad(P)\to M$, we obtain a new bundle $\pi_{\Ad\,\ast}:T\Ad(P)\to TM$. 
Since $\hat{f}$ is a section of $\pi_\Ad:\Ad(P)\to M$, 
the corresponding tangent map $\hat{f}_\ast$ provides a section of $\pi_{\Ad\,\ast}:T\Ad(P)\to TM$. 
The action of $\hat{f}_\ast$ on $\lambda$ can be understood in terms of 
an action of $T\Ad(P)$ on $TP/G$ fibered over $TM$, which we describe below. 

As a first step, note that $\pi_\ast:TP\to TM$ is a principal $TG$-bundle. In fact $TG$, 
endowed with the group structure induced by $G$, is a Lie group 
and $\pi_\ast:TP\to TM$ becomes a principal bundle with the right $TG$-action 
induced by the right $G$-action of the original principal bundle $P$. 
We will denote the group multiplication in $TG$ by juxtaposition and the inverse of $\Xi\in TG$ by $\Xi^{-1}$. 
In analogy with the notation used for the right $G$-action on $P$, 
we will denote the action of $\Xi\in TG$ on $X\in TP$ simply by $X\Xi$. 
If we consider the action $TG\times TG\to TG$ induced by the adjoint action $\Ad:G\times G\to G$ of $G$ on itself, 
in full analogy with the associated bundle $\pi_\Ad:\Ad(P)\to M$, we can introduce 
the corresponding bundle over $TM$ associated to $TP$ with fiber $TG$, namely $TP_{TG}\to TM$. 
It turns out that $TP_{TG}\to TM$ is isomorphic to $\pi_{\Ad\,\ast}:T\Ad(P)\to TM$.\footnote{In fact, 
we have a fiberwise surjective bundle map $P\times G\to\Ad(P)$ covering $\id_M$ 
which maps each pair $(p,g)$ to the corresponding orbit $[p,g]$. 
Looking at the tangent map $TP\times TG\to T\Ad(P)$ covering $\id_{TM}$, which is still fiberwise surjective, 
it is easily understood that injectivity is recovered as soon as the orbit space of the source $TP\times TG$ 
under the right $TG$-action is considered. Therefore we conclude that 
$TP_{TG}\to TM$ and $\pi_{\Ad\,\ast}:T\Ad(P)\to TM$ are isomorphic as fiber bundles.} 
Identifying $T\Ad(P)$ with $TP_{TG}$ via this isomorphism allows us 
to define a left action of $T\Ad(P)$ on $TP/G$ fibered over $TM$: 
\begin{align}\label{eqConnGauge}
A:T\Ad(P)\times TP/G\to TP/G\,, && ([X,\Xi],[X])\mapsto[X\Xi^{-1}]\,.
\end{align}
Note that if $X,Y\in TP$ are such that $\pi_\ast X=\pi_\ast Y$, 
there exists a unique $\Theta\in TG$ such that $X\Theta=Y$. 
Therefore it is always possible to have the same vector tangent to $P$ 
as a representative in both arguments in the map $A$ defined above. This shows that $A$ is well-defined. 

\begin{proposition}\label{prpGaugeConn}
Let $G$ be a Lie group and consider a principal $G$-bundle $P$ over $M$. 
Denote with $A:T\Ad(P)\times TP/G\to TP/G$ the left action of $T\Ad(P)$ on $TP/G$, see\ \eqref{eqConnGauge}. 
For each $f\in\gauge(P)$, $A(\hat{f}_\ast,\cdot)=\sect(\conn(f)):\sect(\conn(P))\to\sect(\conn(P))$, 
where $\hat{f}\in\sect(\Ad(P))$ is defined by $f$ according to\ Proposition\ \ref{prpGaugeGroup}. 
\end{proposition}

\begin{proof}
Take $f\in\gauge(P)$ and $\lambda\in\sect(\conn(P))$ 
and consider $\nu=A\circ(\hat{f}_\ast,\lambda):TM\to TP/G$. 
Indeed, $\nu$ still lies in $\sect(\conn(P))$ since $\rho\circ A=\rho\circ\pr_2$, 
where $\pr_2:T\Ad(P)\times TP/G\to TP/G$ denotes the projection on the second factor.
This defines a map $A(\hat{f},\cdot):\sect(\conn(P))\to\sect(\conn(P))$ 
and the claim is that this map coincides with $\sect(\conn(f))$. 
According to\ Remark\ \ref{remConnFunctor} and recalling that $\ul{f}=\id_M$, 
we only have to check that, given $\lambda\in\sect(\conn(P))$, 
$\nu=A\circ(\hat{f}_\ast,\lambda)$ fulfils the requirement $T(f)/G\circ\nu=\lambda$. 
Let $x\in M$ and $v\in T_xM$. Take $p\in P_x$ and $X\in T_pP$ such that $[X]=\lambda(v)$. 
Recalling the definition of $\hat{f}\in\sect(\ad(P))$ in terms of $\tilde{f}\in\c(P,G)$, 
{\em cfr.}\ Proposition\ \ref{prpGaugeGroup}, $\hat{f}_\ast v=[X,\tilde{f}_\ast X]$. 
Therefore $\nu(v)=[X(\tilde{f}_\ast X)^{-1}]$. Recalling also the definition of 
$\tilde{f}\in\c(P,G)$ in terms of $f\in\gauge(P)$ of\ Proposition\ \ref{prpGaugeGroup} 
and keeping in mind that $\tilde{(f^{-1})}=\tilde{f}(\cdot)^{-1}$, one gets $\nu(v)=[{f^{-1}}_\ast X]$. 
This allows us to conclude that $T(f)/G(\nu(v))=\lambda(v)$ for each $v\in TM$, thus proving the claim. 
\end{proof}

\begin{remark}[Action of $\gauge(P)$ on $\sect(\conn(P))$ for Abelian structure groups]\label{remGaugeConnAb}
By Proposition\ \ref{prpTrivAdBun}, it is possible to give a very explicit formula for the action of $f\in\gauge(P)$ 
on $\lambda\in\sect(\conn(P))$ for a principal $G$-bundle $P$ whose structure group is Abelian. 
First of all, note that in this case $\hat{f}\in\c(M,G)$. Furthermore, $\Ad(P)\simeq M\times G$, 
therefore $T\Ad(P)\simeq TM\times TG$. Let $x\in M$ and $v\in T_xM$ and choose $p\in P_x$ and $X\in T_pP$ 
such that $[X]=\lambda(v)$. Following the proof of\ Proposition\ \ref{prpGaugeConn} 
and denoting $A\circ(\hat{f}_\ast,\lambda)$ with $\nu$, we realize that 
\begin{equation*}
\nu(v)=[X(\hat{f}_\ast v)^{-1}]=[X\hat{f}(x)^{-1}-p\hat{f}(x)^{-1}(\hat{f}_\ast v)\hat{f}(x)^{-1}]
=\lambda(v)-\kappa[p,\hat{f}^\ast\mu(v)]\,,
\end{equation*}
where $\mu\in\f^1(G,\lieg)$ denotes the Maurer-Cartan form of the Abelian Lie group $G$, 
defined by $\mu(\Xi)=g^{-1}\Xi$ for $g\in G$ and $\Xi\in T_gG$, see\ {\em e.g}\ \cite[Definition\ 4.6]{Ish99}. 
In terms of the affine structure of the bundle of connections, the last equation reads 
\begin{equation*}
\sect(\conn(f))\lambda=A(\hat{f}_\ast,\lambda)=\lambda-\hat{f}^\ast\mu\,,
\end{equation*}
where we applied Proposition \ref{prpGaugeConn} to obtain the first equality. 
\end{remark}

Let us mention that there are several equivalent definitions of a connection on a principal $G$-bundle, 
see\ {\em e.g.}\ \cite[Subsection\ 6.1.1]{Ish99} for two possible alternative approaches. 
Definition\ \ref{defConn} has the advantage of highlighting the fact that 
connections can be described as sections of a suitable affine bundle. 
This feature, at least for Abelian structure groups, proves very helpful in the attempt of defining 
a sufficiently rich space of regular functionals on the space of connections, see\ Chapter\ \ref{chYangMills}. 
We introduce now the notion of a connection form, which turns out to be convenient 
to capture the contravariant behavior of the space of connections with respect to principal $G$-bundle maps 
and to specify the curvature associated to a given connection. 

\begin{definition}\label{defConnForm}\index{connection!form}
Let $G$ be a Lie group and consider a principal $G$-bundle $P$ over a manifold $M$. 
A {\em connection form} on $P$ is a $\lieg$-valued $1$-form $\omega\in\f^1(P,\lieg)$ 
fulfilling the following properties: 
\begin{enumerate}
\item $\omega(X_p^\xi)=\xi$ for each $p\in P$ and $\xi\in\lieg$, 
where $X_p^\xi\in T_pP$ is the vertical tangent vector at $p$ generated by $\xi$ 
(see the construction of the Atiyah sequence \eqref{eqAtiyahSeq} for further details); 
\item $\omega$ is equivariant with respect to the right $G$-action induced by the adjoint representation on $\lieg$, 
namely $r_g^\ast\omega=\ad_{g^{-1}}\circ\omega$. 
\end{enumerate}
We denote the space of connection forms on $P$ by $\connf(P)$. 
\end{definition}

The space $\connf(P)$ of connection forms is an affine space modeled 
on the space of horizontal $1$-forms $\f^1_\hor(P,\lieg)^\eqv$ 
which are equivariant with respect to the adjoint right $G$-action on $\lieg$, 
namely $\theta\in\f^1_\hor(P,\lieg)^\eqv$ is a $\lieg$-valued $1$-form on $P$ 
such that $\theta(X)=0$ for each vertical tangent vector $X\in VP$ 
and $r_g^\ast\theta=\ad_{g^{-1}}\circ\theta$ for each $g\in G$. 
In fact, it is clear from the definition of $\f^1_\hor(P,\lieg)^\eqv$ that shifting $\omega\in\connf(P)$ 
by $\theta\in\f^1_\hor(P,\lieg)^\eqv$ gives a new connection form $\omega+\theta$. 
Conversely, $\omega,\omega^\prime\in\connf(P)$ always differ by a suitable $\theta\in\f^1_\hor(P,\lieg)^\eqv$. 

\begin{remark}[Natural equivalence between $\connf(\cdot)$ and $\sect(\conn(\cdot))$]
\label{remConnFormToConn}
Let $G$ be a Lie group and consider a principal $G$-bundle $P$ over a manifold $M$. 
As already anticipated, the space of sections of the bundle of connections and the space of connection forms 
are very closely related to each other. In fact, a section $\lambda\in\sect(\conn(P))$ of $\conn(P)$ 
specifies a splitting (on the right) of the Atiyah sequence at each point. 
Furthermore, this splitting \quotes{varies smoothly} with the point on the base manifold. 
Indeed smooth right-splittings of the Atiyah sequence are nothing but sections of $\conn(P)$. 
Let now $\omega\in\connf(P)$ be a connection form and consider the map 
\begin{align*}
TP\to P\times\lieg\,, && X\in T_pP\to (p,\omega(X))\,.
\end{align*}
Equivariance of connection forms entails that the map defined above descends to $\tilde{\omega}:TP/G\to\ad(P)$. 
In fact, the right $G$-action on the source is translated into the right $G$-action on the target 
induced by the adjoint representation of $G$ on $\lieg$, whose orbit space defines $\ad(P)$. 
Recalling\ \eqref{eqAtiyahSeq}, the first property listed in\ Definition\ \ref{defConnForm}, 
namely the fact that $\omega$ maps each vertical tangent vector to its Lie algebra generator, 
turns into the identity $\tilde{\omega}\circ\kappa=\id_{\ad(P)}$, 
meaning that $\tilde{\omega}$ splits the Atiyah sequence everywhere on the left. 
Indeed, there is a one-to-one correspondence between right and left splittings of short exact sequences 
(in fact one usually does not distinguish between the two). In this specific case, this correspondence is realized 
as follows. Given $\lambda\in\sect(\conn(P)))$, $\rho\circ\lambda=\id_{TM}$, whence 
$\id_{TP/G}-\lambda\circ\rho$ factors through $\kappa:\ad(P)\to TP/G$ by exactness of the Atiyah sequence, 
thus defining $\tilde{\omega}:TP/G\to\ad(P)$ such that $\kappa\circ\tilde{\omega}=\id_{TP/G}-\lambda\circ\rho$, 
which entails that $\tilde{\omega}\circ\kappa=\id_{\ad(P)}$. 
In particular, we obtain an affine space isomorphism 
\begin{align}\label{eqConnFormToConn}
\connf(P)\to\sect(\conn(P))\,, && \omega\mapsto\lambda_\omega\,,
\end{align}
where $\lambda_\omega\in\sect(\conn(P))$ is defined by 
$\lambda_\omega\circ\rho=\id_{TP/G}-\kappa\circ\tilde{\omega}$. 
The linear part of this isomorphism is given by the vector space isomorphism 
\begin{align*}
\f^1_\hor(P,\lieg)^\eqv\to\f^1(M,\ad(P))\,, && \theta\mapsto\sigma_\theta\,,
\end{align*}
where $\sigma_\theta\in\f^1(M,\ad(P))$ is defined, for each $x\in M$ and $v\in T_xM$, 
by $\sigma_\theta(v)=-[p,\theta(X)]\in\ad(P)_x$ for an arbitrary choice of 
$p\in P_x$ and $X\in T_pP$ such that $\pi_\ast X=v$. 
\end{remark}

As already mentioned, one of the advantages of looking at connection forms 
rather than sections of the bundle of connections is that, in this perspective, 
the way principal $G$-bundle maps act on connections is more easily understood. 
In fact, given a principal bundle map $f:P\to Q$ covering $\ul{f}:M\to N$, 
we realize that the pull-back $f^\ast:\f^1(Q,\lieg)\to\f^1(P,\lieg)$ maps $\connf(Q)$ to $\connf(P)$. 
Denoting the restriction of $f^\ast$ to $\connf(Q)$ with $\connf(f):\connf(Q)\to\connf(P)$, 
it follows that $\connf(\cdot)$ is a contravariant functor from the category of principal $G$-bundles 
to the category of affine spaces. 
We can now compare the contravariant functors $\connf(\cdot)$ and $\sect(\conn(\cdot))$ 
exploiting the isomorphism in\ \eqref{eqConnFormToConn}. 
As a matter of fact, whenever a principal bundle map $f:P\to Q$ is given, we have A commutative diagram: 
\begin{equation*}
\xymatrix@C=2.5em{
\connf(Q)\ar[d]_\simeq\ar[r]^{\connf(f)} & \connf(P)\ar[d]^\simeq\\
\sect(\conn(Q))\ar[r]_{\sect(\conn(f))} & \sect(\conn(P))
}
\end{equation*}
The diagram above can be interpreted saying that \eqref{eqConnFormToConn} 
is a natural equivalence between the contravariant functors $\connf(\cdot)$ and $\sect(\conn(\cdot))$. 

\subsection{Curvature of a connection}
Our next goal is to introduce the curvature associated to a connection form, 
and therefore the curvature of the corresponding section of the bundle of connections as well. 
To define it, a map providing the horizontal part of a vector tangent to a principal bundle 
with respect to a connection form on such bundle is needed. 
Given a principal $G$-bundle $\pi:P\to M$ and a connection form $\omega\in\connf(P)$, 
we provide a prescription to lift each tangent vector $v\in TM$ over $x\in M$ 
to a tangent vector $v_p^\uparrow\in TP$ over any point $p\in P_x$ 
which is horizontal with respect to the connection form $\omega$, namely $\omega(v_p^\uparrow)=0$. 
In fact, consider $X\in T_pP$ such that $\pi_\ast X=v$ ($X$ exists since $\pi$ is the principal bundle projection) 
and take the vertical tangent vector $X_p^{\omega(X)}\in T_pP$ at $p$ generated by $\omega(X)\in\lieg$. 
$v_p^\uparrow=X-X_p^{\omega(X)}\in T_pP$ still lies in the preimage of $v$ under $\pi_\ast$ 
and moreover $\omega(v_p^\uparrow)=\omega(X)-\omega(X)=0$ by the properties of the connection form. 
Let $Y\in T_pP$ be such that $\omega(Y)=0$ and $\pi_\ast Y=v$. Then $Z=Y-v_p^\uparrow$ is vertical, 
therefore $Z=X_p^{\omega(Z)}=0$, whence $Y=v_p^\uparrow$. 
This shows that, given $\omega\in\connf(P)$, there exists a well-defined vector bundle map 
which realizes the lift prescription we were looking for: 
\begin{align*}
{}^\uparrow:\pi^\ast(TM)\to TP\,, && (p,v)\mapsto v_p^\uparrow\,,
\end{align*}
where $\pi^\ast(TM)$ denotes the pull-back of the vector bundle $TM$ under $\pi:P\to M$. 
In particular, we can assign to each vector tangent to $P$, its horizontal part with respect to $\omega$: 
\begin{align}\label{eqHorPart}
\hor_\omega:TP\to TP\,, && X\in T_pP\mapsto(\pi_\ast X)_p^\uparrow\,.
\end{align}
Note that by the arguments presented above $\hor_\omega$ decomposes $TP$ 
in the direct sum of its image $\hor_\omega(TP)\simeq\pi^\ast(TM)$ and $VP\simeq P\times\lieg$. 
In fact, ${}^\uparrow$ is nothing but a right-splitting of the $G$-equivariant counterpart of the Atiyah sequence, 
which is a short exact sequence of $G$-equivariant vector bundle maps covering $\id_P$, 
{\em cfr.}\ \eqref{eqAtiyahSeq}: 
\begin{equation*}
0\lra VP\overset{\subseteq}{\lra}TP\overset{\pi_\ast}{\lra}\pi^\ast(TM)\lra0\,.
\end{equation*}

\begin{definition}\label{defCurvConnForm}\index{curvature}\index{connection!curvature}
Let $G$ be a Lie group and consider a principal $G$-bundle $P$ over a manifold $M$. 
Given a connection form $\omega\in\connf(P)$, consider the associated vector bundle map 
$\hor_\omega:TP\to TP$, which provides the $\omega$-horizontal part of each vector tangent to $P$, 
see\ eq.\ \eqref{eqHorPart} and the previous discussion. The {\em curvature} 
$\curv_\omega\in\f^2_\hor(P,\lieg)^\eqv$ of the connection form $\omega$ is 
the horizontal and $G$-equivariant $\lieg$-valued $2$-form on $P$ defined by 
$\curv_\omega(X,Y)=\dd\omega(\hor_\omega(X),\hor_\omega(Y))$ for each $p\in P$ and $X,Y\in T_pP$. 
\end{definition}

As already mentioned in the definition, the curvature $\curv_\omega\in\f^2_\hor(P,\lieg)^\eqv$ 
of a connection form $\omega\in\connf(P)$ is both horizontal and equivariant 
with respect to the adjoint right $G$-action on $\lieg$, 
namely it vanishes whenever one of its arguments lies in $VP$ 
and it is such that $r_g^\ast\curv_\omega=\ad_{g^{-1}}\circ\curv_\omega$ for each $g\in G$. 
Being horizontal follows immediately from the fact that the kernel of $\hor_\omega$ is exactly $VP$, 
{\em cfr.}\ \eqref{eqHorPart}, while equivariance is inherited from $\omega$ and $\hor_\omega$. 
These features allow us to equivalently represent the curvature $\curv_\omega$ 
by an $\ad(P)$-valued $2$-form $F_\omega$ on $M$ by setting, for each $x\in M$ and $v,w\in T_xM$, 
\begin{equation}\label{eqCurvConn}
F_\omega(v,w)=[p,\curv_\omega(X,Y)]\,,
\end{equation}
for an arbitrary choice of $p\in P_x$ and $X,Y\in T_pP$ 
such that $\pi_\ast X=v$ and $\pi_\ast Y=w$. 
In fact, for each degree $k$ there is a vector space isomorphism $\f^k_\hor(P,\lieg)\simeq\f^k(M,\ad(P))$, 
{\em cfr.}\ \cite[Chapter\ 2, Example\ 5.2]{KN96} or \cite[Satz\ 3.5]{Bau14}. 
Furthermore, we observe that Definition\ \ref{defCurvConnForm} implicitly defines a map 
\begin{align*}
\curv:\connf(P)\to\f^2_\hor(P,\lieg)^\eqv\,, && \omega\mapsto\curv_\omega\,.
\end{align*}
These facts, together with\ \eqref{eqConnFormToConn}, motivate the following definition. 

\begin{definition}\label{defCurvMap}\index{curvature!map}
Let $G$ be a Lie group and consider a principal $G$-bundle $P$ over a manifold $M$. 
The curvature $F_\lambda\in\f^2(M,\ad(P))$ of a connection $\lambda\in\sect(\conn(P))$ 
is the curvature $F_\omega\in\f^2(M,\ad(P))$ of the corresponding connection form $\omega\in\connf(P)$, 
{\em cfr.}\ \eqref{eqConnFormToConn} and\ \eqref{eqCurvConn}. In particular, we have the {\em curvature map} 
\begin{align*}
F:\sect(\conn(P))\to\f^1(M,\ad(P))\,, && \lambda\mapsto F_\lambda\,.
\end{align*}
\end{definition}

Let $P$ and $Q$ be principal $G$-bundles over $M$ and respectively $N$. 
By naturality of the Atiyah sequence and of the exterior derivative $\dd$, 
for each principal $G$-bundle map $f:P\to Q$ and each connection form $\omega\in\connf(Q)$, we get 
\begin{equation*}
\curv(f^\ast\omega)=(\dd f^\ast\omega)\circ\hor_{f^\ast\omega}=\dd\omega\circ\hor_\omega\circ f_\ast
=f^\ast(\curv(\omega))\,.
\end{equation*}
Note that $\ad(f):\ad(P)\to\ad(Q)$ is fiberwise injective. In particular, we can invert it on a given fiber. 
This allows us to define 
\begin{align*}
\f^2(f):\f^2(N,\ad(Q))\to\f^2(M,\ad(P))\,, && \omega\mapsto\omega^\prime\,,
\end{align*}
$\omega^\prime\in\f^2(M,\ad(P))$ being specified 
by the condition $\ad(f)\circ\omega^\prime=\omega\circ\ul{f}_\ast$, 
where both $\omega:TN\to\ad(Q)$ and $\omega^\prime:TM\to\ad(P)$ are regarded here as vector bundle maps. 
where $\omega\in\f^2(N,\ad(Q))$ is regarded as a vector bundle map $\omega:TN\to\ad(Q)$. 
With this definition, the equation displayed above entails that the following diagram commutes: 
\begin{equation}\label{eqCurvNat}
\xymatrix{
\sect(\conn(Q))\ar[d]_F\ar[r]^{\sect(\conn(f))} & \sect(\conn(P))\ar[d]^F\\
\f^2(N,\ad(Q))\ar[r]_{\f^2(f)} & \f^2(M,\ad(P))
}
\end{equation}
In particular, this diagram illustrates how the curvature of a connection transforms 
under the action of a gauge transformation. 

The subsequent theorem provides a very well-known formula for the curvature of a connection form. 
The proof of this result, originally due to Cartan, is easily available in the literature, 
see\ {\em e.g.} \cite[Theorem\ 6.4]{Ish99}. 

\begin{theorem}[Cartan's structural equation]\label{thmCartan}
Let $G$ be a Lie group and consider a principal $G$-bundle $P$ over $M$. Let $\omega\in\connf(P)$ be 
a connection form and denote its curvature with $\curv_\omega\in\f^2_\hor(P,\lieg)^\eqv$. 
The following identity holds true for all $p\in P$ and $X,Y\in T_pP$: 
\begin{equation*}
\curv_\omega(X,Y)=\dd\omega(X,Y)+[\omega(X),\omega(Y)]\,,
\end{equation*}
where $[\cdot,\cdot]:\lieg\times\lieg\to\lieg$ denotes the Lie bracket. 
\end{theorem}

\begin{remark}[Curvature map for Abelian structure groups]\label{remCurvAb}
In view of\ Chapter\ \ref{chYangMills}, we focus our attention on the case of principal $G$-bundles $P$ 
having an Abelian structure Lie group $G$. 

Since $\ad(P)\simeq M\times\lieg$, {\em cfr.}\ Proposition\ \ref{prpTrivAdBun}, 
the curvature $F_\lambda$ of a connection $\lambda\in\sect(\conn(P))$ 
becomes a $\lieg$-valued $2$-form on the base space $M$ of $P$. 
Accordingly, given a principal $G$-bundle map $f:P\to Q$, $\f^2(f):\f^2(N,\ad(Q))\to\f^2(M,\ad(P))$ 
reduces to the pull-back $\ul{f}^\ast:\f^2(N,\lieg)\to\f^2(M,\lieg)$ under the base map $\ul{f}:M\to N$. 
In particular, if we look at a gauge transformation $f\in\gauge(P)$, we realize that 
the curvature is invariant under gauge transformations. In fact, since $\ul{f}=\id_M$ for $f\in\gauge(P)$, 
the bottom arrow in\ \eqref{eqCurvNat} reduces to the identity for $G$ Abelian, 
whence $F\circ\sect(\conn(f))=F$, meaning that the curvature of a connection $\lambda\in\sect(\conn(P))$ 
coincides with the curvature of the connection $\sect(\conn(f))\lambda$ after the gauge transformation $f$. 

Let us also mention that, for $G$ Abelian, the curvature $F:\sect(\conn(P))\to\f^2(M,\lieg)$ turns out to be 
an affine map (we are regarding the vector space $\f^2(M,\lieg)$ as an affine space modeled on itself). 
To exhibit the linear part of $F$ (thus proving that $F$ is indeed affine), we note that 
in this case $\sect(\conn(P))$ is an affine space modeled on $\f^1(M,\lieg)$, $G$ being Abelian. 
Therefore we can consider the linear map $F_V=-\dd:\f^1(M,\lieg)\to\f^2(M,\lieg)$. 
Our aim consists in showing that $F(\lambda+\sigma)=F(\lambda)+F_V(\sigma)$ 
for each $\lambda\in\sect(\conn(P))$ and each $\sigma\in\f^1(M,\lieg)$. 
Denoting with $\omega_\lambda\in\connf(P)$ and with $\theta_\sigma\in\f^1_\hor(P,\lieg)^\eqv\,$\footnote{$G$ 
being Abelian, the requirement of equivariance under the right adjoint $G$-action on $\lieg$ 
actually reduces to invariance.} the forms obtained from $\lambda$ and respectively from $\theta$ 
via the affine isomorphism of\ \eqref{eqConnFormToConn} and its linear part, we can express 
$F(\lambda+\sigma)$ in terms of $\curv(\omega_\lambda+\theta_\sigma)\in\f^2_\hor(P,\lieg)^\eqv$ 
and $F(\lambda)$ in terms of $\curv(\omega_\lambda)\in\f^2_\hor(P,\lieg)^\eqv$, 
{\em cfr.}\ Definition\ \ref{defCurvMap}. Exploiting Cartan's structural equation 
and noting that the Lie bracket on the Lie algebra of an Abelian group is trivial, we conclude that
\begin{equation*}
\curv(\omega_\lambda+\theta_\sigma)=\dd(\omega_\lambda+\theta_\sigma)
=\dd\omega_\lambda+\dd\theta_\sigma=\curv(\omega_\lambda)+\dd\theta_\sigma\,.
\end{equation*}
This identity already shows that, for $G$ Abelian, $\curv:\connf(P)\to\f^2_\hor(P,\lieg)^\eqv$ is an affine map 
whose linear part $\curv_V:\f^1_\hor(P,\lieg)^\eqv\to\f^2_\hor(P,\lieg)^\eqv$ is provided by 
the restriction to $\f^1_\hor(P,\lieg)^\eqv$ of $\dd:\f^1(P,\lieg)\to\f^2(P,\lieg)$.\footnote{As one 
can directly check, $\dd$ always preserves the equivariance property. 
Furthermore, $\dd$ maps horizontal forms to horizontal forms in the Abelian case. 
In fact, for $p\in P$ and $X\in V_pP$, there exists $\xi\in\lieg$ such that $X_p^\xi=X$, 
therefore an extension of $X$ is given by the vertical vector field $q\in P\mapsto X_q^\xi\in T_qP$. 
Furthermore, for $Y\in T_pP$, we can always consider a vector field $\tilde{Y}$ equivariantly extending $Y$, 
namely such that $\tilde{Y}_p=Y$ and $\tilde{Y}_{qg}=r_{g\,\ast}\tilde{Y}_q$ for all $q\in P$ and $g\in G$. 
Using the vector fields $X^\xi$ and $\tilde{Y}$ to evaluate $\dd\theta(X,Y)$ shows that 
this term vanishes whenever $\theta$ lies in $\f^1_\hor(P,\lieg)^\eqv$ and $G$ is Abelian.} 
To conclude that $F(\lambda+\sigma)=F(\lambda)+F_V(\sigma)$, consider $x\in M$ and $v,w\in T_xM$. 
Furthermore, choose $p\in P_x$ and $X,Y\in T_pP$ such that $\pi_\ast X=v$ and $\pi_\ast Y=w$, 
$\pi:P\to M$ being the projection from the total space to the base of the principal bundle $P$. 
Recalling Definition\ \ref{defCurvMap} and taking into account that $\ad(P)\simeq M\times\lieg$, from the equation 
displayed above we deduce $F_{\lambda+\sigma}(v,w)=F_\lambda(v,w)+\dd\theta_\sigma(X,Y)$. 
By the arguments in\ Remark\ \ref{remConnFormToConn} and in particular the linear part of the affine isomorphism 
displayed in\ \eqref{eqConnFormToConn}, we conclude that $\theta_\sigma=-\pi^\ast\sigma$ for $G$ Abelian, 
whence $F_{\lambda+\sigma}(v,w)=F_\lambda(v,w)-\dd\sigma(v,w)$. 
Summing up, for a principal bundle $P$ with an Abelian structure group $G$, 
the curvature $F:\sect(\conn(P))\to\f^2(M,\lieg)$ turns out to be an affine map. 

The commutative diagram in\ \eqref{eqCurvNat} for $G$ Abelian expresses the fact that 
$F:\sect(\conn(\cdot))\to\f^2((\cdot)_\base,\lieg)$ is a natural transformation between contravariant functors 
from the category of principal $G$-bundles to the category of affine spaces $\Aff$. 
Note that, for each principal $G$-bundle $P$, $\f^2((P)_\base,\lieg)=\f^2(M,\lieg)$ is indeed a vector space, 
now regarded as an affine space modeled on itself. 
Even more, each component of the natural transformation $F$ is an affine differential operator 
in the sense of\ Definition\ \ref{defAffDiffOp}. In fact, for $G$ Abelian, on each principal $G$-bundle $P$, 
the linear part of $F:\sect(\conn(P))\to\f^2(M,\lieg)$ is the linear differential operator 
$F_V=-\dd:\f^1(M,\lieg)\to\f^2(M,\lieg)$. This feature expresses in a mathematical language the fact that 
Yang-Mills field theories with Abelian structure groups do not exhibit self-interactions. 
This fact simplifies a lot the analysis of the dynamics of the equation of motion, 
see\ Chapter\ \ref{chYangMills} for further details. 
\end{remark}

\section{Quantization}\label{secQuantization}
In this section we present a procedure developed in \cite{MSTV73} to quantize presymplectic Abelian groups. 
This generalizes the theory of Weyl systems and CCR-representations for symplectic vector spaces, 
which is very well-understood, for example refer to\ \cite{BGP07, BR87}. 
Note that a generalization is also available for presymplectic vector spaces \cite{BHR04}. 
Furthermore, we analyze certain properties of this construction which are relevant for locally covariant field theories. 
This material is also available in\ \cite[Appendix\ A]{BDHS14}. 
For the sake of completeness, we will briefly sketch some of the proofs. 
As a starting point let us introduce the category of presymplectic Abelian groups. 

\begin{definition}\index{presymplectic!Abelian group}\index{presymplectic!homomorphism}
\index{presymplectic!Abelian group!category}
A {\em presymplectic Abelian group} $(H,\rho)$ consists in the assignment of an Abelian group $H$ 
endowed with a presymplectic form $\rho:H\times H\to\bbR$, 
namely a (possibly degenerate) anti-symmetric bi-homomorphism. 
A {\em presymplectic homomorphism} $L$ between the presymplectic Abelian groups $(H,\rho)$ and $(K,\sigma)$ 
is a group homomorphism $L:H\to K$ preserving the presymplectic structures, 
namely $\sigma\circ(L\times L)=\rho$. 
The {\em category of presymplectic Abelian groups} $\PSymA$ has presymplectic Abelian groups as objects 
and presymplectic homomorphisms as morphisms. 
\end{definition}

Following\ \cite{MSTV73}, to a presymplectic Abelian group $(H,\rho)$, 
one can associate a $\ast$-algebra over the field $\bbC$: 
Consider the vector space $\Delta(H,\sigma)$ spanned over $\bbC$ by the set of symbols $\{\weyl_h:h\in H\}$ 
and endow it with the structure of an associative unital $\ast$-algebra 
introducing a product and an involution via the following identities, 
known as {\em Weyl relations} or {\em canonical commutation relations (CCR)}: 
\begin{align}\label{eqWeylRel}\index{Weyl relations}\index{canonical commutation relation}\index{CCR}
\weyl_h\,\weyl_k=\exp\left(-\frac{i}{2}\rho(h,k)\right)\weyl(h+k)\,, &&
\weyl_h^\ast=\weyl_{-h}\,.
\end{align}
Note that $W_0=\bbone$ is the unit of $\Delta(H,\sigma)$. 
To each presymplectic homomorphism $L:(H,\rho)\to (K,\sigma)$, 
one can also associate a $\ast$-homomorphism 
\begin{align}\label{eqQuantMorph}
\Delta(L):\Delta(H,\rho)\to\Delta(K,\sigma)\,, && \Delta(L)\weyl_h=\weyl_{Lh}\,,\quad\forall\,h\in H\,. 
\end{align}
Denoting with $\Alg$ the category of unital $\ast$-algebras, 
one can easily check that $\Delta:\PSymA\to\Alg$ is a covariant functor. 
Note that $\Delta:\PSymA\to\Alg$ preserves injectivity of morphisms. 
In fact, suppose $a\in\Delta(H,\rho)$ is such that $\Delta(L)a=0$. 
By definition $a$ can be expressed as a $\bbC$-linear combination $a=\sum_{h\in S}c_h\weyl_h$ 
labeled by a finite subset $S\subseteq H$, 
where $\{W_h\}$ is a family of linearly independent elements of $\Delta(H\,\rho)$. 
If $L$ is injective, by definition it follows that $\{W_{Lh}\}$ is 
a collection of linearly independent elements of $\Delta(K,\sigma)$. 
Furthermore, by assumption $\Delta(L)a=\sum_{h\in S}c_h\weyl_{Lh}$ vanishes. 
Therefore $c_h=0$ for each $h\in S$, whence $a=0$, thus showing that $\Delta(L)$ is injective. 

We want to endow $\Delta(H,\rho)$ with a $C^\ast$-norm. 
To do so, we first introduce an auxiliary norm $\normb{\cdot}$ on $\Delta(H,\rho)$: 
For each finite subset $S\subseteq H$ and each collection $\{c_h\in\bbC:h\in S\}$ of complex numbers, we set 
\begin{equation}\label{eqNormb}
\Normb{\sum_{h\in S}c_h\weyl_h}=\sum_{h\in S}\vert c_h\vert\,.
\end{equation}
This norm fulfils the following inequalities for each $a,b\in\Delta(H,\sigma)$: 
\begin{align*}
\normb{a\,b}\leq\normb{a}\,\normb{b}\,, && \normb{a^\ast}=\normb{a}\,.
\end{align*}
Therefore the completion of $\Delta(H,\sigma)$ with respect to $\normb{\cdot}$ 
provides a unital Banach $\ast$-algebra $\balg(H,\rho)$ whose elements are of the form $\sum_{h\in S}c_hW_h$ 
for a countable subset $S\subseteq H$ and a collection $\{c_h\in\bbC:h\in S\}$ of complex numbers 
such that $\sum_{h\in S}\vert c_h\vert<+\infty$. 
Furthermore, if a presymplectic homomorphism $L:(H,\rho)\to(K,\sigma)$ is given, 
and we endow $\Delta(H,\rho)$ and $\Delta(K,\sigma)$ with the norms $\normb{\cdot}_1$ 
and respectively $\normb{\cdot}_2$ defined according to\ \eqref{eqNormb}, 
one can check that the $\ast$-homomorphism $\Delta(L):\Delta(H,\rho)\to\Delta(K,\sigma)$ is bounded by $1$, 
namely $\normb{\Delta(L)a}_2\leq\normb{a}_1$ for each $a\in\Delta(H,\rho)$. 
As a consequence, one gets a continuous extension to the completions, 
namely a continuous $\ast$-algebra homomorphism $\balg(L):\balg(H,\rho)\to\balg(K,\sigma)$. 
Even more, if $L$ is injective, $\Delta(L)$ is an isometry, therefore $\balg(L)$ is an isometry too. 
In fact, denoting with $\BAlg$ the category of Banach $\ast$-algebras, 
it turns out that $\balg:\PSymA\to\BAlg$ is a covariant functor which assigns to each 
injective presymplectic homomorphism an isometric $\ast$-homomorphism between Banach $\ast$-algebras. 

The second step to define a $C^\ast$-norm on $\Delta(H,\rho)$ consists in considering states on $\balg(H,\rho)$. 

\begin{definition}\index{state}\index{state!faithful}
Let $A$ be a unital $\ast$-algebra over $\bbC$. 
A {\em state} $\omega:A\to\bbC$ is a linear functional which is positive and normalized, 
namely such that $\omega(a^\ast\,a)\geq0$ for each $a\in A$ and $\omega(\bbone)=1$. 
A state $\omega:A\to\bbC$ is {\em faithful} if $\omega(a^\ast\,a)=0$ implies $a=0$. 
\end{definition}

In case a topological unital $\ast$-algebra is considered, states are usually also required to be continuous. 
Proposition\ 2.17 of\ \cite{MSTV73} shows that one can obtain continuous positive linear functionals 
on $\balg(H,\rho)$ from positive linear functionals on $\Delta(H,\rho)$. We recall this result below. 

\begin{proposition}\label{prpPosLinFunc}
Let $(H,\rho)$ be a presymplectic Abelian group and $\omega:\Delta(H,\rho)\to\bbC$ be a positive linear functional. 
Then $\omega$ is continuous with respect to $\normb{\cdot}$ 
and it can be extended by continuity to a continuous positive linear functional $\omega:\balg(H,\rho)\to\bbC$. 
Furthermore, if $\omega:\Delta(H,\rho)\to\bbC$ is a state on $\Delta(H,\rho)$, 
its extension $\omega:\balg(H,\rho)\to\bbC$ is a state on $\balg(H,\rho)$. 
\end{proposition}

\begin{proof}
Each positive linear functional $\omega$ on a $\ast$-algebra always provides a Cauchy-Schwarz inequality. 
In the case under analysis, one has $\vert\omega(b^\ast\,a)\vert^2\leq\omega(a^\ast\,a)\,\omega(b^\ast\,b)$ 
for each $a,b\in\Delta(H,\rho)$. Applying this inequality to $b=\bbone\in\Delta(H,\rho)$ 
and $a=\weyl_h\in\Delta(H,\rho)$ for an arbitrary choice of $h\in H$, 
one gets $\vert\omega(\weyl_h)\vert\leq\omega(\bbone)$. 
Recalling that $\{\weyl_h\,:\; h\in H\}$ generates $\Delta(H,\rho)$, from the last inequality 
it follows that, for each $a\in\Delta(H,\rho)$, $\vert\omega(a)\vert\leq\omega(\bbone)\,\normb{a}$. 
Then $\omega:\Delta(H,\rho)\to\bbC$ is continuous with respect to $\normb{\cdot}$ 
and its extension by continuity $\omega:\balg(H,\rho)\to\bbC$ is a continuous positive linear functional. 
Furthermore, if we also assume $\omega:\Delta(H,\rho)\to\bbC$ to be normalized, namely $\omega(\bbone)=1$, 
clearly its extension $\omega:\balg(H,\rho)\to\bbC$ is normalized as well, hence a state on $\balg(H,\rho)$. 
\end{proof}

It is easy to exhibit a faithful state on $\balg(H,\rho)$: 
Define a state $\tilde{\omega}:\Delta(H,\rho)\to\bbC$ imposing $\tilde{\omega}(\weyl_h)=0$ 
for all $h\in H\setminus\{0\}$ and $\tilde{\omega}(\bbone)=1$. 
Proposition\ \ref{prpPosLinFunc} shows that $\tilde{\omega}$ can be extended 
to a state $\tilde{\omega}:\balg(H,\rho)\to\bbC$ by continuity. One can check that this state is faithful. 
In fact, suppose $a\in\balg(H,\rho)$ is such that $\tilde{\omega}(a^\ast\,a)=0$. 
Since $a$ can be written as $a=\sum_{h\in S}c_h\weyl_h$ for a countable subset $S\subseteq H$ 
and a collection $\{c_h\in\bbC:h\in S\}$ of complex numbers such that $\sum_{h\in S}\vert c_h\vert<+\infty$, 
the definition of $\tilde{\omega}$ entails that 
\begin{equation*}
\tilde{\omega}(a^\ast\,a)=\sum_{h\in S}\sum_{k\in S}\overline{c_h}c_k
\exp\left(\frac{i}{2}\rho(h,k)\right)\tilde{\omega}(\weyl_{h-k})=\sum_{h\in S}\vert c_h\vert^2\,.
\end{equation*}
Therefore $\tilde{\omega}(a^\ast\,a)=0$ implies $c_h=0$ for all $h\in S$, that is to say $a=0$. 
The existence of at least one faithful state on $\balg(H,\rho)$ enables us 
to introduce another norm on $\balg(H,\rho)$. 

\begin{definition}\label{defMinRegNorm}\index{minimal regular norm}
Let $(H,\rho)$ be a presymplectic Abelian group and denote with $\states$ the set of states on $\balg(H,\rho)$. 
The {\em minimal regular norm} $\norm{\cdot}$ on $\balg(H,\lieg)$ is defined by 
\begin{align*}
\norm{a}=\sup_{\omega\in\states}\sqrt{\omega(a^\ast\,a)}\,, && a\in\balg(H,\rho)\,.
\end{align*}
\end{definition}

Notice that the supremum is bounded from above by $\normb{\cdot}$. 
In fact, for each $\omega\in\states$ and each $a\in\balg(H,\rho)$, $\vert\omega(a)\vert\leq\normb{a}$, 
see the proof of\ Proposition\ \ref{prpPosLinFunc}. 
One can check that $\norm{\cdot}$ is a norm on $\balg(H,\rho)$ compatible with the algebraic product, 
namely $\norm{a\,b}\leq\norm{a}\norm{b}$ for each $a,b\in\balg(H,\rho)$ directly from its definition 
and the fact that there exists a faithful state on $\balg(H,\rho)$. 
The next lemma shows that $\norm{\cdot}$ is also a $C^\ast$-norm on $\balg(H,\rho)$. 

\begin{lemma}
Let $(H,\rho)$ be a presymplectic Abelian group. The minimal regular norm $\norm{\cdot}$ 
introduced in\ Definition\ \ref{defMinRegNorm} is a $C^\ast$-norm on $\balg(H,\rho)$, 
namely $\norm{a^\ast}=\norm{a}$ and $\norm{a^\ast\,a}=\norm{a}^2$ for each $a\in\balg(H,\rho)$. 
In particular, the completion of $\balg(H,\rho)$, or equivalently of $\Delta(H,\rho)$, 
with respect to the minimal regular norm $\norm{\cdot}$ provides the $C^\ast$-algebra $\CCR(H,\rho)$ 
of canonical commutation relations associated to the presymplectic Abelian group $(H,\rho)$. 
\end{lemma}

\begin{proof}
For each state $\omega\in\states$, consider the associated GNS-triple $\GNS{\omega}$, 
where $\pi_\omega:\balg(H,\rho)\to\endo(\dense_\omega)$ is a $\ast$-representation 
on a pre-Hilbert space $\dense_\omega$, while $\Omega_\omega\in\dense_\omega$ is a cyclic vector such that 
$\la\Omega_\omega,\pi_\omega(a)\Omega_\omega\ra_\omega=\omega(a)$ for each $a\in\balg(H,\rho)$. 
References about the GNS theorem for $C^\ast$-algebras are {\em e.g.}\ \cite[Section\ 2.3]{BR87} 
and\ \cite[Section\ 1.4]{BF09}. However, one can easily check that the GNS construction goes through 
also for states on arbitrary unital $\ast$-algebras, yet resulting in a much less regular GNS representation, 
{\em cfr.}\ \cite[Theorem\ 4.1]{BDH13}. The GNS construction applied to a state $\omega$ on $\balg(H,\rho)$ 
provides the pre-Hilbert space $\dense_\omega$ given by the quotient of vector spaces $\balg(H,\rho)/N_\omega$ 
endowed with the scalar product $\la\cdot,\cdot\ra$ defined by $\la[a],[b]\ra_\omega=\omega(b^\ast\,a)$ 
for each $a,b\in\balg(H,\rho)$, where $N_\omega=\{a\in\balg(H,\rho):\omega(a^\ast\,a)=0\}$ is a left ideal 
of $\balg(H,\rho)$. Furthermore, the $\ast$-representation $\pi_\omega:\balg(H,\rho)\to\endo(\dense_\omega)$ 
is defined by $\pi_\omega(a)[b]=[ab]$ for each $a,b\in\balg(H,\rho)$, 
while the cyclic vector is specified by $\Omega_\omega=[\bbone]$. In particular, we get the following identity: 
\begin{align}\label{eqGNSvsState}
\la\pi_\omega(a)[b],\pi_\omega(a)[b]\ra_\omega=\omega(b^\ast\,a^\ast\,a\,b)\,, 
&& \forall\,a,\,b\in\balg(H,\rho)\,.
\end{align}
This entails the inequality below, which in turn ensures continuity of the $\ast$-representation $\pi_\omega$: 
\begin{align*}
\norm{\pi_\omega(a)[b]}_\omega\leq\norm{a}\,\norm{[b]}_\omega\,, && \forall\,a,\,b\in\balg(H,\rho)\,. 
\end{align*}
This shows that $\pi_\omega(a)$ is a bounded operator on $\dense_\omega$ 
with operator norm $\norm{\pi_\omega(a)}_\omega^\bounded\leq\norm{a}\leq\normb{a}$, 
whence $\pi_\omega(a)$ extends by continuity to a bounded linear operator 
on the Hilbert space completion $\hilb_\omega$ of $\dense_\omega$. 
The result is a continuous $\ast$-representation $\pi_\omega:\balg(H,\rho)\to\bounded(\hilb_\omega)$. 
Furthermore, we already know that $\norm{\pi_\omega(a)}_\omega^\bounded\leq\norm{a}$ 
for each $a\in\balg(H,\rho)$ and each $\omega\in\states$, whence 
\begin{align*}
\sup_{\omega\in\states}\norm{\pi_\omega(a)}_\omega^\bounded\leq\norm{a}\,, && \forall\,a\in\balg(H,\rho)\,. 
\end{align*}
The converse inequality holds true by definition of the minimal regular norm 
and exploiting the identity in\ \eqref{eqGNSvsState}: 
\begin{align*}
\norm{a} & =\sup_{\omega\in\states}\sqrt{\omega(a^\ast\,a)}
\leq\sup_{\omega\in\states}\;\sup_{\substack{b\in\balg(H,\rho)\\\omega(b^\ast\,b)=1}}
\sqrt{\omega(b^\ast\,a^\ast\,a\,b)}\\
& =\sup_{\omega\in\states}\;\sup_{\substack{b\in\balg(H,\rho)\\\omega(b^\ast\,b)=1}}
\norm{\pi_\omega(a)[b]}_\omega
=\sup_{\omega\in\states}\norm{\pi_\omega(a)}_\omega^\bounded\,, & \forall\,a\in\balg(H,\rho)\,.
\end{align*}
Since, for each $\omega\in\states$, $\bounded(\hilb_\omega)$ is 
the $C^\ast$-algebra of bounded operators on $\hilb_\omega$, for each $a\in\balg(H,\rho)$, 
both the identity $\norm{a^\ast}=\norm{a}$ and the $C^\ast$-condition $\norm{a^\ast\,a}=\norm{a}^2$ 
follow from the corresponding ones on $\bounded(\hilb_\omega)$ for each $\omega\in\states$. 
This shows that $\norm{\cdot}$ is a $C^\ast$-norm. Therefore, the completion of $\balg(H,\rho)$, 
or equivalently of $\Delta(H,\rho)$, with respect to $\norm{\cdot}$ provides the $C^\ast$-algebra $\CCR(H,\rho)$. 
Furthermore, for each state $\omega\in\states$, 
$\pi_\omega$ extends by continuity to a continuous $\ast$-representation of $\CCR(H,\rho)$. 
\end{proof}

Note that $\norm{\cdot}$ is dominated by $\normb{\cdot}$ on $\balg(H,\rho)$ 
and that any other $C^\ast$-norm on $\Delta(H,\rho)$ induces a $C^\ast$-algebra 
isomorphic to a quotient of $\CCR(H,\rho)$ by one of its ideals, {\em cfr.}\ \cite[Corollary\ 3.9]{MSTV73}. 
Furthermore, each state $\omega\in\states$ extends by continuity to a state on $\CCR(H,\rho)$ 
by definition of the minimal regular norm and using the Cauchy-Schwarz inequality. 
In fact, $\vert\omega(a)\vert^2\leq\omega(a^\ast\,a)\leq\norm{a}^2$ for each $a\in\balg(H,\sigma)$. 

\begin{remark}\label{remRestriction}
Let us stress that, in the case of a symplectic vector space $(V,\tau)$, 
the $C^\ast$-norm on $\Delta(V,\tau)$ is unique, {\em cfr.}\ \cite[Theorem\ 4.2.9]{BGP07}, 
hence $\CCR(V,\tau)$ coincides with the standard CCR-representation 
associated to a symplectic vector space $(V,\tau)$, see for example \cite[Definition\ 4.2.8]{BGP07}. 
Furthermore, notice that the minimal regular norm coincides with the one considered in\ \cite{BHR04} 
in case $(V,\tau)$ is a presymplectic vector space, {\em cfr.}\ \cite[Proposition\ 3.4]{BHR04}. 
\end{remark}

Given a presymplectic homomorphism $L:(H,\rho)\to(K,\sigma)$, 
one gets a continuous $\ast$-homomorphism $\balg(L):\balg(H,\rho)\to\balg(K,\sigma)$. 
Now we would like to get a continuous $\ast$-homomorphism 
$\CCR(L):\CCR(H,\rho)\to\CCR(K,\sigma)$ out of $\balg(L)$. 
This is a straightforward consequence of the fact that the pull-back via $\balg(L)$
of any state $\omega_2\in\states_2$ on $\balg(K,\sigma)$ is 
a state $\omega_2\circ\balg(L)\in\states_1$ on $\balg(H,\rho)$. 
Keeping this fact in mind, the inequality below holds for each $a\in\balg(H,\rho)$: 
\begin{equation}\label{eqContinuity}
\norm{\balg(L)a}_2=\sup_{\omega_2\in\states_2}\sqrt{\omega_2\big(\balg(L)(a^\ast\,a)\big)}
\leq\sup_{\omega_1\in\states_1}\sqrt{\omega_1(a^\ast\,a)}=\norm{a}_1\,,
\end{equation}
where $\norm{\cdot}_1$ and $\norm{\cdot}_2$ denote the minimal regular norms 
respectively on $\balg(H,\rho)$ and on $\balg(K,\sigma)$. Therefore, the extension to the $C^\ast$-completions 
provides the sought $C^\ast$-algebra homomorphism $\CCR(L):\CCR(H,\rho)\to\CCR(K,\sigma)$. 
We collect the results outlined above in the next theorem. 

\begin{theorem}\label{thmQuantFunctor}
The assignment of the unital $C^\ast$-algebra $\CCR(H,\rho)$ to each presymplectic Abelian group $(H,\rho)$ 
and of the continuous $\ast$-homomorphism $\CCR(L):\CCR(H,\rho)\to\CCR(K,\sigma)$ 
to each presymplectic homomorphism $L:(H,\rho)\to(K,\sigma)$ provides a covariant functor 
$\CCR:\PSymA\to\CAlg$ from the category of presymplectic Abelian groups 
to the category of unital $C^\ast$-algebras. 
\end{theorem}

To conclude the present section, we want to show that $\CCR:\PSymA\to\CAlg$ preserves injectivity of morphisms. 
To achieve this result we need a preliminary lemma showing that, given a presymplectic Abelian group $(H,\rho)$,  
any positive linear functional on a unital $\ast$-subalgebra of $\Delta(H,\rho)$ admits an extension 
to a positive linear functional on $\Delta(H,\rho)$. This is a consequence 
of the positive cone version of the Hahn-Banach theorem, {\em cfr.}\ \cite[Theorem\ 2.6.2]{Edw65}. 

\begin{lemma}\label{lemExtPosLinFunct}
Let $(H,\rho)$ be a presymplectic Abelian group and consider a positive linear functional $\tilde{\omega}:A\to\bbC$ 
on a unital $\ast$-subalgebra $A$ of the unital $\ast$-algebra $\Delta(H,\rho)$. 
Then there exists a positive linear functional $\omega:\Delta(H,\rho)\to\bbC$ extending $\tilde{\omega}$, 
namely such that its restriction $\omega\vert_A$ to $A$ coincides with $\tilde{\omega}$. 
\end{lemma}

\begin{proof}
Consider the real vector subspaces of Hermitian elements of $\Delta(H,\rho)$ and respectively $A$: 
\begin{align*}
\herm=\{a\in\Delta(H,\rho)\,:\;a^\ast=a\}\,, && \hermt=\{a\in A\,:\;a^\ast=a\}\subseteq\herm\,.
\end{align*}
The restriction $\tilde{\omega}\vert_{\hermt}:\hermt\to\bbR$ of $\tilde{\omega}:A\to\bbC$ to $\hermt$ 
is a positive $\bbR$-linear functional. 

Denote with $\cone\subseteq\herm$ the positive cone in $\herm$ whose elements are linear combinations 
with positive coefficients of elements in $\herm$ of the form $a^\ast\,a$ for $a\in\Delta(H,\rho)$. 
We show that, for each $a\in\herm$, there exists $\tilde{a}\in\hermt$ such that $\tilde{a}-a\in\cone$. 
In fact, each $a\in\herm$ is a finite sum of Hermitian elements of the form 
$a_{c,h}=c\weyl_h+\bar{c}\weyl_{-h}$, $c\in\bbC$, $h\in H$. Hence, it is enough to prove that 
the property stated above holds true for $a_{c,h}$ for each $c\in\bbC$ and each $h\in H$. 
Take $c\in\bbC$ and $h\in H$ and consider $b=\bbone-c\weyl_h\in\Delta(H,\rho)$. 
Clearly $k=b^\ast\,b$ lies in the positive cone $\cone$ 
and moreover $k=(1+\vert c\vert^2)\bbone-a_{c,h}$. 
Since $\bbone\in\hermt$, introducing $\tilde{a}=(1+\vert c\vert^2)\bbone\in\hermt$, 
we conclude that $\tilde{a}-a_{c,h}\in\cone$, thus proving the statement above. 

The fact that, for each $a\in\herm$, there exists $\tilde{a}\in\hermt$ such that $\tilde{a}-a\in\cone$, 
allows one to apply \cite[Theorem\ 2.6.2]{Edw65} to $\tilde{\omega}\vert_{\hermt}:\hermt\to\bbR$ 
in order to obtain a positive $\bbR$-linear functional $\omega:\herm\to\bbR$ 
extending $\tilde{\omega}\vert_{\hermt}:\hermt\to\bbR$. 

It only remains to extend $\omega:\herm\to\bbR$ to $\Delta(H,\rho)$. 
First, note that each element $a\in\Delta(H,\rho)$ can be uniquely decomposed as a $\bbC$-linear combination 
of Hermitian elements, $a=a_R+ia_I$, for $a_R=(a+a^\ast)/2\in\herm$ and $a_I=(a-a^\ast)/2i\in\herm$. 
Then one defines $\omega:\Delta(H,\rho)\to\bbC$ setting $\omega(a)=\omega(a_R)+i\omega(a_I)$. 
To conclude the proof, note that $\omega=\tilde{\omega}$ on $\hermt$ entails 
$\omega=\tilde{\omega}$ on $A$ too. 
\end{proof}

\begin{proposition}\label{prpQuantInj}
Let $L:(H,\rho)\to(K,\sigma)$ be an injective presymplectic homomorphism. 
Then $\CCR(L):\CCR(H,\rho)\to\CCR(K,\sigma)$ is an isometric $\ast$-homomorphism, in particular injective. 
\end{proposition}

\begin{proof}
The inequality displayed in\ \eqref{eqContinuity} entails that $\norm{\CCR(L)a}_2\leq\norm{a}_1$ 
for each $a\in\CCR(H,\rho)$, where $\norm{\cdot}_1$ and $\norm{\cdot}_2$ denote 
the minimal regular norms on $\CCR(H,\rho)$ and respectively on $\CCR(K,\sigma)$. 
Therefore, it remains only to check that $\norm{\CCR(L)a}_2\geq\norm{a}_1$ for each $\in\CCR(H,\rho)$. 
Let us denote the set of states on $\balg(H,\rho)$ by $\states_1$ 
and the set of states on $\balg(K,\sigma)$ by $\states_2$. 
Recalling Definition\ \ref{defMinRegNorm}, the thesis is implied by the following statement: 
For each state $\omega_1\in\states_1$, 
there exists a state $\omega_2\in\states_2$ such that $\omega_2\circ\balg(L)=\omega_1$. 

Let us take a state $\omega_1\in\states_1$. 
Recalling that the $\ast$-homomorphism $\balg(L):\balg(H,\rho)\to\balg(K,\sigma)$ is an isometry, 
one deduces in particular that $\Delta(L):\Delta(H,\rho)\to\Delta(K,\sigma)$ is injective, 
therefore it admits a unique inverse $\Delta(L)^{-1}:\Delta(L)(\Delta(H,\rho))\to\Delta(H,\rho)$ 
defined on the image. Consider now 
\begin{equation*}
\tilde{\omega}_2=\omega_1\circ\Delta(L)^{-1}:\Delta(L)\big(\Delta(H,\rho)\big)\to\bbC\,. 
\end{equation*}
This is a positive linear functional defined on a unital $\ast$-subalgebra of $\Delta(K,\sigma)$ 
and it is such that $\tilde{\omega}_2(\bbone_2)=1$. Exploiting\ Lemma\ \ref{lemExtPosLinFunct}, 
one finds a positive linear functional $\omega_2:\Delta(K,\sigma)\to\bbC$ extending $\tilde{\omega}_2$. 
In particular, $\omega_2(\bbone_2)=1$. Applying\ Proposition\ \ref{prpPosLinFunc}, 
one gets a continuous extension of $\omega_2$ to $\balg(K,\sigma)$ too. 
We denote it still by $\omega_2:\balg(K,\sigma)\to\bbC$ and we observe that this is a state in $\states_2$. 
Furthermore, by construction, $\omega_2\circ\Delta(L)=\tilde{\omega}_2\circ\Delta(L)=\omega_1$ 
on $\Delta(H,\rho)$, whence $\omega_2\circ\balg(L)=\omega_1$ on $\balg(H,\rho)$ by continuity. 

Keeping all this in mind, for $a\in\balg(H,\rho)$, one gets the following inequality: 
\begin{equation*}
\norm{a}_1=\sup_{\omega_1\in\states_1}\sqrt{\omega_1(a^\ast\,a)}
\leq\sup_{\omega_2\in\states_2}\sqrt{\omega_2\big(\balg(L)(a^\ast\,a)\big)}=\norm{\balg(L)a}_2\,.
\end{equation*}
The same inequality holds true when passing to the $C^\ast$-completions, 
namely $\norm{\CCR(L)a}_2\geq\norm{a}_1$ for each $a\in\CCR(H,\rho)$, thus concluding the proof. 
\end{proof}

\paginavuota

\chapter{Maxwell $k$-forms}\label{chMaxwell}\index{Maxwell $k$-form}
In this chapter we introduce the first gauge field theory we are interested in, 
namely {\em Maxwell $k$-forms}. For $k=1$ this model reproduces the usual dynamics of the vector potential. 
Maxwell $k$-forms have been widely considered in the literature, especially to model pure electromagnetism, 
possibly in the presence of external source currents. A list of papers in the framework 
of algebraic quantum field theory where electromagnetism and its higher analogues for $k\geq2$ 
have been thoroughly analyzed either from the point of view of the Faraday tensor 
or in the spirit adopted here, possibly including external sources as well, 
is\ \cite{Dim92, FP03, Pfe09, Dap11, DS13, FS13, CRV13, SDH14, FL14}. 

We will not include external sources in our description, 
but we will allow for any degree $k$ ranging from $1$ to $m-1$, 
$m$ being the dimension of the spacetime under consideration. 
Our attention will be focused on certain properties of the functionals which will be considered on field configurations. 
In particular we will exploit the tools developed in Section\ \ref{secForms} 
to show that such functionals separate field configurations 
and moreover that all functionals vanishing on-shell can be obtained via the equation of motion. 
Furthermore, we will analyze in detail the functorial behavior of this model both at the classical 
and at the quantum level adopting the perspective of general local covariance, see\ \cite{BFV03}. 
In particular we will exhibit explicit violations of the locality property which arise 
from degeneracies of the classical presymplectic structures and from non-trivial centers of the quantum algebras. 
Even more, locality cannot be restored via suitable quotients, yet the interpretation of spacetime regions 
as subsystems can be recovered in the Haag-Kastler framework\ \cite{HK64}.

\section{Gauge symmetry and dynamics}\label{secGaugeDynForms}
In the present section $M$ will be a fixed globally hyperbolic spacetime of dimension $m$. 
Off-shell field configurations are provided by $k$-forms $A\in\f^k(M)$ over $M$ with arbitrary support 
for any fixed $k\in\{1,\dots,m-1\}$. 
The equation of motion which specifies on-shell field configurations is the standard equation 
for the vector potential generalized to arbitrary degree $k$: 
\begin{align}\label{eqOnShellForms}
\de\dd A=0\,, && A\in\f^k(M)\,, && k\in\{1,\,\dots,\,m-1\}\,.
\end{align}
Note that for $k=m$ this equation becomes trivial, all forms of top degree being closed. 
This motivates the fact that $m$ has been excluded from the range of possible degrees $k$ 
which are considered here. 
On-shell field configurations are not the central objects that describe the dynamics of the vector potential, 
rather gauge equivalence classes are regarded as carrying the actual physical information. In this spirit, 
and extending to arbitrary degree what is usually considered as gauge symmetry for the vector potential, 
we introduce the following notion of gauge equivalence for field configurations: 
For $k\in\{1,\dots,m-1\}$ and $A,A^\prime\in\f^k(M)$, we have 
\begin{align}\label{eqGaugeForms}
A\sim A^\prime && \iff && \exists\,\chi\in\f^{k-1}(M)\,:\;\dd\chi=A^\prime-A\,,
\end{align}
namely $A,A^\prime$ are equivalent provided $A^\prime=A+\dd\chi$, for some $\chi\in\f^{k-1}(M)$. 
If the degree $k$ is $0$, then there is no gauge equivalence 
and the model boils down to to the free real scalar field. 
Being interested in field theories with non-trivial gauge symmetry, we have excluded the degree $k=0$. 
Let us stress that for $k=1$ the model considered describes the vector potential 
for free electromagnetism without any external source. 

\begin{remark}\label{remEquivalentVecPot}
One can take into account a model related to the present one via the Hodge star operator $\ast$. 
Specifically, consider a model where a field configuration $B\in\f^{m-k}(M)$ is on-shell 
when the equation of motion $\dd\de B=0$, instead of\ \eqref{eqOnShellForms}, is fulfilled 
and where gauge equivalence is specified by the condition 
$B\sim B^\prime$ if and only if there exists $\chi\in\f^{m-k+1}(M)$ such that $\de\chi=B^\prime-B$, 
in place of\ \eqref{eqGaugeForms}. 
It is straightforward to check that this model is completely equivalent to the one described previously, 
the relation $B=\ast A$ being defined via the Hodge star $\ast$. 
In particular one can observe that for this \quotes{dual} model, 
in degree $m-k=0$ the equation of motion becomes trivial, all $0$-forms being coclosed, 
and in degree $m-k=m$ there is no gauge equivalence and the equation of motion reduces to $\Box B=0$. 
Therefore we are interested in the dual version of the model 
described previously only in degree $m-k\in\{1,\dots,m-1\}$. 
\end{remark}

We denote the space of solutions to the field equation\ \eqref{eqOnShellForms} as follows: 
\begin{align}\label{eqSolForms}
\sol_M=\ker\left(\de\dd:\f^k(M)\to\f^k(M)\right)\,, && k\in\{1,\,\dots,\,m-1\}\,.
\end{align}
According to\ \eqref{eqGaugeForms}, two $k$-forms are regarded as equivalent 
whenever they differ by an exact $k$-form,\ {\em i.e.}\ their difference lies in 
\begin{align*}
\gau_M=\dd\f^{k-1}(M)\,, && k\in\{1,\,\dots,\,m-1\}\,.
\end{align*}
Since $\dd\dd=0$, it follows that $\gau_M$ is a subspace of $\sol_M$ 
and therefore it makes sense to consider the quotient below 
which defines the space $[\sol_M]$ of gauge equivalence classes of on-shell $k$-forms: 
\begin{align}\label{eqSolModGaugeForms}
[\sol_M]=\frac{\sol_M}{\gau_M}\,, && k\in\{1,\,\dots,\,m-1\}\,.
\end{align}

For the rest of the present section we will be concerned in presenting a characterization 
of the space $[\sol_M]$ of on-shell Maxwell $k$-forms up to gauge. 
This result is achieved in three steps following the strategy of\ \cite[Section\ 2.2]{Dap11}. 
The first step consists in showing that 
it is always possible to set the Lorenz gauge starting from an on-shell $k$-form. 
Next, we show that, up to a gauge transformation, 
one can exploit the causal propagator of the Hodge-d'Alembert differential operator $\Box=\de\dd+\dd\de$, 
see\ Subsection\ \ref{subDynamics}, to represent each solution in the Lorenz gauge 
starting from a coclosed $k$-form with timelike compact support. 
The last step exploits these two facts to represent on-shell Maxwell $k$-forms up to gauge 
in terms of a suitable quotient of $\ftcde^k(M)$. 

\begin{lemma}\label{lemLorenzGaugeFixing}
Let $M$ be an $m$-dimensional globally hyperbolic spacetime and consider $k\in\{1,\dots,m-1\}$. 
Each $k$-form $A\in\sol_M$ is gauge equivalent to a $k$-form $A^\prime\in\sol_M$ 
satisfying the Lorenz gauge condition, 
namely there exists $\chi\in\f^{k-1}(M)$ such that $A^\prime=A+\dd\chi$ fulfils $\de A^\prime=0$. 
\end{lemma}

\begin{proof}
Given a $k$-form $A\in\sol_M$, we can consider the equation $\de\dd\chi=-\de A$, 
where $\chi\in\f^{k-1}(M)$ has to be determined. 
Consider a partition of unity $\{\chi_+,\chi_-\}$ on $M$ such that 
$\chi_+=1$ in a past compact region, while $\chi_-=1$ in a future compact one. 
Exploiting the retarded and advanced Green operators $G^+,G^-$ for $\Box$, 
one can explicitly write down a solution to the equation $\de\dd\chi=-\de A$: 
\begin{equation*}
\chi=-\de\left(G^+(\chi_+A)+G^-(\chi_-A)\right)\,. 
\end{equation*}
Introducing a new $k$-form $A^\prime=A+\dd\chi$, one realizes that $A^\prime\in\sol_M$ 
and moreover $\de A^\prime=\de A+\de\dd\chi=0$. 
Since per construction $A^\prime\sim A$, this concludes the proof. 
\end{proof}

\begin{remark}
Lemma\ \ref{lemLorenzGaugeFixing} shows how to impose the Lorenz gauge on-shell. 
Actually this result can be also achieved off-shell. Given $\omega\in\f^k(M)$, we look for $\chi\in\f^{k-1}(M)$ 
such that $\omega^\prime=\omega+\dd\chi$ satisfies the Lorenz condition, namely $\de\omega^\prime=0$. 
This amounts to solving the equation $\de\dd\chi=-\de\omega$. A solution $\chi\in\f^{k-1}(M)$ of this equation 
can be found using the argument presented in the proof of the last lemma. 
\end{remark}

\begin{lemma}\label{lemLorenzSol}
Let $M$ be an $m$-dimensional globally hyperbolic spacetime and consider $k\in\{1,\dots,m-1\}$. 
Denote with $G$ the causal propagator for $\Box$, see\ Subsection\ \ref{subDynamics}. 
Each $A\in\sol_M$ satisfying the Lorenz gauge condition is gauge equivalent to $G\omega\in\sol_M$ 
for a suitable coclosed $k$-form $\omega$ with timelike compact support, 
namely there exist $\chi\in\f^{k-1}(M)$ and $\omega\in\ftcde^k(M)$ such that $G\omega=A+\dd\chi$. 
\end{lemma}

\begin{proof}
Consider $A\in\sol_M$ satisfying the Lorenz gauge condition,\ {\em i.e.}\ such that $\de A=0$. 
Then $\de\dd A=0$ entails $\Box A=0$ since $\Box=\de\dd+\dd\de$. 
By the properties of the causal propagator $G$, in particular\ Theorem\ \ref{thmExtCausalProp}, 
there exists $\theta\in\ftc^k(M)$ such that $G\theta=A$. 
From the Lorenz gauge condition, together with Proposition\ \ref{prpIntertwiners}, one reads $G\de\theta=0$. 
This ensures the existence of $\rho\in\ftc^{k-1}(M)$ such that $\Box\rho=\de\theta$. 
Furthermore note that via\ Theorem\ \ref{thmExtCausalProp}, the last identity entails $\de\rho=0$. 
Introducing $\omega=\theta-\dd\rho\in\ftc^k(M)$ and $\chi=-G\rho\in\f^{k-1}(M)$, 
one can conclude that $\de\omega=0$ and, moreover, $G\omega=A+\dd\chi$, thus concluding the proof. 
\end{proof}

\begin{theorem}\label{thmSolModGaugeForms}
Let $M$ be an $m$-dimensional globally hyperbolic spacetime and consider $k\in\{1,\dots,m-1\}$. 
Denote with $G$ the causal propagator for $\Box$, see\ Subsection\ \ref{subDynamics}. 
Then $G$ induces the following isomorphism of vector spaces: 
\begin{align*}
\frac{\ftcde^k(M)}{\de\dd\ftc^k(M)}\to[\sol_M]\,, && [\omega]\mapsto[G\omega]\,.
\end{align*}
\end{theorem}

\begin{proof}
Given any $\omega\in\ftcde^k(M)$, $G\omega$ is coclosed by Proposition\ \ref{prpIntertwiners}, 
therefore $\de\dd G\omega=\Box G\omega=0$. 
This means that $G$ maps $\ftcde^k(M)$ to $\sol_M$, hence one can consider the map 
\begin{align*}
\ftcde^k(M)\to[\sol_M]\,, && \omega\mapsto[G\omega]\,.
\end{align*}
First of all, we prove that the map defined above is surjective. 
Take $[A]\in[\sol_M]$ and choose a representative $A\in[A]$. 
Applying Lemma\ \ref{lemLorenzGaugeFixing}, one finds $A^\prime\in\sol_M$ satisfying the Lorenz gauge condition 
which is gauge equivalent to $A$, namely such that $\de A^\prime=0$ and $[A^\prime]=[A]$. 
By Lemma\ \ref{lemLorenzSol} applied to $A^\prime$, 
there exists $\omega\in\ftcde^{k}(M)$ such that $[G\omega]=[A^\prime]$, 
therefore one deduces that $[G\omega]=[A]$. 

It remains only to check that the kernel of the map mentioned above coincides with $\de\dd\ftc^k(M)$. 
The inclusion in one direction is a consequence of the following chain of identities for each $\theta\in\ftc^k(M)$: 
\begin{equation*}
G\de\dd\theta=G(\Box-\dd\de)\theta=-\dd G\de\theta\in\gau_M\,. 
\end{equation*}
This shows that $\de\dd\ftc^k(M)$ is included in the kernel of the map $\ftcde^k(M)\to[\sol_M]$ introduced above. 
To prove also the converse inclusion we take $\omega\in\ftcde^k(M)$ 
lying in the kernel of the map $\ftcde^k(M)\to[\sol_M]$. 
This is to say that there exists $\chi\in\f^{k-1}(M)$such that $\dd\chi=G\omega$, 
in particular we have the identity $\de\dd\chi=0$. 
For $k\geq2$, an argument similar to the one used above to show surjectivity, 
exploiting both Lemma\ \ref{lemLorenzGaugeFixing} and Lemma\ \ref{lemLorenzSol}, 
shows that there exist $\rho\in\ftcde^{k-1}(M)$ and $\xi\in\f^{k-2}(M)$ such that $G\rho=\chi+\dd\xi$. 
For $k=1$, the situation is simpler since the equation $\de\dd\chi=0$ coincides with $\Box\chi=0$ 
and hence one finds $\rho\in\ctc(M)$ such that $G\rho=\chi$. 
In both cases one deduces that $G\omega=G\dd\rho$. 
Because of this identity there exists $\theta\in\ftc^k(M)$ such that $\Box\theta=\omega-\dd\rho$. 
Acting with $\de$ on both sides, we deduce $\Box\de\theta=-\de\dd\rho$. 
Furthermore, we note that, for any $k\geq1$, $\de\dd\rho=\Box\rho$ and hence $\Box(\de\theta+\rho)=0$. 
By Theorem\ \ref{thmExtCausalProp}, this entails that $\de\theta+\rho=0$. 
Together with the identity $\Box\theta=\omega-\dd\rho$, this entails $\omega=\de\dd\theta$. 
This shows that the kernel of the map $\ftcde^k(M)\to[\sol_M]$ introduced above is included in $\de\dd\ftc^k(M)$, 
thus completing the proof. 
\end{proof}

Under suitable conditions, we can exploit the quotient by $\de\dd\ftc^k(M)$ in the source of the isomorphism 
provided by the last theorem in order to choose a representative in $\ftcde^k(M)$ with support 
\quotes{as close as possible} to a Cauchy hypersurface for $M$. 
This property is related to the fact that the Cauchy problem for Maxwell $k$-forms is well-posed 
(in a suitable sense), {\em cfr.}\ \cite[Proposition II.12]{Pfe09}. 
To make this statement more precise, we prove the following lemma, 
which will be used later to exhibit the so-called time-slice axiom for the functor describing 
the classical field theory of Maxwell $k$-forms, see\ Theorem \ref{thmTimeSliceForms}. 

\begin{lemma}\label{lemTimeSliceForms1}
Let $M$ be an $m$-dimensional globally hyperbolic spacetime and consider $k\in\{1,\dots,m-1\}$. 
Denote with $G$ the causal propagator for $\Box$. Furthermore, assume a partition of unity $\{\chi_+,\chi_-\}$ 
is given in such a way that $\supp(\chi_\pm)$ is past/future compact. 
Given $[A]\in[\sol_M]$, there exists $\omega\in\ftcde^k(M)$ 
with support inside $\supp(\chi_+)\cap\supp(\chi_-)$ such that $[G\omega]=[A]$. 
\end{lemma}

\begin{proof}
Given $[A]\in[\sol_M]$, Theorem\ \ref{thmSolModGaugeForms} provides $\theta\in\ftcde^k(M)$ 
such that $[G\theta]=[A]$. First, recall that $\de\dd G\theta=0$. 
It follows that $\de\dd(\chi_+G\theta)=-\de\dd(\chi_-G\theta)$, 
so that $\omega=\de\dd(\chi_+G\theta)$ has timelike compact support. In fact, 
\begin{equation*}
\supp(\omega)\subseteq\supp(\chi_+)\cap\supp(\chi_-)
\end{equation*}
and the right-hand-side is a timelike compact set since it is the intersection of a past compact set 
with a future compact one. 
The subsequent computation shows that $\theta-\omega\in\de\dd\ftc^k(M)$: 
\begin{align*}
\theta-\omega=\de\dd G^+\theta-\de\dd(\chi_+G^+\theta)+\de\dd(\chi_+G^-\theta)=\de\dd\rho\,.
\end{align*}
Note that for the first equality we exploited the identity $\de\dd G^+\theta=\theta$ and the definition of $G$,
while for the last step we made use of the partition of unity $\{\chi_+,\chi_-\}$ 
and we introduced $\rho=\chi_-G^+\theta+\chi_+G^-\theta$. 
It remains only to check that $\rho$ has timelike compact support. As a matter of fact, 
\begin{equation*}
\supp(\rho)\subseteq (J_M^+(\supp(\theta))\cap\supp(\chi_-))\cup(J_M^-(\supp(\theta))\cap\supp(\chi_+))
\end{equation*}
is the union of two timelike compact sets, which are obtained as the intersection of a past compact set 
and a future compact one. Hence $\rho\in\ftc^k(M)$ and therefore $\theta-\omega\in\de\dd\ftc^k(M)$. 
Applying the isomorphism of\ Theorem\ \ref{thmSolModGaugeForms}, 
we deduce $[G\omega]=[G\theta]=[A]$, which concludes the proof. 
\end{proof}

\begin{remark}\label{remSCSolModGaugeForms}
In place of $\sol_M$, one can also consider the space $\solsc{M}$ of solutions 
to the field equation $\de\dd A=0$ having spacelike compact support. 
In agreement with this restriction on the space of solutions, one has to restrict the notion of gauge equivalence, 
which is specified by $\gausc{M}=\dd\fsc^{k-1}(M)$ in this case. 
Adopting the same arguments presented above to forms with spacelike compact support, 
one can prove statements similar to Lemma\ \ref{lemLorenzGaugeFixing}, Lemma\ \ref{lemLorenzSol} 
and Theorem\ \ref{thmSolModGaugeForms}, 
where $\f^k$ and $\ftc^k$ are replaced respectively by $\fsc^k$ and $\fc^k$. 
In particular, these versions with restricted supports are obtained 
exploiting Theorem\ \ref{thmCausalProp} instead of Theorem\ \ref{thmExtCausalProp}. 
\end{remark}

\section{Covariant classical field theory}\label{secClassicalFTForms}
In this section we first introduce a space of linear functionals 
on the space of on-shell Maxwell $k$-forms up to gauge $[\sol_M]$ 
for a given $m$-dimensional globally hyperbolic spacetime $M$ 
and we endow it with a suitable presymplectic structure in the spirit of Peierls\ \cite{Pei52}, 
see \cite[Section 3.2]{SDH14} for a more recent perspective in the context of Maxwell $k$-forms. 
We will choose regular functionals, but still sufficiently many in order to separate points in $[\sol_M]$. 
In order to exhibit separability of on-shell Maxwell $k$-forms up to gauge 
by means of evaluation on the class of functionals considered, 
we will exploit standard Poincar\'e duality for de Rham cohomology, see\ Subsection\ \ref{subCohom}. 
Furthermore, we will implicitly get rid of vanishing functionals in order to obtain a faithful labelling of our functionals. 
Only later we will also provide a more explicit characterization of vanishing functionals 
by means of Poincar\'e duality between causally restricted de Rham cohomology groups, 
see\ Subsection\ \ref{subRestrictedCohom}. 
After the construction of this space of functionals for a fixed globally hyperbolic spacetime, 
we will exhibit the functorial properties of this construction 
in the spirit of the generally covariant locality principle\ \cite{BFV03}. 
However, as shown in\ \cite[Theorem 4.14]{SDH14}, the theory fails to satisfy locality in the strict sense. 
We will see how to recover at least isotony in the Haag-Kastler framework\ \cite{HK64} 
considering a subcategory of $m$-dimensional globally hyperbolic spacetimes with a specified terminal object, 
see\ \cite[Section\ 6]{BDHS14}. 
The discussion for the moment is at a purely classical level. 
In the next section we will quantize the model using a suitable functor 
to show that the relevant properties are preserved by quantization. 
In particular, the failure of locality in the strict sense,
as well as the possibility to recover the Haag-Kastler framework by fixing a target globally spacetime are shown. 
Let us stress that, even though our approach is sometimes slightly different, 
the results we obtain are in agreement with\ \cite{SDH14}. 

To specify the class of functionals on $[\sol_M]$ we are interested in, 
we adopt the strategy of \cite{BDS14b} for affine field theories 
and its modification developed in \cite{BDS14a} to include the case of Abelian gauge field theories. 
Nevertheless, notice that the situation here is simpler since we are dealing with linear spaces. 
Let us fix an $m$-dimensional globally hyperbolic spacetime $M$ and consider a fixed degree $k\in\{1,\dots,m-1\}$. 
Given a $k$-form with compact support $\alpha\in\fc^k(M)$, we introduce the linear functional
\begin{align}\label{eqEvForms}
\ev_\alpha:\f^k(M)\to\bbR\,, && \beta\mapsto(\alpha,\beta)\,.
\end{align}
We collect all functionals of this type in the space of kinematic functionals: 
\begin{equation*}
\kin_M=\left\{\ev_\alpha\,:\;\alpha\in\fc^k(M)\right\}\,.
\end{equation*}
This nomenclature is motivated by the fact that $\kin_M$ keeps track of the kinematics of Maxwell $k$-forms, 
but disregards both the dynamics and the gauge symmetry. 
Non-degeneracy of the pairing $(\cdot,\cdot):\fc^k(M)\times\f^k(M)\to\bbR$ entails that 
$\kin_M$ is isomorphic to $\fc^k(M)$ via the map $\alpha\in\fc^k(M)\mapsto\ev_\alpha\in\kin_M$ 
implicitly defined in\ \eqref{eqEvForms}. 

As a first step we deal with gauge equivalence for field configurations. 
In fact, keeping in mind that purely gauge Maxwell $k$-forms $\gau_M$ should be irrelevant, 
we want to consider only functionals vanishing on $\gau_M$ in order to define observables for Maxwell $k$-forms. 
Given a functional $\ev_\omega$ of $\kin_M$ evaluated on $\dd\chi\in\gau_M$, 
via Stokes' theorem we get the identity: 
\begin{equation*}
(\de\omega,\chi)=(\omega,\dd\chi)=\ev_\omega(\dd\chi)\,.
\end{equation*}
Non-degeneracy of $(\cdot,\cdot):\fc^{k-1}(M)\times\f^{k-1}(M)\to\bbR$ entails that 
$\ev_\omega$ is invariant if and only if $\omega$ is coclosed. 
This motivates our definition of the space $\inv_M$ of gauge invariant functionals for Maxwell $k$-forms: 
\begin{equation*}
\inv_M=\left\{\ev_\omega\in\kin_M\,:\;\omega\in\fcde^k(M)\right\}\,.
\end{equation*}
It is straightforward to check that the isomorphism between $\kin_M$ and $\fc^k(M)$ 
restricts to an isomorphism between $\inv_M$ and $\fcde^k(M)$. 
Notice that, because of gauge invariance, we can now evaluate a functional $\ev_\omega\in\inv_M$ on 
gauge equivalence classes of off-shell Maxwell $k$-forms, namely elements of the quotient $\f^k(M)/\gau_M$, 
simply by the arbitrary choice of a representative. 
It is important to note that gauge invariant functionals in $\inv_M$ separate on-shell Maxwell $k$-forms up to gauge. 

\begin{proposition}\label{prpInvSeparatesSol}
Let $M$ be an $m$-dimensional globally hyperbolic spacetime and consider $k\in\{1,\,\dots,\,m-1\}$. 
For $A\in\sol_M$, $\ev_\omega(A)=0$ for each $\ev_\omega\in\inv_M$ if and only if $A\in\gau_M$. 
\end{proposition}

\begin{proof}
The sufficient condition has already been checked before the definition of $\inv_M$ 
(in fact, this was the motivation for our definition of $\inv_M$). 
We are left with the proof of the necessary condition. Therefore let us consider $A\in\sol_M$ 
such that $\ev_\omega(A)=0$ for each $\ev_\omega\in\inv_M$. This condition reads 
\begin{align*}
(\omega,A)=0\,, && \forall\,\omega\in\fcde^k(M)\,.
\end{align*}
If we plug in the first argument only coexact $k$-forms of compact support, 
our condition entails that $A$ is closed, therefore $A$ defines a cohomology class $[A]\in\hdd^k(M)$. 
By\ Remark\ \ref{remPoincareDualityImproved}, the following is an isomorphism of vector spaces: 
\begin{align*}
\hdd^k(M)\to(\hcde^k(M))^\ast\,, && [\beta]\to{}_\de(\cdot,[\beta])\,,
\end{align*}
where the pairing ${}_\de(\cdot,\cdot):\hcde^k(M)\times\hdd^k(M)\to\bbR$ is defined 
in\ \eqref{eqPoincarePairing}. Since, according to our assumption, ${}_\de(\cdot,[A])$ vanishes on $\hcde^k(M)$, 
this is the trivial element of $(\hcde^k(M))^\ast$. By the isomorphism above, $[A]=0\in\hdd^k(M)$, 
whence $A\in\gau_M$ as expected. 
\end{proof}

So far no information about the dynamics of the system has been taken into account. 
However, it is clear that some gauge invariant functionals always vanish on-shell. 
For example, given $\xi\in\fc^k(M)$, $\ev_{\de\dd\xi}$ is an invariant functionals since $\de\de=0$. 
If we evaluate such functional on $[A]\in[\sol_M]$, applying Stokes' theorem, we obtain $\ev_{\de\dd\xi}([A])=0$. 
To encode dynamics on the space of gauge invariant linear functionals $\inv_M$, 
we quotient by the space $\van_M$ of functionals vanishing on-shell, 
thus obtaining the space of linear classical observables $\obs_M$: 
\begin{align}\label{eqObsForms}
\van_M=\left\{\ev_\omega\in\kin_M\,:\;\ev_\omega([\sol_M])=\{0\}\right\}\,, &&
\obs_M=\frac{\inv_M}{\van_M}\,.
\end{align}
By construction, we have a natural pairing between the space of observables $\obs_M$ 
and the space $[\sol_M]$ of on-shell Maxwell $k$-forms up to gauge: 
\begin{align}\label{eqObsSolPairingForms}
\obs_M\times[\sol_M]\to\bbR\,, && (\ev_{[\omega]},[A])\mapsto\ev_{[\omega]}([A])=\ev_\omega(A)\,.
\end{align}
This definition is well-posed since $\ev_{[\omega]}=\ev_\omega+\van_M$, $[A]=A+\dd\f^{k-1}(M)$, 
$\de\dd A=0$ and $\omega\in\inv_M$. On account of Proposition\ \ref{prpInvSeparatesSol}, 
the space of classical observables introduced above separates on-shell Maxwell $k$-forms up to gauge. 
Moreover, by definition, observables vanishing on all on-shell Maxwell $k$-forms up to gauge are trivial. 
We collect this result in the next theorem.

\begin{theorem}
Let $M$ be an $m$-dimensional globally hyperbolic spacetime and consider $k\in\{1,\,\dots,\,m-1\}$. 
The pairing $\obs_M\times[\sol_M]\to\bbR$ defined in\ \eqref{eqObsSolPairingForms} is non-degenerate. 
\end{theorem}

On a globally hyperbolic spacetime of finite type it is possible 
to have an explicit characterization of the space $\van_M$ of functionals vanishing on-shell. 
This result makes contact with the previous literature, see\ {\em e.g.}\ \cite[Proposition\ 3.3]{SDH14}. 

\begin{proposition}\label{prpVanForms}
Let $M$ be an $m$-dimensional globally hyperbolic spacetime of finite type and consider $k\in\{1,\,\dots,\,m-1\}$. 
Let $\ev_\omega\in\inv_M$ be a gauge invariant linear functional. Then $\ev_\omega$ lies in $\van_M$ 
if and only if there exists $\rho\in\fc^k(M)$ such that $\de\dd\rho=\omega$. 
In particular, under these hypotheses, $\obs_M=\inv_M/\de\dd\fc^k(M)$. 
\end{proposition}

\begin{proof}
$\ev_\omega$ vanishes on $[\sol_M]$ if $\omega=\de\dd\rho$ for some $\rho\in\fc^k(M)$. 
This relies on the fact that $\de\dd$ is formally self-adjoint with respect to the pairing 
$(\cdot,\cdot):\fc^k(M)\times\f^k(M)\to\bbR$. Therefore only the necessary condition is still to be proved. 
Recalling Theorem\ \ref{thmSolModGaugeForms}, 
we can represent the space $[\sol_M]$ of gauge equivalence classes of on-shell Maxwell $k$-forms 
via the map induced on the quotient $\ftcde^k(M)/\de\dd\ftcde^k(M)$ 
by the causal propagator $G$ for the Hodge-d'Alembert differential operator $\Box$ acting on $k$-forms over $M$. 
This fact, together with formal antiself-adjointness of $G$, entails that 
the condition $\ev_\omega([\sol_M])=\{0\}$ on $\ev_\omega\in\inv_M$ translates into the following: 
\begin{align*}
(G\omega,\eta)=0\,, && \forall\,\eta\in\ftcde^k(M)\,.
\end{align*}
If we temporarily consider only $\eta\in\de\ftc^{k+1}(M)$, by Stokes' theorem we deduce $\dd G\omega=0$, 
hence $G\omega$ defines a cohomology class with spacelike compact support $[G\omega]\in\hscdd^k(M)$, 
see\ Subsection\ \ref{subRestrictedCohom}. Since in the condition displayed above only coclosed forms 
with timelike compact support appear in the second argument of the pairing, we deduce 
\begin{align*}
([G\omega],[\eta])_\de=0\,, && \forall\,[\eta]\in\htcde^k(M)\,,
\end{align*}
where the pairing $(\cdot,\cdot)_\de:\hscdd^k(M)\times\htcde^k(M)\to\bbR$ 
is defined in\ \eqref{eqSCTCPairing}. Recalling that $M$ is of finite type, 
Theorem\ \ref{thmSCTCPoincareDuality} ensures non-degeneracy of this pairing. 
We deduce that $[G\omega]=0\in\hscdd^k(M)$, 
therefore there exists $\theta\in\fsc^{k-1}(M)$ such that $\dd\theta=G\omega$. 
$\omega$ being coclosed, the identity $\de\dd\theta=0$ follows. 
For $k\geq2$, Lemma\ \ref{lemLorenzGaugeFixing} and Lemma\ \ref{lemLorenzSol}, 
together with\ Remark\ \ref{remSCSolModGaugeForms}, provide $\psi\in\fcde^{k-1}(M)$ 
and $\chi\in\fsc^{k-2}(M)$ such that $G\psi+\dd\chi=\theta$. 
For $k=1$, $\Box\theta=\de\dd\theta=0$, therefore one finds $\psi\in\cc(M)$ such that $G\psi=\theta$. 
In both cases one deduces the identity $G\dd\psi=\dd\theta=G\omega$, where $\psi$ is coclosed. 
Therefore one can find $\rho\in\fc^k(M)$ such that $\Box\rho=\omega-\dd\psi$. 
Acting with $\de$ on both sides of the last identity, we obtain $\Box\de\rho=-\de\dd\psi=-\Box\psi$ 
on account of $\omega$ and $\psi$ being coclosed, whence $\de\rho=-\psi$ via Theorem\ \ref{thmCausalProp}. 
This allows us to conclude that $\de\dd\rho=\Box\rho-\dd\de\rho=\Box\rho+\dd\psi=\omega$, 
thus completing the proof. 
\end{proof}

We can endow the space of observables $\obs_M$ for Maxwell $k$-forms over $M$ with a presymplectic structure 
defined in terms of the causal propagator $G$ for the Hodge-d'Alembert operator $\Box$. 
This choice can be motivated exploiting Peierls' method, {\em cfr.}\ \cite[Subsection\ 3.2]{SDH14} or 
\cite{Pei52} and \cite[Section\ I.4]{Haa96} for a presentation of this procedure in a broader context. 

\begin{definition}\index{presymplectic!vector space}\index{presymplectic!linear map}
\index{presymplectic!vector space!category}
A {\em presymplectic vector space} $(V,\sigma)$ consists in the assignment of a vector space $V$ 
endowed with a presymplectic form $\sigma:V\times V\to\bbR$, 
namely a (possibly degenerate) anti-symmetric bilinear form on $V$. 
A {\em presymplectic linear map} $L$ between the presymplectic spaces $(V,\sigma)$ and $(W,\tau)$ 
is a linear map $L:V\to W$ preserving the presymplectic structures, namely $\tau\circ(L\times L)=\sigma$. 
The {\em category of presymplectic vector spaces} $\PSymV$ has presymplectic vector spaces as objects 
and presymplectic linear maps as morphisms. 
\end{definition}

\begin{proposition}\label{prpPSymObsForms}
Let $M$ be an $m$-dimensional globally hyperbolic spacetime and consider $k\in\{1,\,\dots,\,m-1\}$. 
Denote with $G$ the causal propagator for the Hodge-d'Alembert operator $\Box$ on $M$, 
see\ Subsection\ \ref{subDynamics}. 
The anti-symmetric bilinear map 
\begin{align*}
\tau_M:\inv_M\times\inv_M\to\bbR\,, 
&& (\ev_\omega,\ev_{\omega^\prime})\mapsto(\omega,G\omega^\prime)\,,
\end{align*}
defined via the pairing $(\cdot,\cdot)$ between $k$-forms with compact overlapping support, 
see\ \eqref{eqPairing2}, 
induces a presymplectic form $\tau_M:\obs_M\times\obs_M\to\bbR$ on the quotient $\obs_M=\inv_M/\van_M$. 
\end{proposition}

\begin{proof}
The map $\tau_M$ is clearly well-defined and bilinear. Anti-symmetry follows from $(\cdot,\cdot)$ being symmetric 
and $G$ being formally antiself-adjoint with respect to $(\cdot,\cdot)$, see\ \eqref{eqGreenAdjoint}. 
More explicitly, given $\ev_\omega,\ev_{\omega^\prime}\in\inv_M$, we get the following chain of identities: 
\begin{align*}
(\omega,G\omega^\prime)=(G\omega^\prime,\omega)=-(\omega^\prime,G\omega)\,.
\end{align*}
Because of anti-symmetry, to conclude the proof it is enough to check 
that $\tau_M(\ev_\theta,\ev_\omega)=0$ for $\ev_\theta\in\van_M$ and $\ev_\omega\in\inv_M$. 
Since per assumption $\omega\in\fcde^k(M)$, the intertwining properties of Proposition\ \ref{prpIntertwiners} 
ensure that $\de\dd G\omega=\Box G\omega-\dd G\de\omega=0$, that is to say $G\omega\in\sol_M$. 
Therefore, by definition, we have 
\begin{equation*}
\tau_M(\ev_\theta,\ev_\omega)=(\theta,G\omega)=\ev_\theta(G\omega)=0\,,
\end{equation*}
thus showing that $\tau_M$ induces a presymplectic form $\tau_M$ on the quotient space $\obs_M$.  
\end{proof}

The last proposition provides a presymplectic form $\tau_M$ on the space of linear classical observables $\obs_M$ 
for on-shell Maxwell $k$-forms up to gauge. 
It seems natural to ask whether the presymplectic form introduced above 
is actually degenerate. This information turns out to be particularly relevant 
for the analysis of the functorial behavior of the model we are developing. 
Specifically, locality (in the sense of\ \cite[Definition\ 2.1]{BFV03}) requires 
that causal embeddings give rise to injective morphisms at the field theoretical level. 
This fact would be automatic if $(\obs_M,\tau_M)$ were a symplectic vector space 
for each globally hyperbolic spacetime $M$ and the maps induced at the level of observables by causal embeddings 
preserved the symplectic structures, see the subsequent Proposition \ref{prpSymInj}. 
However, we will see that degeneracies of the presymplectic form occur for Maxwell $k$-forms 
on globally hyperbolic spacetimes with suitable spacetime topologies. 
The next proposition provides a characterization of the null space of the presymplectic form 
and the subsequent example will present a globally hyperbolic spacetime where the null space is non-trivial. 

\begin{proposition}\label{prpRadForms}
Let $M$ be an $m$-dimensional globally hyperbolic spacetime and consider $k\in\{1,\dots,m-1\}$. 
The null space $\rad_M$ of the symplectic form $\tau_M$ has the following explicit expression: 
\begin{equation*}
\rad_M=\left\{\ev_{[\omega]}\in\obs_M\,:\;\omega\in\de(\fc^{k+1}(M)\cap\dd\ftc^k(M))\right\}\,,
\end{equation*}
In particular, if $M$ has compact Cauchy hypersurfaces, $\rad_M$ is trivial and $\tau_M$ is symplectic. 
\end{proposition}

\begin{proof}
We show the inclusion $\supseteq$ first. Consider $\omega\in\de(\fc^{k+1}(M)\cap\dd\ftc^k(M))$. 
This means there exists $\eta\in\ftc^k(M)$ such that $\dd\eta\in\fc^{k+1}(M)$ and $\de\dd\eta=\omega$. 
In particular, we note that $G\omega=G\Box\eta-G\dd\de\eta=-\dd G\de\eta$.
For any $\ev_{[\omega^\prime]}\in\obs_M$, 
\begin{align*}
\tau_M(\ev_{[\omega^\prime]},\ev_{[\omega]})=(\omega^\prime,G\omega)=-(\omega^\prime,\dd G\de\eta)
=-(\de\omega^\prime,G\de\eta)=0\,,
\end{align*}
where we exploited the fact that $\omega^\prime$ is coclosed for the last equality. 
This shows that $\ev_{[\omega]}\in\rad_M$. 

For the inclusion $\subseteq$, consider $\ev_{[\omega]}\in\obs_M$ 
such that $\tau_M(\ev_{[\omega^\prime]},\ev_{[\omega]})=0$ for each $\ev_{[\omega^\prime]}\in\obs_M$. 
Choosing a representative $\ev_\omega\in\inv_M$ for $\ev_{[\omega]}$, this condition translates into 
$(\omega^\prime,G\omega)=0$ for each $\omega^\prime\in\fcde^k(M)$. 
Consider first only $\omega^\prime\in\de\fc^{k+1}(M)$. Stokes' theorem entails that $\dd G\omega=0$, 
meaning that we can consider the cohomology class $[G\omega]\in\hdd^k(M)$. 
We can now reinterpret our assumption by saying that ${}_\de(\cdot,[G\omega])=0$ in $(\hcde^k(M))^\ast$, 
see\ \eqref{eqPoincarePairing} for the de\-fi\-ni\-tion of ${}_\de(\cdot,\cdot):\hcde^k(M)\times\hdd^k(M)\to\bbR$. 
According to\ Remark\ \ref{remPoincareDualityImproved}, the map $\hdd^k(M)\to(\hcde^k(M))^\ast$ 
induced by this pairing is an isomorphism. Therefore $[G\omega]=0$ in $\hdd^k(M)$, 
hence there exists $\chi\in\f^{k-1}(M)$ such that $\dd\chi=G\omega$. 
Since $\ev_\omega\in\inv_M$, $\de\omega=0$ and hence $\de\dd\chi=0$ as well. 
For $k\geq2$, applying Theorem\ \ref{thmSolModGaugeForms} to $[\chi]\in\sol_M$ (note that the degree is $k-1$), 
there exist $\alpha\in\ftcde^{k-1}(M)$ and $\beta\in\f^{k-2}(M)$ such that $G\alpha+\dd\beta=\chi$. 
For $k=1$, we have $\Box\chi=\de\dd\chi=0$, therefore there exists $\alpha\in\ctc(M)$ such that $G\alpha=\chi$. 
In both cases one concludes that $G\dd\alpha=\dd G\alpha=\dd\chi=G\omega$ 
for a suitable $\alpha\in\ftcde^{k-1}(M)$. 
The exact sequence in\ Theorem\ \ref{thmExtCausalProp} ensures the existence of $\eta\in\ftc^k(M)$ 
such that $\Box\eta=\omega-\dd\alpha$. 
In particular, $\Box\de\eta=-\de\dd\alpha=-\Box\alpha$, whence $\de\eta=-\alpha$. 
This identity implies that $\de\dd\eta=\Box\eta-\dd\de\eta=\Box\eta+\dd\alpha=\omega$. 
Furthermore, $\Box\dd\eta=\dd\de\dd\eta=\dd\omega$, hence, by the properties of the retarded/advanced 
Green operator $G^\pm$ for $\Box$, we conclude $\dd\eta=G^\pm\dd\omega$. 
Note that the identity holds true for both $G^+$ and $G^-$, meaning that the support of $\dd\eta$ 
is contained in the intersection of the causal future and the causal past of $\supp(\omega)$. 
Being a closed subset both of a past spacelike compact set and of a future spacelike compact one, 
we deduce that $\supp(\dd\eta)$ is compact. 
Therefore $\omega=\de\dd\eta\in\de(\fc^{k+1}(M)\cap\dd\ftc^k(M))$. 

The second part of the statement follows from the observation that timelike compact sets are also compact 
whenever $M$ admits a compact Cauchy hypersurface. 
\end{proof}

The subsequent example exhibits a class of globally hyperbolic spacetimes 
where the presymplectic form has a non-trivial null space. 

\begin{example}[Non-trivial null space]\label{exaNonTrivRadForms}
This example is a slight modification of\ \cite[Example\ 3.7]{SDH14}. Let $m\geq2$ and $k\in\{1,\dots,m-1\}$. 
We consider an $m$-dimensional globally hyperbolic spacetime $M$ with Cauchy hypersurface diffeomorphic to 
$\bbR^k\times\bbT^{m-k-1}$. 
Points in $M$ are denoted by tuples $(x^0,x^1,\dots,x^k,\varphi^{k+1},\dots,\varphi^{m-1})$. 
We consider a function $a\in\cc(\bbR)$ such that $\int_\bbR a(r)\,\dd r=1$ 
and we define $b\in\c(\bbR)$ such that $b(s)=\int_s^\infty a(r)\,\dd r$. 
For each $i\in\{0,\dots,k\}$, we set $\alpha_i={x^i}^\ast(a\,\dd s)$; 
moreover we set $\beta={x^1}^\ast b$. 
We can introduce $\theta=\alpha_0\wedge\alpha_1\wedge\cdots\wedge\alpha_k\in\fcdd^{k+1}(M)$ 
and $\eta=\beta\alpha_0\wedge\alpha_2\wedge\cdot\wedge\alpha_k\in\ftc^k(M)$. 
We observe that $[\theta]\in\hcdd^{k+1}(M)$ is non-trivial. In fact, for $k\in\{1,\dots,m-2\}$, 
and denoting with $\mu_j={\varphi^j}^\ast\mu$ the pull-back to $M$ of the volume form $\mu$ 
on $\bbT$ (normalized to $1$), we can consider $\nu=\mu_{k+1}\wedge\cdots\wedge\mu_{m-1}$. 
Since $\dd\mu=0$, we can consider the cohomology class $[\nu]\in\hdd^{m-k-1}(M)$. 
One can explicitly evaluate the Poincar\'e pairing between $[\theta]$ and $[\nu]$, see\ \eqref{eqPoincarePairing}: 
\begin{equation*}
\la[\theta],[\nu]\ra=\int_M\theta\wedge\nu
=\left(\int_\bbR a(r)\dd r\right)^{k+1}\left(\int_{\bbT}\mu\right)^{m-k-1}=1\,,
\end{equation*}
For $k=m-1$, we simply have 
\begin{equation*}
\la[\theta],[1]\ra=\int_M\theta=\left(\int_\bbR a(r)\dd r\right)^m=1\,.
\end{equation*}
In particular, this shows that $[\theta]\in\hcdd^{k+1}(M)$ is non-trivial. 
Furthermore, noting that $\dd\beta=-\alpha_1$, we deduce $\dd\eta=\theta$. 
Therefore $\ev_{[\omega]}\in\rad_M$ for $\omega=\de\theta$. To show that $\rad_M$ is actually non-trivial, 
we still have to prove that $\ev_\omega\in\inv_M\setminus\van_M$. 
Recalling Proposition\ \ref{prpVanForms} and noting that $M$ is of finite type, 
it is sufficient to show that $\omega\in\fcde^k(M)\setminus\de\dd\fc^k(M)$. 
In fact, assuming by contradiction, that there exists $\rho\in\fc^k(M)$ such that $\de\dd\rho=\omega$, 
we get the following chain of identities: 
\begin{equation*}
\Box\dd\rho=\dd\de\dd\rho=\dd\omega=\dd\de\theta=\Box\theta\,,
\end{equation*}
which implies $\dd\rho=\theta$, contradicting the fact that $[\theta]\in\hcdd^{k+1}(M)$ is non-trivial. 
\end{example}

Now we would like to understand whether the assignment of the presymplectic space $(\obs_M,\tau_M)$ 
to each globally hyperbolic spacetime $M$ behaves functorially. 

\begin{theorem}\label{thmFunctorForms}
Let $m\geq2$ and $k\in\{1,\,\dots,\,m-1\}$. 
Consider two $m$-dimensional globally hyperbolic spacetimes $M$ and $N$. 
Given a causal embedding $f:M\to N$, the push-forward $f_\ast:\fc^k(M)\to\fc^k(N)$ 
induces a presymplectic linear map $\PSV(f)$ between the presymplectic spaces 
$(\obs_M,\tau_M)$ and $(\obs_N,\tau_N)$. Furthermore, by setting $\PSV(M)=(\obs_M,\tau_M)$ 
for each $m$-dimensional globally hyperbolic spacetime, we get a covariant functor $\PSV:\GHyp\to\PSymV$.
\end{theorem}

\begin{proof}
Given $\ev_\omega\in\inv_M$, we check that $\ev_{f_\ast\omega}$ lies in $\inv_N$. 
In fact, $f^\ast(\gau_N)\subseteq\gau_M$ since $f^\ast\dd_N=\dd_Mf^\ast$. 
Furthermore, by a change of variables in the integral, for each $\alpha\in\fc^k(M)$ and each $\beta\in\f^k(N)$, 
we have $(f_\ast\alpha,\beta)_N=(\alpha,f^\ast\beta)_M$, where the subscripts are used to distinguish 
the pairing over $M$ from the one over $N$. Therefore 
\begin{equation*}
\ev_{f_\ast\omega}(\gau_N)=(f_\ast\omega,\gau_N)_N=(\omega,f^\ast(\gau_N))_M
\subseteq(\omega,\gau_M)_M=\{0\}\,,
\end{equation*}
thus showing that $\ev_{f_\ast\omega}$ is invariant. 
Consider now $\ev_\theta\in\van_M$. Since also $f^\ast\de_N=\de_Mf^\ast$, 
we have the inclusion $f^\ast\sol_N\subseteq\sol_M$. By the same argument as above, we obtain 
\begin{equation*}
\ev_{f_\ast\theta}(\sol_N)=(f_\ast\theta,\sol_N)_N
=(\theta,f^\ast(\sol_N))_M\subseteq(\theta,\sol_M)_M=\{0\}\,,
\end{equation*}
thus showing that $\ev_{f_\ast\theta}$ vanishes on $\sol_N$. 
Summing up, we can conclude that $f_\ast:\fc^k(M)\to\fc^k(N)$ induces a linear map $\PSV(f):\obs_M\to\obs_N$. 

In order to check that $\PSV(f):\obs_M\to\obs_N$ preserves the presymplectic structures, 
we take $\ev_{[\omega]},\ev_{[\omega^\prime]}\in\obs_M$ and perform the following computation: 
\begin{equation*}
(f_\ast\omega,G_Nf_\ast\omega^\prime)_N=(\omega,f^\ast G_Nf_\ast\omega^\prime)_M
=(\omega,G_M\omega^\prime)_M\,,
\end{equation*}
where we used the identity $f^\ast G_Nf_\ast=G_M$, which holds true on $\fc^k(M)$ 
on account of Proposition\ \ref{prpGreenNat}. This shows that $\PSV(f)^\ast\tau_N=\tau_M$, 
whence $\PSV(f):(\obs_M,\tau_M)\to(\obs_N,\tau_N)$ is a presymplectic linear map. 

Since $\PSV(f)$ is defined via the push-forward $f_\ast:\fc^k(M)\to\fc^k(N)$ along $f$, 
$\PSV$ inherits its functorial properties from the covariant functor $\fc^k(\cdot)$. This concludes the proof. 
\end{proof}

We now analyze certain properties of the covariant functor $\PSV:\GHyp\to\PSymV$ to make contact 
with the formulation of field theory on curved spacetimes proposed in\ \cite[Definition\ 2.1]{BFV03}. 

\begin{theorem}\label{thmCausalityForms}\index{causality}
Let $m\geq2$ and $k\in\{1,\,\dots,\,m-1\}$. 
The covariant functor $\PSV:\GHyp\to\PSymV$ fulfils the {\em causality} property: 
Consider the $m$-dimensional globally hyperbolic spacetimes $M_1$, $M_2$ and $N$. Furthermore, 
suppose that $f:M_1\to N$ and $h:M_2\to N$ are causal embeddings with causally disjoint images in $N$, 
namely $f(M_1)\cap J_N(h(M_2))=\emptyset$. Then 
\begin{equation*}
\tau_N(\PSV(f)\obs_{M_1},\PSV(h)\obs_{M_2})=\{0\}\,.
\end{equation*}
\end{theorem}

\begin{proof}
Let $\ev_{[\omega_1]}\in\obs_{M_1}$ and $\ev_{[\omega_2]}\in\obs_{M_2}$. By definition we have
\begin{equation*}
\tau_O(\PSV(f)\ev_{[\omega_1]},\PSV(h)\ev_{[\omega_2]})=(f_\ast\omega_1,G_Nh_\ast\omega_2)_N\,.
\end{equation*}
The inclusion $\supp(G_Nh_\ast\omega_2)\subseteq J_N(h(M_2))$ follows 
from the properties of the causal propagator $G_N$. 
Per hypothesis $\supp(f_\ast\omega_1)\subseteq f(M_1)$ does not meet $J_N(h(M_2))$. 
Therefore the right-hand-side of the last equation vanishes, thus concluding the proof. 
\end{proof}

To prove the next theorem, namely the time-slice axiom for the functor $\PSV$, we need a preliminary result. 

\begin{lemma}\label{lemTimeSliceForms2}
Let $M$ be an $m$-dimensional globally hyperbolic spacetime and consider $k\in\{1,\dots,m-1\}$. 
Furthermore, assume that a partition of unity $\{\chi_+,\chi_-\}$ 
with $\supp(\chi_\pm)$ past/future compact is given. Given $\omega\in\fcde^k(M)$, 
there exists $\xi\in\fcde^k(M)$ with support inside $\supp(\chi_+)\cap\supp(\chi_-)$  
such that $\omega-\xi\in\de\dd\fc^k(M)$. Furthermore, if $\xi^\prime$ satisfies the same condition, 
then $\xi^\prime-\xi\in\de\dd\fc^k(M)$ is supported inside $\supp(\chi_+)\cap\supp(\chi_-)$. 
\end{lemma}

\begin{proof}
First, recall that $\de\omega=0$ and $\Box G\omega=0$, therefore $\de\dd G\omega=0$. 
It follows that $\de\dd(\chi_+G\omega)=-\de\dd(\chi_-G\omega)$, 
whence $\xi=\de\dd(\chi_+G\omega)$ has compact support. In fact, we have the inclusion 
\begin{equation*}
\supp(\xi)\subseteq J_M(\supp(\omega))\cap\supp(\chi_+)\cap\supp(\chi_-)
\end{equation*}
showing that the right-hand-side is compact being the intersection of a spacelike compact, 
a past compact and a future compact set. 
The subsequent computation shows that $\omega-\xi$ has the desired properties: 
\begin{align*}
\omega-\xi=\de\dd G^+\omega-\de\dd(\chi_+G^+\omega)+\de\dd(\chi_+G^-\omega)=\de\dd\eta\,.
\end{align*}
Note that in the first equality we exploited the identity $\de\dd G^+\omega=\omega$ and the definition of $G$,
while in the second equality we used the partition of unity $\{\chi_+,\chi_-\}$ to introduce 
$\eta=\chi_-G^+\omega+\chi_+G^-\omega$. It turns out that $\eta$ has compact support in $N$ since 
\begin{equation*}
\supp(\eta)\subseteq (J_N^+(\supp(\omega))\cap\supp(\chi_-))\cup(J_N^-(\supp(\omega))\cap\supp(\chi_+))
\end{equation*}
is the union of two compact sets, the first being the intersection between a past spacelike compact and a future 
compact set, while the second being the intersection between a future spacelike compact and a past compact set. 

The second part of the statement follows immediately from the first. 
\end{proof}

\begin{theorem}\label{thmTimeSliceForms}\index{time-slice axiom}
Let $m\geq2$ and $k\in\{1,\,\dots,\,m-1\}$. The covariant functor $\PSV:\GHyp\to\PSymV$ 
fulfils the {\em time-slice axiom}: Consider the $m$-dimensional globally hyperbolic spacetimes $M$ and $N$. 
Furthermore, assume that $f:M\to N$ is a Cauchy morphism. Then $\PSV(f):\PSV(M)\to\PSV(N)$ is an isomorphism. 
\end{theorem}

\begin{proof}
First of all, note that each causal embedding $f:M\to N$ can be factored into 
an isomorphism $M\to f(M)$ and an inclusion $f(M)\to N$ according to the diagram below: 
\begin{equation*}
\xymatrix{
M\ar[r]^\simeq\ar[rd]_f & f(M)\ar[d]^\subseteq\\
& N
}
\end{equation*}
Note that the one above is a commutative diagram in the category $\GHyp$. 
In particular, $f(M)$ is a globally hyperbolic spacetime in its own right since $f$ is a causal embedding. 
Being a functor, $\PSV$ turns isomorphisms into isomorphisms, therefore it remains only to check that, 
under our assumptions, the inclusion $f(M)\to N$ gives rise to an isomorphism via $\PSV$. 

Since $f$ is a Cauchy morphism, $f(M)$ is an open neighborhood of a Cauchy hypersurface of $N$. 
In particular, we can find Cauchy hypersurfaces $\Sigma_+$ and $\Sigma_-$ for $N$ which lie inside $f(M)$ 
and such that $\Sigma_+\cap J^-_N(\Sigma_-)=\emptyset$. This enables us to choose a partition of unity 
$\{\chi_+,\chi_-\}$ subordinated to the open cover $\{I_N^+(\Sigma_-),I_N^-(\Sigma_+)\}$ of $N$. 

To simplify the notation, we set $O=f(M)\subseteq N$ and we keep the inclusion map implicit. 
Assume $\ev_{[\omega]}\in\obs_N$ is given and choose a representative $\ev_\omega\in\inv_N$. 
Applying Lemma\ \ref{lemTimeSliceForms2}, we find $\ev_\xi\in\inv_O$ 
such that $\ev_\omega-\ev_\xi\in\van_N$ (see the comment before\ \eqref{eqObsForms}) and  
moreover, if $\ev_{\xi^\prime}$ does the job too, then $\ev_{\xi^\prime}-\ev_\xi$ lies in $\van_N\cap\inv_O$. 
Furthermore, if we assume $\ev_\omega\in\van_N$, this implies $\ev_\xi\in\van_N\cap\inv_O$ as well. 
Summing up, if we were able to show that $\van_N\cap\inv_O=\van_O$, 
then our procedure would define a linear map $K:\obs_N\to\obs_O$. 

Let $\ev_\theta\in\van_N\cap\inv_O$ and consider $[A]\in\sol_O$. 
Lemma\ \ref{lemTimeSliceForms1} provides $\alpha\in\ftcde^k(O)$ 
supported in $J_N^+(\Sigma_-)\cap J_N^-(\Sigma_+)$ such that $[G\alpha]=[A]\in\sol_O$ 
(here $G$ denotes the causal propagator for $\Box$ on $O$). 
Being supported between the Cauchy hypersurfaces $\Sigma_+$ and $\Sigma_-$, 
$\alpha$ can be extended by zero to a $k$-form with timelike compact support on the whole $N$. 
In particular, we can consider $[G\alpha]\in[\sol_N]$ (here $G$ denotes the causal propagator for $\Box$ on $N$) 
and conclude that this restricts to $[A]\in[\sol_O]$. 
Since $\ev_\theta\in\inv_O$, it can be evaluated on $[A]\in\sol_O$. 
Since the outcome of this evaluation is expressed in term of an integral over $O$ 
of an $m$-form with compact support inside $O$, we can equivalently perform this integral over $N$. 
\begin{equation*}
\ev_\theta([A])=\ev_\theta([G\alpha])=(\theta,G\alpha)_O=(\theta,G\alpha)_N=\ev_\theta([G\alpha])\,,
\end{equation*}
where $[G\alpha]\in[\sol_N]$ appears after the last equality, 
while its restriction to $O$ appears after the first one. Since $\ev_\theta\in\van_N$, the identity above allows us 
to conclude that $\ev_\theta([A])=\ev_\theta([G\alpha])=0$, namely $\ev_\theta\in\van_O$. 
The inclusion $\van_O\subseteq\van_N\cap\inv_O$ is trivial 
because each element in $[\sol_N]$ restricts to one in $[\sol_O]$, 
therefore we conclude that the following linear map is well-defined by the procedure presented above: 
\begin{align*}
K:\obs_N\to\obs_O\,, && \ev_{[\omega]}\to\ev_{[\xi]}\,.
\end{align*}

To conclude the proof it remains only to check that $K$ is the inverse of the presymplectic linear map $L$ 
obtained via the functor $\PSV$ applied to the inclusion map of $O$ in $N$. 
Let $\ev_{[\xi]}\in\obs_O$. $L\ev_{[\xi]}\in\obs_N$ is obtained simply choosing a representative $\ev_\xi\in\inv_O$ 
and extending it by zero to the whole $N$. Since per construction $L\ev_{[\xi]}$ admits a representative 
whose restriction to $O$ coincides with $\ev_\xi\in\inv_O$, 
to represent $KL\ev_{[\xi]}\in\obs_O$ we can again consider $\ev_\xi\in\inv_O$, 
meaning that $KL\ev_{[\xi]}=\ev_{[\xi]}$. Conversely, consider $\ev_{[\omega]}\in\obs_N$. 
Once a representative $\ev_\omega\in\inv_N$ has been chosen, by the argument presented above 
it is possible to find $\ev_\xi\in\inv_O$ such that $\ev_\omega-\ev_\xi\in\van_N$. 
This defines $K\ev_{[\omega]}=\ev_{[\xi]}\in\obs_O$. As above, we can evaluate $LK\ev_{[\omega]}\in\obs_N$ 
choosing a representative of $K\ev_{[\omega]}\in\obs_O$ and extending it by zero to the whole $N$. 
$\ev_\xi\in\inv_O$ is indeed such a representative and moreover, by construction, 
its extension differs from $\ev_\omega$ by an element in $\van_N$, 
therefore $LK\ev_{[\omega]}=\ev_{[\omega]}$. This shows that $L$ is indeed an isomorphism in $\PSymV$. 
\end{proof}

\subsection{Failure of locality}\label{subLocFailForms}
Here we discuss the locality property, as stated in\ \cite[Definition\ 2.1]{BFV03}, in the case of Maxwell $k$-forms, 
namely the requirement that causal embeddings should induce injective morphisms via $\PSV:\GHyp\to\PSymV$. 
As anticipated, we show first that this would be the case 
if $\PSV:\GHyp\to\PSymV$ were valued in the full subcategory of symplectic vector spaces. 
However, we have already exhibited examples where this does not happen, 
{\em cfr.}\ Example\ \ref{exaNonTrivRadForms}. 
We prove also that the kernel of a presymplectic map is always a subspace of the null space of its domain. 
This is indeed the reason why we look for a non-trivial element of the null space 
in order to explicitly show in the subsequent example a causal embedding 
which induces a morphism in $\PSymV$ with a non-trivial kernel. 

\begin{proposition}\label{prpSymInj}
Let $(V,\sigma)$ and $(W,\tau)$ be presymplectic vector spaces 
and consider a presymplectic linear map $L:(V,\sigma)\to(W,\tau)$. 
Then $\ker(L)$ is a vector subspace of the null space of $\sigma$. 
In particular, if $\sigma$ is non-degenerate, namely $(V,\sigma)$ is actually symplectic, then $L$ is injective. 
\end{proposition}

\begin{proof}
The first part of the statement follows from the fact that $L$ preserves the presymplectic forms. 
Let $v\in\ker(L)$. Then, $\sigma(v,v^\prime)=\tau(Lv,Lv^\prime)=0$ for each $v^\prime\in V$. 
This means that $v$ lies in the null space of $\sigma$. 

The second part of the statement follows from the extra assumption of non-degeneracy. 
In fact, under this hypothesis the null space of $\sigma$ is trivial, 
therefore its subspace $\ker(L)$ must be trivial as well. 
\end{proof}

In the example below we show that for Maxwell $k$-forms there are causal embeddings 
giving rise to presymplectic linear maps with a non-trivial kernel. 

\begin{example}[Non-injective presymplectic maps from causal embeddings]\label{exaNonInjForms}
This example is similar to the one in\ \cite[Example 6.9]{BSS14}. Let $m\geq3$ and consider $k\in\{1,\dots,m-2\}$. 
As a starting point, consider an $m$-dimensional globally hyperbolic spacetime $N$ 
with Cauchy hypersurface $\Sigma^\prime=\bbR^{k+1}\times\bbT^{m-k-2}$. 
Exploiting\ Theorem\ \ref{thmGlobHyp}, $N$ is foliated as $\bbR\times\Sigma^\prime$, 
therefore we can introduce the projection $t:N\to\bbR$ on the first factor such that 
$\Sigma^\prime$ is the locus $t=0$ in $N$, $\Sigma^\prime=\{0\}\times\bbR^{k+1}\times\bbT^{m-k-2}$. 
Consider $\Sigma=\{0\}\times\bbR^{k-1}\times(\bbR^2\setminus\{0\})\times\bbT^{m-k-2}
\subseteq\Sigma^\prime$ and define $M$ to be the Cauchy development of $\Sigma$ in $N$. 
$M$ turns out to be a globally hyperbolic spacetime causally embedded in $N$, see\ \cite[Lemma A.5.9]{BGP07}. 
We note that $\bbR^2\setminus\{0\}$ is diffeomorphic to $\bbR\times\bbT$. Since $\Sigma$ is per construction 
a Cauchy hypersurface for $M$, by\ Theorem \ref{thmGlobHyp} $M$ admits a foliation 
$\bbR\times\bbR^{k-1}\times(\bbR\times\bbT)\times\bbT^{m-k-2}=\bbR\times\bbR^k\times\bbT^{m-k-1}$. 
Therefore we can apply to $M$ the construction of\ Example \ref{exaNonTrivRadForms} 
to exhibit a non-trivial element $\ev_{[\omega]}\in\rad_M$. 
Under the causal embedding $M\subseteq N$, this observable is simply extended by zero. 
Since per construction $\omega\in\de\fcdd^{k+1}(M)$, after the extension $\omega\in\de\fcdd^{k+1}(N)$. 
Note that $\hcdd^{k+1}(N)\simeq\{0\}$, hence $\fcdd^{k+1}(N)=\dd\fc^k(N)$ 
and therefore $\ev_\omega\in\van_N$. This shows that the presymplectic linear map 
obtained applying the functor $\PSV$ to the causal embedding $M\subseteq N$ is not injective. 

For $m\geq2$ and $k=m-1$, we take an $m$-dimensional globally hyperbolic spacetime $N$ 
with Cauchy hypersurface $\Sigma^\prime=\bbR^{m-1}$. By the usual foliation, 
$\Sigma^\prime$ is seen as the $t=0$ hypersurface in $N$, namely $\Sigma^\prime=\{0\}\times\bbR^{m-1}$.  
We consider $\Sigma=\{0\}\times\bbR^{m-2}\times(\bbR^1\setminus\{0\})\subseteq\Sigma^\prime$ 
and we note that it is made of two disconnected components $\Sigma_1$ and $\Sigma_2$, 
therefore $M=M_1\sqcup M_2$, which is obtained as the Cauchy development of $\Sigma$ in $N$, 
is made of two disconnected components as well. 
Let us consider $\theta_i\in\fc^m(M_i)$ such that $\int_{M_i}\theta_i=(-1)^i$, $i\in\{1,2\}$. 
We define $\theta\in\fc^m(M)$ by setting $\theta=\theta_i$ on $M_i$, $i\in\{1,2\}$. 
Applying\ Theorem \ref{thmTCCohomology} to $M$, 
we deduce $\htcdd^m(M)\simeq\hdd^{m-1}(\Sigma)\simeq\{0\}$, 
$\Sigma$ being the disjoint union of two manifolds diffeomorphic to $\bbR^{m-1}$with $m\geq2$. 
This means that $\theta\in\fc^m(M)\cap\dd\ftc^{m-1}(M)$, 
therefore $\omega=\theta$ defines $\ev_{[\omega]}\in\rad_M$. 
Per construction, $[\theta]\in\hcdd^m(M)$ is non-trivial. In fact, 
$\hcdd^m(M)\simeq\hcdd^m(M_1)\oplus\hcdd^m(M_2)$, 
$M_1$ and $M_2$ being the two disconnected components of $M$, 
and neither $[\theta_1]\in\hcdd^m(M_1)$ nor $[\theta_2]\in\hcdd^m(M_2)$ are trivial, 
since their integrals do not vanish. By the argument given in the last part of\ Example\ \ref{exaNonTrivRadForms}, 
this fact implies $\ev_{[\omega]}\in\rad_M$ is non-trivial. 
Also in this case the inclusion $M\subseteq N$ is a causal embedding, see\ \cite[Lemma A.5.9]{BGP07}. 
Now we show that the presymplectic linear map obtained applying the functor $\PSV$ 
to the causal embedding $M\subseteq N$ is not injective. 
In fact, the extension by zero of $\theta$ to $N$ has vanishing integral, 
\begin{equation*}
\int_N\theta=\int_M\theta=\int_{M_1}\theta_1+\int_{M_2}\theta_2=-1+1=0\,,
\end{equation*}
meaning that $\la[\theta],[1]\ra=0$, 
where $[\theta]\in\hcdd^m(N)$ denotes the cohomology class of the extension by zero of $\theta$ 
and $[1]\in\hdd^0(N)$ is a generator of $\hdd^0(N)\simeq\bbR$, $N$ being connected. 
Since the pairing $\la\cdot,\cdot\ra:\hcdd^m(N)\times\hdd^0(M)\to\bbR$ is non-degenerate 
({\em cfr.}\ Theorem \ref{thmPoincareDuality}, keeping in mind that $N$ is of finite type), 
it follows that $[\theta]\in\hcdd^m(N)$ is trivial. Therefore, mapping $\ev_{[\omega]}\in\rad_M$ to $N$ 
gives an element in $\van_N$, thus showing that the presymplectic linear map obtained 
applying the functor $\PSV$ to the causal embedding $M\subseteq N$ is not injective. 
\end{example}

\begin{theorem}\label{thmNonLocForms}
Let $m\geq2$ and $k\in\{1,\dots,m-1\}$. The covariant functor $\PSV:\GHyp\to\PSymV$ 
violates the locality property according to\ \cite[Definition\ 2.1]{BFV03}, namely 
there exists a causal embedding $f$ such that $\PSV(f)$ has non-trivial kernel. 
\end{theorem}

\begin{proof}
Example\ \ref{exaNonInjForms} exhibits a causal embedding such that the presymplectic map induced 
by $\PSV$ has non-trivial kernel (consider the first part of the example for $m\geq3$, $k\in\{1,\dots,m-2\}$ 
and the second for $m\geq2$, $k=m-1$). 
\end{proof}

\begin{remark}\label{remLocViolationForms}
The last theorem exhibits the violation of locality of the functor $\PSV$. 
In particular, it shows that, for $m\geq3$, $k\in\{1,\dots,m-2\}$, the violation of locality occurs 
even in the full subcategory of $\GHyp$ whose objects are connected globally hyperbolic spacetimes. 
Note that this includes the most prominent physical situation, namely the vector potential of free electromagnetism 
over four-dimensional globally hyperbolic spacetimes, that is $m=4$ and $k=1$. 
\end{remark}

Example\ \ref{exaNonInjForms} shows a quite explicit violation of the locality property. 
The idea used to produce such violation is the fact that causal embeddings do not always preserve 
compactly supported cohomology groups. In fact, we are going to show that, for $k\in\{1,\dots,m-1\}$, 
the failure of injectivity of $\PSV(f):\PSV(M)\to\PSV(N)$, $f:M\to N$ being a causal embedding 
between $m$-dimensional globally hyperbolic spacetimes, can be traced back to the failure of 
injectivity of the push-forward $\hcdd^{k+1}(f):\hcdd^{k+1}(M)\to\hcdd^{k+1}(N)$ 
for the compactly supported de Rham $(k+1)$-th cohomology group. 

The next lemma shows that, for $k\in\{1,\dots,m-1\}$, the cohomology functor $\hcdd^{k+1}(\cdot)$ 
is naturally included in the covariant functor $\PSV$. 

\begin{lemma}\label{lemLocCohoForms}
Let $m\geq2$ and $k\in\{1,\dots,m-1\}$. On an $m$-dimensional globally hyperbolic spacetime $M$, 
consider the presymplectic linear map 
\begin{align*}
\Delta:\hcdd^{k+1}(M)\to\PSV(M)\,, && [\theta]\to[\de\theta]\,,
\end{align*}
where $\hcdd^{k+1}(M)$ is regarded as a presymplectic space endowed with the trivial presymplectic structure. 
$\Delta:\hcdd^{k+1}(\cdot)\to\PSV$ defines a natural transformation 
between covariant functors taking values in the category $\PSymV$ of presymplectic vector spaces, 
which is injective when restricted to the full subcategory $\GHypF$ 
of $m$-dimensional globally hyperbolic spacetimes of finite type. 
\end{lemma}

\begin{proof}
First of all, note that $\Delta:\hcdd^{k+1}(M)\to\PSV(M)$ is well-defined. 
This follows from the fact that $\de:\fc^{k+1}(M)\to\fc^{k}(M)$ maps $\dd\fc^k(M)$ to $\de\dd\fc^k(M)$. 
Furthermore, $\de:\hcdd^{k+1}(M)\to\PSV(M)$ preserves the presymplectic structures. 
In fact, for each $\theta,\eta\in\fcdd^{k+1}(M)$, $\tau_M([\de\theta],[\de\eta])=0$. 
The fact that $\de:\fc^{k+1}(\cdot)\to\fc^{k}(\cdot)$ is a natural transformation entails that 
$\Delta:\hcdd^{k+1}(\cdot)\to\PSV$ is such as well. It remains only to check that, 
for each $m$-dimensional globally hyperbolic spacetime $M$, $\Delta:\hcdd^{k+1}(M)\to\PSV(M)$ is injective. 
Consider $\theta\in\fcdd^{k+1}(M)$ and, recalling Proposition\ \ref{prpVanForms}, 
suppose there exists $\eta\in\fc^k(M)$ such that $\de\dd\eta=\de\theta$. 
Applying $\dd$ on both sides and keeping in mind that $\dd\theta=0$, we get $\Box\dd\eta=\Box\theta$. 
Since both $\eta$ and $\theta$ have compact supports, we conclude that $\theta=\dd\eta$. 
Therefore $\Delta:\hcdd^{k+1}(M)\to\PSV(M)$ is injective as claimed. 
\end{proof}

\begin{theorem}\label{thmLocCohoForms}
Let $m\geq2$ and $k\in\{1,\dots,m-1\}$. Consider a causal embedding $f:M\to N$ between 
$m$-dimensional globally hyperbolic spacetimes of finite type. Then $\PSV(f):\PSV(M)\to\PSV(N)$ is injective 
if and only if $\hcdd^{k+1}(f):\hcdd^{k+1}(M)\to\hcdd^{k+1}(N)$ is injective too. 
\end{theorem}

\begin{proof}
According to\ Lemma\ \ref{lemLocCohoForms}, a causal embedding $f:M\to N$ induces the commutative diagram 
\begin{equation*}
\xymatrix@C=3em{
\hcdd^{k+1}(M)\ar[d]_{\Delta}\ar[r]^{\hcdd^{k+1}(f)} & \hcdd^{k+1}(N)\ar[d]^\Delta\\
\PSV(M)\ar[r]_{\PSV(f)} & \PSV(N)
}
\end{equation*}
where the vertical arrows are injective since $M$ and $N$ are of finite type. 
Clearly, if $\PSV(f)$ is injective, also $\hcdd^{k+1}(f)$ is such. 
For the converse, suppose that $\hcdd^{k+1}(f)$ is injective. 
Note that the null space $\rad_M$ of the presymplectic form $\tau_M$ is 
a subspace of $\Delta(\hcdd^{k+1}(M))$, {\em cfr.}\ Proposition\ \ref{prpRadForms}. 
From the diagram above we deduce that the only element shared by $\ker(\PSV(f))$ 
and $\Delta(\hcdd^{k+1}(M))$ is zero. However, according to\ Proposition\ \ref{prpSymInj}, 
$\ker(\PSV(f))$ is a subspace of $\rad_M$. We deduce that $\PSV(f)$ is injective. 
\end{proof}

One might wonder whether it is possible to recover locality by taking quotients in a suitable sense. 
We will show that this is not the case. 
Let us first define what we mean by a suitable quotient of a covariant functor taking values in $\PSymV$. 

\begin{definition}\label{defQSubfunctorV}
\index{subfunctor}\index{quotientable subfunctor}\index{subfunctor!quotientable}
Let $\mathsf{C}$ be a category and consider a covariant functor $\mathfrak{F}:\mathsf{C}\to\PSymV$. 
A covariant functor $\mathfrak{S}:\mathsf{C}\to\PSymV$ is a {\em subfunctor} of $\mathfrak{F}$ 
if the following conditions are fulfilled: 
\begin{enumerate}
\item For each object $c$ of the category $C$, 
$\mathfrak{S}(c)$ is a presymplectic vector subspace of $\mathfrak{F}(c)$; 
\item For each morphism $\gamma:c\to d$ in the category $C$, 
$\mathfrak{S}(\gamma):\mathfrak{S}(c)\to\mathfrak{S}(d)$ is the restriction to $\mathfrak{S}(c)$ 
of $\mathfrak{F}(\gamma):\mathfrak{F}(c)\to\mathfrak{F}(d)$. 
\end{enumerate}
A {\em quotientable subfunctor} $\mathfrak{Q}:\mathsf{C}\to\PSymV$ of the covariant functor 
$\mathfrak{F}:\mathsf{C}\to\PSymV$ is a subfunctor such that, for each object $c$ of the category $C$, 
$\mathfrak{Q}(c)$ is a presymplectic vector subspace of the null space of $\mathfrak{F}(c)$. 
\end{definition}

For $m\geq2$ and $k\in\{1,\dots,m-1\}$, Lemma\ \ref{lemLocCohoForms} shows that 
$\hcdd^{k+1}(\cdot):\GHyp\to\PSymV$ is a subfunctor of $\PSV:\GHyp\to\PSymV$. 
However, this is not quotientable since there exist $m$-dimensional globally hyperbolic spacetimes $M$ 
for which $\fc^{k+1}(M)\cap\dd\ftc^k(M)\neq\dd\fc^k(M)$, see\ Example\ \ref{exaNonTrivRadForms}. 

Once a quotientable subfunctor is given, it is possible to perform quotients at the level of functors. 
The next proposition makes this statement precise. 

\begin{proposition}\label{prpQuotientFunctor}\index{quotient functor}
Let $\mathsf{C}$ be a category and consider a covariant functor $\mathfrak{F}:\mathsf{C}\to\PSymV$. 
Furthermore, let $\mathfrak{Q}:\mathsf{C}\to\PSymV$ be a quotientable subfunctor of $\mathfrak{F}$. 
To each object $c$ in $\mathsf{C}$ assign the object 
$\mathfrak{F}/\mathfrak{Q}(c)=\mathfrak{F}(c)/\mathfrak{Q}(c)$ in $\PSymV$. 
Furthermore, to each morphism $\gamma:c\to d$ in $\mathsf{C}$ assign the morphism 
$\mathfrak{F}/\mathfrak{Q}(\gamma):\mathfrak{F}/\mathfrak{Q}(c)\to\mathfrak{F}/\mathfrak{Q}(d)$ 
in $\PSymV$ induced by $\mathfrak{F}(\gamma):\mathfrak{F}(c)\to\mathfrak{F}(d)$. 
Then $\mathfrak{F}/\mathfrak{Q}:\mathsf{C}\to\PSymV$ is a covariant functor, 
called {\em quotient functor} of $\mathfrak{F}$ by $\mathfrak{Q}$. 
\end{proposition}

\begin{proof}
By\ Definition\ \ref{defQSubfunctorV}, for each object $c$ in $\mathsf{C}$, 
$\mathfrak{Q}(c)$ is a presymplectic vector subspace of the null space of $\mathfrak{F}(c)$. 
Therefore the quotient of the vector space underlying $\mathfrak{F}(c)$ 
by the vector space underlying $\mathfrak{Q}(c)$ can be endowed with the presymplectic structure 
which is consistently induced from the one of $\mathfrak{F}(c)$. 
This defines $\mathfrak{F}/\mathfrak{Q}(c)$ as an object in $\PSymV$. 

Furthermore, $\mathfrak{Q}$ being a subfunctor of $\mathfrak{F}$, 
given a morphism $\gamma:c\to d$ in $\mathsf{C}$, we know that 
$\mathfrak{Q}(\gamma):\mathfrak{Q}(c)\to\mathfrak{Q}(d)$ is the restriction 
to $\mathfrak{Q}(c)$ of $\mathfrak{F}(\gamma):\mathfrak{F}(c)\to\mathfrak{F}(d)$. 
Furthermore, $\mathfrak{Q}(d)$ is a subspace of the null space of $\mathfrak{F}(d)$. 
This entails that $\mathfrak{F}(\gamma)$ descends to the quotients, thus providing a presymplectic map 
$\mathfrak{F}/\mathfrak{Q}(\gamma):\mathfrak{F}/\mathfrak{Q}(c)\to\mathfrak{F}/\mathfrak{Q}(d)$. 

$\mathfrak{F}/\mathfrak{Q}$ inherits all its functorial properties from $\mathfrak{F}$ 
and therefore it is a covariant functor as expected. 
\end{proof}

The forthcoming example is a very explicit construction involving two causal embeddings 
having the same globally hyperbolic spacetime $M$ as source. 
These embeddings are such that there is an element in the null space of $\PSV(M)$ 
which lies in the kernel of one of the induced presymplectic linear maps 
and which is mapped out of the null space of the target by the other induced presymplectic linear map. 
This provides the counterexample to show that there are no quotientable subfunctors of $\PSV$ 
which implement injectivity at the level of morphisms. 
Therefore recovering locality in the sense of\ \cite[Definition\ 2.1]{BFV03} is not possible. 

\begin{example}\label{exaNoQuotientForms}
Let $m\geq3$ and $k\in\{1,\dots,m-2\}$. 
Consider the $m$-dimensional globally hyperbolic spacetime $N$ defined as follows: 
\begin{enumerate}
\item The underlying manifold is $\bbR\times\bbT^{m-1}$ 
and its points are denoted by $m$-tuples $(t,\theta^1,\dots,\theta^{m-1})$; 
\item The metric is given by 
$g_N=-{\dd t}^2+\sum_{i=1}^{m-1}{\dd\theta^i}^2$; 
\item The orientation is specified by 
$\dd t\wedge\bigwedge_{i=1}^{m-1}\dd\theta^i$, 
while the time-orientation is specified by the vector field $\partial_t$. 
\end{enumerate}

Furthermore, consider the $m$-dimensional globally hyperbolic spacetime $O$ defined by the following data: 
\begin{enumerate}
\item The underlying manifold is $\bbR\times\bbR^{k+1}\times\bbT^{m-k-2}$ 
and its points are denoted by $m$-tuples $(t,y^1,\dots,y^{k+1},\eta^{k+2},\dots,\eta^{m-1})$; 
\item The metric is given by 
\begin{equation*}
g_O=-{\dd t}^2+\sum_{i=1}^{k-1}{\dd y^i}^2
+\alpha(r^2)\left({\dd y^k}^2+{\dd y^{k+1}}^2\right)+\beta(r^2){\dd r}^2
+\sum_{j=k+2}^{m-1}{\dd\eta^j}^2\,,
\end{equation*}
where $r=r(y^k,y^{k+1})=({y^k}^2+{y^{k+1}}^2)^{1/2}$ 
and $\alpha,\beta\in\c(\bbR)$ are functions taking values in the interval $[0,1]$ 
such that, for $\xi\le1$, $\alpha(\xi)=1$ and $\beta(\xi)=0$, 
while, for $\xi\ge4$, $\alpha(\xi)=\xi^{-1}$ and $\beta(\xi)=1-\xi^{-1}$. 
\item The orientation is specified by 
$\dd t\wedge\bigwedge_{i=1}^{k+1}\dd y^i\wedge\bigwedge_{j=k+2}^{m-1}\dd\eta^j$, 
while the time-orientation is specified by the vector field $\partial_t$. 
\end{enumerate}
Note that, in the region $r\geq2$ and using the coordinates specified above, 
$g_O$ has the same formal expression as $g_N$. 
This feature plays a central role in order to show that the $m$-dimensional globally spacetime $M$ 
defined below is causally embedded in both $N$ and $O$. 

Consider now the acausal hypersurface $\Sigma=\{0\}\times I^k\times\bbT^{m-k-1}$ of $N$, 
where $I\subset\bbT$ is an arbitrary open interval of $(0,2\pi)$. 
Define the $m$-dimensional globally hyperbolic spacetime $M$ to be the Cauchy development of $\Sigma$ in $N$ 
endowed with the induced metric $g_M=g_N\vert_M$, orientation and time-orientation. 
As a byproduct of this procedure, we obtain a causal embedding $f:M\to N$. 

We would like to causally embed $M$ in $O$ too. 
Exploiting the fact that $M\subseteq N$ as sets, we represent points of $M$ using the coordinates for $N$. 
Keeping this in mind, we introduce the following embedding: 
\begin{align*}
h:M\to O\,, && (t,\,\theta^1,\,\dots,\,\theta^{m-1})
\mapsto(t,\,\theta^1,\,\dots,\,\theta^{k-1},\,y^k,\,y^{k+1},\,\theta^{k+2},\,\dots,\,\theta^{m-1})\,,
\end{align*}
where $y^k=(\theta^k+2)\cos(\theta^{k+1})$ and $y^{k+1}=(\theta^k+2)\sin(\theta^{k+1})$. 
Per construction, the image of $h$ lies in the region of $O$ where $g_O$ has the same formal expression as $g_N$. 
Since the metric $g_M$ on $M$ is nothing but the restriction of the metric $g_N$ 
and one can check that $h$ preserves both the orientation and the time-orientation, 
$h$ turns out to be a causal embedding. 

The topology being the same, barring minor modifications, we can apply\ Example\ \ref{exaNonTrivRadForms} 
to $M$ to exhibit a non-trivial element $\ev_{[\omega]}\in\rad_M$. 
In fact, it is enough to take $a\in\cc(\bbR)$ to be supported inside the interval $I$, 
now interpreted as a subset of $\bbR$. Since $\hcdd^{k+1}(O)=\{0\}$, 
$\ev_{[\omega]}$ lies in the kernel of $\PSV(h)$, see\ Example \ref{exaNonInjForms} 
and\ Lemma\ \ref{lemLocCohoForms}. 

Note that $N$ has a compact Cauchy hypersurface, therefore, according to\ Proposition\ \ref{prpRadForms}, 
the null space $\rad_N$ of the presymplectic form on $N$ is trivial. 
This observation entails that either $\PSV(f)$ annihilates $\ev_{[\omega]}$ 
or $\PSV(f)\ev_{[\omega]}$ is an element of $\obs_N$ which lies outside of the null space $\rad_N$. 
In fact, the argument which shows that $\ev_{[\omega]}$ is non-trivial in $\obs_M$ 
implies that same is true for $\PSV(f)\ev_{[\omega]}$ in $\obs_N$. 
This is already enough to conclude that $\PSV(f)\ev_{[\omega]}$ lies outside of $\rad_N=\{0\}$. 

Let now $m\geq2$ and $k=m-1$. 
Consider the $m$-dimensional globally hyperbolic spacetime $N$ defined as follows: 
\begin{enumerate}
\item The underlying manifold is $\bbR\times\bbT^{m-2}\times(\bbT\sqcup\bbT)$, 
where $\sqcup$ denotes the disjoint union of manifolds. 
Points of this manifold are specified by $(m+1)$-tuples $(t,\theta^1,\dots,\theta^{m-2},j,\theta^{m-1})$, 
where $j\in\{0,1\}$ specifies to which disconnected component a given point belongs; 
\item On both components the metric is given by $g_N=-\dd t^2+\sum_{i=1}^{m-1}{\dd\theta^i}^2$; 
\item On both components the orientation is specified by $\dd t\wedge\bigwedge_{i=1}^{m-1}\dd\theta^i$, 
while the time-orientation is specified by the vector field $\partial_t$. 
\end{enumerate}

Furthermore, consider the $m$-dimensional Minkowski spacetime $O$: 
\begin{enumerate}
\item The underlying manifold is $\bbR\times\bbR^{m-1}$ 
and points of this manifold are specified by $m$-tuples $(t,x^1,\dots,x^{m-1})$; 
\item The metric is given by $g_O=-\dd t^2+\sum_{i=1}^{m-1}{\dd x^i}^2$; 
\item The orientation is specified by $\dd t\wedge\bigwedge_{i=1}^{m-1}\dd x^i$, 
while the time-orientation is specified by the vector field $\partial_t$. 
\end{enumerate}

Consider now the acausal hypersurface $\Sigma=\{0\}\times I^{m-2}\times (I\sqcup I)$ of $N$, 
where $I\subset\bbT$ is an arbitrary open interval of $(0,2\pi)$. 
Define the $m$-dimensional globally hyperbolic spacetime $M$ to be the Cauchy development of $\Sigma$ in $N$ 
endowed with the induced metric $g_M=g_N\vert_M$, orientation and time-orientation. 
As a byproduct of this procedure, we obtain a causal embedding $f:M\to N$. 

To embed $M$ in $O$ we proceed as follows. As in the previous case, 
we keep in mind the inclusion of $M$ into $N$ as sets and we adopt the coordinates for $N$ to specify points of $M$. 
Then we introduce the following embedding: 
\begin{align*}
h:M\to O && (t,\,\theta^1,\,\dots,\,\theta^{m-2}\,,j,\,\theta^{m-1})
\mapsto(t,\,\theta^1,\,\theta^2,\,\dots,\,\theta^{m-2},\,2\pi j+\theta^{m-1})\,.
\end{align*}
This is clearly a causal embedding of the $2$-component globally hyperbolic spacetime $M$ into $O$, 
which is performed in such a way that the image of $M$ is made of two causally disjoint regions of $O$. 

We can exploit Example \ref{exaNonInjForms} to exhibit a non-trivial element $\ev_{[\omega]}$ in $\rad_M$ 
which is annihilated by $\PSV(h)$. As before, by\ Proposition\ \ref{prpRadForms}, 
$\rad_N$ is trivial since its Cauchy hypersurface is compact. Therefore, 
if $\PSV(f)\ev_{[\omega]}\neq0$ in $\obs_N$, then $\PSV(f)\ev_{[\omega]}$ lies outside $\rad_N$. 
This is indeed the case for the very same reason which ensures that $\ev_{[\omega]}\neq0$ in $\obs_M$. 
In fact, one can also check that $\hcdd^m(f):\hcdd^m(M)\to\hcdd^m(N)$ is injective 
(actually this is an isomorphism) and then apply Lemma\ \ref{lemLocCohoForms}. 
\end{example}

\begin{theorem}\label{thmLocNotRecoveredForms}
Let $m\geq2$ and $k\in\{1,\dots,m-1\}$. The covariant functor $\PSV:\GHyp\to\PSymV$ 
has no quotientable subfunctor $\mathfrak{Q}:\GHyp\to\PSymV$ 
which recovers the locality property of\ \cite[Definition\ 2.1]{BFV03}, 
namely such that $\PSV/\mathfrak{Q}(f)$ is injective for each causal embedding $f$. 
\end{theorem}

\begin{proof}
Let $m\geq2$ and $k\in\{1,\dots,m-1\}$.
By contradiction, suppose $\mathfrak{Q}:\GHyp\to\PSymV$ is a quotientable subfunctor of $\PSV:\GHyp\to\PSymV$ 
such that locality is recovered taking the quotient by $\mathfrak{Q}$. 
According to\ Example\ \ref{exaNoQuotientForms} (consider the first part of the example for $m\geq3$, 
$k\in\{1,\dots,m-2\}$ and the second for $m\geq2$, $k=m-1$), we have a diagram in $\GHyp$ of the form 
\begin{equation*}
\xymatrix{
N && O\\
&M\ar[lu]^f\ar[ru]_h
}
\end{equation*}
and a non-trivial element $\ev_{[\omega]}$ of $\PSV(M)$ which is annihilated by $\PSV(h)$. 
Since the quotient by $\mathfrak{Q}$ has to recover injectivity, $\ev_{[\omega]}$ has to lie in $\mathfrak{Q}(M)$. 
However, by\ Definition\ \ref{defQSubfunctorV}, $\mathfrak{Q}(f)$ is the restriction of $\PSV(f)$ 
to $\mathfrak{Q}(M)$. In particular, $\PSV(f)\ev_{[\omega]}$ must lie in $\mathfrak{Q}(N)$. 
This contradicts the fact that $\mathfrak{Q}$ is quotientable, 
since we know from\ Example\ \ref{exaNoQuotientForms} that 
$\PSV(f)\ev_{[\omega]}$ does not lie in the null space of $\PSV(N)$. 
\end{proof}

\begin{remark}\label{remNoQuotientLocForms}
The last theorem shows the impossibility to recover locality by quotientable subfunctors of $\PSV$. 
In particular, it shows that, for $m\geq3$, $k\in\{1,\dots,m-2\}$, 
this is the case even in the full subcategory of $\GHyp$ whose objects are connected. 
Note that this includes the most prominent physical situation, namely the vector potential of free electromagnetism 
on four-dimensional globally hyperbolic spacetimes, $m=4$ and $k=1$. 
\end{remark}

\subsection{Recovering isotony \backtick{a} la Haag-Kastler}\label{subHKForms}
Even though the covariant functor $\PSV$ violates locality in the sense of\ \cite[Definition\ 2.1]{BFV03} 
and, moreover, there exists no quotientable subfunctor which allows to recover this property, 
it is still possible to recover isotony in the spirit of\ \cite{HK64}. 
This procedure is the analogue of the one in \cite[Section\ 6]{BDHS14}. 

The idea is to fix an $m$-dimensional globally hyperbolic spacetime $M$ (interpreted as the \quotes{universe}) 
and consider only its causally compatible open subsets (interpreted as regions of the universe $M$), 
so that a unique causal embedding into $M$ (induced by the inclusion) is specified for each region. 
This ultimately enables us to quotient out the kernels of all morphisms induced by inclusions, 
therefore preserving on each region only those observables 
which provide some information about Maxwell $k$-forms defined on the whole universe $M$ 
and not about those which exist as non-trivial field configurations only on a certain region. 

\begin{definition}
Let $m\geq2$ and consider an $m$-dimensional globally hyperbolic spacetime $M$. 
We define the category $\GHyp_M$ to be the subcategory of $\GHyp$ 
whose objects are causally compatible open subsets of $M$ (often referred to as {\em regions} of $M$) 
and whose morphisms are only those induced by inclusions, 
namely there exists a morphism $\iota_{O\,O^\prime}$ from a region $O$ to a region $O^\prime$ 
if and only if $O\subseteq O^\prime$, such morphism is unique and it is specified by the inclusion. 
\end{definition}

By definition there exists only one morphism from each region $O^\prime$ of $M$ into $M$, 
namely the one induced by the inclusion $O^\prime\subseteq M$. 
In the language of category theory, $M$ is the only terminal object of $\GHyp_M$. 
Furthermore, if a morphism from $O$ to $O^\prime$ is given, 
then we have the chain of inclusions $O\subseteq O^\prime\subseteq M$, 
which is pictorially represented by the following commutative diagram, where all arrows are simply inclusions: 
\begin{equation}\label{eqUniverse}
\xymatrix{
& M\\
O\ar[rr]_{\iota_{O\,O^\prime}}\ar[ru]^{\iota_{O\,M}} && O^\prime\ar[lu]_{\iota_{O^\prime\,M}}
}
\end{equation}
For $m\geq2$ and $k\in\{1,\dots,m-1\}$ fixed, we introduce the functor $\PSV:\GHyp_M\to\PSymV$ 
(still denoted by the same symbol with a slight abuse of notation) 
obtained by restriction of the covariant functor $\PSV:\GHyp\to\PSymV$ to the subcategory $\GHyp_M$. 
Our aim is to construct a new covariant functor out of $\PSV:\GHyp_M\to\PSymV$, 
which, besides causality and time-slice axiom, satisfies isotony as well, 
in the sense that induced morphisms are always injective. 
Due to the restriction from $\GHyp$ to $\GHyp_M$, now the class of morphisms is very much restricted. 
In place of the term \quotes{locality}, which, according to\ \cite[Definition\ 2.1]{BFV03}, 
was always referred to the requirement of injectivity for the presymplectic maps induced 
via $\PSV:\GHyp\to\PSymV$ by morphisms in $\GHyp$, we adopt here the term \quotes{isotony} 
for the weaker requirement of injectivity for presymplectic maps 
induced by morphisms in the subcategory $\GHyp_M$. 
This term refers to a similar property, which was originally formulated in\ \cite{HK64}. 

\begin{lemma}\label{lemKerSubfunctorForms}
Let $m\geq2$ and $k\in\{1,\dots,m-1\}$ and consider an $m$-dimensional globally hyperbolic spacetime $M$. 
The following assignment defines a quotientable subfunctor $\mfker_M:\GHyp_M\to\PSymV$ 
of the covariant functor $\PSV:\GHyp_M\to\PSymV$: 
To each object $O$ in $\GHyp_M$ assign the object $\mfker_M(O)=\ker(\PSV(\iota_{O\,M}))$ 
(endowed with the trivial presymplectic structure) in $\PSymV$ 
and to each morphism $\iota_{O\,O^\prime}:O\to O^\prime$ in $\GHyp_M$ 
assign the morphism $\mfker_M(\iota_{O\,O^\prime}):\mfker_M(O)\to\mfker_M(O^\prime)$ in $\PSymV$ 
obtained as the restriction of $\PSV(\iota_{O\,O^\prime}):\PSV(O)\to\PSV(O^\prime)$ to $\mfker_M(O)$. 
\end{lemma}

\begin{proof}
Given an object $O$ in $\GHyp_M$, its inclusion $\iota_{O\,M}$ 
in the terminal object $M$ of $\GHyp_M$ is given as well. 
Therefore it makes sense to consider the vector space $\ker(\PSV(\iota_{O\,M}))$. 
By\ Proposition\ \ref{prpSymInj}, this vector space is a subspace of $\rad_O$. 
Obviously the presymplectic form of $\PSV(O)$ vanishes when restricted to $\rad_O$. 
For this reason it is natural to define $\mfker_M(O)$
by endowing $\ker(\PSV(\iota_{O\,M}))$ with the trivial presymplectic structure. 
Furthermore, the fact that $\mfker_M(O)$ is a subspace of the null space of the presymplectic form of $\PSV(O)$ 
entails that, if $\mfker:\GHyp_M\to\PSymV$ were a subfunctor of the covariant functor $\PSV:\GHyp_M\to\PSymV$, 
then it would also be quotientable. 

Consider now a morphism $\iota_{O\,O^\prime}:O\to O^\prime$ in $\GHyp_M$. 
First of all, note that, $M$ being the terminal object in $\GHyp_M$, 
we have also the morphisms $\iota_{O\,M},\iota_{O^\prime\,M}$ in $\GHyp_M$ 
and these fit into the commutative diagram\ \eqref{eqUniverse} together with $\iota_{O\,O^\prime}$. 
Applying the covariant functor $\PSV:\GHyp_M\to\PSymV$, we obtain a new commutative diagram: 
\begin{equation*}
\xymatrix{
& \PSV(M)\\
\PSV(O)\ar[rr]_{\PSV(\iota_{O\,O^\prime})}\ar[ru]|{\PSV(\iota_{O\,M})} 
&& \PSV(O^\prime)\ar[lu]|{\PSV(\iota_{O^\prime\,M})}
}
\end{equation*}
From this diagram we deduce that $\PSV(\iota_{O\,O^\prime})$ 
maps the kernel of $\PSV(\iota_{O\,M})$ to the kernel of $\PSV(\iota_{O^\prime\,M})$. 
This means that restricting $\PSV(\iota_{O\,O^\prime})$ to the kernel of $\PSV(\iota_{O\,M}))$ 
provides a homomorphism $\mfker_M(\iota_{O\,O^\prime}):\mfker_M(O)\to\mfker(O^\prime)$, 
which preserves the relevant presymplectic structures since they are trivial. 
Therefore we get a morphism in $\PSymV$. 
In particular, since $\mfker_M:\GHyp_M\to\PSymV$ inherits its functorial behavior from $\PSV:\GHyp_M\to\PSymV$, 
it is a subfunctor of $\PSV:\GHyp_M\to\PSymV$ and therefore also quotientable by the first part of the proof. 
\end{proof}

\begin{theorem}\label{thmIsotonyForms}\index{isotony}
Let $m\geq2$ and $k\in\{1,\dots,m-1\}$. Consider an $m$-dimensional globally hyperbolic spacetime $M$. 
The covariant functor $\PSV_M:\GHyp_M\to\PSymV$ defined as the quotient of $\PSV:\GHyp_M\to\PSymV$ 
by its quotientable subfunctor $\mfker_M:\GHyp_M\to\PSymV$ ({\em cfr.}\ Lemma\ \ref{lemKerSubfunctorForms}) 
satisfies the {\em isotony} property, 
namely $\PSV_M$ assigns an injective morphism in $\PSymV$ to each morphism in $\GHyp_M$. 
Furthermore, both causality and the time-slice axiom hold for $\PSV_M:\GHyp_M\to\PSymV$. 
\end{theorem}

\begin{proof}
Proposition\ \ref{prpQuotientFunctor} and Lemma\ \ref{lemKerSubfunctorForms} ensure that 
$\PSV_M:\GHyp_M\to\PSymV$ is a well-defined covariant functor. 

For each morphism $\iota_{O\,M}$ in $\GHyp_M$, $\PSV_M(\iota_{O\,M})$ is injective per construction. 
In fact, $\PSV_M(O)$ is obtained from $\PSV(O)$ exactly by taking the quotient 
by $\mfker_M(O)=\ker(\PSV(\iota_{O\,M}))$. 
For a morphism $\iota_{O\,O^\prime}$ in $\GHyp_M$, we have a commutative diagram in $\PSymV$ 
induced from the one in\ \eqref{eqUniverse} by the functorial properties of $\PSV_M$: 
\begin{equation*}
\xymatrix{
& \PSV_M(M)=\PSV(M)\\
\PSV_M(O)\ar[rr]_{\PSV_M(\iota_{O\,O^\prime})}\ar[ru]|{\PSV_M(\iota_{O\,M})} 
&& \PSV_M(O^\prime)\ar[lu]|{\PSV_M(\iota_{O^\prime\,M})}
}
\end{equation*}
We already know that $\PSV_M(\iota_{O\,M})$ is injective. 
Commutativity of the diagram implies that $\PSV_M(\iota_{O\,O^\prime})$ is injective too, 
thus showing that the isotony property holds true. 

Both causality and the time-slice axiom are inherited 
from the corresponding properties of the functor $\PSV:\GHyp\to\PSymV$, 
{\em cfr.}\ Theorem\ \ref{thmCausalityForms} and Theorem\ \ref{thmTimeSliceForms}. 
\end{proof}

\begin{remark}\label{remIsotonyForms}
This procedure to recover isotony \backtick{a} la Haag-Kastler\ \cite{HK64} has a physical interpretation, 
see\ also\ \cite[Section\ 7]{BDHS14}. 
Let $m\geq2$ and $k\in\{1,\cdots,m-1\}$. For each region $O$ of a globally hyperbolic spacetime $M$, 
we take the quotient by exactly those observables in $\obs_O$ which would become trivial in $\obs_M$. 
After a quotient on a space of functionals is performed, 
the space upon which the elements of the resulting quotient can be evaluated is a subspace of the original one. 
In the present context, this can be thought of as a further restriction 
beyond the on-shell condition for Maxwell $k$-forms over the region $O$. 
This restriction is such that the outcome of a \quotes{measurement} in some region $O$ 
cannot provide some information which is not available in the whole universe $M$. 
In fact, after the quotient is performed, isotony tells us that the resulting space of observables $\PSV_M(O)$ 
associated to a region $O$ is nothing but a subspace of the space of observables $\PSV_M(M)=\PSV(M)$ 
of the whole universe $M$. 
\end{remark}

\section{Covariant quantum field theory}\label{secQFTForms}
We now quantize the functor $\PSV:\GHyp\to\PSymV$ describing observables for Maxwell $k$-forms 
on $m$-dimensional globally hyperbolic spacetimes for $m\geq2$ and $k\in\{1,\dots,m-1\}$, 
see\ Theorem\ \ref{thmFunctorForms}. 
Quantization is performed via the covariant functor $\CCR:\PSymA\to\CAlg$ 
introduced in\ Section\ \ref{secQuantization} restricted to the category $\PSymV$ of presymplectic vector spaces.  
This functor assigns a unital $C^\ast$-algebra implementing Weyl relations to each presymplectic vector space. 
Clearly, the resulting covariant functor $\QFT=\CCR\circ\PSV:\GHyp\to\CAlg$ 
inherits the causality property and the time-slice axiom from $\PSV:\GHyp\to\PSymV$, 
{\em cfr.}\ Theorem\ \ref{thmCausalityForms} and\ Theorem\ \ref{thmTimeSliceForms}. 
However, locality fails also for $\QFT:\GHyp\to\CAlg$ and it is not possible to recover it taking quotients. 
This follows from the analogous results for $\PSV:\GHyp\to\PSymV$, 
see\ Theorem\ \ref{thmNonLocForms} and\ Theorem\ \ref{thmLocNotRecoveredForms}. 
Finally, one can show that isotony in the sense of\ \cite{HK64} survives quantization. 
In fact, fixing an $m$-dimensional globally hyperbolic spacetime $M$, 
the covariant functor $\QFT_M=\CCR\circ\PSV_M:\GHyp_M\to\CAlg$, 
which is the quantized counterpart of the functor introduced in\ Subsection\ \ref{subHKForms}, 
inherits isotony, together with causality and the time-slice axiom, from $\PSV_M:\GHyp\to\PSymV$, 
see\ Theorem\ \ref{thmIsotonyForms}. This is due to the fact that $\CCR$ preserves injectivity of morphisms, 
{\em cfr.}\ Proposition\ \ref{prpQuantInj}. 

Notice that, for an $m$-dimensional globally hyperbolic spacetime $M$, 
the generator of $\QFT(M)$ corresponding to an element $\ev_{[\omega]}\in\PSV(M)$ 
will be denoted by $\weyl_{[\omega]}$. 

\begin{theorem}
Let $m\geq2$ and $k\in\{1,\,\dots,\,m-1\}$. Consider the covariant functor $\PSV:\GHyp\to\PSymV$ 
introduced in\ Theorem\ \ref{thmFunctorForms} and the covariant functor $\CCR:\PSymV\to\CAlg$ 
obtained by restricting the one in\ Theorem\ \ref{thmQuantFunctor} to presymplectic vector spaces. 
The covariant functor $\QFT=\CCR\circ\PSV:\GHyp\to\CAlg$ fulfils 
the quantum counterparts of both causality and the time-slice axiom, namely: 
\begin{description}
\item[Causality] If $f:M_1\to N$ and $h:M_2\to N$ are causal embeddings with causally disjoint images in $N$, 
namely such that $f(M_1)\cap J_N(h(M_2))=\emptyset$, then the $C^\ast$-subalgebras 
$\QFT(f)(\QFT(M_1))$ and $\QFT(h)(\QFT(M_2))$ of the $C^\ast$-algebra $\QFT(N)$ commute with each other; 
\item[Time-slice axiom] If $f:M\to N$ is a Cauchy morphism, then $\QFT(f):\QFT(M)\to\QFT(N)$ 
is an isomorphism of $C^\ast$-algebras. 
\end{description}
The covariant functor $\QFT:\GHyp\to\CAlg$ violates locality, 
namely there exists a causal embedding $f:M\to N$ such that $\QFT(f):\QFT(M)\to\QFT(N)$ is not injective. 
\end{theorem}

\begin{proof}
$\QFT:\GHyp\to\CAlg$ is a covariant functor since it is defined composing two covariant functors. 
Causality follows from the Weyl relations\ \eqref{eqWeylRel} and\ Theorem\ \ref{thmCausalityForms}. 
The time-slice axiom follows from $\CCR:\PSymV\to\CAlg$ being a functor and\ Theorem\ \ref{thmTimeSliceForms}. 
To exhibit a violation of locality, consider Theorem\ \ref{thmNonLocForms}. 
The violation at the classical level means that there exists a causal embedding $f:M\to N$ 
and an element $\ev_{[\omega]}\in\PSV(M)$ such that $\PSV(f)\ev_{[\omega]}=0$. 
Consider now the corresponding generator $\weyl_{[\omega]}$ of $\QFT(M)$ 
and the unit $\weyl_0=\bbone_M\in\QFT(M)$ and introduce $a=\bbone_M-\weyl_{[\omega]}\in\QFT(M)$. 
Since $\PSV(f)\ev_{[\omega]}=0$, it follows by definition of $\CCR$, {\em cfr.}\ \eqref{eqQuantMorph}, 
that $\QFT(f)a=\bbone_N-\weyl_0=0$, namely $\QFT(f)$ is not injective. 
\end{proof}

To state the next theorem about the impossibility to recover locality by quotients, 
one has to introduce first suitable notions of subfunctors and quotientable subfunctors 
for functors taking values in the category $\CAlg$ of unital $C^\ast$-algebras. 
We adopt the notions of\ \cite[Section\ 5]{BDHS14}. 

\begin{definition}\label{defQSubfunctorAlg}
\index{subfunctor}\index{quotientable subfunctor}\index{subfunctor!quotientable}
Let $\mathsf{C}$ be a category and consider a covariant functor $\mathfrak{F}:\mathsf{C}\to\CAlg$. 
A covariant functor $\mathfrak{S}:\mathsf{C}\to\CAlg$ is a {\em subfunctor} of $\mathfrak{F}$ 
if the following conditions are fulfilled: 
\begin{enumerate}
\item For each object $c$ of the category $C$, 
$\mathfrak{S}(c)$ is a $C^\ast$-subalgebra of $\mathfrak{F}(c)$ (non-unital subalgebras are allowed); 
\item For each morphism $\gamma:c\to d$ in the category $C$, 
$\mathfrak{S}(\gamma):\mathfrak{S}(c)\to\mathfrak{S}(d)$ is the restriction to $\mathfrak{S}(c)$ 
of $\mathfrak{F}(\gamma):\mathfrak{F}(c)\to\mathfrak{F}(d)$. 
\end{enumerate}
A {\em quotientable subfunctor} $\mathfrak{Q}:\mathsf{C}\to\CAlg$ of the covariant functor 
$\mathfrak{F}:\mathsf{C}\to\CAlg$ is a subfunctor such that, for each object $c$ of the category $C$, 
$\mathfrak{Q}(c)$ is a closed two-sided $\ast$-ideal of the $C^\ast$-algebra $\mathfrak{F}(c)$. 
Furthermore, a subfunctor $\mathfrak{Q}:\mathsf{C}\to\CAlg$ is {\em proper} 
if $\mathfrak{S}(c)$ is a proper $C^\ast$-subalgebra of $\mathfrak{F}(c)$ for each object $c$ of the category $C$. 
\end{definition}

The next proposition adapts to the $C^\ast$-algebraic setting 
the one for presymplectic vector spaces, {\em cfr.}\ Proposition\ \ref{prpQuotientFunctor}. 

\begin{proposition}\label{prpQuotientFunctorAlg}\index{quotient functor}
Let $\mathsf{C}$ be a category and consider a covariant functor $\mathfrak{F}:\mathsf{C}\to\CAlg$. 
Furthermore, let $\mathfrak{Q}:\mathsf{C}\to\CAlg$ be a proper quotientable subfunctor of $\mathfrak{F}$. 
To each object $c$ in $\mathsf{C}$ assign the object 
$\mathfrak{F}/\mathfrak{Q}(c)=\mathfrak{F}(c)/\mathfrak{Q}(c)$ in $\CAlg$. 
Furthermore, to each morphism $\gamma:c\to d$ in $\mathsf{C}$ assign the morphism 
$\mathfrak{F}/\mathfrak{Q}(\gamma):\mathfrak{F}/\mathfrak{Q}(c)\to\mathfrak{F}/\mathfrak{Q}(d)$ 
in $\CAlg$ induced by $\mathfrak{F}(\gamma):\mathfrak{F}(c)\to\mathfrak{F}(d)$. 
Then $\mathfrak{F}/\mathfrak{Q}:\mathsf{C}\to\CAlg$ is a covariant functor, 
called {\em quotient functor} of $\mathfrak{F}$ by $\mathfrak{Q}$. 
\end{proposition}

\begin{proof}
The proof is very similar to the one of\ Proposition\ \ref{prpQuotientFunctor}. 
In fact, one only has to replace presymplectic vector spaces with $C^\ast$-algebras 
and null spaces of presymplectic forms with closed two-sided $\ast$-ideals of $C^\ast$-algebras. 
\end{proof}

Below we only consider proper quotientable subfunctors. This is done in order to avoid the situation 
where, after performing the quotient, the resulting functor gives trivial $C^\ast$-algebras. 

\begin{theorem}
Let $m\geq2$ and $k\in\{1,\dots,m-1\}$. The covariant functor $\QFT:\GHyp\to\CAlg$ 
has no proper quotientable subfunctor $\mathfrak{Q}:\GHyp\to\CAlg$ 
which recovers locality in the sense of\ \cite[Definition\ 2.1]{BFV03}, 
namely such that $\QFT/\mathfrak{Q}(f)$ is injective for each causal embedding $f:M\to N$. 
\end{theorem}

\begin{proof}
Let $m\geq2$ and $k\in\{1,\dots,m-1\}$.
By contradiction, suppose $\mathfrak{Q}:\GHyp\to\CAlg$ is a quotientable subfunctor of $\QFT:\GHyp\to\CAlg$ 
such that locality is recovered taking the quotient by $\mathfrak{Q}$. 
According to\ Theorem\ \ref{thmLocNotRecoveredForms} (consider the first part of the example for $m\geq3$, 
$k\in\{1,\dots,m-2\}$ and the second for $m\geq2$, $k=m-1$), we have a diagram in $\GHyp$ of the form 
\begin{equation*}
\xymatrix{
N && O\\
&M\ar[lu]^f\ar[ru]_h
}
\end{equation*}
and a non-trivial element $\ev_{[\omega]}$ of $\PSV(M)$ which is annihilated by $\PSV(h)$, 
but such that $\PSV(f)\ev_{[\omega]}$ does not lie in the null space of $\PSV(N)$. 
It follows that $\bbone_M-\weyl_{[\omega]}\in\QFT(M)$ is annihilated by $\QFT(h)$. 
Therefore $\bbone_M-\weyl_{[\omega]}$ must lie in the $\ast$-ideal $\mathfrak{Q}(M)$. 
In fact, per hypothesis, taking the quotient recovers injectivity. 
On the other side, defining $\ev_{[\omega^\prime]}=\PSV(f)\ev_{[\omega]}\in\PSV(N)$, 
there exists $\ev_{[\theta]}\in\PSV(N)$ such that $\tau_N(\ev_{[\theta]},\ev_{[\omega^\prime]})\notin2\pi\bbZ$. 
Using Weyl relations\ \eqref{eqWeylRel}, one gets 
\begin{equation*}
\weyl_{[-\theta]}\,\big(\bbone_N-\weyl_{[\omega^\prime]}\big)\,\weyl_{[\theta]}
=\bbone_N-\myexp^{-i\tau_N(\ev_{[\omega^\prime]},\ev_{[\theta]})}\,\weyl_{[\omega^\prime]}\,,
\end{equation*}
whence 
\begin{equation*}
\bbone_N-\weyl_{[\omega^\prime]}-\myexp^{i\tau_N(\ev_{[\omega^\prime]},\ev_{[\theta]})}\,
\weyl_{[-\theta]}\,\big(\bbone_N-\weyl_{[\omega^\prime]}\big)\,\weyl_{[\theta]}
=\big(1-\myexp^{i\tau_N(\ev_{[\omega^\prime]},\ev_{[\theta]})}\,\big)\,\bbone_N
\end{equation*}
is a non-zero multiple of the identity, 
$\tau_N(\ev_{[\theta]},\ev_{[\omega^\prime]})$ not being an integer multiple of $2\pi$. 
Since we proved that $\bbone_M-\weyl_{[\omega]}$ lies in the $\ast$-ideal $\mathfrak{Q}(M)$ and 
$\mathfrak{Q}$ being a subfunctor, we deduce that $\bbone_N-\weyl_{[\omega^\prime]}$ lies $\mathfrak{Q}(N)$. 
However, $\mathfrak{Q}(N)$ is a $\ast$-ideal because $\mathfrak{Q}$ is quotientable by assumption. 
In particular, $\weyl_{[-\theta]}\,(\bbone_N-\weyl_{[\omega^\prime]})\,\weyl_{[\theta]}$ 
must lie in $\mathfrak{Q}(N)$ as well, 
and so does any linear combination with $\bbone_N-\weyl_{[\omega^\prime]}$, 
see in particular the left-hand-side of the equation displayed above. 
Yet, this is a non-zero multiple of the identity, therefore $\mathfrak{Q}(N)$ is a $\ast$-ideal of $\QFT(N)$
containing $\bbone_N$, whence it coincides with $\QFT(N)$. 
This contradicts the hypothesis that $\mathfrak{Q}(N)$ is a proper $\ast$-ideal of $\QFT(N)$. 
\end{proof}

Notice that, exactly as in the classical situation, see\ Remark\ \ref{remLocViolationForms} 
and\ Remark\ \ref{remNoQuotientLocForms}, in dimension $m\geq3$ and for any degree $k\in\{1,\dots,m-2\}$, 
for the covariant quantum field theory $\QFT:\GHyp\to\CAlg$ the violation of locality and the impossibility 
to recover it via quotients already occurs for $m$-dimensional connected globally hyperbolic spacetimes. 

It remains only to show that at least isotony in the sense of Haag and Kastler\ \cite{HK64} 
can be lifted to the quantum theory. 

\begin{theorem}\index{isotony}
Let $m\geq2$, $k\in\{1,\dots,m-1\}$ and take an $m$-dimensional globally hyperbolic spacetime $M$. 
Consider the covariant functor $\PSV_M:\GHyp_M\to\PSymV$ introduced in\ Theorem\ \ref{thmIsotonyForms}. 
Then the covariant functor $\QFT_M=\CCR\circ\PSV_M:\GHyp_M\to\CAlg$ satisfies the isotony property, 
namely $\QFT_M$ assigns an injective morphism in $\CAlg$ to each morphism in $\GHyp_M$. 
Furthermore, both causality and the time-slice axiom hold for $\QFT_M:\GHyp_M\to\CAlg$. 
\end{theorem}

\begin{proof}
According to\ Theorem\ \ref{thmIsotonyForms}, $\PSV_M$ assigns an injective morphism in $\PSymV$ 
to each morphism in $\GHyp_M$. Since $\CCR$ preserves injectivity, {\em cfr.}\ Proposition\ \ref{prpQuantInj}, 
we conclude that $\QFT_M:\GHyp_M\to\CAlg$ satisfies the isotony property. 
Both causality and the time-slice axiom are inherited from $\PSV_M:\GHyp_M\to\PSymV$ 
exactly as in the general case. 
\end{proof}

\paginavuota

\chapter{$U(1)$ Yang-Mills model}\label{chYangMills}\index{Yang-Mills model}
We present an approach for the quantization of the Yang-Mills model with $U(1)$ as structure group. 
The dynamical degrees of freedom of this model are connections on principal $U(1)$-bundles 
over globally hyperbolic spacetimes, which form an affine space modeled on an infinite dimensional vector space. 
The underlying gauge symmetry is induced by the action of principal bundle automorphisms 
covering the identity on the base manifold. 
Physically, the model describes free electromagnetism without external sources. 
As in the previous chapter, our aim is to implement general local covariance\ \cite{BFV03}. 
Also in this context we will encounter issues in relations to injectivity of morphisms at the level of observables, 
both in the classical and in the quantum case. In particular, we will adopt two different approaches. 
The first one, based on\ \cite{BDS14a}, considers gauge invariant affine smooth functionals 
on on-shell $U(1)$-connections. It turns out that in this situation the requirement of invariance 
under gauge transformations restricts too much the class of functionals, 
thus leading to the failure in distinguishing certain gauge-inequivalent field configurations. 
In particular, flat connections, which are tightly related to the Aharonov-Bohm effect, 
see\ {\em e.g.}\ \cite[Example 6.6.1]{MS00}, are not detected by gauge invariant affine functionals. 
This observation motivates our second attempt, based on\ \cite{BDHS14}. 
Instead of affine functionals, we consider their complex exponential. 
This fact significantly weakens the constraint imposed by gauge invariance. In particular, it turns out that now 
on-shell field configurations can be separated (up to gauge) by evaluation on this class of exponential functionals.  
In particular, also flat connections can be detected in this approach. 
We also show that these functionals are closely related to Wilson loops. 
Both approaches provide covariant functors which satisfy \cite[Definition\ 2.1]{BFV03} up to injectivity. 
The failure of injectivity turns out to be related to subfunctors describing certain cohomology groups. 
In the first case, injectivity can be recovered by a suitable quotient, 
which might be interpreted as enforcing the idea that no electric charge should be detectable 
and hence admissible field configurations should carry no electric flux. 
In the second case, similarly to Maxwell $k$-forms, {\em cfr.}\ Theorem\ \ref{thmLocNotRecoveredForms}, 
we prove that there exists no quotientable subfunctor which allows to recover injectivity. 
Once more, a procedure in the spirit of the Haag-Kastler axioms \cite{HK64} provides a setting 
which allows to recover isotony, {\em cfr.}\ Subsection\ \ref{subHKForms}. 

Note that, unless otherwise stated, in this chapter $G$ denotes the Lie group $U(1)$, 
while $\lieg$ stands for the associated Lie algebra $i\bbR$.

\section{Gauge symmetry and dynamics}\label{secGaugeDynamicsYM}
To start with, we specialize the discussion of\ Section\ \ref{secPrBun} about connections for principal $G$-bundles 
to the case in which the Lie group $G$ is the Abelian Lie group $U(1)$. In particular, we introduce 
the space of off-shell configurations relevant to the present model and we recall the notion of gauge equivalence. 
Afterwards, by means of the curvature map, see\ Definition\ \ref{defCurvMap}, we present the equation of motion. 

In this section $M$ is a fixed $m$-dimensional globally hyperbolic spacetime, 
and we consider a principal $G$-bundle $P$ on top of it. 
The affine space of connections $\lambda\in\sect(\conn(P))$ provides off-shell field configurations for the model 
we are presenting, see\ Section\ \ref{secConn} for the notion of a connection on a principal $G$-bundle. 
Note that $\sect(\conn(P))$ is an affine space modeled on the vector space $\f^1(M,\lieg)$ 
on account of the comment following Definition\ \ref{defConn} and of Proposition\ \ref{prpTrivAdBun}. 

We do not regard connections on $P$ as the physically relevant objects of the model, 
rather we consider equivalence classes of such connections 
with respect to the natural notion of gauge transformation which is encoded by the geometry of principal bundles, 
see\ Section\ \ref{secGaugeTr}. As explained in\ Remark\ \ref{remGaugeConnAb}, since $G$ is Abelian, 
gauge transformations on $P$ can be identified with smooth $G$-valued functions on $M$ 
and a gauge transformation $f\in\c(M,G)$ acts on a connection $\lambda\in\sect(\conn(P))$ 
by translation exploiting the affine structure of $\sect(\conn(P))$: 
\begin{equation}\label{eqGaugeConnAb}
\lambda\mapsto\lambda-f^\ast\mu\,,
\end{equation}
where $\mu\in\f^1(G,\lieg)$ is the Maurer-Cartan form on the Abelian Lie group $G$. 
Let us introduce the set $\gau_P$ of gauge shifts, 
which encompasses all possible shifts by translation induced by gauge transformations of the principal bundle $P$: 
\begin{equation}\label{eqGaugeShifts}
\gau_P=\left\{f^\ast\mu\,:\;f\in\c(M,G)\right\}\,.
\end{equation}
As we will see in\ Subsection \ref{subGaugeConnAb}, $\gau_P$ is only an Abelian group. 
We consider the orbit space $[\sect(\conn(P))]$ of $\sect(\conn(P))$ 
under translation by $\gau_P$ as the relevant object. This amounts to the following equivalence relation: 
For $\lambda,\lambda^\prime\in\sect(\conn(P))$, we set
\begin{align*}
\lambda\sim\lambda^\prime && \iff && \exists\,f\in\c(M,G)\,:\;\lambda-f^\ast\mu=\lambda^\prime\,.
\end{align*}
Note that in general the orbit space $[\sect(\conn(P))]$ is no longer an affine space. 
In fact, taking the orbits under translation by $\gau_P$, which is only an Abelian group, breaks the affine structure. 

To introduce the equation of motion, we exploit the curvature map, see\ Remark\ \ref{remCurvAb}. 
The curvature map $F:\sect(\conn(P))\to\f^2(M,\lieg)$ provides an affine differential operator 
in the sense of\ \cite[Definition\ 3.1]{BDS14b}, {\em cfr.}\ Remark\ \ref{remCurvAb}. 
In particular, we remind the reader that the linear part of $F$ is given by $F_V=-\dd:\f^1(M,\lieg)\to\f^2(M,\lieg)$. 
Furthermore, we can directly check that the curvature does not change under gauge transformations, 
meaning that $F$ descends to a map $[\sect(\conn(P))]\to\f^2(M,\lieg)$. 
In fact, $\mu$ is the Maurer-Cartan form for an Abelian Lie group, hence it is closed. 

We introduce a new affine differential operator composing the curvature map with the codifferential: 
\begin{align}\label{eqMaxwellOp}
\MW:\sect(\conn(P))\to\f^1(M,\lieg)\,, && \lambda\mapsto\de F(\lambda)\,.
\end{align}
$\MW$ is an affine differential operator of second order on account of the fact that it is obtained 
composing an affine differential operator with a linear differential operator, both being of first order. 
In fact, its linear part is given by $\MW_V=\de\circ F_V=-\de\dd:\f^1(M,\lieg)\to\f^1(M,\lieg)$. 
Let us mention that, up to a sign, $\MW_V$ is the differential operator 
which rules the dynamics for Maxwell $1$-forms, see\ Chapter\ \ref{chMaxwell}. 
The dynamics for the Yang-Mills model with structure group $G=U(1)$ is specified 
by the affine differential operator $\MW$: 
\begin{align}\label{eqOnShellConn}
\MW(\lambda)=0\,, && \lambda\in\sect(\conn(P))\,.
\end{align}

\begin{remark}[Existence of solutions]\label{remYMSolExist}
If the base space is a globally hyperbolic spacetime $M$, solutions of eq.\ \eqref{eqOnShellConn} always exist. 
Let $\lambda\in\sect(\conn(P))$ be a connection and consider the equation $\de\dd\omega=\MW(\lambda)$ 
for $\omega\in\f^1(M,\lieg)$. Since $\MW(\lambda)=\de F(\lambda)$ is coexact, a solution $\omega$ exists. 
In fact, considering a partition of unity $\{\chi_+,\chi_-\}$ on $M$ such that 
$\chi_+=1$ in a past compact region, while $\chi_-=1$ in a future compact one, 
one can check that $\omega=\de(G^+(\chi_+F(\lambda))+G^-(\chi_-F(\lambda)))\in\f^1(M,\lieg)$ 
solves the equation $\de\dd\omega=\MW(\lambda)$, 
where $G^\pm$ is the retarded/advanced Green operator for $\Box$ acting on $\lieg$-valued $1$-forms. 
From $\MW_V=-\de\dd$ it follows that $\lambda+\omega$ solves eq.\ \eqref{eqOnShellConn}. 
\end{remark}

In agreement with the equation of motion displayed above, we introduce the space of on-shell connections: 
\begin{equation}\label{eqSolYM}
\sol_P=\left\{\lambda\in\sect(\conn(P))\,:\;\MW(\lambda)=0\right\}\,.
\end{equation}
This is indeed an affine space modeled on the space of solutions of the equation $\de\dd\omega=0$ 
for $1$-forms over $M$, {\em cfr.}\ \eqref{eqSolForms} for $k=1$. 
In the next remark we comment upon certain naturality property of $\MW$ and the consequent functorial behaviour 
of $\sol_P$. In order to make sense of the following statements, let us define a suitable category 
of principal $G$-bundles over globally hyperbolic spacetimes. 

\begin{definition}\index{principal bundle!category}
Let $m\geq2$ and $G=U(1)$. The {\em category of principal $G$-bundles over globally hyperbolic spacetimes} 
$\PrBun_\GHyp$ has principal $G$-bundles over $m$-dimensional globally hyperbolic spacetimes as objects 
and principal bundle maps covering causal embeddings as arrows. 
\end{definition}

\begin{remark}[Naturality of $\MW$]\label{remMWNaturality}
According to\ Remark\ \ref{remCurvAb}, the curvature map $F:\sect(\conn(\cdot))\to\f^2((\cdot)_\base,\lieg)$ 
is a natural transformation between contravariant functors from $\PrBun_\GHyp$ to $\Aff$. 
Furthermore, according to\ Remark\ \ref{remdddeBoxNaturality}, 
$\de:\f^2((\cdot)_\base,\lieg)\to\f^1((\cdot)_\base,\lieg)$ 
is a natural transformation between contravariant functors from $\PrBun_\GHyp$ to $\Vec$. 
Regarding vector spaces as affine spaces modeled on themselves, we get a covariant functor 
from the category $\Vec$ of vector spaces to the category $\Aff$ of affine spaces, 
therefore $\MW=\de\circ F:\sect(\conn(\cdot))\to\f^1((\cdot)_\base,\lieg)$ 
is a natural transformation between contravariant functors from $\PrBun_\GHyp$ to $\Aff$. 
This means that $f:P\to Q$ induces the following commutative diagram in $\Aff$:
\begin{equation*}
\xymatrix@C=2.5em{
\sect(\conn(Q))\ar[d]_{\MW}\ar[r]^{\sect(\conn(f))} & \sect(\conn(P))\ar[d]^{\MW}\\
\f^1(N,\lieg)\ar[r]_{\ul{f}^\ast} & \f^1(M,\lieg)
}
\end{equation*}
This diagram shows that $\sect(\conn(f))$ maps $\sol_Q$ to $\sol_P$. 
Therefore we can regard the assignment of the affine space $\sol_P$ to each object $P$ in $\PrBun_\GHyp$ 
and of the affine map $\sect(\conn(f)):\sol_Q\to\sol_P$ as a contravariant functor 
from $\PrBun_\GHyp$ to the category $\Aff$ of affine spaces. 
\end{remark}

As already mentioned, only gauge equivalence classes of connections are relevant 
and not their individual representatives. 
Since gauge transformations do not affect the dynamics (in fact, they do not even change the curvature), 
we can introduce the space of gauge equivalence classes of on-shell connections on $P$: 
\begin{equation*}
[\sol_P]=\sol_P/\gau_P.
\end{equation*}
As $[\sect(\conn(P))]$, $[\sol_P]$ is not an affine space due to the equivalence relation induced by $\gau_P$. 

\subsection{How gauge transformations act on connections}\label{subGaugeConnAb}
As before, $P$ is a fixed principal bundle with the Lie group $G=U(1)$ as structure group. 
Its base space is a globally hyperbolic spacetime $M$ of dimension $m$. 
Before we proceed with the construction of functionals defined on the space $[\sol_P]$ 
of on-shell connections up to gauge, it is convenient to shed light 
on the action of gauge transformations on connections as in\ \cite[Section\ 4]{BDS14a}. 

Let $\mu\in\f^1(G,\lieg)$ denote the Maurer-Cartan form on an Abelian Lie group $G$ 
and recall that $\mu$ is closed, whence, for each $f\in\c(M,G)$, $f^\ast\mu\in\f^1(M,\lieg)$ is closed too. 
Therefore we can introduce the following homomorphism of Abelian groups: 
\begin{align}\label{eqConnGaugeHom}
T:\c(M,G)\to\f^1_\dd(M,\lieg)\,, && f\mapsto f^\ast\mu\,,
\end{align}
where $\c(M,G)$ is endowed with the structure of an Abelian group specified by pointwise multiplication in $G$, 
while $\f^1_\dd(M,\lieg)$ is regarded as an Abelian group under addition. 
Eq.\ \eqref{eqGaugeConnAb} means that gauge transformations act on connections by translation 
with $1$-forms in the image $T(\c(M,G))\subseteq\f^1_\dd(M,\lieg)$ of the homomorphism $T$. 

\begin{remark}
To check that the map in\ \eqref{eqConnGaugeHom} is a group homomorphism, 
consider $f,g\in\c(M,G)$ and evaluate $(fg)^\ast\mu$ on a tangent vector $v\in TM$ over $x\in M$. 
Denoting with $\Delta:G\times G\to G$ the group multiplication, we can rewrite $fg$ 
as a composition $\Delta\circ(f,g)$. This allows us to compute $(fg)_\ast v$: 
\begin{equation*}
(fg)_\ast v=l_{f(x)\,\ast}g_\ast v+r_{g(x)\,\ast}f_\ast v\,,
\end{equation*}
where $l_h=\Delta(h,\cdot):G\to G$ is the left-action of $h\in G$ on $G$, 
while $r_h=\Delta(\cdot,h):G\to G$ is the right-action of $h\in G$ on $G$. 
Since $G$ is Abelian, $l_h=r_h$ for each $h\in G$. 
Therefore, by definition of the Maurer-Cartan form, we conclude that
\begin{equation*}
(fg)^\ast\mu(v)=\mu(l_{f(x)\,\ast}g_\ast v)+\mu(r_{g(x)\,\ast}f_\ast v)=f^\ast\mu(v)+g^\ast\mu(v)\,,
\end{equation*}
which proves that $T$ is a group homomorphism. 
\end{remark}

The next lemma characterizes the kernel of the homomorphism $T$ introduced in\ \eqref{eqConnGaugeHom}, 
hence it provides full information about its image, 
which is indeed isomorphic to the quotient of the domain $\c(M,G)$ of $T$ by its kernel. 

\begin{lemma}\label{lemGaugeConnKer}
Let $G$ be an Abelian Lie group and consider the group homomorphism $T:\c(M,G)\to\f^1_\dd(M,\lieg)$ 
defined in\ \eqref{eqConnGaugeHom}. The kernel of $T$ consists of locally constant $G$-valued functions on $M$: 
\begin{equation*}
\ker(T)=\left\{f\in\c(M,G)\,:\;f\mbox{ is locally constant}\right\}\,.
\end{equation*}
\end{lemma}

\begin{proof}
It is clear that, if $f$ is locally constant, $Tf=f^\ast\mu$ vanishes since $f_\ast$ is zero on each fiber of $TM$. 
For the converse, consider $f\in\ker(T)$. Therefore, for each $x\in M$ and $v\in T_xM$, $\mu(f_\ast v)=0$. 
Since the Maurer-Cartan form provides an isomorphism between the vector spaces $T_gG$ and $\lieg$ 
for each $g\in G$, we deduce that $f_\ast v=0$ for each $x\in M$ and $v\in T_xM$, whence $f$ is locally constant. 
\end{proof}

Let $\exp:\lieg\to G$ denote the exponential map and observe that, since $G$ is Abelian, 
this is a group homomorphism, where $\lieg$ is regarded as an Abelian group under addition. 

\begin{lemma}\label{lemExpConAbSurj}
Let $G$ be a connected Abelian Lie group. The group homomorphism $\exp:\lieg\to G$ 
defined by the exponential map is surjective. 
\end{lemma}

\begin{proof}
Let $g\in G$ be arbitrary but fixed. $\exp$ is a local diffeomorphism onto an open neighborhood $U\subseteq G$ 
of the identity. Since $G$ is connected, $U$ generates the whole $G$. Therefore we find $\xi_1,\dots,\xi_n\in\lieg$ 
such that $\exp(\xi_1)\cdots\exp(\xi_n)=g$. Since $\exp$ is a homomorphism for $G$ Abelian, 
we conclude that $\xi=\xi_1+\dots+\xi_n\in\lieg$ is such that $\exp(\xi)=g$. 
\end{proof}

We can exhibit a large class of gauge transformations 
by composition of an arbitrary $\lieg$-valued smooth function on $M$ with $\exp$. 
Gauge transformations of this form shift connections by exact $1$-forms. 
In fact, for $\chi\in\c(M,\lieg)$ and $f=\exp(\chi)\in\c(M,G)$, 
by the properties of the exponential map of an Abelian Lie group, we deduce that $f^\ast\mu=\dd\chi$. 
This means that $T$ descends to a group homomorphism 
\begin{equation}\label{eqConnGaugeHomQuot}
\tilde{T}:\frac{\c(M,G)}{\exp(\c(M,\lieg))}\to\hdd^1(M,\lieg)\,.
\end{equation}
$\tilde{T}$ gives us information about the obstruction to reproduce all gauge transformations $\c(M,G)$ 
via $\exp$ acting on $\c(M,\lieg)$. Since we know that this kind of exponential gauge transformations 
shift connections by exact forms, this part of the gauge group is closely related 
to the one considered for Maxwell $1$-forms, {\em cfr.}\ Section\ \ref{secGaugeDynForms}. 
In the following lemma we will show that the obstructions to having only gauge transformations 
of exponential type are captured by a subgroup of $\hdd^1(M,\lieg)$. 

\begin{lemma}\label{lemGaugeConnInj}
Let $G$ be a connected Abelian Lie group. The group homomorphism 
$\tilde{T}:\c(M,G)/\exp(\c(M,\lieg))\to\hdd^1(M,\lieg)$ defined in\ \eqref{eqConnGaugeHomQuot} is injective. 
\end{lemma}

\begin{proof}
Let $f\in\c(M,G)$ be such that $f^\ast\mu=\dd\chi$ for some $\chi\in\c(M,\lieg)$. 
Defining $g=\exp(-\chi)\in\c(M,G)$, we have $T(fg)=T(f)+T(g)=0$ since $g^\ast\mu=-\dd\chi$. 
From\ Lemma\ \ref{lemGaugeConnKer} we deduce that $fg$ is locally constant. 
Exploiting connectedness of $G$, we can apply Lemma\ \ref{lemExpConAbSurj} to conclude that 
there exists a locally constant $\lieg$-valued function $\tilde{\chi}\in\c(M,\lieg)$ 
such that $fg=\exp(\tilde{\chi})$, whence $f=\exp(\tilde\chi+\chi)$. 
\end{proof}

The above discussion does not depend on the specific choice of a connected Abelian Lie group. 
Yet, we return now to the case of interest to us, namely $G=U(1)$. 
We will employ techniques from sheaf cohomology to characterize $\gau_P$ 
in terms of tractable mathematical objects. For sheaf theory, we refer the reader to the literature, 
{\em e.g.}\ \cite{Har11}. For an Abelian Lie group $H$, consider the sheaf $\c_M(\cdot,H)$ 
of smooth $H$-valued functions on a manifold $M$ (we will consider $H$ to be $G=U(1)$, $\lieg=i\bbR$ 
and $\bbZ$ regarded as a zero-dimensional Lie group). 
Recalling that a sequence of sheaves is exact if and only if, at each point of $M$, 
the corresponding sequence of stalks is exact, one can check that the short exact sequence of Abelian groups 
\begin{equation*}
0\lra\bbZ\overset{2\pi i}{\lra}\lieg=i\bbR\overset{\exp}{\lra}G=U(1)\lra0\,,
\end{equation*}
where $2\pi i:\bbZ\to\lieg$ denotes the homomorphism $z\in\bbZ\mapsto2\pi iz\in\lieg$, 
gives rise to a corresponding short exact sequence of sheaves: 
\begin{equation*}
0\lra\c_M(\cdot,\bbZ)\overset{2\pi i}{\lra}\c_M(\cdot,\lieg)\overset{\exp}{\lra}\c_M(\cdot,G)\lra0\,.
\end{equation*}
Exactness of the sequence of sheaves does not ensure that the homomorphism 
$\exp:\c(M,\lieg)\to\c(M,G)$ is surjective on global sections of the sheaf $\c_M(\cdot,\lieg)$. 
In fact, the obstruction to surjectivity is described by the long exact sequence in sheaf cohomology displayed below: 
\begin{equation*}
\xymatrix@C=1.9pc@R=.6pc{
0\ar[r] & \c(M,\bbZ)\ar[r] & \c(M,\lieg)\ar[r] & \c(M,G)\ar`r[r]`/9pt[d]`^d[lll]`^r/9pt[ddlll] [ddll] & \\
& & & & \\
& \coho^1(M,\c_M(\cdot,\bbZ))\ar[r] & \coho^1(M,\c_M(\cdot,\lieg))\ar[r] 
& \coho^1(M,\c_M(\cdot,G))\ar[r] & \cdots
}
\end{equation*}
$\c_M(\cdot,\lieg)$ is a soft sheaf because there exists an extension to $M$ 
for each real valued function defined on a closed subset of $M$, {\em cfr.}\ \cite[Example\ 9.4, Chapter\ II]{Bre97}. 
As a consequence of this fact, all cohomology groups $\coho^k(M,\c_M(\cdot,\lieg))$ vanish for $k\geq1$ 
and we get the following exact sequence out of the long one displayed above: 
\begin{equation*}
0\lra\c(M,\bbZ)\overset{\subseteq}{\lra}\c(M,\lieg)\overset{\exp}{\lra}\c(M,G)
\lra\coho^1(M,\c_M(\cdot,\bbZ))\lra0\,.
\end{equation*}
The next proposition follows from exactness of this sequence and from\ Lemma \ref{lemGaugeConnInj}. 

\begin{proposition}\label{prpGaugeShiftsSheafCoho}
Let $M$ be a manifold, $G$ the Abelian group $U(1)$ and $\lieg=i\bbR$ its Lie algebra. 
The image of the homomorphism 
\begin{equation*}
\tilde{T}:\c(M,G)/\exp(\c(M,\lieg))\to\hdd^1(M,\lieg)\,,
\end{equation*} 
defined in\ \eqref{eqConnGaugeHomQuot}, is isomorphic to the first cohomology group 
$\coho^1(M,\c_M(\cdot,\bbZ))$ of the sheaf $\c_M(\cdot,\bbZ)$ of 
locally constant $\bbZ$-valued functions on $M$.\footnote{This is the same as the sheaf of 
smooth $\bbZ$-valued functions on $M$ since $\bbZ$ is endowed with the discrete topology 
to be regarded as a $0$-dimensional Lie group.} 
\end{proposition}

Since $M$ is a manifold, $\coho^1(M,\c_M(\cdot,\bbZ))$ is isomorphic 
both to the first singular cohomology group $\coho^1(M,\bbZ)$ with integer coefficients 
and to the first {\v C}ech cohomology group $\cech^1(M,\bbZ)$ with integer coefficients. 

Exploiting the universal coefficient theorem for cohomology, {\em cfr.}\ \cite[Corollary\ 5.3.2]{Mun96}, 
applied to the divisible group $\lieg=i\bbR$, we get an isomorphism 
$\coho^k(M,\lieg)\simeq\hom(\coho_k(M,\bbZ),\lieg)$ between the $k$-th singular cohomology group 
with $\lieg$-coefficients and the group of $\lieg$-valued homomorphisms on the $k$-th homology group 
with integer coefficients. The same theorem for the coefficient group $\bbZ$ tells us that 
$\coho^1(M,\bbZ)\simeq\hom(\coho_1(M,\bbZ),\bbZ)$ 
since the zeroth-homology group $\coho_0(M,\bbZ)$ is free Abelian. 
Therefore, the injective homomorphism $2\pi i:\bbZ\to\lieg$ induces an injective homomorphism 
$2\pi i:\coho^1(M,\bbZ)\to\coho^1(M,\lieg)$ in cohomology. 
Via de Rham's theorem, see\ {\em e.g.}\ \cite[Appendix\ A]{Mas91}, 
the $k$-th singular cohomology group $\coho^k(M,\lieg)$ is isomorphic 
to the $k$-th de Rham cohomology group $\hdd^k(M,\lieg)$. 
Therefore, we can consider the image $\hdd^1(M,\lieg)_\bbZ$ of $\coho^1(M,\bbZ)$ in $\hdd^1(M,\lieg)$. 
Note that $\hdd^1(M,\lieg)_\bbZ$ is only an Abelian subgroup of $\hdd^1(M,\lieg)$. 

\begin{corollary}\label{corGaugeShiftsExplicit}
Let $G=U(1)$ and $\lieg=i\bbR$. Consider a principal $G$-bundle $P$ over $M$. 
One can characterize the image of the homomorphism $T:\c(M,G)\to\fdd^1(M,\lieg)$ 
defined in\ \eqref{eqConnGaugeHom}, which by definition coincides with the Abelian group $\gau_P$, 
{\em cfr.}\ \eqref{eqGaugeShifts}: 
\begin{equation*}
\gau_P=\left\{\eta\in\fdd^1(M,\lieg)\,:[\eta]\in\hdd^1(M,\lieg)_\bbZ\right\}\,,
\end{equation*}
where $\hdd^1(M,\lieg)_\bbZ$ denotes the image of the injective $\bbZ$-module homomorphism 
$2\pi i:\coho^1(M,\bbZ)\to\coho^1(M,\lieg)\simeq\hdd^1(M,\lieg)$. 
\end{corollary}

In conclusion, the effect of the gauge group on connections consists in shifting by closed $1$-forms 
$\eta\in\fdd^1(M,\lieg)$ whose de Rham cohomology class $[\eta]\in\hdd^1(M,\lieg)$ 
is, up to a factor $2\pi i$, the image of a singular cohomology class in $\coho^1(M,\bbZ)$. 
This means that the value of the integral of $\eta$ along singular $1$-cycles with $\bbZ$-coefficients 
is an integer multiple of $2\pi i$. More formally, the $\bbZ$-module homomorphism displayed below 
\begin{align*}
\int_\cdot\,[\eta]:\coho_1(M,\bbZ)\to\lieg\,, && [\sigma]\mapsto\int_{[\sigma]}[\eta]\,,
\end{align*}
takes values in $2\pi i\bbZ$. 

\begin{remark}[$\hdd^1(M,\lieg)_\bbZ$ generates $\hdd^1(M,\lieg)$]\label{remZCohoGenRCoho}
As already mentioned, the zeroth-homology of $M$ is trivial, 
therefore $\coho^1(M,\bbZ)\simeq\hom(\coho_1(M,\bbZ),\bbZ)$ via the universal coefficient theorem. 
Denoting with $T$ the torsion part of $\coho_1(M,\bbZ)$ and with $F$ its torsion-free part, 
we get a short exact sequence 
\begin{equation*}
0\lra T\lra\coho_1(M,\bbZ)\lra F\lra0\,. 
\end{equation*}
Applying the contravariant functor $\hom(\cdot,\bbZ)$, we get the exact sequence 
\begin{equation*}
0\lra\hom(F,\bbZ)\lra\hom(\coho_1(M,\bbZ),\bbZ)\lra\hom(T,\bbZ)\,.
\end{equation*}
where $\hom(T,\bbZ)$ is trivial, $T$ being a torsion group. Therefore we have the following chain of isomorphisms: 
\begin{equation*}
\coho^1(M,\bbZ)\simeq\hom(\coho_1(M,\bbZ),\bbZ)\simeq\hom(F,\bbZ)\,.
\end{equation*}
Recall that $\lieg=i\bbR$. Along the same lines, we get $\coho^1(M,\lieg)\simeq\hom(F,\lieg)$. 
If $F$ is finitely generated, than it is also free Abelian, whence 
\begin{equation*}
\coho^1(M,\lieg)\simeq\hom(F,\lieg)\simeq\hom(F,\bbZ)\otimes_\bbZ\lieg
\simeq\coho^1(M,\bbZ)\otimes_\bbZ\lieg\,.
\end{equation*}
We conclude that, if $\coho_1(M,\bbZ)$ is finitely generated ({\em e.g.}\ when $M$ is of finite type), 
then $\hdd^1(M,\lieg)_\bbZ$ is an Abelian subgroup generating $\hdd^1(M,\lieg)$ over the field $\bbR$. 
In particular, a $\bbZ$-module basis of $\hdd^1(M,\lieg)_\bbZ$ 
is also an $\bbR$-module basis of $\hdd^1(M,\lieg)$. 
\end{remark}

\section{Observables via affine functionals}\label{secAffObsYM}
As already mentioned, we will denote with $G$ the Lie Abelian group $U(1)$. 
$\lieg$ will denote the corresponding Lie algebra $i\bbR$. 

In this section we will adopt the perspective of\ \cite{BDS14a}. 
Let us consider a principal $G$-bundle $P$ over an $m$-dimensional globally hyperbolic spacetime $M$. 
At this stage, all the tools needed to introduce suitable functionals 
on the space $[\sol_P]$ of on-shell connections up to gauge are available. 

To start with, we recall that the bundle of connections $\conn(P)$ on $P$ is an affine bundle over $M$ 
modeled on the vector bundle $\hom(TM,M\times\lieg)\simeq T^\ast M\otimes\lieg$, 
{\em cfr.}\ Subsection\ \ref{subConnBun} and Proposition\ \ref{prpTrivAdBun}. 
Therefore, the space $\sect(\conn(P))$ of connections is an affine space 
modeled on the vector space $\f^1(M,\lieg)$, 
{\em cfr.}\ Proposition\ \ref{prpTrivAdBun} and the comment after Definition\ \ref{defConn}. 
In particular, we can apply to $\sect(\conn(P))$ the techniques developed in\ Section\ \ref{secAffine}. 
For the purpose of constructing well-behaved functionals on the affine space $\sect(\conn(P))$, 
we exploit the pairing with compactly supported sections of the vector dual bundle $\conn(P)^\dagger$, 
see\ \eqref{eqAffPairing} for the more general situation of an arbitrary affine bundle. 
In this spirit, to each compactly supported section $\phi\in\sc(\conn(P)^\dagger)$ we associate the affine functional 
\begin{align}\label{eqAffFuncYM}
\ev_\phi:\sect(\conn(P))\to\bbR\,, && \lambda\mapsto\int_M\phi(\lambda)\,\vol\,,
\end{align}
where $\vol=\ast1$ is the canonical volume form defined on $M$. 
We collect all these affine functionals in the space of kinematic functionals: 
\begin{equation*}
\kin_P=\left\{\ev_\phi\in(\sect(\conn(P)))^\dagger\,:\;\phi\in\sc(\conn(P)^\dagger)\right\}\,.
\end{equation*}
According to\ Proposition\ \ref{prpAffSeparability}, $\sc(\conn(P)^\dagger)$ is not a good labeling space 
for functionals on $\sect(\conn(P)$ of the type specified above. 
In fact, we also have a characterization of those functionals which vanish systematically 
(later also referred to as being \quotes{trivial}). Adopting the strategy of Remark\ \ref{remAffTriv}, 
we introduce the vector subspace $\triv_P$ of $\sc(\conn(P)^\dagger)$ encompassing all trivial functionals: 
\begin{equation}\label{eqTrivYM}
\triv_P=\left\{a\,\bfone\in\sc(\conn(P)^\dagger)\,:\;a\in\cc(M)\,,\;\int_Ma\,\vol=0\right\}\,,
\end{equation}
where $\bfone\in\sect(\conn(P)^\dagger)$ is the section defined by the condition that 
its value at each point $x\in M$ is the affine map $a\in\conn(P)_x\mapsto1\in\bbR$. 
From\ Proposition\ \ref{prpAffSeparability} and Remark\ \ref{remAffTriv}, 
it is clear that $\sc(\conn(P)^\dagger)/\triv_P\simeq\kin_P$, 
meaning that the quotient of $\sc(\conn(P)^\dagger)$ by $\triv_P$ provides a labeling space for $\kin_P$. 
In fact, the isomorphism is induced by the map $\phi\in\sc(\conn(P)^\dagger)\mapsto\ev_\phi\in\kin_P$. 
To avoid a too heavy notation, we will denote with $\phi$ elements of $\sc(\conn(P)^\dagger)$, 
as well as elements of its quotient by $\triv_P$ 
and we identify $\kin_P$ with the quotient $\sc(\conn(P)^\dagger)/\triv_P$. 

Since the physically relevant object is not the connection itself, 
but rather its gauge equivalence class, not all the kinetic functionals in $\kin_P$ can be regarded 
as being suitable to extract physical information from field configurations. 
We have to restrict ourselves to gauge invariant functionals: 
\begin{equation*}
\inv_P=\left\{\ev_\phi\in\kin_P\,:\;\ev_\phi(\lambda+\gau_P)
=\{\ev_\phi(\lambda)\}\,,\;\forall\,\lambda\in\sect(\conn(P))\right\}\,,
\end{equation*}
where $\gau_P$ denotes the Abelian group of gauge shifts introduced in\ \eqref{eqGaugeShifts}, 
whose structure has been thoroughly analyzed in\ Subsection\ \ref{subGaugeConnAb}. 

The next lemma characterizes explicitly the space $\inv_P$ of gauge invariant functionals 
whenever $M$ is of finite type. Note that we will adopt the following notation: 
The linear part $\phi_V$ of a section $\phi\in\sect(\conn(P))$ should be a section of 
the dual of the vector bundle $T^\ast M\otimes\lieg$. Via the pairing defined in\ \eqref{eqContraction}, 
we can identify the dual of $T^\ast M\otimes\lieg$ with $T^\ast M\otimes\lieg^\ast$. 
After this identification, $\phi_V$ is regarded as a $\lieg^\ast$-valued $1$-form on $M$. 
Accordingly, the linear part of an affine functional $\ev_\phi\in\kin_P$ 
is specified by the linear map $(\phi_V,\cdot):\f^1(M,\lieg)\to\bbR$, 
where $(\cdot,\cdot):\fc^k(M,\lieg^\ast)\times\f^k(M,\lieg)\to\bbR$ is the non-degenerate pairing 
defined in\ \eqref{eqPairing2}, except for the irrelevant evaluation of $\lieg^\ast$ on $\lieg$. 

\begin{proposition}\label{prpInvAffIsCoexact}
Let $M$ be an $m$-dimensional globally hyperbolic spacetime and consider a principal $G$-bundle $P$ over $M$. 
Each gauge invariant affine functional $\ev_\phi\in\inv_P$ is such that $	\de\phi_V=0$, 
while each $\ev_\phi\in\kin_P$ with $\phi_V\in\de\fc^2(M,\lieg^\ast)$ lies in $\inv_P$, 
namely it is gauge invariant. 

Furthermore, if $M$ is of finite type, gauge invariant affine functionals are those kinetic functionals 
whose linear part is a $\lieg^\ast$-valued $1$-form with compact support, 
which is coexact in the sense of de Rham cohomology with compact support: 
\begin{equation*}
\inv_P=\left\{\ev_\phi\in\kin_P\,:\;\phi_V\in\de\fc^2(M,\lieg^\ast)\right\}\,.
\end{equation*}
\end{proposition}

\begin{proof}
For the inclusion of the right-hand-side in the left-hand-side of the equation displayed above, 
assume $\phi\in\sc(\conn(P)^\dagger)$ has linear part $\phi_V\in\de\fc^2(M,\lieg^\ast)$. 
Given $\lambda\in\sect(\conn(P))$ and taking into account that the elements of $\gau_P$ are 
closed $\lieg^\ast$-valued $1$-forms, see the comment before\ \eqref{eqConnGaugeHom}, we have 
\begin{equation*}
\ev_\phi(\lambda+\gau_P)=\ev_\phi(\lambda)+(\phi_V,\gau_P)=\{\ev_\phi(\lambda)\}\,,
\end{equation*}
where $(\cdot,\cdot)$ is the pairing defined in\ \eqref{eqPairing2}. 
Note that in the second equality we exploited Stokes' theorem 
after having expressed $\phi_V$ as $\de\omega$ for a suitable $\omega\in\fc^2(M,\lieg^\ast)$. 
This also shows that gauge invariance follows from the fact that the linear part is coexact. 

For the converse inclusion, suppose $\phi\in\sc(\conn(P)^\dagger))$ is such that $\ev_\phi\in\inv_P$ 
and denote with $\phi_V\in\fc^1(M,\lieg^\ast)$ its linear part. 
For each $\lambda\in\sect(\conn(P))$, we have the chain of identities below: 
\begin{align*}
\ev_\phi(\lambda)+(\phi_V,\gau_P)=\ev_\phi(\lambda+\gau_P)=\{\ev_\phi(\lambda)\}\,.
\end{align*}
If we consider only exponential gauge transformations, see the comment before\ \eqref{eqConnGaugeHomQuot}, 
we encompass the Abelian subgroup $\dd\c(M,\lieg)$ of $\gau_P$. 
Via this argument we deduce from the equation above that $(\phi_V,\dd\c(M,\lieg))=\{0\}$, whence $\de\phi_V=0$. 
This already shows that gauge invariance entails for $\ev_\phi$ that its linear part $\phi_V$ is coclosed. 
We show that one can improve this result assuming that $M$ is of finite type. 
Exploiting the characterization of $\gau_P$ presented in\ Corollary\ \ref{corGaugeShiftsExplicit}, 
we deduce that ${}_\de([\phi_V],\hdd^1(M,\lieg)_\bbZ)=\{0\}$, where $[\phi_V]\in\hcde^1(M,\lieg^\ast)$ 
and ${}_\de(\cdot,\cdot):\hcde^1(M,\lieg^\ast),\hdd^1(M,\lieg))\to\bbR$ is the pairing between cohomology 
classes defined in\ \eqref{eqPoincarePairing} (up to the irrelevant evaluation of $\lieg^\ast$ on $\lieg$). 
Recall that, $M$ being of finite type and according to\ Remark\ \ref{remZCohoGenRCoho}, 
the Abelian subgroup $\hdd^1(M,\lieg)_\bbZ$ generates $\hdd^1(M,\lieg)$ over the field $\bbR$. 
This fact, together with linearity of the pairing ${}_\de(\cdot,\cdot)$, entails that 
${}_\de([\phi_V],\hdd^1(M,\lieg))=\{0\}$. Using again the hypothesis that $M$ is of finite type, 
we can exploit Theorem\ \ref{thmPoincareDuality} to conclude that 
$\phi_V$ lies in $\de\fc^2(M,\lieg^\ast)$, thus completing the proof. 
\end{proof}

The last lemma has a drawback, which we present in the forthcoming theorem. 
Recall that $F:\sect(\conn(P))\to\f^2(M,\lieg)$ denotes the curvature map 
introduced in\ Definition\ \ref{defCurvMap} 
and later specialized to the case of an Abelian structure group in\ Remark\ \ref{remCurvAb}. 

\begin{theorem}\label{thmInvAffDetectsCurv}
Let $M$ be an $m$-dimensional globally hyperbolic spacetime of finite type 
and consider a principal $G$-bundle $P$ over $M$. Two connections $\lambda,\lambda^\prime\in\sect(\conn(P))$ 
have the same curvature, namely $F(\lambda)=F(\lambda^\prime)$, 
if and only if $\ev_\phi(\lambda)=\ev_\phi(\lambda^\prime)$ for all $\ev_\phi\in\inv_P$. 
\end{theorem}

\begin{proof}
Since $\sect(\conn(P))$ is an affine space modeled on the vector space $\f^1(M,\lieg)$, 
there exists $\omega\in\f^1(M,\lieg)$ such that $\lambda+\omega=\lambda^\prime$. 
As shown in\ Remark\ \ref{remCurvAb}, the linear part of the affine map $F:\sect(\conn(P))\to\f^2(M,\lieg)$ 
is simply $F_V=-\dd:\f^1(M,\lieg)\to\f^2(M,\lieg)$. 
Therefore, $\lambda$ and $\lambda^\prime$ have the same curvature 
if and only if $\omega$ is closed, namely $\dd\omega=0$.
Furthermore, $\ev_\phi(\lambda)$ and $\ev_\phi(\lambda^\prime)$ coincide for all $\ev_\phi\in\inv_P$ 
if and only if $\dd\omega=0$ is closed. In fact, for all $\ev_\phi\in\inv_P$, we have 
\begin{equation*}
\ev_\phi(\lambda)=\ev_\phi(\lambda^\prime)=\ev_\phi(\lambda)+(\phi_V,\omega)\,.
\end{equation*}
Therefore, by\ Proposition\ \ref{prpInvAffIsCoexact} (we use here the hypothesis that $M$ is of finite type), 
we have to check that $(\de\fc^2(M,\lieg^\ast),\omega)=\{0\}$ if and only if $\omega$ is closed, 
but this statement is true due to Stokes' theorem, thus concluding the proof. 
\end{proof}

\begin{remark}[$\inv_P$ tests only the curvature]\label{remInvAffDetectsCurv}
Already at this stage we can highlight a weak point of the space $\inv_P$ of gauge invariant affine functionals. 
As shown by the last theorem, functionals of this type are only sensitive to the curvature of a connection. 
In particular, flat connections cannot be detected using $\inv_P$. 
This has relevant physical implications. Despite the fact that flat connections, 
which correspond to the physical situation of the Aharonov-Bohm effect, are available in our framework, 
we fail in detecting them using gauge invariant affine functionals. 
We will return on this issue in\ Section\ \ref{secYangMillsChar}, 
where a different choice of functionals will weaken the constraint imposed by gauge invariance, 
eventually leading to the capability to detect all gauge equivalence classes of connections, in particular the flat ones. 
\end{remark}

The last remark motivates our interest in the class of gauge invariant functionals introduced in the next example. 

\begin{example}[Dual of $F$ and curvature affine functionals]\label{exaCurvObs}
Consider the affine differential operator $F:\sect(\conn(P))\to\f^2(M,\lieg)$. 
For the definition of affine differential operators and how to obtain formal duals 
refer to\ Subsection\ \ref{subAffDiffOp}. 
According to the argument outlined in\ Remark\ \ref{remAffDiffOpUniqueDual}, 
$F^\ast:\f^2(M,\lieg^\ast)\to\sect(\conn(P)^\dagger)$ is in general uniquely defined only 
on compactly supported forms provided that we quotient out $\triv_P$ from the target space. 
See also\ Example\ \ref{exaAffDiffOpNonUniqueDual} for the failure of uniqueness without this precaution. 
Specifically, $F^\ast:\fc^2(M,\lieg^\ast)\to\kin_P$ is uniquely specified by the condition stated below: 
\begin{align}\label{eqCurvDual}
\int_M\left(F^\ast(\beta)\right)(\lambda)\,\vol=(\beta,F(\lambda))\,, 
&& \forall\beta\in\fc^2(M,\lieg^\ast)\,,\;\forall\lambda\in\sect(\conn(P))\,,
\end{align}
where $(\cdot,\cdot):\fc^k(M,\lieg^\ast)\times\f^k(M,\lieg)\to\bbR$ is the pairing defined in\ \eqref{eqPairing2}, 
except for the irrelevant evaluation of $\lieg^\ast$ on $\lieg$. 

Note that we can use $F^\ast:\fc^2(M,\lieg^\ast)\to\kin_P$ 
to exhibit examples of gauge invariant affine functionals. 
In fact, let $\beta\in\fc^2(M,\lieg^\ast)$ and consider $\ev_\phi=F^\ast(\beta)$. 
The linear part of the functional $\ev_\phi\in\kin_P$ is given by $(\phi_V,\cdot)=-(\de\beta,\cdot)$. 
Therefore, for each $\lambda\in\sect(\conn(P))$, we have 
\begin{equation*}
\ev_\phi(\lambda+\gau_P)=\ev_\phi(\lambda)-(\de\beta,\gau_P)=\{\ev_\phi(\lambda)\}\,,
\end{equation*}
where we used the inclusion $\gau_P\subseteq\fdd^1(M,\lieg)$ to conclude. 
Note that gauge invariant functionals $F^\ast(\beta)$ of curvature type are already sufficiently many 
to test the curvature of any connection, which is anyway the maximal information one can get with $\inv_P$ 
according to\ Theorem\ \ref{thmInvAffDetectsCurv} and to\ Remark\ \ref{remInvAffDetectsCurv} 
at least for base manifolds of finite type. 
\end{example}

We return to the problem of defining suitable functionals on the space of gauge equivalence classes of on-shell 
connections on a principal bundle $P$ with structure group $G=U(1)$ over a globally hyperbolic spacetime $M$. 
Up to now we dealt with gauge invariance, but we still have to account for the dynamics. 
The equation of motion is presented in\ \eqref{eqOnShellConn} 
in terms of the affine differential operator $\MW=\de\circ F:\sect(\conn(P))\to\f^1(M,\lieg)$ 
introduced in\ \eqref{eqMaxwellOp}. 
To implement the dynamics at the level of functionals, we quotient by all those which vanish on-shell. 
Specifically, we introduce first the subspace $\van_P$ of $\inv_P$, which comprises all functionals vanishing on-shell, 
and then we define the space of affine observables taking the quotient of $\inv_P$ by $\van_P$: 
\begin{align*}
\van_P=\left\{\ev_\phi\in\inv_P\,:\;\ev_\phi(\sol_P)=\{0\}\right\}\,, && \obs_P=\frac{\inv_P}{\van_P}\,.
\end{align*}
Elements of $\obs_P$ are denoted by $\ev_{[\phi]}$, where the square brackets are used to remind 
that the quotient by $\van_P$ has been performed. 
Notice that $\obs_P$ has a natural pairing with $[\sol_P]$ specified by 
\begin{align}\label{eqAffObs}
\obs_P\times[\sol_P]\to\bbR\,, && (\ev_{[\phi]},[\lambda])\mapsto\ev_\phi(\lambda)\,,
\end{align}
where $\ev_\phi\in\ev_{[\phi]}$ and $\lambda\in[\lambda]$ are arbitrary representatives. 
This map is well-defined since $\ev_\phi$ is gauge invariant and $\lambda$ is on-shell. 
We already know that, at least for base manifolds of finite type, 
$\inv_P$ fails in separating points of $\sect(\conn(P))$, see\ Theorem\ \ref{thmInvAffDetectsCurv}. 
In particular, it is not possible to detect flat connections. Since flat connections are always on-shell, 
$\obs_P$ fails in separating points of $[\sol_P]$ via the pairing\ \eqref{eqAffObs} 
as much as $\inv_P$ fails in detecting flat connections. 
The converse, namely the fact that $[\sol_P]$ separates points of $\obs_P$, is true by construction. 
However, this property holds at the price of a quite implicit definition for the $\van_P$. 
In the following some effort is devoted to characterize explicitly the space of affine functionals vanishing on-shell. 

Indeed, as we show in the subsequent example, 
it is fairly easy to exhibit examples of gauge invariant affine functionals $\ev_\phi\in\inv_P$ 
which always vanish on solutions of the equation of motion, 
namely such that $\ev_\phi(\lambda)=0$ for all $\lambda\in\sol_P$. 

\begin{example}[Dual of $\MW$ and affine functionals vanishing on-shell]\label{exaVanConn}
Consider the formal dual $\MW^\ast$ of $\MW$, which can be obtained from the formal dual $F^\ast$ of $F$ 
exploiting the fact that $\MW=\de\circ F$. 
As for $F^\ast$, $\MW^\ast:\f^1(M,\lieg^\ast)\to\sect(\conn(P)^\dagger)$ 
is in general uniquely defined only on compactly supported forms 
provided that we quotient out $\triv_P$ from the target space. 
Specifically, $\MW^\ast:\fc^1(M,\lieg^\ast)\to\kin_P$ is uniquely specified by the condition 
\begin{align}\label{eqMaxwellDual}
\int_M\left(\MW^\ast(\alpha)\right)(\lambda)\,\vol=(\alpha,\MW(\lambda))\,, 
&& \forall\alpha\in\fc^1(M,\lieg^\ast)\,,\;\forall\lambda\in\sect(\conn(P))\,.
\end{align}
In particular, from\ \eqref{eqCurvDual}, we deduce that $\MW^\ast=F^\ast\circ\dd$. 

Using $\MW^\ast$ we can easily define functionals 
which always vanish on solutions of the equation $\MW(\lambda)=0$. 
Let $\alpha\in\fc^1(M,\lieg^\ast)$ and consider $\ev_\phi=\MW^\ast(\alpha)\in\kin_P$. 
Note that the affine functional $\ev_\phi$ is gauge invariant on account of\ Example\ \ref{exaCurvObs}. 
In fact, the image of $\MW^\ast=F^\ast\circ\dd$ is indeed included in the image of $F^\ast$. 
It is evident from\ \eqref{eqMaxwellDual} that the affine functional $\ev_\phi$ vanishes systematically on $\sol_P$. 
\end{example}

\begin{proposition}\label{prpVanObsConn}
Let $M$ be an $m$-dimensional globally hyperbolic spacetime and consider a principal $G$-bundle $P$ over $M$. 
$\ev_\phi\in\MW^\ast(\fc^1(M,\lieg^\ast))$ always vanishes on-shell. 
Assuming further that $M$ is of finite type, if $\ev_\phi\in\inv_P$ vanishes on-shell, 
then $\ev_\phi\in\MW^\ast(\fc^1(M,\lieg^\ast))$. In particular, 
in this situation $\van_P=\MW^\ast(\fc^1(M,\lieg^\ast))$ and $\obs_P=\inv_P/\MW^\ast(\fc^1(M,\lieg^\ast))$. 
\end{proposition}

\begin{proof}
The first part of the statement was shown in\ Example\ \ref{exaVanConn}. 
For the second part, assume that $\ev_\phi\in\inv_P$ is such that 
$\ev_\phi(\lambda)=0$ for each $\lambda\in\sol_P$. 
By\ Proposition\ \ref{prpInvAffIsCoexact}, $\phi_V$ is coclosed. 
As already noted immediately after its definition in\ \eqref{eqSolYM}, 
$\sol_P$ is an affine space modeled on the kernel of $\de\dd:\f^1(M,\lieg)\to\f^1(M,\lieg)$. 
This fact entails that $\ev_\phi(\lambda+\omega)=0$ for all $\lambda\in\sol_P$ and $\omega\in\ker(\de\dd)$. 
Hence, we have $(\phi_V,\omega)=0$ for all $\omega\in\ker(\de\dd)$. 
Therefore, we can apply Proposition\ \ref{prpVanForms} to $\phi_V$ 
and conclude that there exists $\rho\in\fc^1(M,\lieg^\ast)$ such that $\de\dd\rho=\phi_V$. 
Now consider $\ev_\psi=-\MW^\ast(\rho)$. We want to show that $\ev_\psi=\ev_\phi$. 
Choose an on-shell connection $\tilde{\lambda}\in\sol_P$, 
whose existence is ensured by\ Remark\ \ref{remYMSolExist}. 
Since $\sect(\conn(P))$ is an affine space modeled on $\f^1(M,\lieg)$, it is sufficient to prove that 
$\ev_\psi(\tilde{\lambda}+\omega)=\ev_\phi(\tilde{\lambda}+\omega)$ for all $\omega\in\f^1(M,\lieg)$. 
Notice that both $\ev_\phi$ (per hypothesis) and $\ev_\psi$ (per construction) vanish on $\tilde{\lambda}$ 
since it is on-shell. Therefore $\ev_\psi=\ev_\phi$ if and only if their linear parts coincide, 
which is the case due to $\psi_V=\de\dd\rho=\phi_V$. 
\end{proof}

\begin{remark}[Observables testing magnetic and electric charges]\label{remChargeObs}
Let $P$ be a principal $G$-bundle over an $m$-dimensional globally hyperbolic spacetime $M$. 
The dual of the curvature map $F^\ast:\fc^2(M,\lieg^\ast)\to\inv_P$, {\em cfr.}\ Remark \ref{exaCurvObs}, 
allows us to define observables testing the magnetic and electric charges. 

Take $\theta\in\fcde^2(M,\lieg^\ast)$ 
and define $\ev_{[\phi]}\in\obs_P$ as the equivalence class of $\ev_\phi=F^\ast(\theta)$. 
Notice that $\ev_\phi$ has vanishing linear part. In fact, the linear part of $\phi_V=-\de\theta=0$. 
Taking a gauge equivalence class of on-shell connections $[\lambda]\in[\sol_P]$, one can evaluate $\ev_{[\phi]}$: 
\begin{equation*}
\ev_{[\phi]}([\lambda])=(\theta,F(\lambda))={}_\de([\theta],[F(\lambda)])\,.
\end{equation*}
Therefore, evaluating $[\lambda]$ with $\ev_{[\phi]}$ amounts to test the $\dd$-cohomology class 
$[F(\lambda)]\in\hdd^2(M,\lieg)$, namely the Chern class of $P$.\footnote{Notice that 
$[F(\lambda)]\in\hdd^2(M,\lieg)$ does not really depend on $\lambda$, 
but only on the topology of the principal bundle $P$. In fact, two connections always differ 
by a $\lieg$-valued $1$-form, whence their curvatures always differ by an exact $\lieg$-valued $2$-form.} 
In particular, since the pairing ${}_\de(\cdot,\cdot):\hcde^2(M,\lieg^\ast)\times\hdd^2(M,\lieg)\to\bbR$ 
induces an isomorphism $\hdd^2(M,\lieg)\to(\hcde^2(M,\lieg^\ast))^\ast$, 
see\ Remark\ \ref{remPoincareDualityImproved}, 
the class of observables presented above makes it possible to determine the Chern class of $P$. 
As a matter of fact, in this procedure to test the Chern class of $P$, 
one could replace coclosed $2$-forms with compact support with closed $2$-cycles. 
This makes evident the fact that observables of the type described really measure 
the magnetic charge, namely the flux of $F(\lambda)$ through a closed $2$-surface embedded in $M$. 

The on-shell condition we impose on connections, namely $\de F(\lambda)=\MW(\lambda)=0$, 
entails that the curvature associated to a connection $\lambda\in\sol_P$, besides being closed, 
is also a coclosed form. In fact, to each gauge equivalence class of connections $[\lambda]\in[\sol_P]$, 
one can assign the cohomology class $[F(\lambda)]\in\hde^2(M,\lieg)$.\footnote{Notice that 
$[F(\lambda)]\in\hde^2(M,\lieg)$ does not depend on the choice of a representative $\lambda$ 
in the gauge equivalence class $[\lambda]$. 
This is due to the fact that gauge transformations always shift connections 
by $\lieg$-valued closed $1$-forms, {\em cfr.}\ \eqref{eqConnGaugeHom}, 
whence the curvature $F(\lambda)$ is not affected by gauge transformations.} 
Taking $\eta\in\fcdd^2(M,\lieg)$, one can introduce $\ev_{[\psi]}\in\obs_P$ 
specifying one of its representatives according to $\ev_\psi=F^\ast(\eta)$. 
In this case, the linear part of $\ev_\psi$ is given by $\psi_V=-\de\eta$. 
Let us evaluate $\ev_{[\psi]}$ on $[\lambda]$: 
\begin{equation*}
\ev_{[\psi]}=(\eta,F(\lambda))=([\eta],[F(\lambda)])_\de\,.
\end{equation*}
We conclude that this class of observables tests the $\de$-cohomology class $[F(\lambda)]\in\hde^2(M,\lieg)$. 
In fact, for a given gauge equivalence class of connections $[\lambda]\in[\sol_P]$, 
one can determine the associated cohomology class $[F(\lambda)]\in\hde^2(M,\lieg)$ 
simply using the class of observables presented here. This follows from the fact that the pairing 
$(\cdot,\cdot)_\de:\hcdd^2(M,\lieg^\ast)\times\hde^2(M,\lieg)\to\bbR$ induces an isomorphism 
$\hde^2(M,\lieg)\to(\hcdd^2(M,\lieg^\ast))^\ast$, see\ Remark\ \ref{remPoincareDualityImproved}. 
The physical quantity which one determines with this procedure represents the electric charge 
associated to a gauge equivalence class of on-shell connections $[\lambda]\in[\sol_P]$. 
To make this interpretation more evident, one should consider $[\ast F(\lambda)]\in\hdd^{m-2}(M,\lieg)$ 
in place of $[F(\lambda)]\in\hde^2(M,\lieg)$ and replace $\lieg$-valued closed $2$-forms with compact support 
with $(m-2)$-cycles. Doing so, one realizes that testing $\ev_{[\psi]}$ on $[\lambda]$ 
amounts to measuring the flux of $\ast F(\lambda)$ through a closed $(m-2)$-surface embedded in $M$. 
This is the typical flux of the electric field through a closed $2$-surface in the physical situation, namely $m=4$. 

For further details on magnetic and electric charges refer to\ \cite[Section\ 6]{BDS14a} 
and\ \cite[Remark\ 5.5]{BDHS14}. 
\end{remark}

We can endow the space of observables $\obs_P$ for gauge equivalence classes of Yang-Mills connections 
on a principal $G$-bundle $P$ over an $m$-dimensional globally hyperbolic spacetime $M$ 
with a presymplectic structure defined out of the causal propagator $G$ for the Hodge-d'Alembert operator $\Box$. 
This choice can be motivated applying Peierls' method to the Lagrangian density 
\begin{equation}\label{eqLagrangian}
\mathcal{L}[\lambda]=h(F(\lambda))\wedge\ast F(\lambda)\,,
\end{equation}
where $h:\lieg\to\lieg^\ast$ is an $\ad$-equivariant isomorphism such that 
$(h\xi)(\eta)=(h\eta)(\xi)$ for all $\xi,\eta\in\lieg$,\footnote{Since $\lieg=i\bbR$ 
and the adjont action of $G$ on $\lieg$ is trivial, $G$ being Abelian, 
$h$ is simply a real number different from zero, which can be interpreted as an electric charge constant, 
see the comment after\ \cite[eq. (3.11)]{BDHS14}.} 
and evaluation of $\lieg^\ast$ on $\lieg$ is understood. 
For further details on Peierls' method, see\ \cite[Remark\ 3.5]{BDS14a} for the present situation or 
\cite{Pei52} and \cite[Section\ I.4]{Haa96} for a more general discussion about this topic. 
Notice that $h$ is fixed once and for all. 
We will often use the inverse $h^{-1}:\lieg^\ast\to\lieg$ to define an inner product on $\lieg^\ast$. 
For example, we will consider the pairing $(\cdot,\cdot)_h$ between $\lieg^\ast$-valued $k$-forms on $M$ 
with compact overlapping support, which is obtained from the pairing $(\cdot,\cdot)$ defined in \eqref{eqPairing2}: 
\begin{equation}\label{eqLiePairing}
(\omega,\omega^\prime)_h=(\omega,h^{-1}(\omega^\prime))\,,
\end{equation}
where $\omega,\omega^\prime\in\f^k(M,\lieg^\ast)$ have supports with compact overlap 
and the evaluation of $\lieg^\ast$ on $\lieg$ is understood. 
Note that the requirements on $h$ entail that this pairing is symmetric upon interchange of the arguments. 

\begin{proposition}\label{prpPSymObsConn}
Let $M$ be an $m$-dimensional globally hyperbolic spacetime and consider a principal $G$-bundle $P$ over $M$. 
Denote with $G:\fc^1(M,\lieg^\ast)\to\fsc^1(M,\lieg^\ast)$ the causal propagator 
for the Hodge-d'Alembert operator $\Box$ on $M$ acting on $\lieg^\ast$-valued $1$-forms. 
The anti-symmetric bilinear map defined below induces 
a presymplectic form $\tau_P:\obs_P\times\obs_P\to\bbR$ on the quotient $\obs_P=\inv_P/\van_P$: 
\begin{align*}
\tau_P:\inv_P\times\inv_P\to\bbR\,, && (\ev_\phi,\ev_\psi)\mapsto(\phi_V,G\psi_V)_h\,,
\end{align*}
where $\phi_V,\psi_V\in\fc^1(M,\lieg^\ast)$ represent the linear parts of $\ev_\phi,\ev_\psi\in\inv_P$ 
and $(\cdot,\cdot)_h$ denotes the pairing between $\lieg^\ast$-valued $1$-forms 
with compact overlapping support, see\ \eqref{eqLiePairing}, 
\end{proposition}

\begin{proof}
It is clear from its definition that $\tau_P:\inv_P\times\inv_P\to\bbR$ is bilinear. 
Anti-symmetry follows from the causal propagator $G:\fc^1(M,\lieg^\ast)\to\fsc^1(M,\lieg^\ast)$ 
being formally antiself-adjoint with respect to $(\cdot,\cdot)_h$ 
and from $(\cdot,\cdot)_h$ being symmetric with respect to the interchange of its arguments. 
It remains only to check that $\tau_P(\ev_\phi,\ev_\psi)=0$ for all $\ev_\phi\in\van_P$ and $\ev_\psi\in\inv_P$. 
From\ Proposition\ \ref{prpInvAffIsCoexact}, we know that $\psi_V$ is coclosed, 
whence $\omega=h^{-1}(G\psi_V)\in\fsc^1(M,\lieg)$ is a solution of the equation $\de\dd\omega=0$, 
which is the linear part of eq.\ \eqref{eqOnShellConn}. 
Consider a connection $\tilde{\lambda}\in\sol_P$, which exists on account of\ Remark\ \ref{remYMSolExist}. 
Then $\lambda=\tilde{\lambda}+\omega$ lies in $\sol_P$ as well. 
Evaluating $\ev_\phi$ on $\lambda$, we obtain the following result: 
\begin{equation*}
\ev_\phi(\lambda)=\ev_\phi(\tilde{\lambda})+(\phi_V,\omega)=\ev_\phi(\tilde{\lambda})+(\phi_V,G\psi_V)_h\,,
\end{equation*}
where the last step follows from the definition of $(\cdot,\cdot)_h$ in \eqref{eqLiePairing}. 
Since $\ev_\phi\in\van_P$, $\ev_\phi(\tilde{\lambda})=0$ and $\ev_\phi(\lambda)=0$ too, 
$(\phi_V,G\psi_V)_h$ vanishes as well, 
proving that $\tau_P:\inv_P\times\inv_P\to\bbR$ descends to the quotient $\obs_P$. 
This shows that $\tau_P:\obs_M\times\obs_M\to\bbR$ is a a well-defined presymplectic form. 
\end{proof}

As in\ Chapter \ref{chMaxwell}, after introducing the presymplectic structure on the relevant space of observables, 
a natural question which arises is whether the presymplectic structure is degenerate or not. 
This has important consequences for the properties of the construction which will follow. 
In particular, on account of\ Proposition\ \ref{prpSymInj}, the null space of the presymplectic form provides 
an upper bound for the kernels of all presymplectic linear maps, thus helping us in understanding 
to what extent the requirement of locality of\ \cite[Definition\ 2.1]{BFV03} is violated by the present model. 

\begin{proposition}\label{prpRadConn}
Let $M$ be an $m$-dimensional globally hyperbolic spacetime and consider a principal $G$-bundle $P$ over $M$. 
Denote with $\rad_P$ the null space of the presymplectic structure $\tau_P$ on $\obs_P$ 
introduced in\ Proposition\ \ref{prpPSymObsConn} and take $\ev_\phi\in\inv_P$. 
The following implications hold true: 
\begin{enumerate}
\item If $\phi_V\in\de(\fc^2(M,\lieg^\ast)\cap\dd\ftc^1(M,\lieg^\ast))$, then $\ev_{[\phi]}\in\rad_P$; 
\item If $\ev_{[\phi]}\in\rad_P$, then $\phi_V\in\de\fcdd^2(M,\lieg^\ast)$. 
\end{enumerate}
Furthermore, under the assumption that $M$ is of finite type, 
the null space $\rad_P$ has the following explicit characterization: 
\begin{equation*}
\rad_P=\left\{\ev_\phi\in\inv_P\,:\;\phi_V\in\de\fcdd^2(M,\lieg^\ast)\right\}/\van_P\,.
\end{equation*}
\end{proposition}

\begin{proof}
To start with, consider $\ev_\phi\in\inv_P$ such that 
$\phi_V\in\de(\fc^2(M,\lieg^\ast)\cap\dd\ftc^1(M,\lieg^\ast))$, 
namely there exists $\xi\in\ftc^1(M,\lieg^\ast)$ such that $\dd\xi$ has compact support and $\de\dd\xi=\phi_V$. 
Consider the causal propagator $G:\fc^k(M,\lieg^\ast)\to\fsc^k(M,\lieg^\ast)$ 
for $\Box=\de\dd+\dd\de$ acting on $\lieg^\ast$-valued $k$-forms on $M$. 
Since $\dd$ intertwines $G$ for different degrees $k$, we have the following chain of identities: 
\begin{equation*}
G\phi_V=G\de\dd\xi=G(\Box\xi-\dd\de\xi)=-\dd G\de\xi\,.
\end{equation*}
Therefore, for each $\ev_\psi\in\inv_P$, we can conclude 
\begin{equation*}
\tau_P(\ev_{[\psi]},\ev_{[\phi]})=(\psi_V,G\phi_V)_h=-(\psi_V,\dd G\de\xi)_h=-(\de\psi_V,G\de\xi)_h=0\,,
\end{equation*}
where we used both Stokes' theorem and the fact that $\psi_V$ is coclosed, 
{\em cfr.}\ Proposition\ \ref{prpInvAffIsCoexact}. This proves the first statement. 

For the second statement, assume that $\ev_\phi\in\inv_P$ is such that $\ev_{[\phi]}$ lies in $\rad_P$. 
Since, by\ Proposition\ \ref{prpInvAffIsCoexact}, each $\ev_\psi\in\kin_P$ with $\psi_V\in\de\fc^2(M,\lieg^\ast)$ 
lies in $\inv_P$, our hypothesis entails that $(\de\beta,G\phi_V)_h=0$ for each $\beta\in\fc^2(M,\lieg^\ast)$, 
whence $G\dd\phi_V=0$. On account of this fact, there exists $\eta\in\fc^2(M,\lieg^\ast)$ such that 
$\Box\eta=\dd\phi_V$. Applying $\dd$ on both sides and recalling that $\eta$ has compact support, 
we deduce that $\dd\eta=0$. Furthermore, 
recalling that $\ev_\phi\in\inv_P$ entails $\de\phi_V=0$ according to\ Proposition\ \ref{prpInvAffIsCoexact}, 
we get $\Box\de\eta=\de\dd\phi_V=\Box\phi_V$, whence $\phi_V=\de\eta\in\de\fcdd^2(M,\lieg^\ast)$. 

Assuming that $M$ is of finite type, we can refine the first statement. 
In fact, according to\ Proposition\ \ref{prpInvAffIsCoexact}, under this extra hypothesis, 
we have that each $\ev_\psi\in\inv_P$ has $\psi_V\in\de\fc^2(M,\lieg^\ast)$. 
Therefore, given $\ev_\phi\in\inv_P$ with $\phi_V\in\de\fcdd^2(M,\lieg^\ast)$, 
we have to check that $\tau_P(\ev_{[\psi]},\ev_{[\phi]})=(\psi_V,G\phi_V)_h$ vanishes for all $\ev_\psi\in\inv_P$. 
Per hypothesis, there exists $\eta\in\fcdd^2(M,\lieg^\ast)$ such that $\de\eta=\phi_V$. 
Furthermore, by the argument mentioned above and exploiting the fact that $M$ is of finite type, 
for each $\ev_\psi\in\inv_P$, we can find $\beta\in\fc^2(M,\lieg^\ast)$ such that $\de\beta=\psi_V$. 
All these facts together entail that 
\begin{equation*}
\tau_P(\ev_{[\psi]},\ev_{[\phi]})=(\psi_V,G\phi_V)_h=(\de\beta,G\de\eta)_h=(\beta,G\dd\de\eta)=0\,.
\end{equation*}
For the last equality use the identity $\dd\de\eta=\Box\eta$, which follows from $\dd\eta=0$. 
\end{proof}

Clearly this kind of presymplectic structure always has a non-trivial null space. This is simply because 
only the linear parts of the elements in $\obs_P$ enter the definition of the presymplectic form. 
However, this kind of degeneracy is not causing any issue to locality since elements in $\obs_P$ 
with trivial linear part are mapped injectively along arrows in $\PrBun_\GHyp$. 
This will be discussed in more detail in\ Remark\ \ref{remMagObsConn}. 
To construct an example of a principal bundle $P$ with structure group $G=U(1)$ 
over a globally hyperbolic spacetime $M$ such that 
the corresponding presymplectic form $\tau_P$ exhibits more interesting degeneracies, 
we refer to\ Example\ \ref{exaNonTrivRadForms}. 

\begin{example}[Non-trivial null space]\label{exaNonTrivRadConn}
Let $M$ be an $m$-dimensional globally hyperbolic spacetime 
with Cauchy hypersurface diffeomorphic to $\bbR^1\times\bbT^{m-2}$. 
On $M$ consider the trivial principal bundle $P=M\times G$. 
The argument of\ Example\ \ref{exaNonTrivRadForms} for $k=1$ provides 
a $2$-form $\theta\in\fcdd^2(M)$ with compact support which is closed, 
but non-exact in the sense of cohomology with compact support. 
In fact, $[i\theta]\neq0\in\hcdd^2(M,\lieg^\ast)$.\footnote{Recall that 
$\lieg=i\bbR$ and therefore $\lieg^\ast\simeq i\bbR$.} 
Therefore, via the formal dual $F^\ast:\fc^2(M,\lieg^\ast)\to\kin_P$ of the curvature map 
we can introduce $\ev_\phi=F^\ast(i\theta)\in\inv_P$, see\ Example\ \ref{exaCurvObs} for further details. 
Since the linear part of $\ev_\phi$ is given by $-i\de\theta$, $\ev_{[\phi]}$ lies in $\rad_P$, 
{\em cfr.}\ Proposition\ \ref{prpRadConn}. 
We still have to show that $\ev_{[\phi]}$ is non-trivial in $\obs_P$. 
Our choice of $M$ admits a finite good cover, therefore $\van_P=\MW^\ast(\fc^1(M,\lieg^\ast))$ 
according to\ Proposition\ \ref{prpVanObsConn}. We argue by contradiction assuming that 
there exists $\eta\in\fc^1(M,\lieg^\ast)$ such that $\MW^\ast(\eta)=\ev_\phi$. 
This implies the identity $-\de\dd\eta=-i\de\theta$ between the linear parts. 
$\theta$ being closed, acting on both sides with $\dd$, we get $\Box\dd\eta=i\Box\theta$, 
whence $\dd\eta=i\theta$, thus contradicting the fact that $[i\theta]$ is non-trivial in $\hcdd^2(M,\lieg^\ast)$. 
\end{example}

The next step consists in verifying that the assignment of the presymplectic space $(\obs_P,\tau_P)$ 
to each principal bundle $P$ with structure group $G=U(1)$ 
over an $m$-dimensional globally hyperbolic spacetime $M$ gives rise to a covariant functor. 

\begin{lemma}\label{lemDualConnFunctor}
For $m\geq2$ and $G=U(1)$, let $P$ and $Q$ be principal $G$-bundles 
over $m$-dimensional globally hyperbolic spacetimes, respectively $M$ and $N$, and consider 
a principal bundle map $f:P\to Q$ covering a causal embedding $\ul{f}:M\to N$. Introduce the linear map 
\begin{align*}
\sc(\conn(f)^\dagger):\sc(\conn(P)^\dagger)\to\sc(\conn(Q)^\dagger)\,, && \phi\mapsto\psi\,,
\end{align*}
where $\psi\in\sc(\conn(Q)^\dagger)$ is defined out of $\phi\in\sc(\conn(P)^\dagger)$ according to 
\begin{align*}
\psi_{\ul{f}(x)}\circ\conn(f)\vert_x=\phi_x:\conn(P)_x\to\bbR\,, && \forall\,x\in M\,.
\end{align*}
In particular, $\sc(\conn(f)^\dagger)$ satisfies the following identity 
for each $\phi\in\sc(\conn(P)^\dagger)$ and each $\lambda\in\sect(\conn(Q))$: 
\begin{equation}\label{eqPushforwardKinConn}
\ev_{\sc(\conn(f)^\dagger)\phi}(\lambda)=\ev_\phi\big(\sect(\conn(f))(\lambda)\big)\,,
\end{equation}
where $\sect(\conn(f)):\sect(\conn(Q))\to\sect(\conn(P))$ is defined in\ Remark\ \ref{remConnFunctor}. 

The assignment of the vector space $\sc(\conn(P)^\dagger)$ to each object $P$ in $\PrBun_\GHyp$ 
and of the linear map $\sc(\conn(f)^\dagger)$ to each morphism $f$ in $\PrBun_\GHyp$ 
defines a covariant functor $\sc(\conn(\cdot)^\dagger):\PrBun_\GHyp\to\Vec$ 
taking values in the category of vector spaces $\Vec$. 
\end{lemma}

\begin{proof}
Take $\phi\in\sc(\conn(P)^\dagger)$. 
At each point $x\in M$, $\phi$ gives an affine map $\phi_x:\conn(P)_x\to\bbR$. 
We want to map $\phi_x$ via $f$ to an affine map $\conn(Q)_{\ul{f}(x)}\to\bbR$ at the point $\ul{f}(x)$. 
We do it in the spirit of\ Remark\ \ref{remBunConnFunctor}, but in the opposite direction. 
More specifically, exploiting the fact that $\ul{f}:M\to N$ is an embedding 
and that $f:P\to Q$ is an isomorphism on each fiber, we can invert $f$ onto its image. 
Therefore, to define the affine bundle map $\conn(f^{-1}):\conn(Q_{\ul{f}(M)})\to\conn(P)$ 
according to\ Remark\ \ref{remBunConnFunctor}, we consider the principal bundle isomorphism 
$f^{-1}:Q_{\ul{f}(M)}\to P$ covering the diffeomorphism $\ul{f}^{-1}:f(M)\to M$. 
Using the restriction of $\conn(f^{-1})$ to the fiber over $\ul{f}(x)$, $x\in M$, we can introduce the affine map 
$\psi_x=\phi_x\circ\conn(f^{-1})\vert_{\ul{f}(x)}:\conn(Q)_{\ul{f}(x)}\to\bbR$. 
This defines a new smooth section $\psi\in\sect(\conn(Q_{\ul{f}(M)})^\dagger)$. 
Since $\phi$ has compact support in $M$, $\psi$ has compact support in $\ul{f}(M)$. 
$\ul{f}(M)$ is open in $N$, $\ul{f}$ being a causal embedding. Therefore $\psi$ can be extended 
by zero to a smooth section $\psi\in\sc(\conn(Q)^\dagger)$ with compact support. 
This shows that the map $\sc(\conn(f)^\dagger):\sc(\conn(P)^\dagger)\to\sc(\conn(Q)^\dagger)$ is well-defined. 
Linearity follows immediately from the definition. 

Take $\phi\in\sc(\conn(P)^\dagger)$ and $\lambda\in\sect(\conn(Q))$. 
Recalling Remark\ \ref{remConnFunctor} and comparing with the construction illustrated above, 
one realizes that the evaluation of $\sc(\conn(f)^\dagger)\phi$ on $\lambda$ at the point $\ul{f}(x)\in N$ 
gives the same result of the evaluation of $\phi$ on $\sect(\conn(f))(\lambda)$ at the point $x\in M$. 
Since $\sect(\conn(f)^\dagger)\phi$ vanishes on $N\setminus\ul{f}(M)$, the following equality holds: 
\begin{equation*}
\int_N\left(\sc\left(\conn(f)^\dagger\right)\phi\right)(\lambda)\,\vol
=\int_M\phi\left(\sect(\conn(f))(\lambda)\right)\,\vol\,,
\end{equation*}
whence $\ev_{\sc(\conn(f)^\dagger)\phi}(\lambda)=\ev_\phi\big(\sect(\conn(f))(\lambda)\big)$ as claimed. 

Let $P$ be an object in $\PrBun_\GHyp$ and consider the principal bundle map $\id_P:P\to P$ covering $\id_M$. 
Since $\conn(\id_P)=\id_{\conn(P)}$, $\sc(\conn(\id_P)^\dagger)=\id_{\sc(\conn(P)^\dagger)}$. 
Furthermore, given two composable morphisms $f:O\to P$ and $h:P\to Q$ in $\PrBun_\GHyp$ 
and denoting with $\ul{f}:L\to M$ and $\ul{h}:M\to N$ the corresponding causal embeddings, 
we have that $(h\circ f)^{-1}:Q_{\ul{h}\circ\ul{f}(L)}\to O$ coincides 
with the composition $f^{-1}\circ h^{-1}\vert_{\ul{h}(\ul{f}(L))}:Q_{\ul{h}(\ul{f}(L))}\to O$ 
of $f^{-1}:P_{\ul{f}(L)}\to O$ with the restriction of $h^{-1}:Q_{\ul{h}(M)}\to P$ 
to the subbundle over $\ul{h}(\ul{f}(L))$. From this fact, it follows that 
$\sc(\conn(h\circ f)^\dagger)=\sc(\conn(h)^\dagger)\circ\sc(\conn(f)^\dagger)$. 
Therefore $\sect(\conn(\cdot)^\dagger):\PrBun_\GHyp\to\Vec$ is a covariant functor. 
\end{proof}

\begin{theorem}\label{thmObsFunctorConn}
Let $m\geq2$ and $G=U(1)$. Consider the principal $G$-bundles $P$ and $Q$ 
over the $m$-dimensional globally hyperbolic spacetimes $M$ and respectively $N$. 
Given a principal bundle map $f:P\to Q$ covering a causal embedding $\ul{f}:M\to N$, 
the linear map $\sc(\conn(f)^\dagger):\sc(\conn(P)^\dagger)\to\sc(\conn(P)^\dagger)$ 
defined in\ Lemma\ \ref{lemDualConnFunctor} induces the presymplectic linear map 
\begin{align*}
\PSV(f):(\obs_P,\tau_P)\to(\obs_Q,\tau_Q)\,, && \ev_{[\phi]}\mapsto\ev_{[\sect(\conn(f)^\dagger)\phi]}\,.
\end{align*}
Setting $\PSV(P)=(\obs_P,\tau_P)$ for each principal $G$-bundle $P$ over an $m$-di\-men\-sio\-nal 
globally hyperbolic spacetime $M$, $\PSV:\PrBun_\GHyp\to\PSymV$ turns out to be a covariant functor. 
\end{theorem}

\begin{proof}
As a first step, we note that $\sect(\conn(f)^\dagger)$ induces a linear map 
\begin{align*}
L_f:\kin_P\to\kin_Q\,, && \ev_\phi\mapsto\ev_{\sect(\conn(f)^\dagger)\phi}\,.
\end{align*}
In fact, on account of\ Lemma\ \ref{lemDualConnFunctor} and of\ eq.\ \eqref{eqPushforwardKinConn} in particular, 
for each $\phi\in\sc(\conn(P)^\dagger)$ and each $\lambda\in\sect(\conn(Q))$, 
one has the identity $\ev_{\sc(\conn(f)^\dagger)\phi}(\lambda)=\ev_\phi(\sect(\conn(f))(\lambda))$. 
If $\phi\in\triv_P$, the right-hand-side vanishes for all $\lambda\in\sect(\conn(Q))$, 
whence $\sc(\conn(f)^\dagger)\phi\in\triv_P$ by\ Proposition\ \ref{prpAffSeparability}. 

Now we check that $L_f$ maps $\inv_P$ to $\inv_Q$ and $\van_P$ to $\van_Q$. 
If this is the case, the linear map $\PSV(f):\obs_P\to\obs_Q$ is well-defined. 
Let $\ev_\phi\in\inv_P$. Recalling the action of gauge transformations on connections, 
see\ \eqref{eqGaugeConnAb}, for each $h\in\c(N,G)$ and $\lambda\in\sect(\conn(Q))$, we have 
\begin{equation*}
(L_f(\ev_\phi))(\lambda-h^\ast\mu)=\ev_\phi(\sect(\conn(f))(\lambda-h^\ast\mu))\,,
\end{equation*}
where $\mu$ denotes the Maurer-Cartan form on $G=U(1)$. 
On account of Remark\ \ref{remConnFunctor} and of\ Proposition\ \ref{prpTrivAdBun}, the linear part 
$\sect(\conn(f))_V:\f^1(N,\ad(Q))\to\f^1(M,\ad(P))$ of $\sect(\conn(f))$ boils down to the pull-back 
$\ul{f}^\ast:\f^1(N,\lieg)\to\f^1(M,\lieg)$ for $\lieg$-valued $1$-forms along the base map $\ul{f}:M\to N$. 
Therefore 
\begin{equation*}
\sect(\conn(f))\left(\lambda-h^\ast\mu\right)=\sect(\conn(f))(\lambda)-\ul{f}^\ast h^\ast\mu
=\sect(\conn(f))(\lambda)-(h\circ\ul{f})^\ast\mu\,.
\end{equation*}
Since $h\circ\ul{f}\in\c(M,G)$ and $\ev_\phi\in\inv_P$, we conclude that 
\begin{align*}
(L_f(\ev_\phi))(\lambda-h^\ast\mu) & =\ev_\phi\left(\sect(\conn(f))(\lambda)-(h\circ\ul{f})^\ast\mu\right)
=\ev_\phi(\sect(\conn(f))(\lambda))\\
& =(L_f(\ev_\phi))(\lambda)\,,
\end{align*}
whence $L_f(\ev_\phi)\in\inv_Q$. Note that we used \eqref{eqPushforwardKinConn} for the last equality. 

Recalling Remark\ \ref{remMWNaturality}, $\MW:\sect(\conn(\cdot))\to\f^1((\cdot)_\base,\lieg)$ 
is a natural transformation between contravariant functors from $\PrBun_\GHyp$ to $\Aff$. 
In particular the affine map $\sect(\conn(f)):\sect(\conn(Q))\to\sect(\conn(P))$ restricts to 
an affine map $\sect(\conn(f)):\sol_Q\to\sol_P$. Therefore, given $\ev_\phi\in\van_P$, 
for each $\lambda\in\sol_Q$, exploiting \eqref{eqPushforwardKinConn}, we get 
\begin{equation*}
(L_f(\ev_\phi))(\lambda)=\ev_\phi(\sect(\conn(f))(\lambda))=0\,,
\end{equation*}
thus showing that $L_f(\ev_\phi)$ lies in $\van_Q$. 

Up to now, we proved that $\PSV(f):\obs_P\to\obs_Q$ is well-defined. 
We still have to show that it preserves the relevant presymplectic structures. 
Take $\ev_\phi,\ev_\psi\in\inv_P$ and consider $L_f(\ev_\phi),L_f(\ev_\psi)\in\inv_Q$. 
In order to evaluate $\tau_Q$ on $L_f(\ev_\phi),L_f(\ev_\psi)\in\inv_Q$, 
we need the $\lieg^\ast$-valued $1$-forms with compact support on $N$ which represent their linear parts. 
Denote the linear parts of $\ev_\phi$ and of $\ev_\psi$ respectively with $\phi_V\in\fc^1(M,\lieg^\ast)$ 
and with $\psi_V\in\fc^1(M,\lieg^\ast)$. As already mentioned above, the linear part of $\sect(\conn(f))$ 
is simply $\ul{f}^\ast:\f^1(N,\lieg)\to\f^1(M,\lieg)$, whence 
$\ul{f}_\ast:\fc^1(M,\lieg^\ast)\to\fc^1(N,\lieg^\ast)$ provides the linear parts of $\ev_\phi$ and of $\ev_\psi$, 
namely $\ul{f}_\ast\phi_V\in\fc^1(N,\lieg^\ast)$ and respectively $\ul{f}_\ast\psi_V\in\fc^1(N,\lieg^\ast)$. 
To keep track of the base space, we introduce a subscript 
on the causal propagator $G$ for $\Box$ acting on $\lieg^\ast$-valued $1$-forms. 
Recalling the definition of the presymplectic form, see\ Proposition\ \ref{prpPSymObsConn}, we conclude that 
\begin{equation*}
(\ul{f}_\ast\phi_V,G_N\ul{f}_\ast\psi_V)_{h\,N}=(\phi_V,\ul{f}^\ast G_N\ul{f}_\ast\psi_V)_{h\,M}\,,
\end{equation*}
where also $(\cdot,\cdot)_h$, the pairing between $\lieg^\ast$-valued $1$-forms defined in\ \eqref{eqLiePairing}, 
carries a subscript referred to the globally hyperbolic spacetime upon which the pairing is defined. 
According to\ Proposition\ \ref{prpGreenNat}, $\ul{f}^\ast G_N\ul{f}_\ast=G_M$ on $\fc^1(M,\lieg^\ast)$, whence 
\begin{equation*}
\tau_Q\left(\PSV(f)\ev_{[\phi]},\PSV(f)\ev_{[\psi]}\right)=\tau_P\left(\ev_{[\phi]},\ev_{[\psi]}\right)\,.
\end{equation*}

It remains only to check that $\PSV:\PrBun_\GHyp\to\PSymV$ is a covariant functor. 
As a matter of fact, this follows from $\sc(\conn(\cdot)^\dagger):\PrBun_\GHyp\to\Vec$ being a covariant functor. 
In fact, $L_{\id_P}=\id_{\kin_P}$ and, for each pair of composable morphisms $f:O\to P$ and $h:P\to Q$, 
one has $L_{h\circ f}=L_h\circ L_f$. Therefore, the restriction to gauge invariant affine functionals first 
and then the quotient by those which vanish on-shell show that $\PSV(\id_P)=\id_{(\obs_P,\tau_P)}$ 
and $\PSV(h\circ f)=\PSV(h)\circ\PSV(f)$, whence $\PSV:\PrBun_\GHyp\to\PSymV$ is indeed a covariant functor. 
\end{proof}

To make contact with the axiomatic approach to quantum field theory on curved spacetimes 
proposed in\ \cite{BFV03}, we turn our attention to the causality property 
and the time-slice axiom for the covariant functor $\PSV:\GHyp\to\PSymV$. 
We will discuss the failure of the locality property (as well as a possible procedure to recover it) 
in\ Subsection\ \ref{subLocalityConn}. 

\begin{theorem}\label{thmCausalityConn}\index{causality}
Let $m\geq2$ and $G=U(1)$. 
The covariant functor $\PSV:\PrBun_\GHyp\to\PSymV$ fulfils the {\em causality} property: 
Consider the principal $G$-bundles $P_1$, $P_2$ and $Q$ 
over the $m$-dimensional globally hyperbolic spacetimes $M_1$, $M_2$ and $N$. 
Furthermore, assume $f:P_1\to Q$ and $g:P_2\to Q$ are principal bundle maps covering the causal embeddings 
$\ul{f}:M_1\to N$ and respectively $\ul{h}:M_2\to N$ whose images are causally disjoint in $N$, 
namely such that $\ul{f}(M_1)\cap J_N(\ul{h}(M_2))=\emptyset$. Then 
\begin{equation*}
\tau_Q(\PSV(f)\obs_{P_1},\PSV(h)\obs_{P_2})=\{0\}\,.
\end{equation*}
\end{theorem}

\begin{proof}
Let $\ev_\phi\in\inv_{P_1}$ and $\ev_\psi\in\inv_{P_2}$. Recalling the definition of 
the presymplectic form $\tau_Q$, {\em cfr.}\ Proposition\ \ref{prpPSymObsConn}, 
and of the functor $\PSV$, {\em cfr.}\ Theorem\ \ref{thmObsFunctorConn}, one gets 
\begin{equation*}
\tau_Q\left(\PSV(f)\ev_{[\phi]},\PSV(h)\ev_{[\psi]}\right)=\big(\ul{f}_\ast\phi_V,G\ul{h}_\ast\psi_V\big)_h\,,
\end{equation*}
where $\phi_V\in\fc^1(M_1,\lieg^\ast)$ and $\psi_V\in\fc^1(M_2,\lieg^\ast)$ represent 
the linear parts of $\ev_\phi$ and respectively of $\ev_\psi$. 
The first argument of the pairing $(\cdot,\cdot)_h$ is supported inside $\ul{f}(M_1)$, 
while the second argument is supported inside $J_N(\ul{h}(M_2))$ by the properties of the causal propagator $G$. 
The hypothesis entails that the supports do not overlap, 
therefore the right-hand-side of the equation displayed above vanishes as claimed. 
\end{proof}

The following lemmas will be used for the proof of the time-slice axiom, see\ Theorem\ \ref{thmTimeSliceConn}. 

\begin{lemma}\label{lemTimeSliceConn1}
Let $G=U(1)$ and consider a principal $G$-bundle $P$ over an $m$-dimensional globally hyperbolic spacetime $M$. 
Furthermore, assume that a partition of unity $\{\chi_+,\chi_-\}$ 
with $\supp(\chi_\pm)$ past/future compact is given. For each $\ev_\phi\in\inv_P$, 
there exists $\psi\in\sc(\conn(P)^\dagger)$ with support inside $\supp(\chi_+)\cap\supp(\chi_-)$ 
such that $\ev_\psi-\ev_\phi\in\MW^\ast(\fc^1(M,\lieg^\ast))$, whence $\ev_\psi\in\inv_P$. 
\end{lemma}

\begin{proof}
Given $\ev_\phi\in\inv_P$, we choose one of its representatives $\phi\in\sc(\conn(Q)^\dagger)$. 
Its linear part $\phi_V\in\fc^1(M,\lieg^\ast)$ is coclosed according to\ Proposition\ \ref{prpInvAffIsCoexact}. 
Applying\ Lemma\ \ref{lemTimeSliceForms2}, we find $\xi\in\fcde^1(M,\lieg^\ast)$ supported inside 
$\supp(\chi_+)\cap\supp(\chi_-)$ and $\tilde{\xi}\in\fc^1(M,\lieg^\ast)$ such that $\xi-\de\dd\tilde{\xi}=\phi_V$. 
Choosing a reference connection $\tilde{\lambda}\in\sol_P$, 
whose existence is ensured by\ Remark\ \ref{remYMSolExist}, 
and taking a function $a\in\cc(M)$ with support inside $\supp(\chi_+)\cap\supp(\chi_-)$ such that 
$\int_Ma\,\vol=\int_M\phi(\tilde{\lambda})\,\vol$, 
we can introduce $\psi=a\,\bfone_P+\ast^{-1}(\xi\wedge\ast(\cdot-\tilde{\lambda})\in\sc(\conn(P)^\dagger)$, 
where $\bfone_P\in\sect(\conn(P)^\dagger)$ denotes the distinguished section 
whose value at any point $x\in M$ is specified by the affine map $\zeta\in\conn(P)_x\to1\in\bbR$. 
Note that, per construction, $\supp(\psi)$ in included in $\supp(\chi_+)\cap\supp(\chi_-)$. 
Exploiting the fact that $\sect(\conn(P))$ is an affine space modeled on the vector space $\f^1(M,\lieg)$, 
we show that $\ev_\psi-\ev_\phi\in\MW^\ast(\fc^1(M,\lieg^\ast))$, 
see\ Example\ \ref{exaVanConn} for the definition of $\MW^\ast:\fc^1(M,\lieg^\ast)\to\kin_P$. 
In fact, for each $\omega\in\f^1(M,\lieg)$, we have 
\begin{align*}
\ev_\psi(\tilde{\lambda}+\omega)+\left(\tilde{\xi},\MW(\tilde{\lambda}+\omega)\right)
& =\int_Ma\,\vol+\left(\xi-\de\dd\tilde{\xi},\omega\right)\\
& =\int_M\phi(\tilde{\lambda})\,\vol+(\phi_V,\omega)
=\ev_\phi(\tilde{\lambda}+\omega)\,.
\end{align*}
Note that, together with the fact that $\tilde{\lambda}$ is on-shell, 
for the second equality we exploited $\MW_V=-\de\dd$, see the comment below\ \eqref{eqMaxwellOp}. 
We conclude that $\ev_\psi+\MW^\ast(\tilde{\xi}\,)=\ev_\phi$. 
Since both $\ev_\phi$ and $\MW^\ast(\tilde{\xi}\,)$ lie in $\inv_P$, {\em cfr.}\ Example\ \ref{exaVanConn}, 
$\ev_\psi\in\inv_P$ too. 
\end{proof}

\begin{lemma}\label{lemTimeSliceConn2}
Let $G=U(1)$ and consider a principal $G$-bundle $P$ over an $m$-dimensional globally hyperbolic spacetime $M$. 
Let $O\subseteq M$ be a causally compatible open neighborhood of a Cauchy hypersurface for $M$. 
For each $\eta\in\gau_{P_O}$ there exists $\eta^\prime\in\gau_P$ such that 
the restriction of $\eta^\prime$ to $O$ differs from $\eta$ by an exact $1$-form on $O$. 
\end{lemma}

\begin{proof}
As globally hyperbolic spacetimes, both $M$ and $O$ share a Cauchy hypersurface. 
Via\ Theorem\ \ref{thmGlobHyp}, $M$ and $O$ are homotopy equivalent, 
whence the restriction from $M$ to $O$ induces an isomorphism $\coho^1(M,\bbZ)\simeq\coho^1(O,\bbZ)$. 
Recalling\ Corollary\ \ref{corGaugeShiftsExplicit}, given $\eta\in\gau_{P_O}$, its cohomology class 
$[\eta]\in\hdd^1(O,\lieg)_\bbZ$ can be represented via the isomorphism 
$\coho^1(M,\bbZ)\simeq\coho^1(O,\bbZ)$ induced by the restriction from $M$ to $O$. 
In particular, we find $[\eta^\prime]\in\hdd^1(M,\lieg)_\bbZ$ which restricts to $[\eta]$. 
Consider a representative $\eta^\prime\in[\eta^\prime]$. 
Per construction, $\eta^\prime\in\gau_P$ restricts to a $1$-form on $O$ 
which differs from $\eta$ by $\dd\chi$ for a suitable $\chi\in\c(O,\lieg)$
\end{proof}

\begin{lemma}\label{lemTimeSliceConn3}
Let $G=U(1)$ and consider a principal $G$-bundle $P$ over an $m$-dimensional globally hyperbolic spacetime $M$. 
Let $O\subseteq M$ be a causally compatible open neighborhood of a Cauchy hypersurface for $M$. 
Denote with $P_O$ the principal $G$-bundle obtained restricting $P$ to $O$. 
Let $\phi\in\sc(\conn(P)^\dagger)$ be such that $\supp(\phi)$ is included in $O$ and $\ev_\phi\in\inv_P$. 
Therefore $\ev_{\phi\vert_O}$ lies in $\inv_{P_O}$, namely $\ev_{\phi\vert_O}\in\kin_{P_O}$ 
is invariant under gauge transformations of the principal bundle $P_O$ as well. 
\end{lemma}

\begin{proof}
Let us consider $\phi\in\sc(\conn(P)^\dagger)$ according to the hypothesis. By the support properties of 
$\phi\in\sc(\conn(P)^\dagger)$, we can regard $\ev_{\phi\vert_O}$ as an element of $\kin_{P_O}$. 
Using\ Lemma\ \ref{lemTimeSliceConn2}, 
we show that $\ev_\phi\in\inv_P$ entails $\ev_{\phi\vert_O}\in\inv_{P_O}$. 
Fix a reference connection $\tilde{\lambda}\in\sect(\conn(P))$. 
Its restriction $\tilde{\lambda}\vert_O$ to $O$ gives a connection on $P_O$. 
Since $\sect(\conn(P_O))$ is an affine space modeled on $\f^1(O,\lieg)$, 
to conclude that $\ev_{\phi\vert_O}$ lies in $\inv_{P_O}$, it is enough to check that 
$\ev_{\phi\vert_O}(\tilde{\lambda}\vert_O+\omega+\eta)=\ev_{\phi\vert_O}(\tilde{\lambda}\vert_O+\omega)$ 
for all $\omega\in\f^1(O,\lieg)$ and $\eta\in\gau_{P_O}$. 
For each $\eta\in\gau_{P_O}$, exploiting\ Lemma\ \ref{lemTimeSliceConn2}, we find $\eta^\prime\in\gau_P$ 
whose restriction $\eta^\prime\vert_O$ to $O$ differs from $\eta$ by $\dd\chi$ for a suitable $\chi\in\c(O,\lieg)$. 
Recalling that $\de\phi_V=0$, {\em cfr.}\ Proposition\ \ref{prpInvAffIsCoexact}, 
given $\omega\in\f^1(O,\lieg)$, one gets the following chain of identities: 
\begin{align*}
\ev_{\phi\vert_O}(\tilde{\lambda}\vert_O+\omega+\eta)
& =\ev_{\phi\vert_O}(\tilde{\lambda}\vert_O+\omega+\eta^\prime\vert_O+\dd\chi)\\
& =\ev_\phi(\tilde{\lambda}+\eta^\prime)+(\phi_V\vert_O,\omega+\dd\chi)\\
& =\ev_\phi(\tilde{\lambda})+(\phi_V\vert_O,\omega)\\
& =\ev_{\phi\vert_O}(\tilde{\lambda}\vert_O+\omega)\,
\end{align*}
where we exploited the support properties of $\phi$, the fact that $\ev_\phi$ lies in $\inv_P$ 
and we made use of Stokes' theorem to show that the term $(\phi_V\vert_O,\dd\chi)$ vanishes. 
\end{proof}

\begin{theorem}\label{thmTimeSliceConn}\index{time-slice axiom}
Let $m\geq2$ and $G=U(1)$. The covariant functor $\PSV:\PrBun_\GHyp\to\PSymV$ 
fulfils the {\em time-slice axiom}: Consider the principal $G$-bundles $P$ and $Q$ 
over the $m$-dimensional globally hyperbolic spacetimes $M$ and respectively $N$. 
Furthermore, assume that $f:P\to Q$ is a principal bundle map covering a Cauchy morphism $\ul{f}:M\to N$. 
Then $\PSV(f):\PSV(P)\to\PSV(Q)$ is an isomorphism. 
\end{theorem}

\begin{proof}
The present setting is similar to the one of\ Theorem\ \ref{thmTimeSliceForms}. 
In particular, we have the following commutative diagram 
with principal $U(1)$-bundles in the upper part and their base manifolds in the lower part: 
\begin{equation*}
\xymatrix{
P\ar[dd]\ar[r]^\simeq\ar[rrd]_f & f(P)\ar[dd]\ar[rd]^\subseteq\\
&& Q\ar[dd]\\
M\ar[r]^\simeq\ar[rrd]_{\ul{f}} & \ul{f}(M)\ar[rd]^\subseteq\\
&& N
}
\end{equation*}
This diagram shows that to prove the theorem it is enough to check that the functor 
$\PSV:\PrBun_\GHyp\to\PSymV$ provides an isomorphism out of the inclusion of principal bundles $f(P)\to Q$. 
For convenience we denote $\ul{f}(M)$ with $O\subseteq N$. 
This inclusion induces inclusions between spaces of sections with compact support. 
All these inclusions are omitted during this proof to simplify the notation. 
Note that the principal bundle $f(P)$ over $\ul{f}(M)$ coincides with the restriction $Q_O$ of $Q$ to $O$. 

According to the hypothesis, $O$ is a neighborhood of a Cauchy hypersurface for $N$. 
Therefore we can choose Cauchy hypersurfaces $\Sigma_+$ and $\Sigma_-$ for $N$ lying inside $O$ 
and such that $\Sigma_+\cap J^-_N(\Sigma_-)=\emptyset$. This allows us to find a partition of unity 
$\{\chi_+,\chi_-\}$ subordinated to the open cover $\{I_N^+(\Sigma_-),I_N^-(\Sigma_+)\}$ of $N$. 

Consider $\ev_{[\phi]}\in\obs_Q$ and choose a representative $\ev_\phi\in\inv_Q$. 
According to\ Lemma\ \ref{lemTimeSliceConn1}, there exists $\psi\in\sc(\conn(Q_O)^\dagger)$ 
such that $\ev_\psi-\ev_\phi\in\MW^\ast(\fc^1(N,\lieg^\ast))$. In particular $\ev_\psi\in\inv_Q$. 
Lemma\ \ref{lemTimeSliceConn3} entails that $\ev_\psi$ lies also in $\inv_{Q_O}$. 
If we consider $\tilde{\psi}\in\sc(\conn(Q_O)^\dagger)$ 
such that $\ev_{\tilde{\psi}}-\ev_\phi\in\MW^\ast(\fc^1(N,\lieg^\ast))$ too, 
we deduce that $\ev_{\tilde{\psi}}-\ev_\psi$ lies in $\inv_{Q_O}\cap\MW^\ast(\fc^1(N,\lieg^\ast))$. 
Furthermore, if we choose another representative $\ev_{\phi^\prime}\in\inv_Q$ of $\ev_{[\phi]}\in\obs_Q$, 
following the procedure outlined above, we find $\psi^\prime\in\sc(\conn(Q_O)^\dagger)$ such that 
$\ev_{\psi^\prime}-\ev_\psi\in\inv_{Q_O}\cap\van_Q$. This is due to the fact 
that both $\ev_\psi-\ev_\phi$ and $\ev_{\psi^\prime}-\ev_{\phi^\prime}$ lie in 
$\MW^\ast(\fc^1(N,\lieg^\ast))\subseteq\van_Q$, {\em cfr.}\ Example\ \ref{exaVanConn}, 
while $\ev_{\phi^\prime}-\ev_\phi$ lies in $\van_Q$. Therefore, as soon as it is proved that 
$\inv_{Q_O}\cap\van_Q=\van_{Q_O}$, the procedure outlined above defines a linear map 
$K:\obs_Q\to\obs_{Q_O}$, which is the candidate to invert the presymplectic linear map 
induced by the principal bundle map $Q_O\to Q$ via $\PSV$. 

The inclusion $\inv_{Q_O}\cap\van_Q\supseteq\van_{Q_O}$ follows from the fact that 
$\sol_Q$ maps to $\sol_{Q_O}$ under restriction. For the converse inclusion, 
consider $\ev_\phi\in\inv_{Q_O}\cap\van_Q$. We have to check that $\ev_\phi$ vanishes on $\sol_{Q_O}$. 
$\sol_{Q_O}$ is an affine space modeled on the vector space of solutions to the equation $\de\dd\omega=0$, 
where $\omega$ is a $\lieg$-valued $1$-form on $O$. Since $\sol_Q$ maps to $\sol_{Q_O}$ under restriction, 
for an arbitrary choice of a reference connection $\tilde{\lambda}\in\sol_{Q}$, it is enough to show that 
$\ev_\phi(\tilde{\lambda}\vert_O+\omega)=0$ for all $\omega\in\f^1(O,\lieg)$ fulfilling $\de\dd\omega=0$. 
Given $\omega\in\f^1(O,\lieg)$ such that $\de\dd\omega=0$, via Lemma\ \ref{lemTimeSliceForms1} 
we find $\alpha\in\ftcde^1(O,\lieg)$ supported inside $\supp(\chi_+)\cap\supp(\chi_-)$ 
and $\chi\in\c(O,\lieg)$ such that $G\alpha+\dd\chi=\omega$, 
where $G$ denotes the causal propagator for $\Box$ acting on $\lieg$-valued $1$-forms on $O$. 
The support properties of $\alpha$ allow us to extend $\omega$ by zero 
to a coclosed $\lieg$-valued $1$-form on $N$ with timelike compact support, 
still denoted by $\omega\in\ftcde^1(N,\lieg)$ with a slight abuse of notation. 
In particular, we can regard $G\alpha$ as a solution of the equation $\de\dd\omega=0$ on the whole $N$, 
where now $G$ denotes the causal propagator for $\Box$ acting on $\lieg$-valued $1$-forms on $N$. 
The following chain of identities follows: 
\begin{align*}
\ev_\phi(\tilde{\lambda}\vert_O+\omega)=\ev_\phi(\tilde{\lambda}\vert_O+G\alpha\vert_O+\dd\chi)
=\ev_\phi(\tilde{\lambda}+G\alpha)=0\,.
\end{align*}
Since $\ev_\phi\in\inv_{Q_O}\cap\van_Q$, the second equality follows from $\dd\chi\in\gau_{Q_O}$, 
{\em cfr.}\ the comment after the proof of\ Lemma\ \ref{lemExpConAbSurj}, 
while the last one follows from $\tilde{\lambda}+G\alpha\in\sol_Q$. 
This shows that also the inclusion $\inv_{Q_O}\cap\van_Q\subseteq\van_{Q_O}$ holds. 

According to the first part of the proof, we have a well-defined linear map $K:\obs_Q\to\obs_{Q_O}$, 
mapping $\ev_{[\phi]}\in\obs_Q$ to $\ev_{[\psi]}\in\obs_{Q_O}$, where $\psi\in\sc(\conn(Q_O)^\dagger)$ 
is such that $\ev_\psi-\ev_\phi\in\van_Q$ for an arbitrary choice of a representative $\ev_\phi\in\ev_{[\phi]}$. 
Denote with $L:\obs_{Q_O}\to\obs_Q$ the linear presymplectic map 
induced by the principal bundle map $Q_O\to Q$ via the covariant functor $\PSV$. 
We check that $KL=\id_{\obs_{Q_O}}$. In fact, given $\ev_{[\psi]}\in\obs_{Q_O}$, 
the extension by zero of one of its representatives $\psi\in\sc(\conn(Q_O)^\dagger)$ can be used 
to represent $L\ev_{[\psi]}$. Therefore $KL\ev_{[\psi]}$ still admits $\psi\in\sc(\conn(Q_O)^\dagger)$ 
as a representative, whence $KL\ev_{[\psi]}=\ev_{[\psi]}$. 
To show that $LK=\id_{\obs_Q}$, we take $\ev_{[\phi]}\in\obs_Q$, 
consider a representative $\ev_\phi\in\ev_{[\phi]}$ and find $\psi\in\sc(\conn(Q_O)^\dagger)$ 
such that $\ev_\psi-\ev_\phi\in\van_Q$ according to the procedure outlined in the first part. 
By definition, $K\ev_{[\phi]}=\ev_{[\psi]}$. Indeed, the extension by zero of $\psi$ is 
a representative of $LK\ev_{[\phi]}$ and $\ev_\psi-\ev_\phi\in\van_Q$ per construction, 
whence $LK\ev_{[\phi]}=\ev_{[\phi]}$. 
Summing up, it turns out that $K$ is the inverse of $L$; in particular, $K$ is presymplectic as $L$ is. 
We conclude that $L$ is an isomorphism in $\PSymV$, thus completing the proof. 
\end{proof}

\subsection{Failure of locality and how to recover it}\label{subLocalityConn}
In Example\ \ref{exaNonInjConn} below, we exploit the degeneracies of the presymplectic structure 
for certain spacetime topologies, {\em cfr.}\ Remark\ \ref{exaNonTrivRadConn}, 
to exhibit a principal bundle map covering a causal embedding 
which does not induce an injective morphism via the functor $\PSV:\PrBun_\GHyp\to\PSymV$. 
This is an explicit violation of locality in the sense of\ \cite[Definition\ 2.1]{BFV03}. 
At least on the full subcategory $\PrBun_\GHypF$ of principal $U(1)$-bundles 
over $m$-dimensional globally hyperbolic spacetimes of finite type 
(for which we have explicit characterizations of the space of gauge invariant affine functionals, 
of the subspace of functionals which vanish on-shell and of the null space of the presymplectic form too), 
we will show that this model admits a quotientable subfunctor, see\ Definition\ \ref{defQSubfunctorV}, 
which allows to recover locality. However, this comes at the price of restricting the set of connections 
on which the elements of the resulting quotient can be tested. 
This amounts to imposing on connections a stricter condition then just being on-shell. 
Nonetheless, let us remind the reader that affine observables fail to detect flat connections 
and the Aharonov-Bohm effect in particular. 
For more details refer to\ Theorem\ \ref{thmInvAffDetectsCurv} and\ Remark\ \ref{remInvAffDetectsCurv}. 

Before we show the kind of situation where the locality property is explicitly violated 
by the functor $\PSV:\PrBun_\GHyp\to\PSymV$, let us stress that the most evident degeneracies 
of the presymplectic structure which are due to those elements in $\obs_P$ whose linear part is trivial 
are not related with the failure of locality. This is explained in the following remark. 

\begin{remark}\label{remMagObsConn}
Consider a principal $U(1)$-bundle map $f:P\to Q$ covering a causal embedding $\ul{f}:M\to N$ 
and take $\phi\in\sc(\conn(P)^\dagger)$ such that its linear part vanishes, namely $\phi_V=0$. 
Clearly $\ev_\phi\in\inv_P$ and, moreover, $\ev_{[\phi]}\in\rad_P$. 
One can actually express $\phi$ as $a\,\bfone_P$ for a suitable $a\in\cc(M)$, 
where $\bfone_P\in\sect(\conn(P)^\dagger)$ is the section which gives at each point the affine map 
on the corresponding fiber of $\conn(P)$ taking the constant value $1$. Therefore, 
according to\ Lemma\ \ref{lemDualConnFunctor}, $\sc(\conn(f)^\dagger)\phi=(\ul{f}_\ast a)\,\bfone_Q$. 
From this fact we deduce that $\PSV(f)\ev_{[\phi]}=0$ entails $\ev_\phi=0$. 
In fact, $\sc(\conn(f)^\dagger)\phi$ is a representative of $\PSV(f)\ev_{[\phi]}$; 
therefore, supposing that $\PSV(f)\ev_{[\phi]}$ is trivial in $\obs_Q$ 
and choosing a reference connection $\tilde{\lambda}\in\sol_Q$, for each $\omega\in\f^1(N,\lieg)$, one deduces 
\begin{equation*}
\int_Ma\,\vol=\int_N\ul{f}_\ast a\,\vol=
\big(\PSV(f)\ev_{[\phi]}\big)([\tilde{\lambda}])+\big(\ul{f}_\ast\phi_V,\omega\big)=0\,,
\end{equation*}
which means $\phi\in\triv_P$ or equivalently $\ev_\phi=0\in\kin_P$. 
Therefore, $\PSV(f)$ maps injectively elements of $\obs_P$ of the form $\ev_{[a\,\bfone_P]}$ for any $a\in\cc(M)$ 
to $\ev_{[(\ul{f}_\ast a)\,\bfone_P]}\in\obs_Q$, whence these elements in the null space 
of the presymplectic form $\tau_P$ are not those responsible for the failure of locality. 
Note that observables testing the Chern class of $P$ lie in this class, {\em cfr.}\ Remark\ \ref{remChargeObs}. 
\end{remark}

A situation where locality fails is presented below. The elements in the null space of the presymplectic structure 
which are responsible for this failure are those of the type considered in\ Example\ \ref{exaNonTrivRadConn}. 

\begin{example}[Non-injective presymplectic maps]\label{exaNonInjConn}
The argument is very similar to the one in\ Example\ \ref{exaNonInjForms} for $k=1$ and $m\geq2$. 
We consider the same globally hyperbolic spacetimes $M$ and $N$ 
presented in the example mentioned above, together with the corresponding causal embedding $\ul{f}:M\to N$. 
On top of $M$ and $N$, we take the trivial principal $U(1)$-bundles $P$ and $Q$. 
Obviously, we get a principal bundle map $f:P\to Q$ covering $\ul{f}:M\to N$ which acts trivially on the fibers. 
For $m\geq3$ we consider the gauge invariant affine functional $\ev_\phi\in\inv_P$ 
introduced in\ Example\ \ref{exaNonTrivRadConn} by means of the dual of the curvature map. 
Instead, for $m=2$ we take $\ev_\phi=F^\ast(i\theta)\in\inv_P$ with $\theta\in\fcdd^2(M)\setminus\dd\fc^1(M)$ 
as specified in the second part of\ Example\ \ref{exaNonInjForms}. In both cases, 
$\ev_{[\phi]}\in\obs_P$ is a non-trivial element of the null space $\rad_P$ of the presymplectic form $\tau_P$. 
This follows from the given references and the characterization of $\van_P$ and of $\rad_P$ 
presented in\ Proposition\ \ref{prpVanObsConn} and in\ Proposition\ \ref{prpRadConn}, 
which both apply due to the fact that $M$ is of finite type. We note that $\PSV(f)\ev_{[\phi]}\in\obs_Q$ is trivial. 
In fact, $\PSV(f)\ev_{[\phi]}$ has a representative $\ev_{\sc(\conn(f)^\dagger)\phi}\in\inv_P$ 
according to\ Theorem\ \ref{thmObsFunctorConn}. 
Since $\hcdd^2(N,\lieg^\ast)$ is trivial in the construction for the case $m\geq3$, while $[i\ul{f}_\ast\theta]=0$ 
in $\hcdd^2(N,\lieg^\ast)$ in the construction for the case $m=2$, {\em cfr.}\ Example\ \ref{exaNonInjForms}, 
there exists $\eta\in\fc^1(N,\lieg^\ast)$ such that $\dd\eta=i\ul{f}_\ast\theta$. 
Therefore, introducing $\ev_\psi=\MW^\ast(\eta)\in\van_Q$ via the dual of the affine differential operator 
$\MW:\sect(\conn(Q))\to\f^1(N,\lieg)$, {\em cfr.}\ Example\ \ref{exaVanConn}, we conclude that 
for each $\lambda\in\sect(\conn(Q))$ the following chain of identities holds: 
\begin{align*}
\ev_{\sc(\conn(f)^\dagger)\phi)}(\lambda) & =\ev_\phi\big(\sect(\conn(f))(\lambda)\big)
=\big(i\theta,F\big(\sect(\conn(f))(\lambda)\big)\big)\\
& =\big(i\theta,\ul{f}^\ast F(\lambda)\big)=\big(i\ul{f}_\ast\theta,F(\lambda)\big)
=\big(\dd\eta,F(\lambda)\big)=\ev_\psi(\lambda)\,.
\end{align*}
Note that we exploited the naturality of the curvature map, see\ Remark\ \ref{remCurvAb}. 
It follows that $\ev_{\sc(\conn(f)^\dagger(\phi))}=\ev_\psi$, whence $\PSV(f)(\ev_{[\phi]})=0$ in $\obs_Q$, 
thus showing that $\PSV(f):\PSV(P)\to\PSV(Q)$ has a non-trivial kernel. 
\end{example}

\begin{theorem}\label{thmNonLocYM}
Let $m\geq2$ and $G=U(1)$. The covariant functor $\PSV:\PrBun_\GHyp\to\PSymV$ 
violates the locality property as stated in\ \cite[Definition\ 2.1]{BFV03}, namely there exists 
a principal $G$-bundle map $f$ covering a causal embedding such that the kernel of $\PSV(f)$ is non-trivial. 
\end{theorem}

\begin{proof}
Example\ \ref{exaNonInjConn} provides a counterexample to locality 
involving connected globally hyperbolic spacetimes of finite type for spacetime dimension $m\geq3$ 
and a counterexample involving non-connected globally hyperbolic spacetimes of finite type for $m=2$. 
\end{proof}

Consider the full subcategory $\PrBun_\GHypF$ of $\PrBun_\GHyp$ 
whose objects are principal $U(1)$-bundles covering $m$-dimensional globally hyperbolic spacetimes of finite type. 
As already anticipated, we show that the covariant functor $\PSV:\PrBun_\GHyp\to\PSymV$, 
when restricted to $\PrBun_\GHypF$, admits a quotientable subfunctor which recovers locality. 

\begin{lemma}\label{lemQSubfunctorConn}
Let $m\geq2$ and $G=U(1)$. To each principal $G$-bundle $P$ over an $m$-dimensional 
globally hyperbolic spacetime $M$ of finite type, assign the vector subspace 
\begin{equation*}
\mathfrak{Q}(P)=\frac{F^\ast(\fcdd^2(M,\lieg^\ast))}{\van_P}
=\frac{F^\ast(\fcdd^2(M,\lieg^\ast))}{\MW^\ast(\fc^1(M,\lieg^\ast))}
\end{equation*}
of the null space $\rad_P$ of $\PSV(P)$ endowed with the trivial presymplectic structure. 
This assignment gives rise to a quotientable subfunctor $\mathfrak{Q}:\PrBun_\GHypF\to\PSymV$ 
of $\PSV:\PrBun_\GHypF\to\PSymV$, {\em cfr.}\ Definition\ \ref{defQSubfunctorV}. 
\end{lemma}

\begin{proof}
For each object $P$ over $M$ in $\PrBun_\GHypF$, by\ Proposition\ \ref{prpRadConn} and keeping in mind 
the assumption that $M$ is of finite type, $\mathfrak{Q}(P)$ is a presymplectic vector subspace 
(with trivial presymplectic structure) of the radical $\rad_P$ of the presymplectic form $\tau_P$ of $\PSV(P)$. 
Furthermore, naturality of $F:\sect(\conn(\cdot))\to\f^2((\cdot)_\base,\lieg)$ 
and of $\MW:\sect(\conn(\cdot))\to\f^1((\cdot)_\base,\lieg)$ entails that 
$\mathfrak{Q}:\PrBun_\GHypF\to\PSymV$ is a covariant functor. 
Therefore, to conclude that $\mathfrak{Q}:\PrBun_\GHypF\to\PSymV$ is a quotientable subfunctor 
of $\PSV:\PrBun_\GHypF\to\PSymV$ according to\ Definition\ \ref{defQSubfunctorV}, 
we only have to check that, for each morphism $f:P\to Q$ in $\PrBun_\GHypF$, 
$\PSV(f):\PSV(P)\to\PSV(Q)$ maps $\mathfrak{Q}(P)$ to $\mathfrak{Q}(Q)$. 
Actually, also this property follows from naturality 
of the curvature affine map $F:\sect(\conn(\cdot))\to\f^2((\cdot)_\base,\lieg)$, 
see\ Remark\ \ref{remCurvAb}. In fact, given $\ev_\phi=F^\ast(\eta)\in\inv_P$ 
for $\eta\in\fcdd^2(M,\lieg^\ast)$, consider $\PSV(f)\ev_{[\phi]}\in\obs_Q$, 
which has a representative $\ev_\psi=\ev_{\sc(\conn(f)^\dagger)\phi}\in\inv_Q$ 
according to\ Theorem\ \ref{thmObsFunctorConn}. 
For each connection $\lambda\in\sect(\conn(Q))$, the chain of identities presented below holds true: 
\begin{equation*}
\ev_\psi(\lambda)=\big(\eta,F\big(\sect(\conn(f))(\lambda)\big)\big)=\big(\eta,\ul{f}^\ast F(\lambda)\big)
=\big(\ul{f}_\ast\eta,F(\lambda)\big)\,.
\end{equation*}
This proves that $\ev_\psi=F^\ast(\ul{f}_\ast\eta)\in F^\ast(\fcdd^2(N,\lieg^\ast))$, 
whence the presymplectic linear map $\PSV(f):\PSV(P)\to\PSV(Q)$ restricts to 
$\mathfrak{Q}(f):\mathfrak{Q}(P)\to\mathfrak{Q}(Q)$, 
thus showing that $\mathfrak{Q}:\PrBun_\GHypF\to\PSymV$ is a quotientable subfunctor 
of $\PSV:\PrBun_\GHypF\to\PSymV$. 
\end{proof}

\begin{theorem}\label{thmLocRecovedYM}
Let $m\geq2$ and $G=U(1)$. The quotient functor $\PSV^0=\PSV/\mathfrak{Q}:\PrBun_\GHypF\to\PSymV$, 
defined according to\ Proposition\ \ref{prpQuotientFunctor}, fulfils locality, causality and the time-slice axiom. 
In particular, for each morphism $f:P\to Q$ in $\PrBun_\GHypF$, 
the corresponding morphism $\PSV^0(f):\PSV^0(P)\to\PSV^0(Q)$ in $\PSymV$ is injective. 
\end{theorem}

\begin{proof}
By\ Lemma\ \ref{lemQSubfunctorConn}, we have a quotientable subfunctor 
$\mathfrak{Q}:\PrBun_\GHypF\to\PSymV$ of $\PSV:\PrBun_\GHypF\to\PSymV$. 
Applying Proposition\ \ref{prpQuotientFunctor}, 
we get a new covariant functor $\PSV^0=\PSV/\mathfrak{Q}:\PrBun_\GHypF\to\PSymV$. 
As for causality and the time-slice axiom, these properties are inherited from $\PSV:\PrBun_\GHyp\to\PSymV$. 
Only locality has to be checked. Consider a morphism $f:P\to Q$ in $\PrBun_\GHypF$. 
We have to show that $\PSV^0(f):\PSV^0(P)\to\PSV^0(Q)$ is injective. 
Assume that $\ev_{[\phi]}\in\obs_P$ is such that $\PSV(f)\ev_{[\phi]}$ lies in $\mathfrak{Q}(Q)$.  
Applying Theorem\ \ref{thmObsFunctorConn} and\ Lemma\ \ref{lemDualConnFunctor}, 
we find $\eta\in\fcdd^2(N,\lieg^\ast)$ such that $\ev_{\sc(\conn(f)^\dagger)\phi}=F^\ast(\eta)$. 
In particular, $\PSV(f)\ev_{[\phi]}$ lies in the null space $\rad_Q$. 
Therefore $\ev_{[\phi]}$ lies in the null space $\rad_P$ as well. 
In fact, for each $\ev_{[\psi]}\in\obs_P$, we get 
\begin{equation*}
\tau_P\left(\ev_{[\phi]},\ev_{[\psi]}\right)=\tau_Q\left(\PSV(f)\ev_{[\phi]},\PSV(f)\ev_{[\psi]}\right)=0\,.
\end{equation*}
Therefore, Proposition\ \ref{prpRadConn} entails that a representative $\ev_\phi\in\inv_P$ of $\ev_{[\phi]}$ 
has linear part $\phi_V=-\de\xi$ for a suitable $\xi\in\fcdd^2(M,\lieg^\ast)$. 
It follows that $\ev_\phi-F^\ast(\xi)\in\inv_P$ has vanishing linear part. 
The map $\sc(\conn(f)^\dagger):\sc(\conn(P)^\dagger)\to\sc(\conn(Q)^\dagger)$, 
defined in Lemma\ \ref{lemDualConnFunctor}, induces a map $\kin_P\to\kin_Q$.  
We deduce that the image of $\ev_\phi-F^\ast(\xi)$ via this map coincides 
with $F^\ast(\eta-\ul{f}_\ast\xi)$ and its linear part vanishes too. As a consequence, 
$\eta-\ul{f}_\ast\xi\in\fc^1(N,\lieg^\ast)$ is both closed and coclosed, whence $\eta=\ul{f}_\ast\xi$. 
This entails that $\ev_\phi=F^\ast(\xi)$. In fact, given a reference connection $\tilde{\lambda}\in\sect(\conn(Q))$ 
and exploiting the affine structure of $\sect(\conn(P))$, it is enough to check that $\ev_\phi$ coincides 
with $F^\ast(\xi)$ on $\sect(\conn(f))(\tilde{\lambda})+\omega$ for each $\omega\in\f^1(M,\lieg)$: 
\begin{align*}
\ev_\phi & \big(\sect(\conn(f))(\tilde{\lambda})+\omega\big)
=\ev_\phi\big(\sect(\conn(f))(\tilde{\lambda})\big)+(-\de\xi,\omega)\\
& =\ev_{\sc(\conn(f)^\dagger)\phi}(\tilde{\lambda})+(\xi,-\dd\omega)
=\big(\eta,F(\tilde{\lambda})\big)+(\xi,-\dd\omega)\\
& =\big(\ul{f}_\ast\xi,F(\tilde{\lambda})\big)+(\xi,-\dd\omega)
=\big(\xi,F\big(\sect(\conn(f))(\tilde{\lambda})\big)\big)+(\xi,-\dd\omega)\\
& =\big(\xi,F\big(\sect(\conn(f))(\tilde{\lambda})+\omega\big)\big)\,.
\end{align*}
We conclude that $\ev_\phi$ lies in $F^\ast(\fcdd^2(M,\lieg^\ast))$, 
thus proving that $\PSV^0(f):\PSV^0(P)\to\PSV^0(Q)$ is injective. 
\end{proof}

We have shown that a suitable quotient of the covariant functor $\PSV:\PrBun_\GHypF\to\PSymV$ recovers locality 
at least for principal $U(1)$-bundles $P$ over $m$-dimensional globally hyperbolic spacetimes $M$ of finite type. 
However, if $\mathfrak{Q}(P)$ is not trivial,\footnote{This happens whenever $\hcdd^2(M,\lieg^\ast)$ 
is non-trivial.} for each non-zero $\ev_{[\phi]}\in\mathfrak{Q}(P)$, there exists a gauge equivalence class 
of on-shell connections $[\lambda]\in[\sol_P]$ such that $\ev_{[\phi]}([\lambda])\neq0$. 
In fact, $[\sol_P]$ separates $\obs_P$ per construction. Therefore, the quotient 
by $\mathfrak{Q}(P)$ is not compatible with the pairing\ \eqref{eqAffObs} between $\obs_P$ and $[\sol_P]$. 
To restore this pairing, one might interpret the quotient by $\mathfrak{Q}(P)$ on $\obs_P$ 
as a further restriction on connections beyond being on-shell. 
Equivalently, one might say that this amounts to restricting the on-shell condition for connections. 
To understand the new on-shell condition which is implied by the quotient by $\mathfrak{Q}$, 
one has to understand which properties of a connection $\lambda\in\sol_P$ are detected 
by $\mathfrak{Q}(P)$ for some principal $U(1)$-bundle $P$ 
over an $m$-dimensional globally hyperbolic spacetime $M$ of finite type. 

\begin{proposition}
Let $P$ be a principal $U(1)$-bundle over an $m$-dimensional globally hyperbolic spacetime $M$ of finite type 
and consider a connection $\lambda\in\sol_P$. The curvature of $\lambda$ is coexact if and only if 
$\ev_{[\phi]}([\lambda])=0$ for each $\ev_{[\phi]}\in\mathfrak{Q}(P)$, 
see\ Lemma\ \ref{lemQSubfunctorConn} for the definition of $\mathfrak{Q}(P)$. 
\end{proposition}

\begin{proof}
The proof is contained in the second part of\ Remark\ \ref{remChargeObs}. 
\end{proof}

By the last proposition, for a principal $U(1)$-bundle $P$ over an $m$-di\-men\-sio\-nal manifold $M$ of finite type, 
the evaluation of an element in $\PSV^0(P)$ on a gauge equivalence class of connections 
$[\lambda]\in[\sol_P]$ is well-defined if and only if the curvature $F(\lambda)$ 
associated to any representative of the gauge equivalence class $[\lambda]$ is coexact, 
namely $[\ast F(\lambda)]=0\in\hdd^{m-2}(M,\lieg)$. 
Therefore, one might interpret the quotient by $\mathfrak{Q}(P)$ as enforcing a different kind of on-shell condition, 
namely that connections $\lambda\in\sect(\conn(P))$ should have coexact curvature, 
not only coclosed as it is required by the equation of motion $\MW(\lambda)=\de F(\lambda)=0$. 
This new constraint on connections can be rephrased via Gauss' law saying that 
no charge should be allowed. In fact, requiring coexactness for $F(\lambda)$ 
is equivalent to state that the flux of $\ast F(\lambda)$ through any $(m-2)$-cycle in $M$ should vanish, 
{\em cfr.}\ Remark\ \ref{remChargeObs}. 

\subsection{Quantization}
In the present subsection we introduce the covariant quantum field theory $\QFT:\PrBun_\GHyp\to\CAlg$. 
This should be seen as a quantization of the covariant classical field theory $\PSV:\PrBun_\GHyp\to\CAlg$ 
describing observables for Yang-Mills connections on principal $U(1)$-bundles 
over $m$-dimensional globally hyperbolic spacetimes. Let us also stress that these observables are 
defined out of smooth gauge invariant affine functionals, see the first part of\ Section\ \ref{secAffObsYM}. 
Notice that this choice is not harmless. In fact, observables of this kind can only distinguish 
on-shell connections with different curvatures, {\em cfr.}\ Remark\ \ref{remInvAffDetectsCurv}. 
The quantum field theory is obtained composing $\PSV:\PrBun_\GHyp\to\PSymV$ 
with the restriction $\CCR:\PSymV\to\CAlg$ to presymplectic vector spaces of the quantization functor 
presented in\ Section\ \ref{secQuantization}, namely $\QFT=\CCR\circ\PSV:\PrBun_\GHyp\to\CAlg$. 
First of all we will show that $\QFT:\PrBun_\GHyp\to\CAlg$ fulfils both causality and the time-slice axiom, 
but locality is still violated after quantization. 
After that, we will restrict to the full subcategory $\PrBun_\GHypF$ of $\PrBun_\GHyp$ 
whose objects are principal $U(1)$-bundles over $m$-dimensional globally hyperbolic spacetimes of finite type. 
In this slightly restricted setting, it was shown that locality can be recovered classically by a suitable quotient. 
In fact, Theorem\ \ref{thmLocRecovedYM} provides a covariant functor $\PSV^0:\PrBun_\GHypF\to\PSymV$ 
which fulfils locality, causality and the time-slice axiom. 
The same quantization procedure, which was employed before, 
now provides a locally covariant quantum field theory $\QFT^0:\PrBun_\GHypF\to\CAlg$. 
Specifically, locality, causality and time-slice axiom are inherited 
from $\PSV^0:\PrBun_\GHypF\to\PSymV$ via $\CCR:\PSymV\to\CAlg$. 

Let $P$ be a principal $U(1)$-bundle over an $m$-dimensional globally hyperbolic spacetime $M$. 
In the following we will denote by $\weyl_{[\phi]}$ the generator of the $C^\ast$-algebra $\QFT(P)$ 
corresponding to the element $\ev_{[\phi]}\in\PSV(P)$, {\em cfr.}\ Section\ \ref{secQuantization}. 

\begin{theorem}\label{thmCovariantQFTConn}
Let $m\geq2$ and $G=U(1)$. Consider the covariant functor $\PSV:\PrBun_\GHyp\to\PSymV$ 
introduced in\ Theorem\ \ref{thmObsFunctorConn} and the covariant functor $\CCR:\PSymV\to\CAlg$ 
obtained by restricting the one in\ Theorem\ \ref{thmQuantFunctor} to presymplectic vector spaces. 
The covariant functor $\QFT=\CCR\circ\PSV:\PrBun_\GHyp\to\CAlg$ fulfils 
the quantum counterparts of both causality and the time-slice axiom: 
\begin{description}
\item[Causality] If $f:P_1\to Q$ and $h:P_2\to Q$ are principal bundle maps co\-ve\-ring causal embeddings 
$\ul{f}:M_1\to N$ and $\ul{h}:M_2\to N$ with causally disjoint images in $N$, 
namely such that $\ul{f}(M_1)\cap J_N(\ul{h}(M_2))=\emptyset$, then the $C^\ast$-subalgebras 
$\QFT(f)(\QFT(P_1))$ and $\QFT(h)(\QFT(P_2))$ of the $C^\ast$-algebra $\QFT(Q)$ commute with each other; 
\item[Time-slice axiom] If $f:P\to Q$ is a principal bundle map covering a Cauchy morphism $\ul{f}:M\to N$, 
then $\QFT(f):\QFT(P)\to\QFT(Q)$ is an isomorphism of $C^\ast$-algebras. 
\end{description}
The covariant functor $\QFT:\PrBun_\GHyp\to\CAlg$ violates locality, namely there exists a principal bundle map  
$f:P\to Q$ covering a causal embedding $\ul{f}:M\to N$ such that $\QFT(f):\QFT(P)\to\QFT(Q)$ is not injective. 
\end{theorem}

\begin{proof}
Being the composition of two covariant functors, $\QFT:\PrBun_\GHyp\to\CAlg$ is a covariant functor as well. 
Weyl relations\ \eqref{eqWeylRel} and\ Theorem\ \ref{thmCausalityConn} entail quantum causality, 
while the time-slice axiom follows from $\CCR:\PSymV\to\CAlg$ being a functor 
and from\ Theorem\ \ref{thmTimeSliceConn}. 
On account of\ Theorem\ \ref{thmNonLocYM} there exists a principal bundle map $f:P\to Q$ 
covering a causal embedding $\ul{f}:M\to N$ and an element $\ev_{[\phi]}\in\PSV(P)$ 
lying in the kernel of $\PSV(f):\PSV(P)\to\PSV(Q)$. 
From this fact, it follows that the corresponding generator $\weyl_{[\phi]}$ of the $C^\ast$-algebra $\QFT(P)$ 
is mapped to $\bbone_Q\in\QFT(Q)$ via $\QFT(f):\QFT(P)\to\QFT(Q)$. 
Therefore $\bbone_P-\weyl_{[\phi]}\in\QFT(P)$ lies in the kernel of $\QFT(f)$, {\em cfr.}\ \eqref{eqQuantMorph}. 
This is a counterexample to locality for the functor $\QFT:\PrBun_\GHyp\to\CAlg$. 
\end{proof}

\begin{remark}
Even if we will not pursue this point of view here, let us mention that, according to\ \cite[Appendix\ B]{FS14}, 
for affine field theories, such as the $U(1)$ Yang-Mills model, in each presymplectic vector space of observables, 
it is reasonable to consider the functional taking the constant value $1$ as a distinguished point. 
Accordingly, one should replace the category of presymplectic vector spaces  
with the category of pointed presymplectic vector spaces. 
Doing so, one recovers the additivity property with respect to the composition of subtheories.
This means that the covariant field theory corresponding to a multiplet of fields 
is equivalent to the covariant field theory obtained as the sum of the covariant field theories 
corresponding to each one of the fields forming the multiplet. For further details on the assignment of a suitable 
monoidal structure on the category of pointed presymplectic vector spaces, see the reference mentioned above. 
Furthermore, one might observe that the group of automorphisms of the functor $\PSV:\PrBun_\GHyp\to\PSymV$ 
includes a $\bbZ_2$-symmetry with respect to the flip of the sign of functionals 
(refer to \cite{Few13} for more information about the relation 
between automorphisms of field theory functors and symmetries of the corresponding model). 
This symmetry has no counterpart in the Lagrangian of the model and therefore it should not be there. 
Fixing a distinguished point in the presymplectic vector space removes the unexpected $\bbZ_2$-symmetry. 
Summing up, passing to pointed presymplectic vector spaces, as in\ \cite[Appendix\ B]{FS14}, 
ensures additivity under composition of subtheories 
and removes unexpected symmetries from the group of automorphisms of the functor $\PSV$. 

The quantization functor should be adapted to this modified setting. 
In particular, it should take care of the distinguished point in each pointed pre\-sym\-plec\-tic vector space. 
See\ \cite[Appendix\ B]{FS14} for the technique to implement this feature at the algebraic level. 
\end{remark}

\begin{theorem}
Let $m\geq2$ and $G=U(1)$. Consider the covariant functor $\PSV^0:\PrBun_\GHypF\to\PSymV$ 
introduced in\ Theorem\ \ref{thmLocRecovedYM} and the covariant functor $\CCR:\PSymV\to\CAlg$ 
obtained by restricting the one in\ Theorem\ \ref{thmQuantFunctor} to presymplectic vector spaces. 
The covariant functor $\QFT^0=\CCR\circ\PSV^0:\PrBun_\GHypF\to\CAlg$ fulfils 
the quantum counterparts of locality, causality and the time-slice axiom. 
In particular, for each principal bundle map  $f:P\to Q$ covering a causal embedding $\ul{f}:M\to N$ 
between $m$-dimensional globally hyperbolic spacetimes of finite type, 
$\QFT^0(f):\QFT^0(P)\to\QFT^0(Q)$ is injective. 
\end{theorem}

\begin{proof}
Causality and time-slice axiom are inherited from the corresponding properties of the functor 
$\PSV^0:\PrBun_\GHypF\to\PSymV$ exactly as in the proof of\ Theorem\ \ref{thmCovariantQFTConn}. 
Locality follows from the analogous property of $\PSV^0:\PrBun_\GHypF\to\PSymV$ 
due to\ Proposition\ \ref{prpQuantInj}. In fact, $\CCR:\PSymV\to\CAlg$ preserves injectivity of morphisms. 
\end{proof}

\section{Observables via affine characters}\label{secYangMillsChar}
We face again the problem of\ Section\ \ref{secAffObsYM}, 
namely we try to find a suitable space of regular functionals to test gauge equivalence classes 
of on-shell connections $[\sol_P]$ on a principal $U(1)$-bundle $P$ over a globally hyperbolic spacetime $M$. 
As shown in the previous section, the space $\obs_P$ is not completely satisfactory 
since it cannot detect flat connections, which are related to the Aharonov-Bohm effect from a physical point of view. 
In fact, $\obs_P$ can distinguish two gauge equivalence classes of on-shell connections if and only if 
their curvatures differ, see\ Theorem\ \ref{thmInvAffDetectsCurv} and\ Remark\ \ref{remInvAffDetectsCurv}. 
On the contrary, the class of functionals we are going to present in this section 
succeeds in separating on-shell connections up to gauge. In fact, it turns out that Wilson loops can be regarded 
as a distributional version of the class of functionals we are going to introduce.
Once again this model fits into the framework of Brunetti, Fredenhagen and Verch\ \cite{BFV03} up to locality. 
Furthermore, we will show that, similarly to the case of Maxwell $k$-forms, 
{\em cfr.}\ Subsection\ \ref{subLocFailForms}, there is no way to recover injectivity taking quotients. 
However, it is still possible to fit into the Haag-Kastler framework by fixing a target principal bundle. 
The material of this section can be also found in\ \cite{BDHS14}. 
Let us mention that, as previously, throughout this section $G=U(1)$ and $\lieg=i\bbR$. 

Consider a principal $G$-bundle $P$ over an $m$-dimensional globally hyperbolic spacetime $M$. 
Given a compactly supported section $\phi\in\sc(\conn(P)^\dagger)$ of the vector dual of the bundle of connections, 
one can consider the exponential of the functional introduced in\ \eqref{eqAffFuncYM}: 
\begin{align}\label{eqAffCharYM}
\expev_\phi:\sect(\conn(P)\to\bbC\,, && 
\lambda\mapsto\exp\big(i\ev_\phi(\lambda)\big)=\exp\left(i\int_M\phi(\lambda)\,\vol\right)\,,
\end{align}
where $\vol=\ast1$ is the canonical volume form on $M$. 
\index{affine character}Functionals of this form will be often called {\em affine characters}. 
This is due to the fact that, exploiting the affine structure of $\conn(P)$, 
for each $\lambda\in\sect(\conn(P))$ and each $\omega\in\f^1(M,\lieg)$, one has 
\begin{equation*}
\expev_\phi(\lambda+\omega)=\expev_\phi(\lambda)\,\exp\big(i(\phi_V,\omega)\big)\,,
\end{equation*}
where $(\cdot,\cdot)$ denotes the pairing between $\fc^1(M,\lieg^\ast)$ and $\f^1(M,\lieg)$. 
Affine characters on $\sect(\conn(P))$ are collected in a new space of kinematic functionals: 
\begin{equation*}
\expkin_P=\big\{\expev_\phi:\sect(\conn(P))\to\bbC\,:\;\phi\in\sc(\conn(P)^\dagger)\big\}\,.
\end{equation*}
Note that this space is a subgroup of the Abelian group of $G$-valued functions on $\sect(\conn(P))$. 
Notice further that $\sc(\conn(P)^\dagger)$ is not a good labeling space for $\expkin_P$. 
In fact, besides the usual trivial sections of $\triv_P$ introduced in\ \eqref{eqTrivYM}, 
there are also other sections lying in the kernel of the group homomorphism $\phi\mapsto\expev_\phi$ 
implicitly defined by\ \eqref{eqAffCharYM}. As one can easily check, 
these are exactly sections of the form $\phi=a\,\bfone$ for $a\in\cc(M)$ such that $\int_Ma\,\vol\in2\pi\bbZ$: 
\begin{equation}\label{eqTrivZ}
\exptriv_P=\left\{a\,\bfone\in\sc(\conn(P)^\dagger)\,:\;a\in\cc(M)\,,\;\int_Ma\,\vol\in2\pi\bbZ\right\}\\,.
\end{equation}
Note that $\exptriv_P$ is an Abelian group under addition. Furthermore, by construction, the Abelian group 
obtained taking the quotient of $\sc(\conn(P)^\dagger)$ by $\exptriv_P$ is isomorphic to $\expkin_P$. 

Since connections only matter up to gauge transformations, 
we look for gauge invariant affine characters in $\expkin_P$: 
\begin{equation*}
\expinv_P=\left\{\expev_\phi\in\expkin_P\,:\;
\expev_\phi(\lambda+\gau_P)=\{\expev_\phi(\lambda)\}\,,\;\forall\,\lambda\in\sect(\conn(P))\right\}\,.
\end{equation*}
Recall that $\gau_P$ is the Abelian group of gauge shifts introduced in\ \eqref{eqGaugeShifts}. 
The next proposition explicitly characterizes the space of gauge invariant affine characters. 

\begin{proposition}\label{prpInvExpDualIntCoho}
Let $M$ be an $m$-dimensional globally hyperbolic spacetime and consider a principal $G$-bundle $P$ over $M$. 
Each affine character $\expev_\phi\in\expkin_P$ is gauge invariant if and only if 
$\de\phi_V=0$ and ${}_\de([\phi_V],\hdd^1(M,\lieg)_\bbZ)\subseteq2\pi\bbZ$, 
where ${}_\de(\cdot,\cdot)$ denotes the pairing between $\hcde^1(M,\lieg^\ast)$ and $\hdd^1(M,\lieg)$ 
and $\hdd^1(M,\lieg)_\bbZ$ is the image of the injective homomorphism 
$2\pi i:\coho^1(M,\bbZ)\to\hdd^1(M,\lieg)$, 
see the discussion following\ Proposition\ \ref{prpGaugeShiftsSheafCoho}. In particular, 
if $\expev_\phi\in\expkin_P$ is such that $\phi_V\in\de\fc^2(M,\lieg^\ast)$, then it lies in $\expinv_P$. 

Furthermore, assuming that $M$ is of finite type, one gets a more explicit characterization of $\expinv_P$. 
Introduce the injective homomorphism 
\begin{equation*}
I:\hom(\hdd^1(M,\lieg)_\bbZ,2\pi\bbZ)\to\hcde^1(M,\lieg^\ast)\,, 
\end{equation*}
defined by ${}_\de(Iz,[\omega])=z([\omega])$ for all $[\omega]\in\hdd^1(M,\lieg)_\bbZ$. 
Denote the image of $I$ with $\hcde^1(M,\lieg^\ast)_\bbZ$. Then 
\begin{equation*}
\expinv_P=\left\{\expev_\phi\in\expkin_P\,:\;\de\phi_V=0\,,\;[\phi_V]\in\hcde^1(M,\lieg^\ast)_\bbZ\right\}\,.
\end{equation*}
\end{proposition}

\begin{proof}
We first show that the condition is sufficient. This follows from\ Corollary\ \ref{corGaugeShiftsExplicit}. 
In fact, given $\expev_\phi\in\expkin_P$ such that $\de\phi_V=0$ 
and such that ${}_\de([\phi_V],\hdd^1(M,\lieg)_\bbZ)\subseteq2\pi\bbZ$, 
for each $\lambda\in\sect(\conn(P))$ and $\eta\in\gau_P$, one has 
\begin{equation*}
\expev_\phi(\lambda+\eta)=\expev_\phi(\lambda)\,\exp\big(i(\phi_V,\eta)\big)
=\expev_\phi(\lambda)\,\exp\big(i\,{}_\de([\phi_V],[\eta])\big)=\expev_\phi(\lambda)\,.
\end{equation*}
This shows that $\expev_\phi$ lies in $\expinv_P$. 

For the converse implication, take $\expev_\phi\in\expinv_P$. 
Considering only gauge transformations of the form $\exp(\c(M,\lieg))$ we deduce that $\de\phi_V=0$. 
In fact, taking into account the argument before\ eq.\ \eqref{eqConnGaugeHomQuot}, 
for each $\lambda\in\sect(\conn(P))$ and each $\chi\in\c(M,\lieg)$, one reads 
\begin{equation*}
\expev_\phi(\lambda)\,\exp\big(-i(\de\phi_V,\chi)\big)=\expev_\phi(\lambda-\dd\chi)=\expev_\phi(\lambda)\,,
\end{equation*}
whence $(\de\phi_V,\chi)\in2\pi\bbZ$ for each $\chi\in\c(M,\lieg)$. Since $\c(M,\lieg)$ is a vector space 
and $(\cdot,\cdot):\fc^1(M,\lieg^\ast)\times\f^1(M,\lieg)\to\bbR$ is bilinear, 
$\de\phi_V=0$ and one can consider $[\phi_V]\in\hcde^1(M,\lieg^\ast)$. 
The fact that $\expev_\phi$ is gauge invariant translates into the condition 
${}_\de([\phi_V],\hdd^1(M,\lieg)_\bbZ)\subseteq2\pi\bbZ$ due to\ Corollary\ \ref{corGaugeShiftsExplicit}. 

Clearly, the condition equivalent to gauge invariance is satisfied if $\expev_\phi\in\expkin_P$ is such that 
$\phi_V\in\de\fc^2(M,\lieg^\ast)$ since in this case ${}_\de([\phi_V],\hdd^1(M,\lieg)_\bbZ)=\{0\}$. 

Let us also assume that $M$ is of finite type. 
Therefore $\hdd^1(M,\lieg)_\bbZ$ is a finitely generated free Abelian group 
which admits a $\bbZ$-module basis $\{[\omega_i]\}$ generating $\hdd^1(M,\lieg)$ over the field $\bbR$, 
see the argument after\ Proposition\ \ref{prpGaugeShiftsSheafCoho}. 
Applying Theorem\ \ref{thmPoincareDuality}, one gets an isomorphism 
$\hcde^1(M,\lieg^\ast)\to(\hdd^1(M,\lieg))^\ast$, $[\alpha]\mapsto{}_\de([\alpha],\cdot)$. 
Given $z\in\hom(\hdd^1(M,\lieg)_\bbZ,2\pi\bbZ)$, we define $Iz\in\hcde^1(M,\lieg^\ast)$ 
imposing ${}_\de(Iz,[\omega])=z([\omega])$ for each $[\omega]\in\hdd^1(M,\lieg)_\bbZ$. 
Since ${}_\de(\cdot,\cdot)$ is bilinear and non-degenerate 
and $\hdd^1(M,\lieg)_\bbZ$ generates $\hdd^1(M,\lieg)$ over $\bbR$, $Iz$ is uniquely specified. 
Therefore one gets the sought injective homomorphism 
\begin{equation*}
I:\hom(\hdd^1(M,\lieg)_\bbZ,2\pi\bbZ)\to\hcde^1(M,\lieg^\ast)\,.
\end{equation*}
Taking $\expev_\phi\in\expkin_P$ such that $\de\phi_V=0$, one realizes that $\expev_\phi\in\expinv_P$ 
if and only if $[\phi_V]\in\hcde^1(M,\lieg^\ast)$ lies in the image of the homomorphism $I$. 
In fact, given $z\in\hom(\hdd^1(M,\lieg)_\bbZ,2\pi\bbZ)$, 
by definition one has ${}_\de(Iz,\hdd^1(M,\lieg)_\bbZ)=z(\hdd^1(M,\lieg))\subseteq2\pi\bbZ$. 
For the converse, assuming that ${}_\de([\phi_V],\hdd^1(M,\lieg)_\bbZ)\subseteq2\pi\bbZ$, 
one can define $z\in\hom(\hdd^1(M,\lieg)_\bbZ,2\pi\bbZ)$ on generators 
setting $z([\omega_i])={}_\de([\phi_V],[\omega_i])$, whence, $[\phi_V]=Iz$ per construction. 
\end{proof}

The complex exponential weakens the constraint imposed by gauge invariance 
if compared to the case of $\inv_P$, see\ Proposition\ \ref{prpInvAffIsCoexact}. 
In fact, via exponentiation, all elements in $\inv_P$ define elements of $\expinv_P$, 
but, depending on the spacetime topology, there can be elements of $\expinv_P$ 
which are not of the form $\exp(i\ev_\phi))$ for any $\ev_\phi\in\inv_P$. 
For this reason, contrary to $\inv_P$, $\expinv_P$ succeeds in separating gauge equivalence classes 
of connections at least in case the base manifold is of finite type, 
{\em cfr.}\ Theorem\ \ref{thmInvAffDetectsCurv} and\ Remark\ \ref{remInvAffDetectsCurv}. 
We prove this fact in the next theorem. 
In particular, $\expinv_P$ detects flat connections 
and hence it is suitable to provide observables which can detect the Aharonov-Bohm effect. 
As we will see in the remark below, Wilson loops on $P$ can be interpreted 
as distributional counterparts of the more regular gauge invariant affine characters in $\expinv_P$. 

\begin{theorem}\label{thmInvAffCharSep}
Let $M$ be an $m$-dimensional globally hyperbolic spacetime of finite type 
and consider a principal $G$-bundle $P$ over $M$. 
Two connections $\lambda,\lambda^\prime\in\sect(\conn(P))$ are gauge equivalent, 
namely there exists $\eta\in\gau_P$ such that $\lambda+\eta=\lambda^\prime$, 
if and only if $\expev_\phi(\lambda)=\expev_\phi(\lambda^\prime)$ for all $\expev_\phi\in\expinv_P$. 
\end{theorem}

\begin{proof}
Exploiting the affine structure of $\sect(\conn(P))$, 
one finds $\omega\in\f^1(M,\lieg)$ such that $\lambda+\omega=\lambda^\prime$. 
Evaluating $\expev_\phi\in\expinv_P$ on $\lambda^\prime$, one gets 
\begin{equation*}
\expev_\phi(\lambda^\prime)=\expev_\phi(\lambda)\,\exp\big(i(\phi_V,\omega)\big)\,.
\end{equation*}
Therefore, we have to show that $\omega$ lies in $\gau_P$ 
if and only if $(\phi_V,\omega)$ is an integer multiple of $2\pi$ for each $\expev_\phi\in\expinv_P$. 

Let us start from the implication \quotes{$\Longleftarrow$}. According to\ Proposition\ \ref{prpInvExpDualIntCoho}, 
$\expev_\phi\in\expkin_P$ with $\phi_V\in\de\fc^2(M,\lieg^\ast)$ lies in $\expinv_P$. 
Therefore the condition $(\phi_V,\omega)\in2\pi\bbZ$ for all $\expev_\phi\in\expinv_P$ 
entails $(\de\alpha,\omega)\in2\pi\bbZ$ for each $\alpha\in\fc^2(M,\lieg^\ast)$ in particular, 
whence $\dd\omega=0$. Therefore one can consider $[\omega]\in\hdd^1(M,\lieg)$. 
If one could prove that $[\omega]$ lies in $\hdd^1(M,\lieg)_\bbZ$, 
then Corollary\ \ref{corGaugeShiftsExplicit} entails $\omega\in\gau_P$. 
Recalling the hypothesis and exploiting Proposition\ \ref{prpInvExpDualIntCoho}, 
we deduce that ${}_\de([\beta],[\omega])\in2\pi\bbZ$ for each $[\beta]\in\hcde^1(M,\lieg^\ast)_\bbZ$. 
Therefore ${}_\de(Iz,[\omega])\in2\pi\bbZ$ for each $z\in\hom(\hdd^1(M,\lieg)_\bbZ,2\pi\bbZ)$, 
where $I:\hom(\hdd^1(M,\lieg)_\bbZ,2\pi\bbZ)\to\hcde^1(M,\lieg^\ast)$ is defined 
in\ Proposition\ \ref{prpInvExpDualIntCoho}. As in the proof of the proposition just mentioned, 
let us consider a $\bbZ$-module basis $\{[\omega_i]\}$ of $\hdd^1(M,\lieg)_\bbZ$ 
generating $\hdd^1(M,\lieg)$ over $\bbR$, see also\ Subsection\ \ref{subGaugeConnAb}. 
We introduce the dual $\bbZ$-module basis $\{z_i\}\subseteq\hom(\hdd^1(M,\lieg)_\bbZ,2\pi\bbZ)$ 
of $\{[\omega_i]\}$ imposing $z^i([\omega_j])=2\pi\delta^i_j$ for each $i,j$. 
Since $2\pi a^i={}_\de(Iz^i,[\omega])\in2\pi\bbZ$, 
we can define $[\tilde{\omega}]=a^i[\omega_i]\in\hdd^1(M,\lieg)_\bbZ$. 
Per construction we get the following: 
\begin{equation*}
{}_\de(Iz^i,[\tilde{\omega}])=z^i([\tilde{\omega}])=z^i(a^j[\omega_j])
=a^j\,2\pi\delta^i_j=2\pi a_i={}_\de(Iz^i,[\omega])\,.
\end{equation*}
Since $\{[\omega_i]\}$ spans $\hdd^1(M,\lieg)$ over $\bbR$, 
$\{Iz^i\}$ spans $\hcde^1(M,\lieg^\ast)$ over $\bbR$. 
Therefore, the equation displayed above entails $[\omega]=[\tilde{\omega}]\in\hdd^1(M,\lieg)_\bbZ$. 

The converse implication \quotes{$\implies$} is straightforward. 
In fact, $\eta\in\gau_P$ entails $[\eta]\in\hdd^1(M,\lieg)_\bbZ$. 
Therefore, for each $\expev_\phi\in\expinv_P$, by Proposition\ \ref{prpInvExpDualIntCoho}, 
${}_\de([\phi_V],\hdd^1(M,\lieg)_\bbZ)\subseteq2\pi\bbZ$, whence $(\phi_V,\eta)\in2\pi\bbZ$ 
and $\expev_\phi(\lambda+\eta)=\expev_\phi(\lambda)$ for each $\lambda\in\sect(\conn(P))$. 
\end{proof}

\begin{remark}[Wilson loops as gauge invariant affine characters]
Let $P$ be a principal $U(1)$-bundle over an $m$-dimensional globally hyperbolic spacetime $M$ 
and consider a closed curve $\gamma:\bbT\to M$. The pull-back bundle $\gamma^\ast P$ 
is a principal $U(1)$-bundle over $\bbT$ fitting into the following commutative diagram 
where the vertical arrows are the bundle projections: 
\begin{equation*}
\xymatrix{
\gamma^\ast P\ar[r]^{\bar{\gamma}}\ar[d] & P\ar[d]\\
\bbT\ar[r]_\gamma & M
}
\end{equation*}
Since $\coho^2(\bbT,\bbZ)$ is trivial, $\gamma^\ast P$ is a trivial principal bundle. 
Thus, one can consider a global section $\sigma\in\sect(\gamma^\ast P)$. 
Using $\sigma$, one can evaluate the Wilson loop of a connection along the closed path $\gamma$. 
Given a connection $\lambda\in\sect(\conn(P))$, consider the corresponding connection form 
$\omega_\lambda\in\connf(P)$, see\ Remark\ \ref{remConnFormToConn}\,. Via pull-back 
one gets a $\lieg$-valued $1$-form $\sigma^\ast(\bar{\gamma}^\ast\omega_\lambda)\in\f^1(\bbT,\lieg)$, 
which can be integrated over $\bbT$, thus providing a group element $\expev_\gamma(\lambda)\in G=U(1)$: 
\begin{align*}
\expev_\gamma:\sect(\conn(P))\to\bbC\,, && \lambda\mapsto\expev_\gamma(\lambda)
=\exp\left(\int_{\bbT}\sigma^\ast(\bar{\gamma}^\ast\omega_\lambda)\right)\,.
\end{align*}
The map $\expev_\gamma$ does not depend on the choice of the trivialization $\sigma$ of $\gamma^\ast P$. 
In fact, two trivializations differ by a gauge transformation $f\in\c(\bbT,U(1))$, 
which shifts the exponent in the formula displayed above by an integer multiple of $2\pi i\in\lieg$. 
For the same reason $\expev_\gamma$ is a gauge invariant functional on $\sect(\conn(P))$. 
$\expev_\gamma(\lambda)$ is the standard Wilson loop of the connection $\lambda$ 
on the principal $U(1)$-bundle $P$ along the closed path $\gamma$. 
At first glance, one realizes that the exponent, regarded as a function of $\lambda$, 
is a distributional section of $\conn(P)^\dagger$ supported in $\gamma(\bbT)$. 
In particular, $\expev_\gamma$ is a gauge invariant affine character, namely the following identity holds:  
\begin{align*}
\expev_\gamma(\lambda+\theta)=\expev_\gamma(\lambda)\,\exp\left(\int_{\bbT}\gamma^\ast\theta\right)\,, 
&& \forall\,\lambda\in\sect(\conn(P))\,,\;\forall\,\theta\in\f^1(M,\lieg)\,.
\end{align*}
Notice that we will refrain from taking into account this kind of distributional sections in the following since, 
according to\ Theorem\ \ref{thmInvAffCharSep}, smooth sections of $\conn(P)^\dagger$ with compact support 
are already sufficiently many to define an Abelian group $\expinv_P$ of more regular gauge invariant affine 
characters which succeed in separating connections up to gauge, at least when the base manifold if of finite type. 
\end{remark}

Similarly to the situation of\ Section\ \ref{secAffObsYM}, there are gauge invariant affine characters 
which take the constant value $1$ on-shell (we will refer to these functionals as the \quotes{vanishing} ones). 
In fact, just by exponentiation of any $\ev_\phi\in\van_P$, 
one gets $\expev_\phi=\exp(i\ev_\phi)$ such that $\expev_\phi(\lambda)=1$ for each $\lambda\in\sol_P$. 
For this reason, we introduce the subgroup of gauge invariant affine characters vanishing on-shell 
$\expvan_P\subseteq\expinv_P$ and we quotient them out to define observables via affine characters: 
\begin{align}\label{eqAffChObsDefYM}
\expvan_P=\left\{\expev_\phi\in\expinv_P\,:\;\expev_\phi(\sol_P)=\{1\}\right\}\,, && 
\expobs_P=\frac{\expinv_P}{\expvan_P}\,.
\end{align}
Extending our convention to the present setting, we will denote elements of $\expobs_P$ by $\expev_{[\phi]}$. 
Per construction, $\expobs_P$ is naturally paired to $[\sol_P]$: 
\begin{align}\label{eqAffChObs}
\expobs_P\times[\sol_P]\to\bbC\,, && (\expobs_{[\phi]},[\lambda])\mapsto\expobs_\phi(\lambda)\,,
\end{align}
By definition of $\expvan_P$, the pairing defined above is non-degenerate in the first argument, 
namely gauge equivalence classes of on-shell connections separate points in $\expobs_P$. 
Furthermore, the same property holds for the second argument as well 
whenever the base space $M$ of the principal $U(1)$-bundle $P$ is of finite type, 
as shown by\ Theorem\ \ref{thmInvAffCharSep}. 

The definition of gauge invariant affine characters vanishing on-shell is quite implicit 
and one might be interested in having a more explicit characterization of this Abelian group. 
This result can be achieved arguing similarly to\ Proposition\ \ref{prpVanObsConn} under the assumption that 
the base manifold of the principal bundle taken into account is of finite type. 
In fact, consider $\expev_\phi\in\expvan_P$ and choose a reference connection $\tilde{\lambda}\in\sol_P$. 
Then, for each $\omega\in\f^1(M,\lieg)$ such that $\de\dd\omega=0$, one has 
\begin{equation*}
\exp\big(i(\phi_V,\omega)\big)=\expev_\phi(\tilde{\lambda})\,\exp\big(i(\phi_V,\omega)\big)
=\expev_\phi(\tilde{\lambda}+\omega)=1\,.
\end{equation*}
This entails that $(\phi_V,\omega)=0$ for each $\omega\in\f^1(M,\lieg)$ such that $\de\dd\omega=0$. 
Furthermore, $\de\phi_V=0$ according to\ Proposition\ \ref{prpInvExpDualIntCoho}. 
Therefore, we can apply Proposition\ \ref{prpVanForms} to $\phi_V$ 
and conclude that there exists $\rho\in\fc^1(M,\lieg^\ast)$ such that $\de\dd\rho=\phi_V$. 
Introducing $\expev_\psi=\exp(-i\MW^\ast(\rho))$, one can check that $\expev_\psi=\expev_\phi$. 
In fact, taking $\tilde{\lambda}\in\sol_P$ as a reference, one can span $\sect(\conn(P))$ 
simply by translating the reference $\tilde{\lambda}$ by elements of $\f^1(M,\lieg)$. 
Since both $\expev_\phi$ and $\expev_\psi$ give $1$ upon evaluation on $\tilde{\lambda}$, 
they coincide if both $\phi$ and $\psi$ have the same linear parts. 
This is indeed the case since $\psi_V=\de\dd\rho=\phi_V$. 

We collect the results presented above in the following theorem. 

\begin{theorem}\label{thmExpObsSeparateConn}
Let $M$ be an $m$-dimensional globally hyperbolic spacetime and consider a principal $G$-bundle $P$ over $M$. 
Then the space of gauge equivalence classes of connections $[\sol_P]$ separates points in $\expobs_P$ 
via the pairing $\expobs_P\times[\sol_P]\to\bbC$ defined in\ \eqref{eqAffChObs}. 
Furthermore, assuming $M$ is of finite type, the space $\expev_P$ of observables defined via affine characters 
separates points in $[\sol_P]$ via the same pairing and $\expvan_P=\exp(i\MW^\ast(\fc^1(M,\lieg^\ast)))$, 
where $\MW^\ast:\fc^1(M,\lieg^\ast)\to\kin_P$ denotes the dual of the affine differential operator 
$\MW:\sect(\conn(P))\to\f^1(M,\lieg)$, see\ Example\ \ref{exaVanConn}. 
\end{theorem}

The last theorem motivates our choice to consider gauge invariant affine characters to test connections up to gauge. 
A more direct motivation for the need of exponential functionals 
has been recently worked out in detail in\ \cite[Appendix\ A]{BSS14}. 
On account of the arguments presented in this reference, the space of gauge equivalence classes of connections 
is a Fr\'echet manifold isomorphic to the product of a torus and a Fr\'echet space. 
In particular, notice that the torus factor is the one which is responsible for the Aharonov-Bohm effect. 
These observations entail that the gauge invariant affine functionals 
introduced in\ Section\ \ref{secAffObsYM} cannot distinguish points on the torus factor 
(namely the Aharonov-Bohm effect cannot be detected)  
and that the gauge invariant affine characters considered in the present section are needed 
in order to achieve this result. 

\begin{remark}[Observables testing magnetic and electric charges]
In full analogy with\ Remark\ \ref{remChargeObs}, 
for a principal $G$-bundle $P$ over an $m$-dimensional globally hyperbolic spacetime $M$, 
in the present approach we still have observables testing both the $\dd$-cohomology class 
and the $\de$-cohomology class of the curvature $F(\lambda)\in\fdd^2(M,\lieg)\cap\fde^2(M,\lieg)$ 
associated to a gauge equivalence class of on-shell connections $[\lambda]\in[\sol_P]$. 
In fact, one can repeat the arguments of the above mentioned remark 
simply taking into account $\ev_\phi$ and $\ev_\psi$ in $\inv_P$ as defined there 
and introducing $\expev_{[\phi]}$ and $\expev_{[\psi]}$ in $\expobs_P$ 
by specifying their representatives $\exp(i\ev_\phi)$ and respectively $\exp(i\ev_\psi)$ in $\expinv_P$. 

Further details on this topic can be found in\ \cite[Section\ 6]{BDS14a} and in\ \cite[Remark\ 5.5]{BDHS14}. 
\end{remark}

Our Abelian group $\expobs_P$ of observables defined out of affine characters can be endowed 
with a presymplectic bilinear form induced by the Lagrangian\ \eqref{eqLagrangian}: 
\begin{align}\label{eqAffCharPSymFormYM}
\upsilon_P:\expobs_P\times\expobs_P\to\bbR\,, && (\expev_{[\phi]},\expev_{[\psi]})\mapsto(\phi_V,G\psi_V)_h\,.
\end{align} 
Note that this presymplectic form is defined by the same formula we had in\ Proposition\ \ref{prpPSymObsConn}. 
In fact, this formula provides a well-defined presymplectic form both on $\obs_P$ and $\expobs_P$ 
since gauge invariant functionals both in $\inv_P$ and in $\expinv_P$ have coclosed linear parts, 
{\em cfr.}\ Proposition\ \ref{prpInvAffIsCoexact} for $\inv_P$ 
and\ Proposition\ \ref{prpInvExpDualIntCoho} for $\expinv_P$. 
In particular, one can adapt the proof of\ Proposition\ \ref{prpPSymObsConn} to this case as well. 
Doing so, $(\expobs_P,\upsilon_P)$ becomes a presymplectic Abelian group. 

\sk

As for the case of the presymplectic space $(\obs_P,\tau_P)$, 
we want to characterize the null space of the presymplectic Abelian group $(\expobs_P,\upsilon_P)$: 
\begin{equation}\label{eqRadicalExpYM}
\exprad_P=\left\{\expev_{[\phi]}\in\expobs_P\,:\;\upsilon_P(\expev_{[\phi]},\expobs_P)=\{0\}\right\}\,.
\end{equation}
Actually, in view of the analysis of the quantum case, we are also interested in the following subgroup: 
\begin{equation}\label{eqCenterExpYM}
\expcnt_P=\left\{\expev_{[\phi]}\in\expobs_P\,:\;
\upsilon_P(\expev_{[\phi]},\expobs_P)\subseteq2\pi\bbZ\right\}\,.
\end{equation}
As we will see in\ Remark\ \ref{remCntAlgYM}, $\expcnt_P$ labels commuting generators of the quantum algebra. 
Notice that by definition, $\exprad_P$ is a subgroup of $\expcnt_P$. 

\begin{proposition}\label{prpRadExpYM}
Let $M$ be an $m$-dimensional globally hyperbolic spacetime and consider a principal $G$-bundle $P$ over $M$. 
Consider the presymplectic form $\upsilon_P$ on $\expobs_P$ introduced in\ \eqref{eqAffCharPSymFormYM} 
and take $\expev_\phi\in\expinv_P$. Then the following implications hold true: 
\begin{enumerate}
\item If $\phi_V\in\de(\fc^2(M,\lieg^\ast)\cap\dd\ftc^1(M,\lieg^\ast))$, then $\expev_{[\phi]}\in\exprad_P$; 
\item If $\expev_{[\phi]}\in\exprad_P$, then $\phi_V\in\de\fcdd^2(M,\lieg^\ast)$. 
\end{enumerate}
Furthermore, under the assumption that $M$ is of finite type, 
the null space $\exprad_P$ has the following explicit characterization: 
\begin{equation*}
\exprad_P
=\left\{\ev_\phi\in\inv_P\,:\;\phi_V\in\de(\fc^2(M,\lieg^\ast)\cap\dd\ftc^1(M,\lieg^\ast))\right\}/\van_P\,.
\end{equation*}
\end{proposition}

\begin{proof}
The proof of the first and of the second statements is the same as for\ Proposition\ \ref{prpRadConn}. 
In fact, the argument only relies on two facts which are shown in\ Proposition\ \ref{prpInvExpDualIntCoho}: 
\begin{enumerate}
\item If $\expev_\psi$ lies in $\expinv_P$, then it is defined out of a section $\psi\in\sc(\conn(P)^\dagger)$ 
with coclosed linear part $\psi_V\in\fcde^1(M,\lieg^\ast)$; 
\item Each $\expev_\psi\in\expkin_P$ which is defined out of a section $\psi\in\sc(\conn(P)^\dagger)$ 
with coexact linear part $\psi_V\in\de\fc^2(M,\lieg^\ast)$ actually lies in $\expinv_P$. 
\end{enumerate}

In case $M$ is of finite type, we can refine the second part of the theorem. 
Assuming that $\expev_{[\phi]}\in\exprad_P$, 
one finds $\eta\in\fcdd^2(M,\lieg^\ast)$ such that $\de\eta=\phi_V$. 
This entails that $h^{-1}(G\phi_V)$ is closed and one can consider its cohomology class 
$[h^{-1}(G\phi_V)]\in\hdd^1(M,\lieg)$,\footnote{Recall that 
$h:\lieg\to\lieg^\ast$ has been fixed once and for all in\ \eqref{eqLagrangian}.} 
Furthermore one has ${}_\de([\psi_V],[h^{-1}(G\phi_V)])=0$ for each $\expev_\psi\in\expinv_P$. 
Recalling the characterization of $\expinv_P$ provided by\ Proposition\ \ref{prpInvExpDualIntCoho} 
for $M$ of finite type, one concludes that ${}_\de([\alpha],[h^{-1}(G\phi_V)])=0$ 
for each $[\alpha]\in\hcde^1(M,\lieg^\ast)_\bbZ$. 
Since $\hcde^1(M,\lieg^\ast)_\bbZ$ generates $\hcde^1(M,\lieg^\ast)$ over $\bbR$, 
see the proof of\ Theorem\ \ref{thmInvAffCharSep}, exploiting also Theorem\ \ref{thmPoincareDuality}, 
one concludes that $[h^{-1}(G\phi_V)]=0$ in $\hdd^1(M,\lieg)$, 
meaning that there exists $\chi\in\c(M,\lieg^\ast)$ such that $\dd\chi=G\phi_V$. 
Since $\de\phi_V=0$, we deduce $\Box\chi=\de\dd\chi=0$. 
Therefore there exists $a\in\ctc(M,\lieg^\ast)$ such that $Ga=\chi$, 
whence we find $\theta\in\ftc^1(M,\lieg^\ast)$ which fulfils $\Box\theta=\phi_V-\dd a$. 
In particular $\de\theta=-a$ and hence $\phi_V=\de\dd\theta$. 
Recalling that $\phi_V=\de\eta$ with $\eta\in\fc^2(M,\lieg^\ast)$ too, one deduces that $\dd\theta=\eta$, 
thus leading to the conclusion that $\phi_V\in\de(\fc^2(M,\lieg^\ast)\cap\dd\ftc^1(M,\lieg^\ast))$. 
\end{proof}

\begin{proposition}\label{prpCntExpYM}
Let $M$ be an $m$-dimensional globally hyperbolic spacetime of finite type 
and consider a principal $G$-bundle $P$ over $M$. Consider the presymplectic structure $\upsilon_P$ 
on $\expobs_P$ introduced in\ \eqref{eqAffCharPSymFormYM} and take $\ev_\phi\in\expinv_P$. 
Under these hypotheses, one finds 
\begin{equation*}
\expcnt_P=\left\{\expev_\phi\in\expinv_P\,:\;\phi_V\in\de\fcdd^2(M,\lieg^\ast)\,,\;
[h^{-1}(G\phi_V)]\in\hdd^1(M,\lieg)_\bbZ\right\}/\expvan_P\,.
\end{equation*}
\end{proposition}

\begin{proof}
Let us show first of all the inclusion \quotes{$\supseteq$}. Given $\expev_\phi\in\expinv_P$ such that 
$\phi_V\in\de\fcdd^2(M,\lieg^\ast)$ and $[h^{-1}(G\phi_V)]\in\hdd^1(M,\lieg)_\bbZ$, 
for each $\expev_{[\psi]}\in\expinv_P$, one gets 
\begin{align*}
\upsilon_P(\expev_{[\psi]},\expev_{[\phi]})=
{}_\de([\psi_V],[h^{-1}(G\phi_V)])=z([h^{-1}(G\phi_V)])\in2\pi\bbZ\,,
\end{align*}
where $z\in\hom(\hdd^1(M,\lieg)_\bbZ,2\pi\bbZ))$ is such that $Iz=[\psi_V]$ 
as stated by Proposition\ \ref{prpInvExpDualIntCoho} under the assumption that $M$ is of finite type. 

To prove the inclusion \quotes{$\subseteq$}, consider $\expev_\phi\in\expinv_P$ such that 
$\upsilon_P(\expobs_P,\expev_{[\phi]})\subseteq2\pi\bbZ$. 
Since each $\expev_\psi\in\expkin_P$ with $\psi_V\in\de\fc^2(M,\lieg^\ast)$ lies in $\expinv_P$, 
our hypothesis entails $(\de\eta,G\phi_V)_h\in2\pi\bbZ$ for each $\eta\in\fc^2(M,\lieg^\ast)$, 
whence $h^{-1}(G\phi_V)$ is closed. In particular, one can find $\alpha\in\fc^2(M,\lieg^\ast)$ 
satisfying $\Box\alpha=\dd\phi_V$, implying that $\dd\alpha=0$. Since $\phi_V$ is coclosed, 
one also deduces $\phi_V=\de\alpha\in\de\fcdd^2(M,\lieg^\ast)$. 
According to our hypothesis and recalling Proposition\ \ref{prpInvExpDualIntCoho} for $M$ of finite type, 
one has ${}_\de([\alpha],[h^{-1}(G\phi_V)])\in2\pi\bbZ$ for each $[\alpha]\in\hcde^1(M,\lieg^\ast)_\bbZ$. 
Repeating the argument in the proof of\ Theorem\ \ref{thmInvAffCharSep}, 
one can introduce a $\bbZ$-module basis $\{[\omega_i]\}$ of $\hdd^1(M,\lieg)_\bbZ$ 
and its dual $\bbZ$-module basis $\{[\alpha^i]\}$ of $\hcde^1(M,\lieg^\ast)_\bbZ$ 
im\-po\-sing ${}_\de([\alpha^i],[\omega_j])=2\pi\delta^i_j$. 
Notice that $\{[\omega_i]\}$ generates $\hdd^1(M,\lieg)$ over $\bbR$, 
while $\{[\alpha^i]\}$ generates $\hcde^1(M,\lieg^\ast)_\bbZ$ over $\bbR$. 
Setting $2\pi a^i={}_\de([\alpha^i],[h^{-1}(G\phi_V)])\in2\pi\bbZ$, 
one can define $[\tilde{\omega}]=a^i[\omega_i]\in\hdd^1(M,\lieg)_\bbZ$ 
and compare it with $[h^{-1}(G\phi_V)]\in\hdd^1(M,\lieg)$: 
\begin{equation*}
{}_\de([\alpha^i],[\tilde{\omega}])=a^j{}_\de([\alpha^i],[\omega_j])=2\pi a^j\delta^i_j
=2\pi a^i={}_\de([\alpha^i],[h^{-1}(G\phi_V)])\,.
\end{equation*}
Since the chain of identities displayed above holds for each element of the $\bbZ$-module basis $\{[\alpha^i]\}$ 
of $\hcde^1(M,\lieg^\ast)_\bbZ$, which is also a basis over $\bbR$ of $\hcde^1(M,\lieg^\ast)$, 
and the pairing ${}_\de(\cdot,\cdot):\hcde^1(M,\lieg^\ast)\times\hdd^1(M,\lieg)\to\bbR$ 
is non-degenerate according to\ Theorem\ \ref{thmPoincareDuality}, 
one deduces that $[h^{-1}(G\phi_V)]=[\tilde{\omega}]\in\hdd^1(M,\lieg)_\bbZ$, thus concluding the proof. 
\end{proof}

\begin{example}[Non-trivial null space]\label{exaNonTrivExpRadConn}
Let $M$ be an $m$-dimensional globally hyperbolic 
with Cauchy hypersurface diffeomorphic to $\bbR\times\bbT^{m-2}$. 
On $M$ consider the trivial principal bundle $P=M\times U(1)$. 
Example\ \ref{exaNonTrivRadForms} for $k=1$ provides a $2$-form $\theta\in\fcdd^2(M)\cap\dd\ftc^1(M)$ 
such that $[\theta]\in\hcdd^2(M)$ is non-trivial. 
Exploiting the dual of the curvature map $F:\sect(\conn(P))\to\f^2(M,\lieg)$, 
one can introduce $\expev_\phi=\exp(iF^\ast(i\theta))$, {\em cfr.}\ Example\ \ref{exaCurvObs}. 
Since $\phi_V=-i\de\theta$ and exploiting Proposition\ \ref{prpRadExpYM}, 
one deduces that $\expev_{[\phi]}\in\exprad_P$. 
It remains only to check that $\expev_{[\phi]}$ is not the trivial element of $\expobs_P$. 
Observing that $M$ is of finite type and applying Theorem\ \ref{thmExpObsSeparateConn}, 
by contradiction one can assume that there exists $\eta\in\fc^1(M,\lieg^\ast)$ 
such that $\expev_\phi=\exp(i\MW^\ast(\eta))$. As a consequence, $-i\de\theta=\phi_V=-\de\dd\eta$, 
whence $\theta=\dd(-i\eta)$, thus contradicting the fact that $[\theta]\in\hcdd^2(M)$ is non-trivial. 
This shows that there are examples of principal $U(1)$-bundles over globally hyperbolic spacetimes 
for which both $\exprad_P$ and $\expcnt_P\supseteq\exprad_P$ are non-trivial. 
\end{example}

The next theorem shows that, at least up to locality, also the approach we are adopting in this section 
fits into the scheme of\ \cite[Definition\ 2.1]{BFV03} adapted to the present setting, 
where the target is the category $\PSymA$ of presymplectic Abelian groups. 

\begin{theorem}\label{thmFunctorCausTSAExpYM}\index{causality}\index{time-slice axiom}
Let $m\geq2$ and $G=U(1)$. Consider the principal $G$-bundles $P$ and $Q$ 
over the $m$-dimensional globally hyperbolic spacetimes $M$ and respectively $N$. 
Given a principal bundle map $f:P\to Q$ covering a causal embedding $\ul{f}:M\to N$, 
the linear map $\sc(\conn(f)^\dagger):\sc(\conn(P)^\dagger)\to\sc(\conn(Q)^\dagger)$ 
defined in\ Lemma\ \ref{lemDualConnFunctor} induces the presymplectic linear map 
\begin{align*}
\PSA(f):(\expobs_P,\upsilon_P)\to(\expobs_Q,\upsilon_Q)\,, 
&& \expev_{[\phi]}\mapsto\expev_{[\sect(\conn(f)^\dagger)\phi]}\,.
\end{align*}
Setting $\PSA(P)=(\expobs_P,\upsilon_P)$ for each principal $G$-bundle $P$ over an $m$-di\-men\-sio\-nal 
globally hyperbolic spacetime $M$, $\PSA:\PrBun_\GHyp\to\PSymA$ turns out to be a covariant functor 
which fulfils both causality and the time-slice axiom: 
\begin{description}
\item[Causality] Consider the principal $G$-bundles $P_1$, $P_2$ and $Q$ 
over the $m$-di\-men\-sio\-nal  globally hyperbolic spacetimes $M_1$, $M_2$ and $N$. 
Furthermore, assume $f:P_1\to Q$ and $h:P_2\to Q$ are principal bundle maps covering the causal embeddings 
$\ul{f}:M_1\to N$ and respectively $\ul{h}:M_2\to N$ whose images are causally disjoint in $N$, 
namely $\ul{f}(M_1)\cap J_N(\ul{h}(M_2))=\emptyset$. Then 
\begin{equation*}
\upsilon_Q\big(\PSA(f)\expobs_{P_1},\PSA(h)\expobs_{P_2}\big)=\{0\}\,;
\end{equation*}
\item[Time-slice axiom] Consider the principal $G$-bundles $P$ and $Q$ 
over the $m$-dimensional globally hyperbolic spacetimes $M$ and respectively $N$. 
Furthermore, assume that $f:P\to Q$ is a principal bundle map covering a Cauchy morphism $\ul{f}:M\to N$. 
Then $\PSA(f):\PSA(P)\to\PSA(Q)$ is an isomorphism. 
\end{description}
\end{theorem}

\begin{proof}
We will only sketch the proof since the arguments already presented in\ Theorem\ \ref{thmObsFunctorConn}, 
Theorem\ \ref{thmCausalityConn} and\ Theorem \ref{thmTimeSliceConn} require only minimal adjustments. 
In fact, it is straightforward to check that $\sc(\conn(f)^\dagger)$ maps $\exptriv_P$ to $\exptriv_Q$, 
{\em cfr.}\ \eqref{eqTrivZ}. In particular one gets a map 
\begin{align*}
L_f:\expkin_P\to\expkin_Q\,, && \expev_\phi\mapsto\expev_{\sect(\conn(f)^\dagger)\phi}\,.
\end{align*}
Furthermore, exponentiation of\ \eqref{eqPushforwardKinConn} entails the identity 
\begin{align*}
(L_f\expev_\phi)(\lambda)=\expev_\phi\big(\sect(\conn(f))(\lambda)\big)\,, 
&& \forall\,\lambda\in\sect(\conn(Q))\,.
\end{align*}
On account of this formula, the argument presented in the proof of\ Theorem\ \ref{thmObsFunctorConn} 
shows that $L_f$ maps $\expinv_P$ to $\expinv_Q$; moreover, it is again the naturality of $\MW$ 
which ensures that $L_f$ maps $\expvan_P$ to $\expvan_Q$ also in this case. 
Therefore the homomorphism $\PSA(f):\expobs_P\to\expobs_Q$ is well-defined. 
Since the presymplectic form is only sensitive to the linear part, 
the argument of\ Theorem\ \ref{thmObsFunctorConn} proves that 
$\PSA:(\expobs_P,\upsilon_P)\to(\expobs_Q,\upsilon_Q)$ is a presymplectic homomorphism. 
Again, the fact that $\sc(\conn(\cdot)^\dagger):\PrBun_\GHyp\to\Vec$ is a covariant functor 
entails that so is $\PSA:\PrBun_\GHyp\to\PSymA$. 

Causality goes through exactly as in the case of\ Theorem\ \ref{thmCausalityConn} 
because it only relies on the supports of the linear parts of the sections which define observables. 
For the time-slice axiom, one observes that it is possible 
to prove analogues of\ Lemma\ \ref{lemTimeSliceConn1} and of\ Lemma \ref{lemTimeSliceConn3} 
for gauge invariant affine characters $\expinv_P$. In fact, the proofs rely only on the fact that, 
for each $\expev_\phi\in\expinv_P$, $\phi_V$ is coclosed, {\em cfr.}\ Proposition\ \ref{prpInvExpDualIntCoho}. 
Using the analogues of these lemmas, one can prove the time-slice axiom in the present setting 
mimicking the proof of\ Theorem\ \ref{thmTimeSliceConn}. 
\end{proof}

\subsection{Failure of locality}
The aim of this subsection is to show that 
the covariant functor $\PSA:\PrBun_\GHyp\to\PSymA$ violates the locality property 
of\ \cite[Definition\ 2.1]{BFV03} and moreover that locality cannot be restored by quotients. 
The situation here is very close to the one for Maxwell $1$-forms, see\ Subsection\ \ref{subLocFailForms}. 
In fact, most of the results of the present subsection rely on the failure of locality for Maxwell $1$-forms. 

\begin{example}[Non-injective presymplectic homomorphisms]\label{exaNonInjExpConn}
We exploit\ Example\ \ref{exaNonInjForms} for $k=1$ and arbitrary spacetime dimension $m\geq2$. 
Consider the same globally hyperbolic spacetimes $M$ and $N$ defined therein, 
together with the causal embedding $\ul{f}:M\to N$. 
On top of $M$ and $N$, take the trivial principal $U(1)$-bundles $P$ and $Q$. 
Obviously, there is a principal bundle map $f:P\to Q$ covering $\ul{f}:M\to N$ which acts trivially on the fibers. 
For $m\geq3$ we consider $\expev_\phi\in\expinv_P$ 
introduced in\ Example\ \ref{exaNonTrivExpRadConn} by means of the dual of the curvature map. 
For $m=2$ we take $\expev_\phi=\exp(iF^\ast(i\theta))\in\expinv_P$ 
with $\theta\in(\fcdd^2(M)\cap\dd\ftc^1(M))\setminus\dd\fc^1(M)$ 
as specified in the second part of\ Example\ \ref{exaNonInjForms}. 
In both cases, $\expev_{[\phi]}\in\expobs_P$ is a non-trivial element of the null space $\exprad_P$ 
of the presymplectic form $\upsilon_P$. Specifically, Proposition\ \ref{prpInvExpDualIntCoho} entails that 
$\expev_\phi\in\expinv_P$, while Theorem\ \ref{thmExpObsSeparateConn}, 
together with the fact that $M$ is of finite type, entails that $\expev_\phi\notin\expvan_P$. 
In fact, $\expev_\phi\in\expvan_P$ would imply $\theta\in\dd\fc^1(M)$, 
which is a contradiction, see the last part of\ Example\ \ref{exaNonTrivExpRadConn}. 
We note that $\PSA(f)\expev_{[\phi]}\in\expobs_Q$ is trivial. 
In fact, choosing the representative $\expev_{\sc(\conn(f)^\dagger)\phi}\in\expinv_P$ 
and recalling that $\hcdd^2(N,\lieg^\ast)$ is trivial in the case $m\geq3$, while $[i\ul{f}_\ast\theta]=0$ 
in $\hcdd^2(N,\lieg^\ast)$ in the case $m=2$, {\em cfr.}\ Example\ \ref{exaNonInjForms}, 
one finds $\eta\in\fc^1(N,\lieg^\ast)$ such that $\dd\eta=i\ul{f}_\ast\theta$. 
Therefore, introducing $\expev_\psi=\exp(i\MW^\ast(\eta))\in\expvan_Q$, 
for each $\lambda\in\sect(\conn(Q))$ one gets the following chain of identities: 
\begin{align*}
\expev_{\sc(\conn(f)^\dagger)\phi}(\lambda) & =\expev_\phi\big(\sect(\conn(f))(\lambda)\big)
=\exp\Big(i\big(i\theta,F\big(\sect(\conn(f))(\lambda)\big)\big)\Big)\\
& =\exp\Big(i\big(i\theta,\ul{f}^\ast F(\lambda)\big)\Big)=\exp\Big(i\big(i\ul{f}_\ast\theta,F(\lambda)\big)\Big)\\
& =\exp\Big(i\big(\dd\eta,F(\lambda)\big)\Big)=\exp\Big(i\big(\MW^\ast(\eta)\big)(\lambda)\Big)\\
& =\expev_\psi(\lambda)\,.
\end{align*}
For the last computation we exploited the fact that, according to\ Remark\ \ref{remCurvAb}, 
the curvature map is a natural transformation. Therefore $\expev_{\sc(\conn(f)^\dagger)\phi}=\expev_\psi$, 
whence $\PSA(f)\expev_{[\phi]}=1$ in $\expobs_Q$. 
In particular, the kernel of $\PSA(f):\PSA(P)\to\PSA(Q)$ is non-trivial. 
\end{example}

\begin{theorem}\label{thmLocViolatedExpYM}
Let $m\geq2$ and $G=U(1)$. The covariant functor $\PSA:\PrBun_\GHyp\to\PSymA$ 
violates the locality property according to\ \cite[Definition\ 2.1]{BFV03}, namely 
there exists a principal bundle map $f$ such that $\PSA(f)$ has non-trivial kernel. 
\end{theorem}

\begin{proof}
Counterexamples to the locality property are provided by\ Example\ \ref{exaNonInjExpConn} and references therein. 
In particular, it is shown that the violation of locality occurs already in the case 
of connected globally hyperbolic spacetimes of finite type in dimension $m\geq3$, 
while a non-connected spacetime enters the counterexample for the case $m=2$. 
\end{proof}

Similarly to the case of Maxwell $1$-forms, {\em cfr.}\ Theorem\ \ref{thmLocCohoForms}, one can show that, 
for a given principal $U(1)$-bundle map $f:P\to Q$ covering a causal embedding $\ul{f}:M\to N$ 
between $m$-dimensional globally hyperbolic spacetimes of finite type, 
the failure of locality for $\PSA(f):\PSA(P)\to\PSA(Q)$ is only due to 
the kernel of $\hcdd^2(\ul{f}):\hcdd^2(M,\lieg^\ast)\to\hcdd^2(N,\lieg^\ast)$ being non-trivial. 

\begin{lemma}\label{lemLocCohoYM}
Let $m\geq2$ and $G=U(1)$. On a principal $G$-bundle $P$ over an $m$-dimensional 
globally hyperbolic spacetime $M$, consider the presymplectic homomorphism
\begin{align*}
\curvstar:\hcdd^2(M,\lieg^\ast)\to\PSA(P)\,, && [\theta]\to\expev_{[\phi]}\,,
\end{align*}
where $\expev_\phi=\exp(iF^\ast(\theta))$ for an arbitrary choice of a representative $\theta\in[\theta]$. 
Here $\hcdd^2(M,\lieg^\ast)$ is regarded as a presymplectic Abelian group 
endowed with the trivial presymplectic structure. 
$\curvstar:\hcdd^2((\cdot)_\base)\to\PSA$ defines a natural transformation 
between covariant functors taking values in the category $\PSymA$ of presymplectic Abelian groups, 
which is injective when restricted to the full subcategory $\PrBun_\GHypF$ of principal $G$-bundles 
over $m$-dimensional globally hyperbolic spacetimes of finite type. 
\end{lemma}

\begin{proof}
First of all, note that $\curvstar:\hcdd^2(M,\lieg^\ast)\to\PSA(P)$ is well-defined. 
In fact, $F^\ast:\fc^2(M,\lieg^\ast)\to\inv_P$ maps $\dd\fc^1(M,\lieg^\ast)$ 
to $\MW^\ast(\fc^1(M,\lieg^\ast))$, {\em cfr.}\ Example\ \ref{exaCurvObs} and\ Example\ \ref{exaVanConn}, 
and via exponentiation $\inv_P$ is mapped to $\expinv_P$. 
Furthermore, $\curvstar:\hcdd^2(M,\lieg^\ast)\to\PSA(P)$ preserves the presymplectic structures 
since, for $\eta,\theta\in\fcdd^2(M,\lieg^\ast)$, the linear parts of $F^\ast(\eta)$ and of $F^\ast(\theta)$ 
are $-\de\eta$ and respectively $-\de\theta$. Therefore, it follows that 
\begin{equation*}
\upsilon_P\big(\curvstar([\eta]),\curvstar([\theta])\big)=(\de\eta,G\de\theta)_h
=(\eta,G\dd\de\theta)_h=(\eta,G\Box\theta)_h=0\,.
\end{equation*}
Notice that $\curvstar:\hcdd^2(M,\lieg^\ast)\to\PSA(P)$ is injective if $M$ is of finite type: 
Assuming that $\theta\in\fcdd^2(M,\lieg^\ast)$ is such that $\curvstar([\theta])$ vanishes in $\expobs_P$, 
one finds $\eta\in\fc^1(M,\lieg^\ast)$ such that $\exp(iF^\ast(\theta))=\exp(i\MW^\ast(\eta))$, 
{\em cfr.}\ Theorem\ \ref{thmExpObsSeparateConn}, whence $\de\theta=\de\dd\eta$, 
which entails ultimately $[\theta]=[\dd\eta]=0$ in $\hcdd^2(M,\lieg^\ast)$. 
$\curvstar:\hcdd^2((\cdot)_\base,\lieg^\ast)\to\PSA$ is a natural transformation on account of the naturality 
of the curvature map $F:\sect(\conn(\cdot))\to\f^2((\cdot)_\base,\lieg)$, see\ Remark\ \ref{remCurvAb}. 
\end{proof}

\begin{theorem}\label{thmLocCohoYM}
Let $m\geq2$ and $G=U(1)$. 
Consider a principal bundle map $f:P\to Q$ covering a causal embedding $\ul{f}:M\to N$ between 
$m$-dimensional globally hyperbolic spacetimes of finite type. Then $\PSA(f):\PSA(P)\to\PSV(Q)$ is injective 
if and only if $\hcdd^2(\ul{f},\lieg^\ast):\hcdd^2(M)\to\hcdd^2(N,\lieg^\ast)$ is injective too. 
\end{theorem}

\begin{proof}
From a principal bundle map $f:P\to Q$ covering a causal embedding $\ul{f}:M\to N$ between 
$m$-dimensional globally hyperbolic spacetimes of finite type, on account of\ Lemma\ \ref{thmLocViolatedExpYM}, 
one gets the commutative diagram in $\PSymA$ displayed below, whose vertical arrows are injective: 
\begin{equation*}
\xymatrix@C=3.3em{
\hcdd^2(M,\lieg^\ast)\ar[d]_{\curvstar}\ar[r]^{\hcdd^2(\ul{f},\lieg^\ast)} 
& \hcdd^2(N,\lieg^\ast)\ar[d]^\curvstar\\
\PSA(P)\ar[r]_{\PSA(f)} & \PSA(Q)
}
\end{equation*}
Injectivity of $\PSA(f)$ entails the same property for $\hcdd^2(\ul{f},\lieg^\ast)$. 
For the converse implication, suppose that $\hcdd^2(\ul{f},\lieg^\ast)$ is injective 
and take $\expev_{[\phi]}\in\expobs_P$ in the kernel of $\PSA(f)$. 
Since Proposition\ \ref{prpSymInj} holds true also for presymplectic Abelian groups, 
namely $\ker(\PSA(f))$ is a subgroup of $\exprad_P$, we deduce that $\expev_{[\phi]}\in\exprad_P$. 
In particular, according to\ Proposition\ \ref{prpRadExpYM}, 
there exists $\theta\in\fcdd^2(M,\lieg^\ast)$ such that $-\de\theta=\phi_V$. 
By construction, $\curvstar([\theta])$ agrees with $\expev_{[\phi]}$ up to a constant phase $c\in U(1)$, 
namely $\curvstar([\theta])=c\,\expev_{[\phi]}$. 
Regarding the constant phase $c$ as a constant functional in $\expobs_P$ and observing that constant functionals 
are left unchanged by $\PSA(f)$, see\ also\ Remark\ \ref{remMagObsConn}, one deduces that 
\begin{equation*}
\curvstar([\ul{f}_\ast\theta])=\PSA(f)\big(\curvstar([\theta])\big)
=\PSA(f)(c\,\expev_{[\phi]})=c\,\PSA(f)\expev_{[\phi]}=c\,.
\end{equation*}
Since $M$ is of finite type, recalling Theorem\ \ref{thmExpObsSeparateConn}, 
the last identity entails that $\de\ul{f}_\ast\theta=\de\dd\xi$ for a suitable $\xi\in\fc^1(N,\lieg^\ast)$, 
therefore $\ul{f}_\ast\theta=\dd\xi$, whence $[\ul{f}_\ast\theta]=0$ and $c=1$. 
Yet, according to our hypothesis, $\hcdd^2(\ul{f},\lieg^\ast)$ is injective, 
hence there exists $\eta\in\fc^1(M,\lieg^\ast)$ such that $\theta=\dd\eta$. 
This shows that $\expev_{[\phi]}=1$, thus concluding the proof. 
\end{proof}

Up to now we have shown that the covariant functor $\PSA:\PrBun_\GHyp\to\PSymA$ 
describing observables defined via affine characters for the Yang-Mills model with structure group $U(1)$ 
over globally hyperbolic spacetimes violates the locality property. 
This result, at least on the full subcategory $\PrBun_\GHypF$ whose objects have base spacetimes of finite type, 
is related to the failure of locality for the cohomology functor 
$\hcdd^2((\cdot)_\base,\lieg^\ast):\PrBun_\GHyp\to\PSymA$, {\em cfr.}\ Theorem\ \ref{thmLocCohoYM}. 
The situation is very similar to the one in\ Chapter\ \ref{chMaxwell} for Maxwell $1$-forms. One might wonder 
whether locality for $\PSA$ can be recovered by suitable quotients. Below we show that this is not possible. 
The argument relies on the similar result presented in\ Theorem\ \ref{thmLocNotRecoveredForms}. 
In fact, for a given principal $U(1)$-bundle $P$ over a globally hyperbolic spacetime $M$, 
this is due to the following facts: 
\begin{enumerate}
\item Compactly supported forms defining gauge invariant linear 
functionals for Maxwell $1$-forms in the null space $\rad_M$ of the presymplectic form $\tau_M$ 
lie in $\de(\fc^2(M)\cap\dd\ftc^1(M))$, see\ Proposition\ \ref{prpRadForms}; 
\item Sections of $\sc(\conn(P)^\dagger)$ whose linear part is of the type 
$\de(\fc^2(M,\lieg^\ast)\cap\dd\ftc^1(M,\lieg^\ast))$ always provide gauge invariant affine characters in 
the null space $\exprad_P$ of the presymplectic form $\upsilon_P$, {\em cfr.}\ Proposition\ \ref{prpRadExpYM}. 
\end{enumerate}

\begin{example}\label{exaNoQuotientYM}
Consider the same geometrical setting of\ Example\ \ref{exaNoQuotientForms} in the case $k=1$. 
However, on top of each globally hyperbolic spacetime, take the trivial principal $U(1)$-bundle. 
Therefore, both in dimension $m\geq3$ and in dimension $m=2$, one has three trivial principal $U(1)$-bundles 
$P$, $Q$ and $R$ respectively covering the globally hyperbolic spacetimes $M$, $N$ and $O$ 
introduced in the first part of\ Example\ \ref{exaNoQuotientForms} for $m\geq3$ and in the second part for $m=2$. 
Furthermore, one has the obvious principal bundle maps $f:P\to Q$ and $h:P\to R$ 
covering the causal embeddings $\ul{f}:M\to N$ and $\ul{h}:M\to O$ 
described in the example mentioned above.\footnote{Note the slightly different notation compared 
to\ Example\ \ref{exaNoQuotientForms}. The causal embeddings, which are indicated there as $f$ and $h$, 
are denoted here by $\ul{f}$ and $\ul{h}$, while $f$ and $h$ label now the principal bundle maps 
covering $\ul{f}$ and $\ul{h}$ which act as the identity on the fibers.} 

Let us first consider the case $m\geq3$. 
The topology of $M$ has been chosen in such a way that, according to Example \ref{exaNonTrivRadConn}, 
one can exhibit a non-constant element $\expev_{[\phi]}\in\exprad_P$ of the null space. 
In particular, notice that $\hcdd^2(M,\lieg^\ast)\simeq\bbR$. 
The topology of $O$ is such that $\hcdd^2(O,\lieg^\ast)$ vanishes, 
whence $\PSA(h)\expev_{[\phi]}=1$ since it is represented by $\exp(i\MW^\ast(\eta))$ 
for a suitable $\eta\in\fc^1(M,\lieg^\ast)$ and such functional is trivial on-shell. 
For $m=2$, we refer to\ Example\ \ref{exaNonInjExpConn} to exhibit a non-trivial element 
$\expev_{[\phi]}\in\exprad_P$ of the null space which vanishes when mapped to $\expobs_R$ via $\PSA(h)$. 
In both cases $m\geq3$ and $m=2$, $N$ has compact Cauchy hypersurfaces, 
hence $\exprad_Q$ contains only constant functionals and actually even $\expcnt_Q$ is such, 
{\em cfr.}\ Proposition\ \ref{prpRadExpYM} and\ Proposition\ \ref{prpCntExpYM}. 
Therefore, either $\PSA(f)\expev_{[\phi]}$ is constant or $\PSA(f)$ maps $\expev_{[\phi]}$ out of $\expcnt_Q$. 
However, the topology of $N$ is such that $\hcdd^2(\ul{f},\lieg^\ast)$ is injective 
(also an isomorphism for $m=2$, in fact), therefore, by\ Theorem\ \ref{thmLocCohoYM}, 
$\PSA(f)$ is injective, whence $\PSA(f)\expev_{[\phi]}\notin\expcnt_Q$. 
\end{example}

To state the next theorem, we need to adapt the definition of a (quotientable) subfunctor to the case 
where the target is the category $\PSymA$ of presymplectic Abelian groups. 
In fact, simply replacing presymplectic vector spaces with presymplectic Abelian groups 
everywhere in\ Definition\ \ref{defQSubfunctorV}, one gets the corresponding notion of a (quotientable) subfunctor 
and, moreover, rephrasing Proposition\ \ref{prpQuotientFunctor}, one can introduce 
the quotient functor whenever a quotientable subfunctor of a covariant functor to $\PSymA$ is assigned. 

\begin{theorem}\label{thmLocNotRecoveredYM}
Let $m\geq2$ and $G=U(1)$. The covariant functor $\PSA:\PrBun_\GHyp\to\PSymA$ 
has no quotientable subfunctor $\mathfrak{Q}:\PrBun_\GHyp\to\PSymA$ 
which recovers the locality property of\ \cite[Definition\ 2.1]{BFV03}, namely such that 
$\PSA/\mathfrak{Q}(f)$ is injective for each principal $G$-bundle map $f$ covering a causal embedding. 
\end{theorem}

\begin{proof}
Let $m\geq2$ and by contradiction, assume there exists a quotientable subfunctor 
$\mathfrak{Q}:\PrBun_\GHyp\to\PSymA$ of $\PSA:\PrBun_\GHyp\to\PSymV$ 
such that the locality property is recovered taking the quotient by $\mathfrak{Q}$. 
Example\ \ref{exaNoQuotientYM} shows that, in both cases $m\geq3$ and $m=2$, 
one gets the diagram in $\PrBun_\GHyp$ displayed below, where all principal bundles are trivial: 
\begin{equation*}
\xymatrix{
Q && R\\
&P\ar[lu]^f\ar[ru]_h
}
\end{equation*}
Furthermore, we know that there exists a non-trivial element $\expev_{[\phi]}$ of $\PSA(P)$ 
which lies in the kernel of $\PSA(h):\PrBun_\GHyp\to\PSymA$. 
Since the quotient by $\mathfrak{Q}$ is supposed to recover locality, 
$\expev_{[\phi]}$ must be contained inside $\mathfrak{Q}(P)$. 
However, $\mathfrak{Q}:\PrBun_\GHyp\to\PSymA$ being a subfunctor of $\PSA$, 
$\mathfrak{Q}(f):\mathfrak{Q}(P)\to\mathfrak{Q}(Q)$ coincides 
with the restriction of $\PSA(f):\PSA(P)\to\PSA(Q)$ to $\mathfrak{Q}(P)$, hence $\PSA(f)\expev_{[\phi]}$ 
should lie in $\mathfrak{Q}(Q)$. This is in contrast with the assumption that $\mathfrak{Q}$ is quotientable. 
In fact, Example\ \ref{exaNoQuotientYM} shows that 
$\PSA(f)\expev_{[\phi]}$ is not even contained in $\expcnt_Q\supseteq\exprad_Q$. 
\end{proof}

\begin{remark}\label{remNoLocQuotExpYM}
Theorem\ \ref{thmLocNotRecoveredYM} proves that 
locality cannot be recovered by quotientable subfunctors of $\PSA:\PrBun_\GHyp\to\PSymA$. 
In particular, for $m\geq3$, this happens even restricting $\PSA$ to the full subcategory of $\PrBun_\GHyp$ 
whose objects have connected $m$-dimensional globally hyperbolic spacetimes as bases. 
Note that this includes the most prominent physical situation, namely $m=4$. 
\end{remark}

\subsection{Recovering isotony \backtick{a} la Haag-Kastler}
In the previous subsection we exhibited the failure of locality 
for the covariant functor $\PSA:\PrBun_\GHyp\to\PSymA$ 
describing affine character observables for the Yang-Mills model with structure group $U(1)$ 
over globally hyperbolic spacetimes. Furthermore, Theorem\ \ref{thmLocNotRecoveredYM} shows that 
the locality property cannot be recovered taking the quotient by a suitable quotientable subfunctor. 
We want to show that isotony can be recovered in the Haag-Kastler framework \cite{HK64}. 
The approach is basically the same as the one in\ Subsection\ \ref{subHKForms} for Maxwell $k$-forms. 
The material of this subsection can be found also in\ \cite[Subsection\ 6]{BDHS14}. 

Let us fix a target principal $U(1)$-bundle $P$ over an $m$-dimensional globally hyperbolic spacetime $M$. 
Instead of looking at the category $\PrBun_\GHyp$, we restrict the functor $\PSA$ to the subcategory 
$\PrBun_P$ of $\PrBun_\GHyp$ whose objects are restrictions $P_O$ of the fixed target principal bundle $P$ 
to causally compatible open subsets $O\subseteq M$. Note that each object comes together 
with an inclusion $\iota_{O\,M}:P_O\to P$, which is indeed a morphism in $\PrBun_\GHyp$. 
Furthermore, morphisms in $\PrBun_P$ are only those induced by inclusions, 
namely there is a morphism $\iota_{O\,O^\prime}$ from $P_O$ to $P_{O^\prime}$ 
if and only if $O\subseteq O^\prime$ and this is the principal bundle map 
covering the causal embedding specified by the inclusion. 
In particular, each morphism $\iota_{O\,O^\prime}$ from $P_O$ to $P_{O^\prime}$ in $\PrBun_P$ 
provides the commutative diagram in the category $\PrBun_P$ displayed below: 
\begin{equation}\label{eqPrBunUniverse}
\xymatrix{
& P\\
P_O\ar[rr]_{\iota_{O\,O^\prime}}\ar[ru]^{\iota_{O\,M}} && P_{O^\prime}\ar[lu]_{\iota_{O^\prime\,M}}
}
\end{equation}
With an abuse of notation, we denote the restriction of $\PSA:\PrBun_\GHyp\to\PSymA$ 
to the subcategory $\PrBun_P$ by the same symbol, 
namely we write $\PSA:\PrBun_P\to\PSymA$ for the restricted functor. 
Our aim is to find a quotientable subfunctor of $\PSA:\PrBun_P\to\PSymA$ 
such that the quotient functor, besides causality and the time-slice axiom, fulfils isotony too, 
namely such that each morphism in $\PrBun_P$ induces an injective morphism in $\PSymA$. 
The existence of such quotientable subfunctor is shown below. 
The proof essentially relies on the observation that the category $\PrBun_P$ has $P$ as its terminal object. 

\begin{lemma}\label{lemKerSubfunctorYM}
Let $m\geq2$, $G=U(1)$ and consider a principal $G$-bundle $P$ 
over an $m$-dimensional globally hyperbolic spacetime $M$. 
The following assignment defines a quotientable subfunctor $\mfker_P:\PrBun_P\to\PSymA$ 
of the covariant functor $\PSA:\PrBun_P\to\PSymA$: 
To each object $P_O$ in $\PrBun_P$, assign the object $\mfker_P(P_O)=\ker(\PSA(\iota_{O\,M}))$ 
(endowed with the trivial presymplectic structure) in $\PSymA$ 
and, to each morphism $\iota_{O\,O^\prime}:P_O\to P_{O^\prime}$ in $\PrBun_P$, 
assign the morphism $\mfker_P(\iota_{O\,O^\prime}):\mfker_P(P_O)\to\mfker_P(P_{O^\prime})$ in $\PSymA$ 
obtained as the restriction of $\PSA(\iota_{O\,O^\prime}):\PSA(P_O)\to\PSA(P_{O^\prime})$ to $\mfker_P(P_O)$. 
\end{lemma}

\begin{proof}
For each object $P_O$ in $\PrBun_P$, one has also a morphism $\iota_{O\,M}$ to the terminal object $P$. 
Introducing the Abelian group $\ker(\PSA(\iota_{O\,M}))$ and recalling that Proposition\ \ref{prpSymInj} 
has an identical formulation for the category $\PSymA$ of presymplectic Abelian groups, 
one deduces that $\ker(\PSA(\iota_{O\,M}))$ is a subgroup of $\exprad_{P_O}$. 
Since the restriction to $\exprad_{P_O}$ of the presymplectic structure of $\PSV(P_O)$ is trivial, 
it is natural to define $\mfker_P(P_O)$ as the Abelian group $\ker(\PSV(\iota_{O\,M}))$ 
endowed with the trivial presymplectic structure. 
Furthermore, since by construction, for each object $P_O$ in $\PrBun_P$, 
$\mfker_P(P_O)$ is a subgroup of the null space of the presymplectic form of $\PSA(P_O)$, 
if $\mfker:\PrBun_P\to\PSymA$ were a subfunctor of $\PSA:\PrBun_P\to\PSymA$, then it would be quotientable. 

So far we dealt with objects in $\PrBun_P$. We still have to take care of morphisms. 
Consider a morphism $\iota_{O\,O^\prime}:P_O\to P_{O^\prime}$ in $\PrBun_P$. 
Since $P$ is the terminal object in $\PrBun_P$, one gets the commutative diagram in the category $\PrBun_P$ 
which is displayed in\ \eqref{eqPrBunUniverse}, whence the covariant functor $\PSA:\PrBun_P\to\PSymA$ 
provides a corresponding commutative diagram in the category $\PSymA$: 
\begin{equation*}
\xymatrix{
& \PSA(P)\\
\PSA(P_O)\ar[rr]_{\PSA(\iota_{O\,O^\prime})}\ar[ru]|{\PSA(\iota_{O\,M})} 
&& \PSA(P_{O^\prime})\ar[lu]|{\PSA(\iota_{O^\prime\,M})}
}
\end{equation*}
It follows that $\PSA(\iota_{O\,O^\prime})$ 
maps the kernel of $\PSA(\iota_{O\,M})$ to the kernel of $\PSA(\iota_{O^\prime\,M})$. 
The presymplectic structures involved being trivial, 
$\mfker_P(\iota_{O\,O^\prime}):\mfker_P(P_O)\to\mfker_P(P_{O^\prime})$ 
obtained from $\PSA(\iota_{O\,O^\prime})$ by restriction to the kernel 
of $\PSA(\iota_{O\,M})$ is indeed a presymplectic homomorphism. 
Clearly the functorial properties of $\PSA:\PrBun_P\to\PSymA$ are inherited by $\mfker_P:\PrBun_P\to\PSymA$. 
Therefore this is a subfunctor of $\PSA:\PrBun_P\to\PSymA$ and hence it is also quotientable 
according to the first part of the proof. 
\end{proof}

\begin{theorem}\label{thmIsotonyYM}\index{isotony}
Let $m\geq2$, $G=U(1)$ and consider a principal $G$-bundle $P$ 
over an $m$-dimensional globally hyperbolic spacetime $M$. 
The covariant functor $\PSA_P:\PrBun_P\to\PSymA$ defined as the quotient of $\PSA:\PrBun_P\to\PSymA$ 
by its quotientable subfunctor $\mfker_P:\PrBun_P\to\PSymA$ ({\em cfr.}\ Lemma\ \ref{lemKerSubfunctorYM}) 
satisfies the isotony property, 
namely $\PSA_P$ assigns an injective morphism in $\PSymA$ to each morphism in $\PrBun_P$. 
Furthermore, both causality and the time-slice axiom hold for $\PSA_P:\PrBun_P\to\PSymA$. 
\end{theorem}

\begin{proof}
The counterparts of\ Proposition\ \ref{prpQuotientFunctor} and of\ Lemma\ \ref{lemKerSubfunctorForms} 
for the category $\PSymA$ of presymplectic Abelian groups ensure that 
$\PSA_P:\PrBun_P\to\PSymA$ is a well-defined covariant functor. 
Injectivity of $\PSA_P(\iota_{O\,M})$ for each morphism $\iota_{O\,M}$ in $\PrBun_P$ holds by construction. 
For a morphism $\iota_{O\,O^\prime}$ in $\PrBun_P$, we have a commutative diagram in $\PSymA$ 
coming from the one in\ \eqref{eqPrBunUniverse} via the functor $\PSA_P$: 
\begin{equation*}
\xymatrix{
& \PSA_P(P)=\PSA(P)\\
\PSA_P(P_O)\ar[rr]_{\PSA_P(\iota_{O\,O^\prime})}\ar[ru]|{\PSA_P(\iota_{O\,M})} 
&& \PSA_P(P_{O^\prime})\ar[lu]|{\PSA_P(\iota_{O^\prime\,M})}
}
\end{equation*}
Since we already know that $\PSA_P(\iota_{O\,M})$ is injective, 
$\PSA_P(\iota_{O\,O^\prime})$ must be injective too since the diagram is commutative. This proves isotony. 
Both causality and the time-slice axiom are inherited 
from the corresponding properties of the functor $\PSA:\PrBun\to\PSymA$, 
{\em cfr.}\ Theorem\ \ref{thmFunctorCausTSAExpYM}. 
\end{proof}

\begin{remark}
The procedure to obtain a functor $\PSA_P:\PrBun_P\to\PSymA$ which fulfils isotony 
has an interpretation similar to\ Remark\ \ref{remIsotonyForms}. 
In short, we fix a target principal $U(1)$-bundle $P$ over an $m$-dimensional globally hyperbolic spacetime $M$ 
and we perform a suitable quotient of the covariant functor $\PSA:\PrBun_P\to\PSymA$ 
which provides $\PSA_P:\PrBun_P\to\PSymA$ fulfilling isotony. 
On each principal subbundle $P_O$ obtained restricting $P$ to a causally compatible open subset $O$ of $M$, 
this procedure identifies all observables in $\mfker_P(P_O)=\ker(\PSA(\iota{O\,M}))$ 
with the constant functional $1$. 
In case there are non-trivial observables in the subgroup of $\mfker_P(P_O)$ of $\PSA(P_O)$
there are also connections in $\sol_{P_O}$ which give non-trivial results upon evaluation on such observables. 
This follows from the definition of $\expvan_{P_O}$, 
{\em cfr.}\ \eqref{eqAffChObsDefYM} and also Theorem\ \ref{thmExpObsSeparateConn}. 
Therefore the quotient enforces implicitly a stricter on-shell condition for connections on a subregion, 
whose meaning is the following: A connection $\lambda\in\sol_{P_O}$ is on-shell 
(in the stricter sense implied by the quotient which recovers isotony for the fixed target $P$) 
if there does not exist any observable $\expev_{[\phi]}\in\PSA(P_O)$ with trivial image in $\PSA(P)$ 
which give a non-trivial outcome upon evaluation on $\lambda$, 
namely $\expev_{[\phi]}([\lambda])$ should be equal to $1$ for each $\expev_{[\phi]}\in\mfker_P(P_O)$. 
Basically, we are imposing that on subregions, observables should only provide information 
about field configurations which is also available in the full system $P$. 
Further details can be found in\ \cite[Section\ 7]{BDHS14}. 
\end{remark}

\subsection{Quantization}\label{subQuantAffCharYM}
We conclude with the quantization of the classical covariant field theory $\PSA:\PrBun_\GHyp\to\PSymA$ 
describing observables for the Yang-Mills model with structure group $U(1)$ 
over $m$-dimensional globally hyperbolic spacetimes. 
The covariant functor $\PSA$ is defined in\ Theorem\ \ref{thmFunctorCausTSAExpYM}, 
where it is also shown that it fulfils both causality and the time-slice axiom. 
Yet, locality is violated on account of\ Theorem\ \ref{thmLocViolatedExpYM} 
and moreover no quotient can recover this property, {\em cfr.}\ Theorem\ \ref{thmLocNotRecoveredYM}. 
Quantization is performed composing with the covariant functor $\CCR:\PSymA\to\CCR$, 
namely we define the covariant functor $\QFT=\CCR\circ\PSA:\PrBun_\GHyp\to\CAlg$. 
We show below that $\QFT:\PrBun_\GHyp\to\CAlg$ is a covariant quantum field theory, 
namely the quantum counterparts of both causality and the time-slice axiom hold. 
The violation of locality is still there after quantization and no proper quotient can be performed to recover locality. 
However, fixing a target principal $U(1)$-bundle $P$ over an $m$-dimensional globally hyperbolic spacetime $M$, 
one can obtain a covariant functor $\PSA_P:\PrBun_P\to\PSymA$ out of $\PSA:\PrBun_\GHyp\to\PSymA$ 
which fulfils isotony, causality and the time-slice axiom, see\ Theorem\ \ref{thmIsotonyYM}. 
The quantization of this functor provides a covariant functor $\QFT_P=\CCR\circ\PSA_P:\PrBun_P\to\CAlg$ 
fulfilling the quantum counterparts of these properties. 

Notice that, for a principal $U(1)$-bundle $P$ over an $m$-dimensional globally hyperbolic spacetime $M$, 
the generator of $\QFT(P)$ corresponding to an element $\expev_{[\phi]}\in\PSA(P)$ 
will be denoted by $\weyl_{[\phi]}$. 

\begin{remark}[Center of the $C^\ast$-algebra $\QFT(P)$]\label{remCntAlgYM}
Let $P$ be a principal $U(1)$-bundle over an $m$-dimensional globally hyperbolic spacetime $M$. 
The generators $\weyl_{[\phi]}$ of the $C^\ast$-algebra $\QFT(P)$ 
corresponding to elements $\expev_{[\phi]}\in\expcnt_P$, {\em cfr.}\ eq.\ \eqref{eqCenterExpYM}, 
generate the center to the algebra $\QFT(P)$. 
In fact, given $\expev_{[\phi]}\in\expcnt_P$ and $\expev_{[\psi]}\in\PSA(P)$, 
one can consider the corresponding algebra generators $\weyl_{[\phi]}$ and $\weyl_{[\psi]}$ and conclude that 
\begin{equation*}
\weyl_{[\phi]}\,\weyl_{[\psi]}
=\myexp^{-\frac{i}{2}\upsilon_P\big(\weyl_{[\phi]},\weyl{[\psi]}\big)}\weyl_{[\phi+\psi]}
=\myexp^{-i\upsilon_P\big(\weyl_{[\phi]},\weyl{[\psi]}\big)}\weyl_{[\psi]}\,\weyl_{[\phi]}
=\weyl_{[\psi]}\,\weyl_{[\phi]}\,.
\end{equation*}
This result follows from the fact that, by definition of $\expcnt_P$, 
$\upsilon_P(\expev_{[\phi]},\expev_{[\psi]})\in\bbZ$ 
for each $\expev_{[\phi]}\in\expcnt_P$ and for each $\expev_{[\psi]}\in\PSA(P)$. 
\end{remark}

\begin{theorem}
Let $m\geq2$ and $G=U(1)$. Consider the covariant functor $\PSA:\PrBun_\GHyp\to\PSymA$ 
introduced in\ Theorem\ \ref{thmFunctorCausTSAExpYM} 
and the covariant functor $\CCR:\PSymA\to\CAlg$ presented in\ Theorem\ \ref{thmQuantFunctor}. 
The covariant functor $\QFT=\CCR\circ\PSA:\GHyp\to\CAlg$ fulfils 
the quantum counterparts of both causality and the time-slice axiom: 
\begin{description}
\item[Causality] If $f:P_1\to Q$ and $h:P_2\to Q$ are principal bundle maps covering causal embeddings 
$\ul{f}:M_1\to N$ and $\ul{h}:M_2\to N$ with causally disjoint images in $N$, 
namely such that $\ul{f}(M_1)\cap J_N(\ul{h}(M_2))=\emptyset$, then the $C^\ast$-subalgebras 
$\QFT(f)(\QFT(P_1))$ and $\QFT(h)(\QFT(P_2))$ of the $C^\ast$-algebra $\QFT(Q)$ commute with each other; 
\item[Time-slice axiom] If $f:P\to Q$ is a principal bundle map covering a Cauchy morphism $\ul{f}:M\to N$, 
then $\QFT(f):\QFT(P)\to\QFT(Q)$ is an isomorphism of $C^\ast$-algebras. 
\end{description}
The covariant functor $\QFT:\PrBun_\GHyp\to\CAlg$ violates locality, namely there exists a principal bundle map  
$f:P\to Q$ covering a causal embedding $\ul{f}:M\to N$ such that $\QFT(f):\QFT(P)\to\QFT(Q)$ is not injective. 
\end{theorem}

\begin{proof}
The composition of covariant functors is a covariant functor, therefore $\QFT:\GHyp\to\CAlg$ is a covariant functor. 
Causality and the time-slice axiom follow respectively from the Weyl relations\ \eqref{eqWeylRel} 
and from functoriality, keeping Theorem\ \ref{thmFunctorCausTSAExpYM} in mind. 
Recalling the classical violation of locality, {\em cfr.}\ Theorem\ \ref{thmLocViolatedExpYM}, 
there exists a principal bundle map $f:P\to Q$ covering a causal embedding $\ul{f}:M\to N$ 
and an element $\expev_{[\phi]}\in\PSA(P)$ in the kernel of $\PSA(f):\PSA(P)\to\PSA(Q)$. 
According to the quantization procedure presented in\ Section\ \ref{secQuantization}, 
the Weyl generator $\weyl_{[\phi]}\in\QFT(P)$ corresponding to $\expev_{[\phi]}$ 
is mapped to $\bbone_Q\in\QFT(Q)$ by $\QFT(f):\QFT(P)\to\QFT(Q)$. 
Therefore $\bbone_P-\weyl{[\phi]}\in\QFT(P)$ lies in the kernel of $\QFT(f)$, 
whence locality is violated at the quantum level too. 
\end{proof}

In the following we show that there exists no proper quotientable subfunctor of $\QFT:\PrBun_\GHyp\to\CAlg$ 
such that locality can be recovered taking the quotient. The notion of a (quotientable) subfunctor 
in the $C^\ast$-algebraic context, as well as the quotient functor, 
can be found in\ Definition\ \ref{defQSubfunctorAlg} and in\ Proposition\ \ref{prpQuotientFunctorAlg}. 
Notice that we consider only proper quotientable subfunctors in order to avoid the situation in which, 
after performing the quotient, the resulting $C^\ast$-algebra is trivial on certain objects of $\PrBun_\GHyp$. 

\begin{theorem}
Let $m\geq2$ and $G=U(1)$. The covariant functor $\QFT:\PrBun_\GHyp\to\CAlg$ 
has no proper quotientable subfunctor $\mathfrak{Q}:\PrBun_\GHyp\to\CAlg$ 
which recovers locality in the sense of\ \cite[Definition\ 2.1]{BFV03}, 
namely such that $\QFT/\mathfrak{Q}(f)$ is injective for each principal bundle map $f:P\to Q$ 
covering a causal embedding $\ul{f}:M\to N$. 
\end{theorem}

\begin{proof}
We prove this theorem by contradiction. Assume that there exists a quotientable subfunctor 
$\mathfrak{Q}:\PrBun_\GHyp\to\CAlg$ of $\QFT:\PrBun_\GHyp\to\CAlg$ 
such that $\QFT/\mathfrak{Q}:\PrBun_\GHyp\to\CAlg$ fulfils locality. 
According to the proof of\ Theorem\ \ref{thmLocNotRecoveredYM}, we have a diagram in $\PrBun_\GHyp$ 
\begin{equation*}
\xymatrix{
Q && R\\
&P\ar[lu]^f\ar[ru]_h
}
\end{equation*}
where all principal bundles are trivial. 
Furthermore, there exists a non-trivial element $\expev_{[\phi]}$ of $\PSA(P)$ in the kernel of $\PSA(h)$ 
such that $\PSA(f)\expev_{[\omega]}$ does not lie in the null space $\exprad_Q$ of $\PSA(Q)$, 
see\ \eqref{eqRadicalExpYM}, and actually not even in the subgroup $\expcnt_Q$ of observables 
which generate the center, see\ \eqref{eqCenterExpYM}. 
On the one hand, $\QFT(h):\QFT(P)\to\QFT(R)$ maps $\bbone_P-\weyl_{[\phi]}\in\QFT(P)$ to zero. 
Since the quotient by $\mathfrak{Q}$ is supposed to recover locality, 
$\bbone_P-\weyl_{[\phi]}$ must lie in the $\ast$-ideal $\mathfrak{Q}(P)$. 
On the other hand, introducing $\expev_{[\phi^\prime]}=\PSA(f)\expev_{[\phi]}\in\PSA(Q)$, there exists 
$\expev_{[\psi]}\in\PSA(Q)$ such that $\upsilon_Q(\expev_{[\psi]},\expev_{[\phi^\prime]})\notin2\pi\bbZ$. 
Weyl relations\ \eqref{eqWeylRel} entail the following identity: 
\begin{equation}\label{eqQuantumNoLocQuotYM}
\bbone_Q-\weyl_{[\phi^\prime]}-\myexp^{i\upsilon_Q(\expev_{[\phi^\prime]},\ev_{[\psi]})}\,
\weyl_{[-\psi]}\,\big(\bbone_Q-\weyl_{[\phi^\prime]}\big)\,\weyl_{[\psi]}
=\big(1-\myexp^{i\upsilon_Q(\expev_{[\phi^\prime]},\expev_{[\psi]})}\,\big)\,\bbone_Q\,.
\end{equation}
Note that this is a non-zero multiple of the identity since $\upsilon_Q(\expev_{[\psi]},\expev_{[\phi^\prime]})$ 
does not lie in $2\pi\bbZ$ according to the argument presented above. 

We already know that $\bbone_P-\weyl_{[\phi]}$ lies in the $\ast$-ideal $\mathfrak{Q}(P)$. 
Furthermore $\mathfrak{Q}$ is a quotientable subfunctor per assumption.
Therefore $\bbone_Q-\weyl_{[\phi^\prime]}$ lies in the $\ast$-ideal $\mathfrak{Q}(Q)$. 
In particular, $\weyl_{[-\psi]}\,(\bbone_Q-\weyl_{[\phi^\prime]})\,\weyl_{[\psi]}$ 
must lie in $\mathfrak{Q}(Q)$ as well, and so does the difference between the two, 
namely the term on the left hand side of\ \eqref{eqQuantumNoLocQuotYM}. 
Yet, this is a non-zero multiple of the identity, therefore $\mathfrak{Q}(Q)$ is a $\ast$-ideal of $\QFT(Q)$
containing the unit $\bbone_Q$, whence $\mathfrak{Q}(Q)=\QFT(Q)$, 
contradicting the hypothesis that $\mathfrak{Q}(Q)$ is a proper $\ast$-ideal of $\QFT(Q)$. 
\end{proof}

In dimension $m\geq3$, the covariant quantum field theory $\QFT:\GHyp\to\CAlg$ violates locality already 
for principal $U(1)$-bundles over connected $m$-dimensional globally hyperbolic spacetimes 
and, moreover, no quotient can recover locality in this restricted context too. 
This follows from the proof of\ Theorem\ \ref{thmLocViolatedExpYM} and\ Remark\ \ref{remNoLocQuotExpYM}. 

It remains only to discuss isotony for a fixed target principal $U(1)$-bundle $P$ 
over an $m$-dimensional globally hyperbolic spacetime $M$. This is the content of the next theorem. 

\begin{theorem}\index{isotony}
Let $m\geq2$ and $G=U(1)$. 
Take a principal $G$-bundle $P$ over an $m$-dimensional globally hyperbolic spacetime $M$. 
Consider the covariant functor $\PSA_P:\PrBun_P\to\PSymA$ introduced in\ Theorem\ \ref{thmIsotonyYM} 
and the covariant functor $\CCR:\PSymA\to\CAlg$ presented in\ Theorem\ \ref{thmQuantFunctor}. 
The covariant functor $\QFT_P=\CCR\circ\PSA_P:\PrBun_P\to\CAlg$ satisfies 
isotony, causality and the time-slice axiom. In particular, 
$\QFT_P$ assigns an injective morphism in $\CAlg$ to each morphism in $\PrBun_P$. 
\end{theorem}

\begin{proof}
Isotony, causality and the time-slice axiom are inherited from the corresponding properties 
of the covariant functor $\PSA_P:\PrBun_P\to\CAlg$, {\em cfr.}\ Theorem\ \ref{thmIsotonyYM}. 
In particular, isotony follows from the fact that, according to\ Proposition\ \ref{prpQuantInj}, 
the functor $\CCR:\PSymA\to\CAlg$ preserves injectivity of morphisms, 
while quantum causality is due to the Weyl relations \eqref{eqWeylRel} 
and the quantum time-slice axiom follows from functoriality. 
\end{proof}

\chapter{Concluding remarks}
In this thesis we have investigated the implementation of general local covariance for Abelian gauge field theories. 
In particular, explicit counterexamples to the locality axiom have been shown 
for both Maxwell $k$-forms (generalizing the vector potential of electromagnetism) 
and the Yang-Mills model with structure group $U(1)$. 
In the case of Maxwell $k$-forms, a no-go theorem shows the impossibility 
to recover locality by means of a quotient (in a suitable functorial sense). 
Still, one can specify a target globally hyperbolic spacetime 
in order to implement Haag-Kastler isotony by means of a suitable quotient. 
In the case of the $U(1)$ Yang-Mills model, we adopted two different approaches, 
both producing locality breaking covariant field theories at the classical and at the quantum level. 
In the first approach, gauge invariant affine functionals are used to define observables. 
As a drawback, this class of observables fails in detecting flat connections (connected to the Aharonov-Bohm effect). 
However, this approach leaves room for a quotient by electric flux observables, which recovers locality. 
In the second approach, affine functionals are replaced by affine characters, namely their exponentiated counterparts. 
The exponential weakens the constraints imposed by gauge invariance, 
thus leading to observables (resembling Wilson loops for $U(1)$-connections) 
which can successfully detect all gauge equivalence classes of on-shell connections (including the flat ones). 
The presence of observables detecting flat connections reduces the degeneracies of the presymplectic structure. 
Eventually, when one considers the class of observables defined via gauge invariant affine characters, 
this leads to a no-go theorem for the existence of a quotient which recovers locality. 
Yet, fixing a principal $U(1)$-bundle over a globally hyperbolic spacetime as the target, 
one can still perform a quotient on the classical theory which recovers isotony 
producing both classical and quantum field theories in the Haag-Kastler framework. 

The results presented so far, together with the analysis of higher Abelian gauge theories developed in\ \cite{BSS14}, 
motivate our interest in extending the investigations about the locality axiom of general local covariance 
to more complicated gauge field theories such as non-Abelian Yang-Mills models and the Einstein equation. 
Despite the availability of approaches based on BRST-BV techniques, which allow to tackle the analysis 
of non-linear gauge field theories in a perturbative fashion\ \cite{Hol08, FR12, FR13, BFR13}, 
it seems that these techniques do not capture fine details of the structure of the relevant gauge groups, 
such as topologically non-trivial gauge transformations. 
These features, however, play a central role in our analysis of the $U(1)$ Yang-Mills model 
as well as in the analysis of\ \cite{BSS14}, 
especially in relation to finding the most appropriate class of functionals to distinguish field configurations up to gauge, 
which we regard as a physically and mathematically well-motivated criterion 
to establish whether a space of observables can be reasonably associated to a given field theoretical model. 
These observations suggest the importance of a non-perturbative treatment of non-linear field theories 
at the classical level. In fact, this approach seems to be more effective in highlighting topological aspects of the model. 
Later, quantization might be performed in a perturbative fashion, 
but still keeping track of the topological features arising from the classical theory. 

Because of the high level of complication introduced by non-linearity of the Yang-Mills and Einstein dynamics, 
as a first step, one might approach the problem for less ambitious, yet more tractable, models, 
such as linear gauge field theories. 
A general framework to deal with this class of models has been developed in\ \cite{HS13}. 
As an example, one might study the locality axiom of general local covariance for 
the linearization of a non-Abelian Yang-Mills model around an exact solution of the full-fledged field equation. 
Similarly, one might analyze in detail the linearized Einstein equation\ \cite{FH13, BDM14, Kha14a}, 
looking for degeneracies of the relevant presymplectic structure, which might eventually originate violations of locality. 

To approach the difficulties one encounters with non-linear partial differential equations and the lack 
of a well-developed framework to deal with them in a non-perturbative manner already at the classical level, 
it seems appropriate to start with non-linear equations of hyperbolic type, 
avoiding for the moment the issues related to non-hyperbolic ones, which are typical of gauge field theories. 
A suitable starting point might be to consider non-linear $\sigma$-models (also known as wave maps). 
In this case, some existence and uniqueness theorems for solutions are available 
in the mathematical literature\ \cite{CH80, CB87, Mul07} (especially in $1+1$ world-sheet dimensions). 
In fact, the control on the solution theory is a central prerequisite 
for the analysis in the framework of general local covariance. 
Notice that, dealing with $\sigma$-models, one has a new source for field configurations of a topological nature, 
such as the Aharonov-Bohm ones for the $U(1)$ Yang-Mills model discussed in\ Chapter\ \ref{chYangMills}. 
In fact, depending on the choice of the world-sheet and of the target manifold, 
one can obtain solutions classified by topological information. 
As a simple example, one might imagine to have a cylinder both as the world-sheet and the target manifold. 
In this case one finds wave maps which carry information about the winding number. 
Furthermore, $\sigma$-models provide the framework to analyze two different kinds of localization, 
one on the target manifold and the other on the world-sheet. 

To conclude, let us also mention that it would be interesting to check if locality is restored in more realistic models, 
where electromagnetism is not considered on its own, as in our case, 
but rather where $U(1)$-connections are coupled to Dirac fields in the Maxwell-Dirac system of electrodynamics. 
The following is a rough argument motivating why locality might be recovered 
when interactions between electromagnetic and charged fields are taken into account: 
As shown in this thesis, see also\ \cite{SDH14, BDS14a, BDHS14}, 
the source of the lack of locality in pure electromagnetism is due to the non-trivial electric fluxes 
of topological origin produced by certain on-shell field configurations. 
Introducing dynamical charged fields, one has also local currents as sources for electric fluxes. 
This might milden the issues with locality related to observables which are measuring the electric flux, 
eventually leading to a locally covariant field theory describing electrodynamics. 
However, if one does not take the principal $U(1)$-bundle as part of the data 
providing the background for the dynamics of the electromagnetic field, 
but rather one wants to describe bundle-connection pairs as dynamical objects on globally hyperbolic spacetimes, 
thus leaving the Chern class of the principal $U(1)$-bundle free as in\ \cite{BSS14}, 
one encounters more observables causing the failure of locality, 
namely those which are measuring the magnetic flux, {\em cfr.}\ \cite[Example\ 6.9]{BSS14}. 
It seems unlikely that charged fields can cure this source of non-locality too 
and it is still not clear which kind of fields might be of help in this respect. 

For these reasons, it seems of major importance to go in the direction of a non-perturbative treatment of 
non-linear field theories to improve our understanding of general local covariance and of its locality axiom in particular. 

\paginavuota


\backmatter


\cleardoublepage\phantomsection
\addcontentsline{toc}{chapter}{Index}
\printindex


\cleardoublepage\phantomsection
\addcontentsline{toc}{chapter}{Bibliography}
\bibliography{thesis} 

\providecommand{\bysame}{\leavevmode\hbox to3em{\hrulefill}\thinspace}
\providecommand{\MR}{\relax\ifhmode\unskip\space\fi MR }
\providecommand{\MRhref}[2]{%
  \href{http://www.ams.org/mathscinet-getitem?mr=#1}{#2}
}
\providecommand{\href}[2]{#2}
\begin{thebibliography}{ACMM86}

\bibitem[AB83]{AB83}
Michael~F. Atiyah and Raoul Bott, \emph{The Yang-Mills equations over Riemann
  surfaces}, Phil. Trans. R. Soc. Lond. A \textbf{308} (1983), 523.

\bibitem[ACMM86]{ACMM86}
M.~Cristina Abbati, Renzo Cirelli, Alessandro Mani{\`a}, and Peter~W. Michor,
  \emph{Smoothness of the action of the gauge transformation group on
  connections}, J. Math. Phys. \textbf{27} (1986), 2469.

\bibitem[Ati57]{Ati57}
Michael~F. Atiyah, \emph{Complex analytic connections in fibre bundles}, Trans.
  Amer. Math. Soc \textbf{85} (1957), 181.

\bibitem[B{\"a}r14]{Bar14}
Christian B{\"a}r, \emph{Green-hyperbolic operators on globally hyperbolic
  spacetimes}, Commun. Math. Phys. Online First (2014).

\bibitem[Bau14]{Bau14}
Helga Baum, \emph{Eichfeldtheorie}, Springer Berlin Heidelberg, 2014.

\bibitem[BDF09]{BDF09}
Romeo Brunetti, Michael D{\"u}tsch, and Klaus Fredenhagen, \emph{Perturbative
  algebraic quantum field theory and the renormalization groups}, Adv. Theor.
  Math. Phys. \textbf{13} (2009), 1541.

\bibitem[BDH13]{BDH13}
Marco Benini, Claudio Dappiaggi, and Thomas-Paul Hack, \emph{Quantum field
  theory on curved backgrounds -- A primer}, Internat. J. Modern Phys. A
  \textbf{28} (2013), 1330023.

\bibitem[BDHS14]{BDHS14}
Marco Benini, Claudio Dappiaggi, Thomas-Paul Hack, and Alexander Schenkel,
  \emph{A {$C^\ast$}-algebra for quantized principal {$U(1)$}-connections on
  globally hyperbolic Lorentzian manifolds}, Commun. Math. Phys. \textbf{332}
  (2014), 477.

\bibitem[BDM14]{BDM14}
Marco Benini, Claudio Dappiaggi, and Simone Murro, \emph{Radiative observables
  for linearized gravity on asymptotically flat spacetimes and their boundary
  induced states}, J. Math. Phys. \textbf{55} (2014), 082301.

\bibitem[BDS14a]{BDS14a}
Marco Benini, Claudio Dappiaggi, and Alexander Schenkel, \emph{Quantized
  Abelian principal connections on Lorentzian manifolds}, Commun. Math. Phys.
  \textbf{330} (2014), 123.

\bibitem[BDS14b]{BDS14b}
\bysame, \emph{Quantum field theory on affine bundles}, Ann. Henri Poincar{\'e}
  \textbf{15} (2014), 171.

\bibitem[BEE96]{BEE96}
John~K. Beem, Paul Ehrlich, and Kevin Easley, \emph{Global Lorentzian
  geometry}, CRC Press, 1996.

\bibitem[Ben14]{Ben14}
Marco Benini, \emph{Optimal space of linear classical observables for Maxwell
  $k$-forms via spacelike and timelike compact de Rham cohomologies}, {\tt
  arXiv}{\rm :1401.7563 [math-ph]} (2014).

\bibitem[BF00]{BF00}
Romeo Brunetti and Klaus Fredenhagen, \emph{Microlocal analysis and interacting
  quantum field theories: Renormalization on physical backgrounds}, Commun.
  Math. Phys. \textbf{208} (2000), 623.

\bibitem[BF09]{BF09}
Christian B{\"a}r and Klaus Fredenhagen (eds.), \emph{Quantum field theory on
  curved spacetimes}, Springer Berlin Heidelberg, 2009.

\bibitem[BFK96]{BFK96}
Romeo Brunetti, Klaus Fredenhagen, and Manfred K{\"o}hler, \emph{The microlocal
  spectrum condition and Wick polynomials of free fields on curved spacetimes},
  Commun. Math. Phys. \textbf{180} (1996), 633.

\bibitem[BFLR12]{BFR12}
Romeo Brunetti, Klaus Fredenhagen, and Pedro Lauridsen~Ribeiro, \emph{Algebraic
  structure of classical field theory I: Kinematics and linearized dynamics for
  real scalar fields}, {\tt arXiv}{\rm :1209.2148 [math-ph]} (2012).

\bibitem[BFR13]{BFR13}
Romeo Brunetti, Klaus Fredenhagen, and Katarzyna Rejzner, \emph{Quantum gravity
  from the point of view of locally covariant quantum field theory}, {\tt
  arXiv}{\rm :1306.1058 [math-ph]} (2013).

\bibitem[BFV03]{BFV03}
Romeo Brunetti, Klaus Fredenhagen, and Rainer Verch, \emph{The generally
  covariant locality principle: A new paradigm for local quantum field theory},
  Commun. Math. Phys. \textbf{237} (2003), 31.

\bibitem[BG12a]{BG12a}
Christian B{\"a}r and Nicolas Ginoux, \emph{Global differential geometry},
  ch.~Classical and quantum fields on Lorentzian manifolds, pp.~359--400,
  Springer Berlin Heidelberg, 2012.

\bibitem[BG12b]{BG12b}
\bysame, \emph{Quantum field theory and gravity}, ch.~CCR- versus
  CAR-quantization on curved spacetimes, pp.~183--206, Springer Basel, 2012.

\bibitem[BGP07]{BGP07}
Christian B{\"a}r, Nicolas Ginoux, and Frank Pf{\"a}ffle, \emph{Wave equations
  on Lorentzian manifolds and quantization}, European Mathematical Society,
  2007.

\bibitem[BHR04]{BHR04}
Ernst Binz, Reinhard Honegger, and Alfred Rieckers, \emph{Construction and
  uniqueness of the $C^\ast$-Weyl algebra over a general pre-symplectic space},
  J. Math. Phys. \textbf{45} (2004), 2885.

\bibitem[BR87]{BR87}
Ola Bratteli and Derek~W. Robinson, \emph{Operator algebras and quantum
  statistical mechanics 1}, Springer Berlin Heidelberg, 1987.

\bibitem[Bre97]{Bre97}
Glen~E. Bredon, \emph{Sheaf Theory}, Springer New York, 1997.

\bibitem[BS05]{BS05}
Antonio~N. Bernal and Miguel S{\'a}nchez, \emph{Smoothness of time functions
  and the metric splitting of globally hyperbolic space-times}, Commun. Math.
  Phys. \textbf{257} (2005), 43.

\bibitem[BS06]{BS06}
\bysame, \emph{Further results on the smoothability of Cauchy hypersurfaces and
  Cauchy time functions}, Lett. Math. Phys. \textbf{77} (2006), 183.

\bibitem[BS07]{BS07}
\bysame, \emph{Globally hyperbolic spacetimes can be defined as 'causal'
  instead of 'strongly causal'}, Class. Quant. Grav. \textbf{24} (2007), 745.

\bibitem[BSS14]{BSS14}
Christian Becker, Alexander Schenkel, and Richard~J. Szabo, \emph{Differential
  cohomology and locally covariant quantum field theory}, {\tt arXiv}{\rm
  :1406.1514 [hep-th]} (2014).

\bibitem[BT82]{BT82}
Raoul Bott and Loring~W. Tu, \emph{Differential forms in algebraic topology},
  Springer New York, 1982.

\bibitem[CB87]{CB87}
Yvonne Choquet-Bruhat, \emph{Global existence theorems for hyperbolic harmonic
  maps}, Ann. Inst. H. Poincar{\'e} Phys. Th{\'e}or. \textbf{46} (1987), 97.

\bibitem[CH80]{CH80}
Gu~Chao-Hao, \emph{On the Cauchy problem for harmonic maps defined on
  two-dimensional Minkowski space}, Comm. Pure Appl. Math. \textbf{33} (1980),
  727.

\bibitem[CRV13]{CRV13}
Fabio Ciolli, Giuseppe Ruzzi, and Ezio Vasselli, \emph{QED representation for
  the net of causal loops}, {\tt arXiv}{\rm :1305.7059 [math-ph]} (2013).

\bibitem[Dap11]{Dap11}
Claudio Dappiaggi, \emph{Remarks on the Reeh-Schlieder property for higher spin
  free fields on curved spacetimes}, Rev. Math. Phys. \textbf{23} (2011), 1035.

\bibitem[DHP09]{DHP09}
Claudio Dappiaggi, Thomas-Paul Hack, and Nicola Pinamonti, \emph{The extended
  algebra of observables for Dirac fields and the trace anomaly of their
  stress-energy tensor}, Rev. Math. Phys. \textbf{21} (2009), 1241.

\bibitem[Dim80]{Dim80}
Jonathan~D. Dimock, \emph{Algebras of local observables on a manifold}, Commun.
  Math. Phys. \textbf{77} (1980), 219.

\bibitem[Dim92]{Dim92}
\bysame, \emph{Quantized electromagnetic field on a manifold}, Rev. Math. Phys.
  \textbf{4} (1992), 223.

\bibitem[DL12]{DL12}
Claudio Dappiaggi and Benjamin Lang, \emph{Quantization of Maxwell's equations
  on curved backgrounds and general local covariance}, Lett. Math. Phys.
  \textbf{101} (2012), 265.

\bibitem[dR84]{dR84}
Georges de~Rham, \emph{Differentiable manifolds}, Springer Berlin Heidelberg,
  1984.

\bibitem[DS13]{DS13}
Claudio Dappiaggi and Daniel Siemssen, \emph{Hadamard states for the vector
  potential on asymptotically flat spacetimes}, Rev. Math. Phys. \textbf{25}
  (2013), 1350002.

\bibitem[Edw65]{Edw65}
Robert~E. Edwards, \emph{Functional analysis: Theory and applications}, Holt,
  Rinehart and Winston, 1965.

\bibitem[Few13]{Few13}
Christopher~J. Fewster, \emph{Endomorphisms and automorphisms of locally
  covariant quantum field theories}, Rev. Math. Phys. \textbf{25} (2013),
  1350008.

\bibitem[FH13]{FH13}
Christopher~J. Fewster and David~S. Hunt, \emph{Quantization of linearized
  gravity in cosmological vacuum spacetimes}, Rev. Math. Phys. \textbf{25}
  (2013), 1330003.

\bibitem[FL14]{FL14}
Christopher~J. Fewster and Benjamin Lang, \emph{Dynamical locality of the free
  Maxwell field}, {\tt arXiv}{\rm :1403.7083 [math-ph]} (2014).

\bibitem[FP03]{FP03}
Christopher~J. Fewster and Michael~J. Pfenning, \emph{A quantum weak energy
  inequality for spin one fields in curved space-time}, J. Math. Phys.
  \textbf{44} (2003), 4480.

\bibitem[FR12]{FR12}
Klaus Fredenhagen and Katarzyna Rejzner, \emph{Batalin-Vilkovisky formalism in
  the functional approach to classical field theory}, Commun. Math. Phys.
  \textbf{314} (2012), 93.

\bibitem[FR13]{FR13}
\bysame, \emph{Batalin-Vilkovisky formalism in perturbative algebraic quantum
  field theory}, Commun. Math. Phys. \textbf{317} (2013), 697.

\bibitem[Fri75]{Fri75}
Friedrich~G. Friedlander, \emph{The wave equation on a curved space-time},
  Cambridge University Press, 1975.

\bibitem[FS13]{FS13}
Felix Finster and Alexander Strohmaier, \emph{Gupta-Bleuler quantization of the
  Maxwell field in globally hyperbolic space-times}, {\tt arXiv}{\rm :1307.1632
  [math-ph]} (2013).

\bibitem[FS14]{FS14}
Christopher~J. Fewster and Alexander Schenkel, \emph{Locally covariant quantum
  field theory with external sources}, {\tt arXiv}{\rm :1402.2436 [math-ph]}
  (2014).

\bibitem[FV12]{FV12}
Christopher~J. Fewster and Rainer Verch, \emph{Dynamical locality and
  covariance: What makes a physical theory the same in all spacetimes?}, Ann.
  Henri Poincar{\'e} \textbf{13} (2012), 1613.

\bibitem[GHV72]{GHV72}
Werner Greub, Stephen Halperin, and Ray Vanstone, \emph{Connections, curvature,
  and cohomology - Volume 1: De Rham Cohomology of manifolds and vector
  bundles}, Academic Press, 1972.

\bibitem[GW14]{GW14}
Christian G{\'e}rard and Micha{\l} Wrochna, \emph{Hadamard states for the
  linearized Yang-Mills equation on curved spacetime}, {\tt arXiv}{\rm
  :1403.7153 [math-ph]} (2014).

\bibitem[Haa96]{Haa96}
Rudolf Haag, \emph{Local quantum physics}, Springer Berlin Heidelberg, 1996.

\bibitem[Har11]{Har11}
G{\"u}nter Harder, \emph{Lectures on algebraic geometry I}, Springer Fachmedien
  Wiesbaden, 2011.

\bibitem[HK64]{HK64}
Rudolf Haag and Daniel Kastler, \emph{An algebraic approach to quantum field
  theory}, J. Math. Phys. \textbf{5} (1964), 848.

\bibitem[Hol07]{Hol07}
Stefan Hollands, \emph{The operator product expansion for perturbative quantum
  field theory in curved spacetime}, Commun. Math. Phys. \textbf{273} (2007),
  1–36.

\bibitem[Hol08]{Hol08}
\bysame, \emph{Renormalized quantum Yang-Mills fields in curved spacetimes},
  Rev. Math. Phys. \textbf{20} (2008), 1033.

\bibitem[HS13]{HS13}
Thomas-Paul Hack and Alexander Schenkel, \emph{Linear bosonic and fermionic
  quantum gauge theories on curved spacetimes}, Gen. Rel. Grav. \textbf{45}
  (2013), 877.

\bibitem[Hus94]{Hus94}
Dale Husemoller, \emph{Fibre bundles}, Springer New York, 1994.

\bibitem[HW01]{HW01}
Stefan Hollands and Robert~M. Wald, \emph{Local Wick polynomials and time
  ordered products of quantum fields in curved spacetime}, Commun. Math. Phys.
  \textbf{223} (2001), 289.

\bibitem[HW03]{HW03}
\bysame, \emph{On the renormalization group in curved spacetime}, Commun. Math.
  Phys. \textbf{237} (2003), 123.

\bibitem[HW10]{HW10}
\bysame, \emph{Axiomatic quantum field theory in curved spacetime}, Commun.
  Math. Phys. \textbf{293} (2010), 85.

\bibitem[Ish99]{Ish99}
Chris~J. Isham, \emph{Modern differential geometry for physicists}, World
  Scientific, 1999.

\bibitem[Jos11]{Jos11}
J{\"u}rgen Jost, \emph{Riemannian geometry and geometric analysis}, Springer
  Berlin Heidelberg, 2011.

\bibitem[Kha14a]{Kha14a}
Igor Khavkine, \emph{Cohomology with causally restricted supports}, {\tt
  arXiv}{\rm :1404.1932 [math-ph]} (2014).

\bibitem[Kha14b]{Kha14b}
\bysame, \emph{Covariant phase space, constraints, gauge and the Peierls
  formula}, Int. J. Mod. Phys. A \textbf{29} (2014), 1430009.

\bibitem[KMS93]{KMS93}
Ivan Kol{\'a}{\v r}, Peter~W. Michor, and Jan Slov{\'a}k, \emph{Natural
  Operations in Differential Geometry}, Springer Berlin Heidelberg, 1993.

\bibitem[KN96]{KN96}
Shoshichi Kobayashi and Katsumi Nomizu, \emph{Foundations of Differential
  Geometry - Volume 1}, Wiley-Interscience, 1996.

\bibitem[Mas91]{Mas91}
William~S. Massey, \emph{A basic course in algebraic topology}, Springer Berlin
  Heidelberg, 1991.

\bibitem[MS00]{MS00}
Luigi Mangiarotti and Gennadi~A. Sardanashvily, \emph{Connections in classical
  and quantum field theory}, World Scientific, 2000.

\bibitem[MSTV73]{MSTV73}
Jerome Manuceau, Michel Sirugue, Daniel Testard, and Andr{\'e} Verbeure,
  \emph{The smallest $C^\ast$-algebra for canonical commutations relations},
  Commun. Math. Phys. \textbf{32} (1973), 231.

\bibitem[M{\"u}l07]{Mul07}
Olaf M{\"u}ller, \emph{The Cauchy problem of Lorentzian minimal surfaces in
  globally hyperbolic manifolds}, Ann. Glob. Anal. Geom. \textbf{32} (2007),
  67.

\bibitem[Mun96]{Mun96}
James~R. Munkres, \emph{Elements of algebraic topology}, Westview Press, 1996.

\bibitem[O'N83]{ONe83}
Barrett O'Neill, \emph{Semi-Riemannian geometry with applications to
  relativity}, Academic Press, 1983.

\bibitem[Pei52]{Pei52}
Rudolph~E. Peierls, \emph{The commutation laws of relativistic field theory},
  Proc. R. Soc. Lond. A \textbf{214} (1952), 143.

\bibitem[Pfe09]{Pfe09}
Michael~J. Pfenning, \emph{Quantization of the Maxwell field in curved
  spacetimes of arbitrary dimension}, Class. Quant. Grav. \textbf{26} (2009),
  135017.

\bibitem[San10]{San10}
Ko~Sanders, \emph{The locally covariant Dirac field}, Rev. Math. Phys.
  \textbf{22} (2010), 381.

\bibitem[San13]{San13}
\bysame, \emph{A note on spacelike and timelike compactness}, Class. Quant.
  Grav. \textbf{30} (2013), 115014.

\bibitem[SDH14]{SDH14}
Ko~Sanders, Claudio Dappiaggi, and Thomas-Paul Hack, \emph{Electromagnetism,
  local covariance, the Aharonov-Bohm effect and Gauss' law}, Commun. Math.
  Phys. \textbf{328} (2014), 625.

\bibitem[Ver01]{Ver01}
Rainer Verch, \emph{A spin-statistics theorem for quantum fields on curved
  spacetime manifolds in a generally covariant framework}, Commun. Math. Phys.
  \textbf{223} (2001), 261.

\bibitem[Wal12]{Wal12}
Stefan Waldmann, \emph{Geometric wave equations}, {\tt arXiv}{\rm :1208.4706
  [math-DG]} (2012).

\bibitem[WZ14]{WZ14}
Micha{\l} Wrochna and Jochen Zahn, \emph{Classical phase space and Hadamard
  states in the BRST formalism for gauge field theories on curved spacetime},
  {\tt arXiv}{\rm :1407.8079 [math-ph]} (2014).

\end{thebibliography}
\bibliographystyle{amsalphamod}
\paginavuota


\cleardoublepage\phantomsection
\addcontentsline{toc}{chapter}{List of publications}
\chapter*{List of publications}

\begin{enumerate}[label={[\arabic*]}]

\item M. Benini, \emph{Optimal space of linear classical observables
for Maxwell $k$-forms via spacelike and timelike compact de Rham cohomologies,} Jan 2014,
\href{http://arxiv.org/pdf/1401.7563.pdf}{\texttt{arXiv}:1401.7563 [math-ph]}.

\item M. Benini, \emph{Relative Cauchy evolution for the vector potential
on globally hyperbolic spacetimes,} Feb 2014, accepted for publication on \textbf{MEMOCS}.

\item M. Benini, C. Dappiaggi, \emph{Models of free quantum field theories
on curved backgrounds,} in \textbf{Advances in Algebraic Quantum Field
Theory}, eds. R. Brunetti, C. Dappiaggi, K. Fredenhagen, J. Yngvason,
to be published by Springer in 2015.

\item M. Benini, C. Dappiaggi, T.-P. Hack, \emph{Quantum field theory
on curved backgrounds -- A primer,} \textbf{Int. J. Mod. Phys. A} 17:28 (2013) 1330023,\\ 
\href{http://dx.doi.org/10.1142/S0217751X13300238}{DOI: 10.1142/S0217751X13300238},
\href{http://arxiv.org/pdf/1306.0527.pdf}{\texttt{arXiv}:1306.0527 [gr-qc]}.

\item M. Benini, C. Dappiaggi, T.-P. Hack, A. Schenkel, \emph{A $C^{*}$-algebra
for quantized Abelian principal $U(1)$-connections on globally hyperbolic
Lorentzian manifolds,} \textbf{Commun. Math. Phys.} 332:1 (2014) 477,\\ 
\href{http://dx.doi.org/10.1007/s00220-014-2100-3}{DOI: 10.1007/s00220-014-2100-3},
\href{http://arxiv.org/pdf/1307.3052.pdf}{\texttt{arXiv}:1307.3052 [math-ph]}.

\item M. Benini, C. Dappiaggi, S. Murro, \emph{Radiative observables
for linearized gravity on asymptotically flat spacetimes and their
boundary induced states,} \textbf{J. Math. Phys.} 55:8 (2014) 082301,\\
\href{http://dx.doi.org/10.1063/1.4890581}{DOI: 10.1063/1.4890581},
\href{http://arxiv.org/pdf/1404.4551.pdf}{\texttt{arXiv}: 1404.4551 [gr-qc]}.

\item M. Benini, C. Dappiaggi, A. Schenkel, \emph{Quantized Abelian
principal connections on Lorentzian manifolds,} \textbf{Commun. Math. Phys.} 330:1 (2014) 123,\\
\href{http://dx.doi.org/10.1007/s00220-014-1917-0}{DOI: 10.1007/s00220-014-1917-0},
\href{http://arxiv.org/pdf/1303.2515.pdf}{\texttt{arXiv}:1303.2515 [math-ph]}.

\item M. Benini, C. Dappiaggi, A. Schenkel, \emph{Quantum field theories
on affine bundles,} \textbf{Ann. Henry Poincar\'e} 15:1 (2014) 171,\\
\href{http://dx.doi.org/10.1007/s00023-013-0234-z}{DOI: 10.1007/s00023-013-0234-z},
\href{http://arxiv.org/pdf/1210.3457.pdf}{\texttt{arXiv}:1210.3457 [math-ph]}.

\end{enumerate}


\clearemptydoublepage
\thispagestyle{empty}\phantom{x}

\clearpage
\thispagestyle{empty}
\phantom{x}
\vfill
\hfill\EANisbn[SC5b]

\end{document}